%% file: main.tex
\pdfoutput=1

\documentclass[a4paper,12pt,authoryear,print,index]{Classes/PhDThesisPSnPDF}

\input{Preamble/preamble}

\input{thesis-info}

\input{Preamble/math_commands}

\input{Preamble/macros}

\ifdefineAbstract
\fi

\ifdefineChapter
\fi

\begin{document}

\frontmatter

\begin{center}
\maketitle
\end{center}

\include{Declaration/declaration}
\include{Abstract/abstract}
\include{Dedication/dedication}
\include{Acknowledgement/acknowledgement}

% *********************** Adding TOC and List of Figures ***********************

\tableofcontents

\listoffigures

\listoftables

\listofchanges

\printnomenclature

% ******************************** Main Matter *********************************
\mainmatter

\include{1-Introduction/introduction}
\include{2-SourceCoding/source_coding}
\include{3-FundamentalLimits/fundamental_limits}
\include{4-RelativeEntropyCodingWithPoissonProcesses/relative_entropy_coding_with_poisson_processes}
\include{5-BranchAndBound/branch_and_bound}
\include{6-COMBINER/combiner}
\include{7-Discussion/discussion}

% ********************************** Back Matter *******************************
% Backmatter should be commented out, if you are using appendices after References
%\backmatter

% ********************************** Bibliography ******************************
\begin{spacing}{0.9}

% To use the conventional natbib style referencing
% Bibliography style previews: http://nodonn.tipido.net/bibstyle.php
% Reference styles: http://sites.stat.psu.edu/~surajit/present/bib.htm

\bibliographystyle{apalike}
\cleardoublepage
\bibliography{References/references} % Path to your References.bib file

% If you would like to use BibLaTeX for your references, pass `custombib' as
% an option in the document class. The location of 'reference.bib' should be
% specified in the preamble.tex file in the custombib section.
% Comment out the lines related to natbib above and uncomment the following line.

%\printbibliography[heading=bibintoc, title={References}]

\end{spacing}

% ********************************** Appendices ********************************

\begin{appendix} % Using appendices environment for more functunality

% required so that cleveref references the appendix sections correctly

\include{Appendix1/appendix1}

\end{appendix}

% *************************************** Index ********************************
\printthesisindex % If index is present

\end{document}

%% file: thesis-info.tex
\title{Data Compression with \texorpdfstring{\\}{} Relative Entropy Coding}
\author{Gergely Flamich}

\dept{Department of Engineering}

\university{University of Cambridge}
\crest{\includegraphics[width=0.2\textwidth]{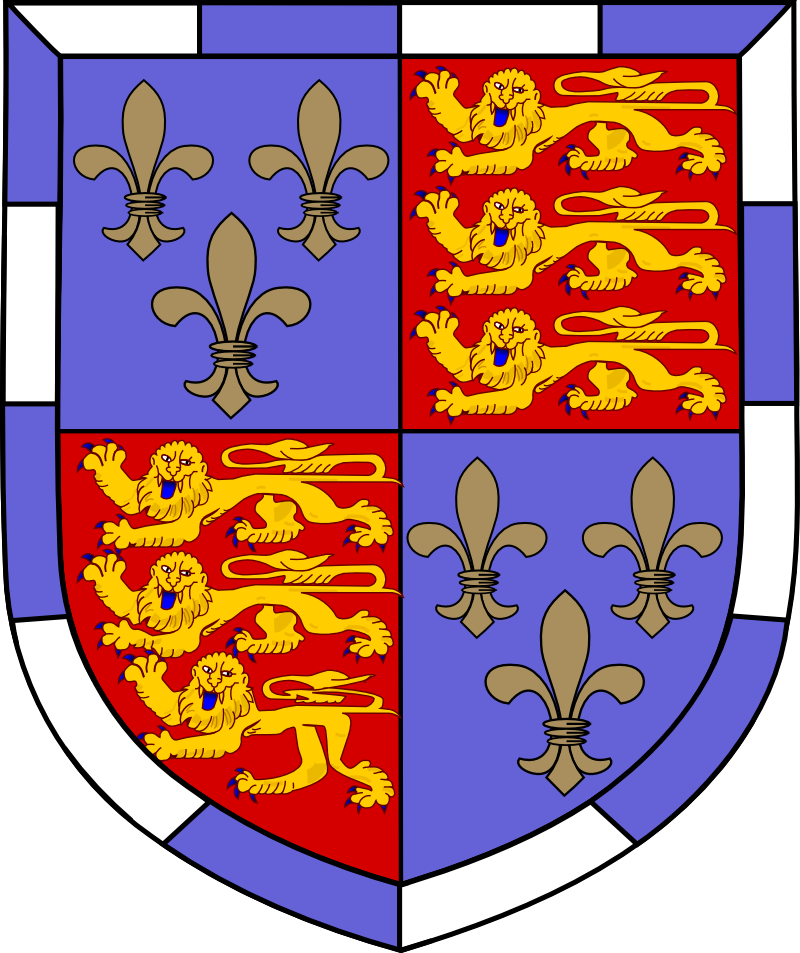}}
\renewcommand{\submissiontext}{This thesis is submitted for the degree of}

\degreetitle{Doctor of Philosophy}

\college{St John's College}

\degreedate{December 2024} 

\subject{LaTeX} \keywords{{LaTeX} {PhD Thesis} {Engineering} {University of
Cambridge}}

%% file: Preamble/math_commands.tex
\usepackage{amsmath,amsfonts,bm}

\def\eqref#1{equation~\ref{#1}}
\def\1{\bm{1}}

\def\rh{{\textnormal{h}}}

\def\rvb{{\mathbf{b}}}

\def\rvh{{\mathbf{h}}}

\def\rvm{{\mathbf{m}}}

\def\rvs{{\mathbf{s}}}

\def\rvw{{\mathbf{w}}}
\def\rvx{{\mathbf{x}}}
\def\rvy{{\mathbf{y}}}
\def\rvz{{\mathbf{z}}}

\def\ervw{{\textnormal{w}}}
\def\ervx{{\textnormal{x}}}

\def\rmC{{\mathbf{C}}}

\def\rmN{{\mathbf{N}}}

\def\vg{{\bm{g}}}

\def\vu{{\bm{u}}}

\def\mA{{\bm{A}}}

\def\mH{{\bm{H}}}

\DeclareMathAlphabet{\mathsfit}{\encodingdefault}{\sfdefault}{m}{sl}
\SetMathAlphabet{\mathsfit}{bold}{\encodingdefault}{\sfdefault}{bx}{n}

\newcommand{\KL}{D_{\mathrm{KL}}}
\newcommand{\Var}{\mathrm{Var}}

\newcommand{\Cov}{\mathrm{Cov}}
\DeclareMathOperator*{\argmax}{arg\,max}
\DeclareMathOperator*{\argmin}{arg\,min}

%% file: Preamble/macros.tex
\usepackage{amsthm}
\usepackage{mathtools}
\usepackage{pifont}
\usepackage[OT2,T1]{fontenc}

\usepackage{thmtools}
\usepackage{thm-restate}

\newtheorem{theorem}{Theorem}[section]
\newtheorem{proposition}{Proposition}[section]
\newtheorem{definition}{Definition}[section]
\newtheorem{lemma}[theorem]{Lemma}
\newtheorem{corollary}{Corollary}[section]

\colorlet{shadecolor}{lightblue!15}

\newenvironment{importantTheorem}
  {\begin{shaded}\begin{theorem}}
  {\end{theorem}\end{shaded}}

\newenvironment{importantCorollary}
  {\begin{shaded}\begin{corollary}}
  {\end{corollary}\end{shaded}}

\DeclarePairedDelimiterX{\infdivx}[2]{[}{]}{%
  #1\;\delimsize\|\;#2%
}
\DeclarePairedDelimiterX{\inftvd}[2]{(}{)}{%
  #1\;\delimsize\|\;#2%
}
\DeclarePairedDelimiterX{\infdivcolon}[2]{[}{]}{%
  #1\;\delimsize;\;#2%
}
\DeclarePairedDelimiterX{\infdivent}[1]{[}{]}{%
  #1
}
\newcommand{\TVD}{\delta_{\mathrm{TV}}\inftvd}
\newcommand{\Ent}{\mathbb{H}\infdivent}
\newcommand{\CumResEnt}{\mathcal{E}\infdivent}
\newcommand{\DiffEnt}{\mathbb{h}\infdivent}
\newcommand{\MI}{I\infdivcolon}
\newcommand{\KLD}{\KL\infdivx}
\newcommand{\CSD}{D_{CS}\infdivx}

\newcommand{\infD}{D_\infty\infdivx}

\DeclarePairedDelimiter{\norm}{\lVert}{\rVert}
\DeclarePairedDelimiter{\round}{\lfloor}{\rceil}
\DeclarePairedDelimiter{\abs}{\lvert}{\rvert}
\DeclarePairedDelimiterX{\innerProd}[2]{\langle}{\rangle}{%
    #1,#2%
}

\def\Normal{\mathcal{N}}
\def\Oh{\mathcal{O}}
\newcommand{\Pois}{\mathrm{Pois}}
\newcommand{\Exponential}{\mathrm{Exp}}
\newcommand{\Bern}{\mathrm{Bern}}

\renewcommand{\Var}{\mathbb{V}}
\newcommand{\Reals}{\mathbb{R}}
\newcommand{\nonnegReals}{\Reals_{\geq 0}}

\newcommand{\Nats}{\mathbb{N}}
\newcommand{\Ints}{\mathbb{Z}}
\newcommand{\posNats}{\mathbb{Z}_{\geq 1}}

\newcommand{\Prob}{\mathbb{P}}

\newcommand{\Ind}{\mathbf{1}}
\newcommand{\Unif}{\mathrm{Unif}}
\newcommand{\Exp}{\mathbb{E}}

\newcommand{\Tree}{\mathcal{T}}

\newcommand{\Laplace}{\mathcal{L}}

\DeclareMathOperator*{\plim}{\mathrm{plim}}
\DeclareMathOperator*{\esssup}{ess\,sup}
\DeclareMathOperator{\pushfwd}{\sharp}

\newcommand{\diag}{\mathrm{diag}}

\newcommand{\enc}{\mathtt{enc}}
\newcommand{\dec}{\mathtt{dec}}

\newcommand{\disteq}{\stackrel{d}{=}}

\newcommand{\XSpace}{\mathcal{X}}
\newcommand{\YSpace}{\mathcal{Y}}
\newcommand{\ZSpace}{\mathcal{Z}}

\newcommand{\Geom}{\mathrm{Geom}}
\newcommand{\DistFamily}{\mathcal{P}}

\newcommand{\Loss}{\mathcal{L}}
\newcommand{\Data}{\mathcal{D}}

\newcommand{\quantiser}{\mathcal{Q}}

\newcommand{\sigmaAlgebra}{\mathcal{F}}

\newcommand{\PMSpace}{\mathfrak{P}}
\newcommand{\Palm}{\mathcal{P}}

\newcommand{\supp}{\mathrm{supp}}

\newcommand{\PoissonProcess}{\Pi}

\newcommand{\graph}{\mathrm{graph}}
\newcommand{\hypograph}{\mathrm{hyp}}

\newcommand{\tikzcircle}[2][red,fill=red]{\tikz[baseline=-0.5ex]\draw[#1,radius=#2] (0,0) circle ;}%
\newcommand{\xmark}{\ding{53}}%

\DeclareSymbolFont{cyrletters}{OT2}{wncyr}{m}{n}
\DeclareMathSymbol{\Sha}{\mathalpha}{cyrletters}{"58}
\DeclareMathSymbol{\sha}{\mathalpha}{cyrletters}{"78}

\newcommand{\splitFun}{\mathtt{split}}

\newcommand{\Frontier}{\mathfrak{F}}
\newcommand{\Dataset}{\mathcal{D}}
\newcommand{\SubTree}{\mathcal{S}}
\newcommand{\Branch}{\mathcal{B}}
\newcommand{\filtration}{\bm{\mathcal{F}}}

\DeclareMathOperator{\lb}{\mathrm{lb}}
\DeclareMathOperator{\exptwo}{\exp_2}

\def\rvmu{{\bm{\mu}}}
\def\rvsigma{{\bm{\sigma}}}
\def\rvnu{{\bm{\nu}}}
\def\rvrho{{\bm{\rho}}}

\def\rvphi{{\bm{\phi}}}
\def\rvalpha{{\bm{\alpha}}}

\newcommand{\Class}{\mathcal{C}}

\newcommand{\rvhbar}{\overline{\rvh}}
\newcommand{\rvhbbar}{\overline{\overline{\rvh}}}

\newcommand{\rvmubar}{\mathbf{\overline{\rvmu}}}
\newcommand{\rvnubar}{\mathbf{\overline{\rvnu}}}
\newcommand{\rvsigmabar}{\mathbf{\overline{\rvsigma}}}
\newcommand{\rvrhobar}{\mathbf{\overline{\rvrho}}}

\newlist{legal}{enumerate}{10}
\setlist[legal]{label*=\arabic*.}

\newcommand{\thinOp}{\mathrm{thin}}

\newcommand{\expb}{{\exp_2}}
\newcommand{\sort}{{\mathtt{sort}}}

%% file: Declaration/declaration.tex
% ******************************* Thesis Declaration ***************************

\begin{declaration}
This thesis is the result of my own work and includes nothing which is the outcome of work done in collaboration except as declared in the preface and specified in the text. It is not substantially the same as any work that has already been submitted, or, is being concurrently submitted, for any degree, diploma or other qualification at the University of Cambridge or any other University or similar institution except as declared in the preface and specified in the text. It does not exceed the prescribed word limit for the relevant Degree Committee.
% I hereby declare that except where specific reference is made to the work of 
% others, the contents of this dissertation are original and have not been 
% submitted in whole or in part for consideration for any other degree or 
% qualification in this, or any other university. This dissertation is my own 
% work and contains nothing which is the outcome of work done in collaboration 
% with others, except as specified in the text and Acknowledgements. This 
% dissertation contains fewer than 65,000 words including appendices, 
% bibliography, footnotes, tables and equations and has fewer than 150 figures.

% Author and date will be inserted automatically from thesis.tex \author \degreedate

\end{declaration}

%% file: Abstract/abstract.tex
% ************************** Thesis Abstract *****************************
% Use `abstract' as an option in the document class to print only the titlepage and the abstract.
\begin{abstract}
Data compression forms the foundation of our Digital Age, enabling us to store, process and communicate more information than ever before.
\replaced[]{%
Over the last few years, machine learning-based compression algorithms have revolutionised the field and surpassed the performance of traditional methods.
Moreover, machine learning unlocked previously infeasible features for compression, such as providing guarantees for users’ privacy or tailoring compression to specific data statistics (e.g., satellite images or audio recordings of animals) or users’ audiovisual perception.
}
{%
Over the last few years, machine learning revolutionised compression, surpassing traditional methods and unlocking features previously infeasible, such as providing guarantees for users’ privacy or tailoring compression to specific data statistics (e.g., satellite images or audio recordings of animals) or users’ audiovisual perception.
}%
This, in turn, has led to an explosion of theoretical investigations and insights that aim to develop new fundamental theories, methods and algorithms better suited for machine learning-based compressors.
\par
In this thesis, I contribute to this trend by investigating relative entropy coding, a mathematical framework that generalises classical source coding theory. 
Concretely, relative entropy coding deals with the efficient communication of uncertain or randomised information. 
\replaced[]{%
One of its key advantages is that it extends compression methods to continuous spaces and can thus be integrated more seamlessly into modern machine learning pipelines than classical quantisation-based approaches.
}{%
Significantly, it extends compression methods to continuous spaces, which means it can be integrated into modern machine learning pipelines a lot more seamlessly compared to classical quantisation-based approaches.
}
Furthermore, it is a natural foundation for developing advanced compression methods that are privacy-preserving or account for the perceptual quality of the reconstructed data. 
\par
The thesis considers relative entropy coding at three conceptual levels: 
After introducing the basics of the framework, 
(1) I prove results that provide new, maximally tight fundamental limits to the communication and computational efficiency of relative entropy coding; 
(2) I use the theory of Poisson point processes to develop and analyse new relative entropy coding algorithms, whose performance attains the theoretic optima and 
(3) I showcase the strong practical performance of relative entropy coding by applying it to image, audio, video and protein data compression using small, energy-efficient, probabilistic neural networks called Bayesian implicit neural representations. 
% I demonstrate the versatility and practical potential of the method by conducting experiments across several data modalities such as image, audio and video compression and show that it %achieves strong performance among competing methods. 
% As such, it could be an energy-efficient alternative to the large, complex machine learning-based compressors that dominate the state-of-the-art nowadays.
\par
Finally, while most of the work I present in the thesis was motivated by practical data compression, many of the techniques I present are more general and have implications and applications in broader information theory, as well as the development and analysis of general-purpose sampling algorithms.
\end{abstract}

%% file: Dedication/dedication.tex
% ******************************* Thesis Dedidcation ********************************

\begin{dedication} 

Szeret\H{o} testv{\'e}reimnek, sz{\"u}leimnek {\'e}s nagysz{\"u}leimnek.

\end{dedication}

%% file: Acknowledgement/acknowledgement.tex
% ************************** Thesis Acknowledgements **************************

\begin{acknowledgements}  
First, I want to express my deepest gratitude to my supervisor, Jos{\'e} Miguel Hern{\'a}ndez Lobato, for his many years of support and advice. 

I would also like to thank Lucas Theis for his mentorship during my internship at Google Brain. His insights into neural data compression have significantly shaped my understanding of the field.

I am also very thankful to Marton Havasi, Stratis Markou, Jiajun He, Szilvia Ujv{\'a}ry, Zongyu Guo, Lennie Wells and Daniel Goc; it was a pleasure to work with them.

Finally, I am grateful to my many loving friends for their support. I especially want to thank Anna Orteu, Margherita Battistara, Rita Dias, Jasmine Vieri, Tim Lameris, Peter Wildemann, Simon Brauberger and Lorenzo Bonito for being the solid rocks in my life during my time at Cambridge.\footnote{I would have included Damon Kutzin in this list as well.
However, he told me he dislikes lengthy, emotional acknowledgements, so I decided to oblige and keep it short.}
\end{acknowledgements}

%% file: 1-Introduction/introduction.tex
%!TEX root = ../thesis.tex
%*******************************************************************************
%*********************************** First Chapter *****************************
%*******************************************************************************

\chapter{Introduction}  %Title of the First Chapter

\ifpdf
    \graphicspath{{Chapter1/Figs/Raster/}{Chapter1/Figs/PDF/}{Chapter1/Figs/}}
\else
    \graphicspath{{Chapter1/Figs/Vector/}{Chapter1/Figs/}}
\fi

%********************************** %First Section  **************************************
\noindent
We truly live in the age of information: the International Telecommunication Union, the United Nations' specialised agency for digital technology, estimated that just over two-thirds of Earth's population, about 5.4 billion people, are now connected to the Internet, and in 2022 alone produced a staggering 5.3 zettabytes (equivalent to 5.3 trillion gigabytes) of internet traffic \citepalias{itu2023facts}.
But how is it possible to communicate such an unfathomable volume of data? 
\par
Excluding enormous empirical engineering efforts enabling efficient exchange, we find that its foundations are rooted in one of the great mathematical achievements of the last one hundred years: information theory. 
The central insight of information theory is connecting the concept of ``information'' to ``predictability''.
\replaced[]{%
Perhaps the clearest example of this connection
}{%
Unarguably, one of the central results of the field
}
is the \textit{source coding theorem}, which establishes that data can only be represented more compactly as long as it contains some predictable structure; in other words, pure randomness is incompressible.
However, the truly miraculous fact on which our global communication ecosystem rests is that the fundamental limit established by the theorem is efficiently achievable. 
There are \textit{fast} compression algorithms that can compute compact representations of data within a few bits of the theoretical limit.
\par
However, the assumptions needed for theoretical guarantees seldom hold, for which we have to compensate with a significant amount of engineering.
Thus, historically, developing new compression algorithms took decades of work by several researchers and engineers.
This began to change about a decade ago, as the burgeoning field of machine learning (ML) started penetrating data compression research.
This advance was inevitable, as learning and compression have long been understood to be joined at the hip, at least theoretically; there was even an early attempt by \citet{schmidhuber1996sequential} to use neural networks to compress text. 
However, for the first time in 2017, Johannes Ball{\'e} and his co-authors succeeded in compressing high-resolution images with performance exceeding that of JPEG 2000 \citep{balle2017end}.
In the years since then, ML has revolutionised the field: by now, all state-of-the-art algorithms for image, audio and video compression all use ML models.
Furthermore, ML models' flexibility has enabled the inclusion of many new aspects into the design process, such as concerns for privacy, perceptual quality, and taking downstream tasks into account.
\par
Nonetheless, applying ML to compression does not come without its challenges. 
First, ML models are nothing without the data on which they are trained.
As such, it must be sourced legally and ethically for scientific studies, especially for commercial uses.
This matter is especially acute in the case of data compression, as these systems are designed to operate on image, audio and video data. 
Hence, it is paramount to obtain the data with the subjects' consent or the creators' permission and to ensure that the ML system's output excludes unintended biases to the best of the practitioner's ability.
To my knowledge, the data I used in the experiments I report in this thesis conform to the above.
% ; see \Cref{sec:dataset_description} for a more elaborate discussion.
%
\par
Second, the energy consumption of ML data compression models directly impacts their practicality, as running large models on mobile or low-power devices is currently not possible and is one of the prime inhibitors to wide-scale deployment.
Hence, as a step towards solving this problem, in \Cref{chapter:combiner}, I present work that sets the basis for deployable, low-energy ML models.
\par
Third, data compression is at odds with the usual ML mantra of ``end-to-end training'', which calls for learning every parameter of every component in the system jointly, at the same time.
The issue is that standard approaches to data compression quantise their input, which is a non-differentiable operation.
This means that gradient descent, the workhorse of modern optimisation algorithms, cannot be applied directly, and one needs to introduce approximations to side-step the issue of quantisation.
While the approximations have been demonstrated to work well in practice, I take a different approach in this thesis.
I will argue for replacing quantisation altogether with \textit{relative entropy coding}, a different approach that avoids the non-differentiability issue entirely.
As I will demonstrate, studying and using relative entropy coding has utility far beyond eliminating the abovementioned discrepancy. 
In \Cref{chapter:combiner}, I show that it can lead to better performance in practice than quantisation-based approaches, while \Cref{chapter:fundamental_limits,chapter:rec_with_pp} show that it has profound theoretical implications for data compression and sampling theory. 
\section{Relative Entropy Coding}
\par
The best way to understand relative entropy coding is through the lens of lossy source coding, also known as lossy data compression.
The motivation for lossy source coding is that we can represent data far more compactly if we do not require that we be able to recover from it an exact replica of the original data.
In other words, we allow the compressor to discard some of the information in the input; as long as we retain semantic information (information pertaining to the \textit{meaning} of the data), we will have made a gain.
Now, we measure the performance of a lossy scheme using two quantities: its \textit{rate}, i.e., how many bits we need on average to encode the data, and its \textit{distortion}, i.e., how different the data we recovered from the compressed representation is from the original.
\par
But how can we design a system that discards information?
To get an idea, imagine that our task is to compress some real-valued input $x$ between $0$ and some integer $n$. 
To this end, upon observing $x$, we round it to the closest integer and use this integer as the compressed representation.
We cannot recover the $x$'s fractional part, which is a consequence of the mechanism we used to discard information: we forced many possible inputs to share the same representation.
To control the rate of our compressor, we can tweak how many inputs share the same representation.
For example, rounding the input to the closest even number halves the number of possible representations, thus halving the rate at the cost of increasing the distortion.
% For example, if instead we round our inputs to the closest even number, we halve the number of possible representations for our inputs, thus halving the rate at the cost of increasing the distortion.
%
\par
As the example illustrates, quantisation is easy to implement and gives us a natural way to control the rate-distortion tradeoff. 
However, its disadvantage is that it does not lend itself to easy analysis, which stems from its deterministic nature.
This means we can only do worst-case analysis, which is usually difficult (if not outright impossible). 
It also does not reflect that we seldom expect to encounter worst-case inputs in practical high-dimensional settings.
A more subtle problem arises when considering the low-rate regime: what would be the compressor's preferred behaviour as we decrease the rate to zero?
Zero-rate lossy compression discards all input information, which a deterministic scheme can only achieve by forcing every possible input to share the same representation.
Imagine now what this would imply for an image compression algorithm: regardless of the picture we input, we would always get the same reconstruction, which could be, e.g., a completely black image.
While this behaviour makes sense, is it what we want?
\par
The most natural way to deal with these issues is to discard information in the input by ``adding noise'' instead of rounding; this is the essence of relative entropy coding.
Such randomisation allows us to perform average-case instead of worst-case analysis, which can be significantly more straightforward.
Furthermore, it arguably has a much more natural low-rate compression behaviour: in the zero-rate case, we would have a fully random representation rather than a fully deterministic one.
This means that we can get a diverse set of realistic-looking images even at zero rate.
While the most basic relative entropy coding algorithms are almost as simple as rounding, they come at a price: most are far too slow and would require thousands of years to encode even thumbnail-sized images.
Unfortunately, they cannot be improved either, as I show in \Cref{chapter:fundamental_limits}.
However, all is not doomed: my result only demonstrates that we cannot be completely careless about how we randomise the input.
In fact, in \Cref{chapter:branch_and_bound}, I develop several algorithms that achieve optimally fast runtime for certain types of randomisation.
\subsection{A Note About the Name}
\par
As with many simple, good ideas, they have been reinvented and named several times across several fields, and relative entropy coding is no different.
Within the information theory and machine learning literature, the same basic idea is known as channel simulation \citep{li2018strong}, channel synthesis \citep{cuff2013distributed}, reverse channel coding \citep{theis2022algorithms} and relative entropy coding \citep{flamich2020compressing}\added[id={PL}]{; see \Cref{sec:rec_history} for a synopsis of its history}.
In this thesis, I will solely refer to the idea as relative entropy coding for two main reasons.
First, I came up with it, and I like it.
Second, and perhaps more importantly, relative entropy coding ties the idea to data compression more closely than the other names by evoking similarity to entropy coding.
While the other names are perfectly good, they reflect the concept's information-theoretic roots by casting it as a dual problem to channel coding.
This, to an extent, has also caused the information theory community to largely ignore concerns with the practical implementation of their relative entropy coding algorithms, as they mostly used them for theoretic analysis, not as functional components in practical compression pipelines.
Hence, by calling it relative entropy coding, I hope to bring about a psychological shift in the community and highlight the theory's vast practical potential. 
\section[Relative Entropy Coding with Poisson Processes]{Relative Entropy Coding with \texorpdfstring{\\}{} Poisson Processes}
\par
As I noted in the previous section, the essence of relative entropy coding is that we encode a perturbed version of the input.
A particularly popular choice is additive perturbation, meaning that for input $\rvx$ we encode $\rvy = \rvx + \epsilon$, where $\epsilon$ is some random noise independent of $\rvx$.
But how can we encode $\rvy$?
\par
First, let us make the standard assumption (which I will use for the rest of the thesis) that the encoder and decoder share a statistical model over $\rvy$ called the \textit{coding distribution} $P$ that is independent of $\rvx$.
When $\rvy$ is discrete, we could simulate a sample $\rvy$ and use an entropy coding algorithm, such as Huffman coding or arithmetic coding \citep{witten1987arithmetic} to encode the sample.
Assuming the coding distribution $P$ is $\rvy$'s true marginal distribution, the average number of bits this procedure uses to encode $\rvy$ is (within a few bits of) the Shannon entropy of $\rvy$, which I denote with $\Ent{\rvy}$.
However, it ignores a crucial part of the problem: in relative entropy coding, we assume that we know the noise model $P_{\rvy \mid \rvx}$, such as the additive example above; in practice, we usually design it ourselves.
Indeed, we can leverage this extra information to encode $\rvy$ using many fewer bits, given by the mutual information $\MI{\rvx}{\rvy}$.
Note that $\MI{\rvx}{\rvy}$ can be finite even in cases when $\rvy$ is not discrete and thus when $\Ent{\rvy}$ is infinite.
From this perspective, relative entropy coding is always at least as efficient as entropy coding.
However, we need a second standard assumption for it to work: the encoder and decoder share a source of \textbf{common randomness} $\rvz$.
In practice, the reader can think of $\rvz$ as the encoder and the decoder sharing their pseudo-random number generator (PRNG) seed, meaning they can generate the same sequence of random variates.
\par
How do we create a scheme with a short codelength? The idea is to do ``selection sampling'': we first set the shared randomness to be
\begin{align*}
\rvz = (\rvy_1, \rvy_2, \hdots ), \quad \rvy_i \stackrel{iid}{\sim} P_Y.
\end{align*}
Given $\rvx \sim P_\rvx$ the encoder selects the $K$-th of these samples according to some selection rule, for which $\rvy_K \sim P_{\rvy \mid \rvx}$.
Then, the sender encodes and transmits the selected index $K$ to the receiver.
Since the receiver has access to the same shared randomness as the encoder by assumption, they can recover the selected sample $\rvy_K$ just from this information.
Perhaps the simplest possible selection rule is rejection sampling \citep{neumann1951various}.
However, it turns out that the index $K$ that rejection sampling would select cannot always be encoded efficiently.
\par
Not to worry, in \cref{chapter:rec_with_pp}, I describe how we can use the theory of Poisson processes to design general-purpose relative entropy coding algorithms with optimal average description length.
Unfortunately, the generality of these algorithms comes at a price: they are impractically slow. 
Thus, in \cref{chapter:branch_and_bound}, I show how we can leverage the special structure of some relative entropy coding problems to develop compression algorithms that are also optimally fast, besides having optimal average description length.
The work I describe in these chapters contributes to an exciting line of recent work that casts the problem of simulating random variates that follow a prescribed distribution as a search problem over random structures.
\section{Compression with Bayesian Implicit Neural Representations}
\par
A characteristic of most ML solutions to data compression is that they mimic classical approaches. 
Where classical methods, such as JPEG or MP3, use the discrete cosine transform, ML approaches ``simply'' replace it with a neural network. 
Unfortunately, these models consist of millions of parameters and require several orders of magnitude more computational power and energy to be trained and run.
\par
However, ML also offers an alternative framework for compressing data that is essentially agnostic to the nature of the data being compressed and is also generally more energy and compute-efficient: \textit{implicit neural representations} (INR).
The motivation for INRs comes from the observation that we can conceive most common types of data as continuous functions that map coordinates to the values of some signal we sampled at a certain resolution.
For example, we can think of images as functions mapping pixel locations to RGB colour intensity values.
The assumption that the signal is continuous suggests we can find an artificial neural network that approximates it arbitrarily well.
Thus, we can train a small neural network using gradient descent to memorise the data coordinate-by-coordinate.
After training, we can (imperfectly) reconstruct the data by evaluating the network at every coordinate location. 
Therefore, the weights of the fitted network constitute a lossy representation of the original data, and hence, encoding the weights constitutes a lossy data compression algorithm!
\par
Unfortunately, using neural networks with deterministic weights only allows us to control the distortion but not the compression rate.
\Cref{chapter:combiner} presents a solution to this problem: I fit neural networks with stochastic weights to the signal.
Then, I use a relative entropy coding algorithm I developed in \Cref{chapter:rec_with_pp} to encode a random weight setting as the representation of the data.
As I will show, this approach, which I call \underline{com}pression with \underline{B}ayesian \underline{i}mplicit \underline{ne}ural \underline{r}epresentations (COMBINER), allows us to control its compression rate as well as its distortion.
Furthermore, I provide experimental evidence that COMBINER performs well on image, video, audio and protein data compression tasks and outperforms significantly more complex competing methods.
\par
The appeal of INR-based compression is that the network architectures we usually need to achieve good performance are relatively simple: they are fully-connected and have only tens of thousands of parameters.
As such, reconstructing the data they encode requires orders of magnitude less compute and energy than the state-of-the-art approaches, which are based on autoencoder or diffusion models with many millions of parameters.
Therefore, I hope COMBINER can serve as a blueprint for implementing the first sustainable, practically implementable ML-based data compression codec. 
\section{Thesis Outline and Contributions}
\par
The thesis follows the flow of the three sections above.
Most of the content is based on previously published ideas and results but contains many new results as well.
Concretely:
\begin{itemize}
\item \textbf{\Cref{chapter:source_coding}} sets the scene for the thesis by exposing lossless and lossy source coding.
The chapter continues with an informal definition of relative entropy coding and explores its applications.
The chapter ends with the formal definition of relative entropy coding, a discussion of some of its properties, its relation to bits-back coding, and a synopsis of its history.
\item \textbf{\Cref{chapter:fundamental_limits}} investigates the fundamental limits of the communication and computational complexity of relative entropy coding.
Starting with communication complexity, I first define a new statistical distance called the channel simulation divergence, which I use to prove the first main result of the thesis, \Cref{thm:csd_rec_lower_bound}.
The theorem characterises the communication complexity of relative entropy coding maximally tightly and significantly extends a result of \citet{li2018strong}.
Then, I use this result to prove \Cref{thm:second_order_behaviour_of_csd}, which establishes a conjecture of \citet{sriramu2024optimal} and characterises the optimal description length of relative entropy coding algorithms for non-singular channels in the asymptotic setting.
\par
Turning to computational complexity, I formally define selection samplers and prove \Cref{thm:selection_sampler_runtime_lower_bound}.
This result shows that under a reasonable model of computation, any general-purpose selection sampler (hence any relative entropy coding algorithm) must scale super-exponentially in the information content of the problem, which demonstrates that practical relative entropy coding \deleted[id={PL}, comment={fixed ungrammatical sentence. PL/email/1}]{without} is hopeless unless we can leverage some extra structure in the problem.
\par
\underline{Contributions:} The sections on the channel simulation divergence and the computational complexity of relative entropy coding are based mainly on my paper with Daniel Goc \citep{goc2024channel}.
The results I include from the paper are my own unless I explicitly state otherwise.
The results on communication complexity are new and my own.
Finally, the results on the computational complexity of selection samplers are based on my paper with Lennie Wells \citep{flamich2024some}.
However, the specific results that appear in the thesis are my own.
\item \textbf{\Cref{chapter:rec_with_pp}}
begins with an overview of Poisson processes and four operations that preserve them: thinning, mapping, restriction and superposition.
Second, I show that each of the first three operations ``induces'' a selection sampler and, in turn, a general-purpose relative entropy coding algorithm: rejection sampling, A* sampling and greedy Poisson rejection sampling.
After defining them, I also analyse their computational and communication efficiency.
Third, I develop parallelised variants of these sampling algorithms using superposition.
Fourth, I consider how to construct approximate samplers from exact ones and show how we can use ideas from relative entropy coding to improve results for general-purpose sampling algorithms.  
\par
\underline{Contributions:} Until \Cref{sec:global_gprs}, the chapter largely follows the papers of \citet{maddison2016poisson} and \citet{li2018strong}.
\Cref{sec:global_gprs,sec:superposition_parallelisation} are based on my paper \citep{flamich2023gprs} while \Cref{sec:approximate_sampling} is based on my paper with Lennie Wells \citep{flamich2024some}.
However, the contributions I include in the thesis are my own.
\item \textbf{\Cref{chapter:branch_and_bound}} describes some of the more advanced computational theory of Poisson processes, which I then use to derive faster, so-called branch-and-bound variants of A* sampling and greedy Poisson rejection sampling for one-dimensional, quasiconcave distributions.
%\replaced{%
There are there are two pairs of main results in the chapter. 
First, \Cref{thm:bnb_a_star_codelength,thm:bnb_gprs_codelength} show that the average description length produced by relative entropy coding algorithms induced by the branch-and-bound variants of the samplers remains essentially the same.
Second, \Cref{thm:bnb_a_star_runtime,thm:bnb_gprs_runtime} show that using the branch-and-bound variants results in an exponential and super-exponential improvement in the runtimes compared to their general-purpose variants, respectively.
The chapter concludes with some numerical experiments demonstrating that the algorithms' empirical performance is in excellent alignment with the theory.
%}{%
% The main results of the chapter are \Cref{thm:bnb_a_star_codelength,thm:bnb_gprs_codelength} which show that the average description length produced by relative entropy coding algorithms induced by the branch-and-bound variants of the samplers remain the same; and \Cref{thm:bnb_a_star_runtime,thm:bnb_gprs_runtime}, which on the other hand show that using the branch-and-bound variants results in an exponential and super-exponential improvement in the runtimes compared to their general-purpose variants, respectively.
%}
\par
\underline{Contributions:}
Branch-and-bound A* sampling was proposed by \citet{maddison2014sampling} and connected to Poisson processes in \citet{maddison2016poisson}.
My contribution is recognising that it could be used for relative entropy coding.
Furthermore, its analysis in the one-dimensional quasiconcave case is due to Stratis Markou and I \citep{flamich2022fast}.
However, the proofs I provide in this thesis are new and my own and are significantly simpler and tighter than the ones presented in \citep{flamich2022fast}.
The branch-and-bound variant of greedy Poisson rejection sampling and its analysis is based on my paper \citet{flamich2023gprs}.
However, once again, the proofs I provide of the main results are simpler and tighter than the ones in the original paper.
\item \textbf{\Cref{chapter:combiner}} briefly introduces nonlinear transform coding and implicit neural representations. 
Then, I introduce the compression with Bayesian implicit neural representations (COMBINER) framework and discuss its practical implementation.
Finally, I present numerical experiments demonstrating that it outperforms competing state-of-the-art methods on various image, video, audio and protein data compression tasks.
\par
\underline{Contributions:} The chapter is entirely based on my two papers \citet{guo2023compression} and \citet{he2024recombiner}.
My contributions in these papers are: 1) recognising that we could use relative entropy coding to improve compression with implicit neural representations and thus proposing the original project; 2) developing the relevant theory for the project (except the EM update equations); and 3) helping with the implementation of the compression algorithms.
\item \textbf{\Cref{chapter:conclusion}} concludes the thesis and discusses open questions and promising future directions for research.
\end{itemize}
\subsection{Required Preliminaries}
\noindent
I attempted to write the thesis to be as self-contained as possible.
Hence, a reader familiar with the basics of probability theory and information theory and a bit of source coding (e.g.\ familiar with the first five chapters of \citet{cover1999elements}) should understand most of the arguments.
I also assume some basic familiarity with elementary measure-theoretic concepts such as the definition of a measure, absolute continuity, and Radon-Nikodym derivatives.
Formally, all integrals are defined in the Lebesgue-Stieltjes sense so that the proofs can be stated in their greatest generality; in most cases, infinitesimal parts such as $dF(x)$ can be substituted without trouble for the more familiar $F'(x)dx$.
\section{Publications}
The thesis draws on eight papers I published during my doctoral degree, though it does not follow them particularly closely except for \citet{guo2023compression} and \citet{he2024recombiner}.
In particular, it is based on:
\begin{itemize}
\item \bibentry{flamich2022fast}
\item \bibentry{flamich2023adaptive}
\item \bibentry{flamich2023gprs}
\item \bibentry{flamich2023grc}
\item \bibentry{guo2023compression}
\item \bibentry{he2024recombiner}
\item \bibentry{goc2024channel}
\item \bibentry{flamich2024some}
\end{itemize}
I also published four additional papers during my degree that I do not describe in this thesis:
\begin{itemize}
\item \bibentry{lin2023minimal}
\item \bibentry{ujvary2023estimating}
\item \bibentry{he2024accelerating}
\item \bibentry{he2024getting}
\end{itemize}
\section{Notation}
\label{sec:notation}
\par
\textbf{Elementary notation.}
In the thesis, $\Reals$ denotes the set of real numbers, while $\nonnegReals$ denotes the set of non-negative reals.
Similarly, $\Nats$ denotes the set of natural numbers (including $0$), $\Ints$ denotes the set of all integers and $\posNats$ denotes the set of positive integers.
For $a, b \in \Reals$, $(a, b)$ and $[a, b]$ denote the open and closed intervals between $a$ and $b$, respectively.
For $a > b$, I adopt the convention that $[a, b] = \emptyset$ is the empty set.
For $a, b \in \Ints$, I define $[a:b] = [a, b] \cap \Ints$. 
The function $\ln(x)$ denotes the natural logarithm, and $\lb(x)$ denotes the binary logarithm of a real number $x$. 
Moreover, $\exp(x) = e^x$ and $\expb(x) = 2^x$ denote the natural and binary exponential functions, respectively.
For a set $A \subseteq \Omega, A^C = \Omega \setminus A$ denotes the complement of $A$ in $\Omega$ when the latter is clear from the context.
For a function $f: \XSpace \to \YSpace$ and a set $A \subseteq \YSpace$, I denote the preimage of $A$ under $f$ as $f^{-1}(A) = \{x \in \XSpace \mid f(x) \in A\}$.
For a measure $\mu$ over $\XSpace$ and measurable $f: \XSpace \to \YSpace$, I denote the pushforward of $\mu$ through $f$ as $f \pushfwd \mu$ defined as $(f \pushfwd \mu)(A) = \mu(f^{-1}(A))$ for measurable $A \subseteq \YSpace$.
I denote the standard Lebesgue measure over $\Reals$ as $\lambda$.
\par
\textbf{Probability notation.}
I denote random variables with bold letters such as $\rvx$, and their probability measures as $P_\rvx$.
Thus, $\rvx \sim P_{\rvx}$ means the random variable $\rvx$ has distribution $P_\rvx$; I might omit the subscript from the probability measure if there is no chance of confusion.
Analogously, $\rvx \sim \rvy$ denotes equality in distribution for two random variables $\rvx, \rvy$, while $\rvx \perp \rvy$ denotes that they are independent.
For probability measures $P_\rvx$ and $P_\rvy$, I denote their product measure as $P_\rvx \otimes P_\rvy$.
I denote the expectation of a random variable $\rvx \sim P$ as $\Exp_{\rvx \sim P}[\rvx], \Exp_P[\rvx]$ or $\Exp[\rvx]$ depending on whether the distribution of the variable is clear from the context or not, and its variance as $\Var[\rvx]$.
Furthermore, I denote its conditional expectation with respect to some information $\sigmaAlgebra$ as $\Exp[\rvx \mid \sigmaAlgebra]$.
For a predicate $\pi$ which applies to an element $x$ of space, I define $\pi$'s indicator function as 
\begin{align*}
\Ind[\pi(x)] = \begin{cases}
    1 &\text{if } \pi(x) = \mathtt{True} \\
    0 &\text{if } \pi(x) = \mathtt{False}.
\end{cases}
\end{align*}
For a predicate $\pi$ which applies to a random variable $\rvx$, I define $\Prob[\pi(\rvx)] = \Exp[\Ind[\pi(\rvx)]]$.
As a mild overload of this notation (but in line with the literature standard), for a $\Reals$-valued random variable $X$ that admits a density with respect to the Lebesgue measure, I sometimes denote its density evaluated at $x$ as $\Prob[X \in dx]$.
\par
\textbf{Information notation.}
I define all relevant information-theoretic quantities using the binary logarithm. 
Concretely, for two discrete, dependent random variables $\rvx, \rvy \sim P_{\rvx, \rvy}$, let $p(x) = P_\rvx(\{x\})$ denote the probability mass of the symbol $x$ and define the conditional probability mass $p(x \mid y)$ analogously. 
Then, I define the entropy of $\Ent{\rvx}$ and the conditional entropy $\Ent{\rvx \mid \rvy}$ as 
\begin{align*}
\Ent{\rvx} &= \Exp_{\rvx \sim P_\rvx}[-\lb p(\rvx)] \\
\Ent{\rvx \mid \rvy} &= \Exp_{\rvx, \rvy \sim P_{\rvx, \rvy}}[-\lb p(\rvx \mid \rvy)],
\end{align*}
respectively.
For two probability measures $Q \ll P$ over the same space with Radon-Nikodym derivative $dQ/dP$, I define the Kullback-Leibler (KL) divergence, also known as the relative entropy, of $Q$ from $P$ as
\begin{align*}
\KLD{Q}{P} = \Exp_{\rvx \sim Q}\left[\lb\left(\frac{dQ}{dP}(\rvx)\right)\right].
\end{align*}
Moreover, I define the mutual information between a pair of dependent random variables $\rvx, \rvy \sim P_{\rvx, \rvy}$ as
\begin{align*}
\MI{\rvx}{\rvy} = \KLD{P_{\rvx, \rvy}}{P_\rvx \otimes P_\rvy}.
\end{align*}
Finally, I denote the infinity norm of the Radon-Nikodym derivative as
\begin{align*}
\norm*{dQ/dP}_\infty = \esssup\left\{dQ/dP\right\},
\end{align*}
where the essential supremum is always taken with respect to the dominating measure $P$.
Then, I define the R{\'e}nyi $\infty$-divergence of $Q$ from $P$ as 
\begin{align*}
\infD{Q}{P} = \lb \left(\norm*{dQ/dP}_\infty\right).
\end{align*}
\added[id={CM}, comment={Fixed ``some operators (like h[]) appear without definition.'' from examiners' report.}]{%
Finally, for a random variable $\rvx \sim P$ that admits a density with respect to the Lebesgue measure $f = \frac{dP}{d\lambda}$, I define its differential entropy as}
\begin{align*}
\DiffEnt{\rvx} = \Exp_{\rvx \sim P}[-\lb f(\rvx)].
\end{align*}

%% file: 2-SourceCoding/source_coding.tex
%!TEX root = ../thesis.tex

\chapter{Relative Entropy Coding as a Source Coding Problem}
\label{chapter:source_coding}
\ifpdf
    \graphicspath{{2-SourceCoding/img}}
\else
    \graphicspath{{2-SourceCoding/img}}
\fi
\par
This chapter sets the theoretical context for the thesis.
First, I describe source coding, the mathematical framework for data compression, and arithmetic coding, a lossless data compression algorithm that I will use as a component of the methods I describe later.
Then, I will motivate the topic of the thesis, relative entropy coding, as a solution to the lossy source coding problem and give an informal introduction to how it can be implemented on a computer and its potential use cases.
Finally, I formally define relative entropy coding, discuss some of its properties and relation to bits-back coding, and conclude with a synopsis of its history.
\section{Source coding}
\label{sec:source_coding}
\noindent
The first thing we should establish in a thesis about data compression is what it is, exactly. 
At a high level, we can characterise it as 
\begin{displayquote}
A compressor computes a compact representation of the data, from which the data can be recovered with as little loss of information as possible.
\end{displayquote}
In this section, I will gradually build up to a formal meaning of this idea.
\par
\textbf{Basic building blocks.}
First, I will need to be able to talk about the set of all possible entities that we might be interested in encoding; I will denote it as $\XSpace$. 
Second, I need an at most countable set of representations for reasons that will become clear shortly. 
While any countable set will do, I will always choose the set of all finite-length binary strings $\{0, 1\}^*$ because almost all digital computers use binary representations. 
For an element $c \in \{0, 1\}^*$, the length of $c$ refers to the number of bits it consists of, and I will denote it by $\abs{c}$.
Consider a function $\enc: \XSpace \to \{0, 1\}^*$ mapping the data to representations, which I will call the \textit{encoder}.
Similarly, consider a \textit{decoder} $\dec: \{0, 1\}^* \to \XSpace$, which recovers the data from a given representation.
With this notation in place, we can characterise compression more formally:
\begin{displayquote}
A compressor is a pair of functions $\enc$ and $\dec$, such that for $\rvx \in \XSpace$ the codelength $\abs{\enc(\rvx)}$ is short and $\rvx \approx \dec(\enc(\rvx))$.
\end{displayquote}
\subsection{Lossless Source Coding}
\label{sec:lossless_source_coding}
\par
To simplify the discussion for the moment, let us eliminate one of the moving parts and require that $\rvx = \dec(\enc(\rvx))$; that is, we need to be able to recover the original message $\rvx$ exactly.
This special case of the problem is called \textit{lossless source coding}.
\par
This simplification allows us to establish the following key notion: what should a ``short'' codelength mean?
The fundamental idea here is that \textbf{compressibility is equivalent to predictability}.
A natural mathematical framework to use, then, is probability theory.
While avoiding probability and relying on the theory of computation instead is possible \citep{li2013introduction}, I will not consider this alternative in this thesis.
Thus, assume that we have built a probabilistic model of the data captured by a probability distribution $P$ over $\XSpace$.
The key technical advantage here is that such a model allows us to talk about the \textit{average} codelength rather than the \textit{worst-case} codelength.
This setup fits most applications better, too: when developing a compression scheme, unless an adversary is trying to ensure that our compressor performs as badly as possible, we can expect the occurrence of inputs that have unusually long description lengths to become negligible as we compress larger and larger quantities of data.
This fact is called the \textit{asymptotic equipartition property} and is a consequence of the law of large numbers.
\par
Thus, from now I define the lossless source coding problem as finding a code $(\enc, \dec)$ that minimises the average description length, that is $R = \Exp_{\rvx \sim P}[\abs{\enc(\rvx)}]$, where $R$ is called the \textit{rate} or \textit{bitrate} of the compressor.
This suggests that a good encoder $\enc$ should assign short representations to data with higher probability (as it will contribute more to the expectation) and longer representations to less predictable data.
However, this goal now raises an important question: if we define the \textit{optimal rate} as
\begin{align}
\label{eq:lossless_rate}
R^* = \min_{(\enc, \dec)} \Exp_{\rvx \sim P}[\abs{\enc(\rvx)}] \quad \text{subject to } \rvx = \dec(\enc(\rvx)),
\end{align}
i.e.\ $R^*$ is the best achievable codelength for a given distribution $P$ over the source, can we characterise it in terms of a simpler quantity?
The answer turns out to be yes, as was shown for the first time by Claude Shannon in his seminal work \citep{shannon1948mathematical}:
\begin{theorem}[Shannon's noiseless source coding theorem.]
\label{thm:shannons_theorem}
Let $\rvx \sim P$ over some space $\XSpace$, and let the optimal rate $R^*$ be defined as in \cref{eq:lossless_rate}.
Then,
\begin{align*}
\Ent{\rvx} \leq R^*.
\end{align*}
Furthermore, this lower bound is achievable in the sense that there exists a (not necessarily unique) code $(\enc, \dec)$ such that
\begin{align*}
R^* \leq \Ent{\rvx} + 1.
\end{align*}
Indeed, any code whose bitrate is within a constant of $\Ent{\rvx}$ is called an \textbf{entropy code}.
\end{theorem}
Besides its immense practical importance, \cref{thm:shannons_theorem} is particularly relevant for the work I present in this thesis, as it will serve as both a blueprint and a basis for many of my theoretical results.
In the next section, I describe one example of an entropy code called arithmetic coding, which I will use later.
\subsection{Lossless Source Coding with Arithmetic Coding}
\label{sec:arithmetic_coding}
\begin{figure}[t]
 \centering
 \hfill
 \begin{subfigure}[t]{0.477\textwidth}
     \centering
     \includegraphics[width=0.87\textwidth]{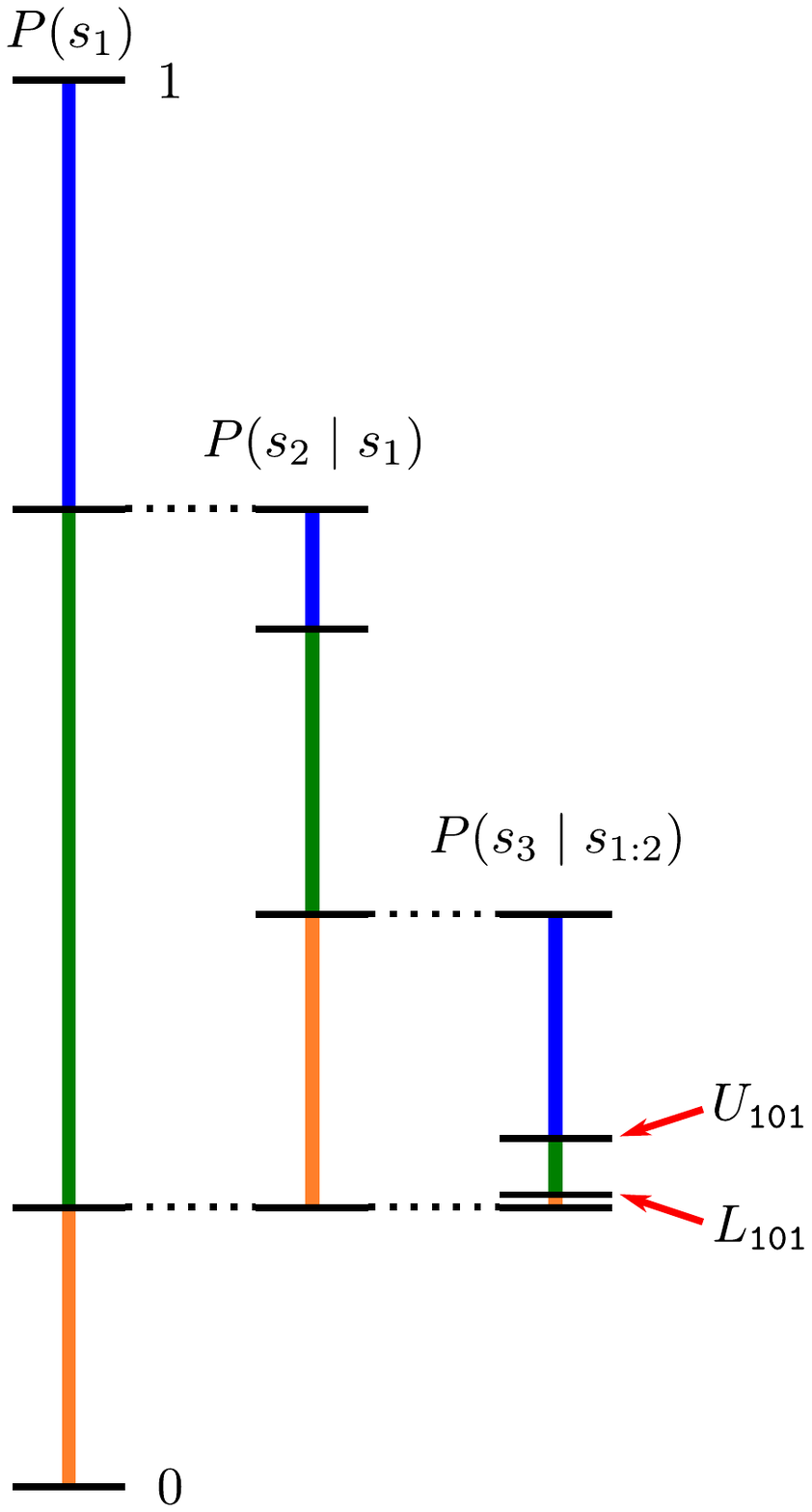}
     \caption{\textbf{Establishing the bounds:} The input string to compress is \texttt{101}.
     AC first finds the required lower and upper bounds, $L_{\texttt{101}}$ and $U_{\texttt{101}}$.}
     \label{fig:ac_first_stage}
 \end{subfigure}
 \hfill
 \begin{subfigure}[t]{0.477\textwidth}
     \centering
     \includegraphics[width=0.87\textwidth]{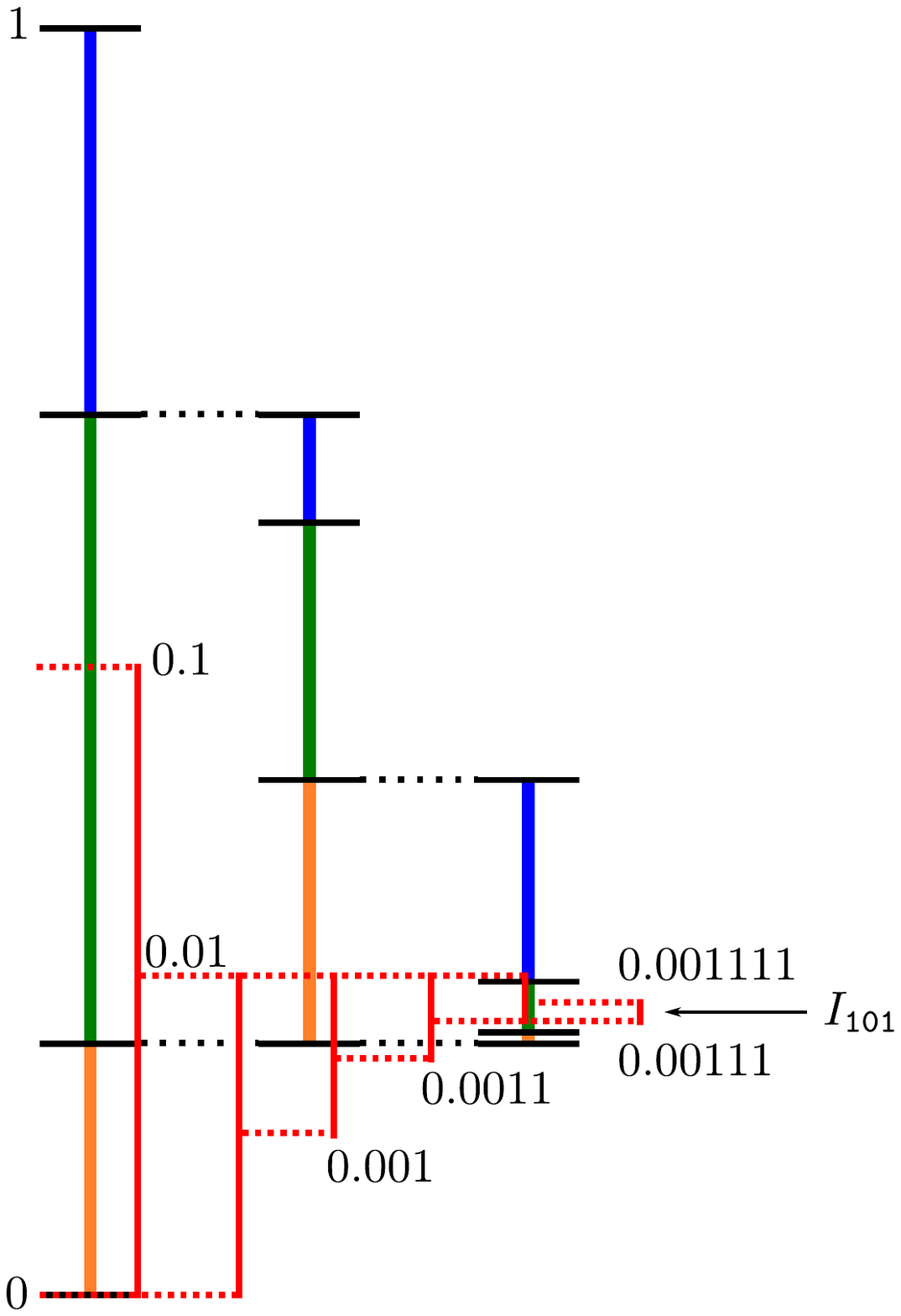}
     \caption{\textbf{Coding the bounds:} AC now finds the largest dyadic
       interval that falls within $[L_{\texttt{101}}, U_{\texttt{101}})$.
       The red intervals show the intermediate dyadic intervals considered
       by AC.}
     \label{fig:ac_second_stage}
 \end{subfigure}
 \caption[The two stages of arithmetic coding]{Example showing the two stages of AC in the encoding of the string $\rvx_{1:3} = \texttt{101}$ from the ternary alphabet $\{0, 1, 2\}$, with joint
 distribution $p(x_{1:3})$ over the symbols. 
 AC terminates after six steps, after finding that $[0.001110, 0.001111)$ falls between the established lower and upper bounds, and hence outputs $C(\texttt{101}) = 001110$.
 In the figures, the symbols $0, 1$ and $2$ are represented by the colours {\color{orange}orange}, {\color{green}green} and {\color{blue}blue}, respectively, and the length of the coloured bars between the black separators is proportional to the probability mass of the symbols. 
 Note that the probability masses can change for each symbol in the sequence.
 }
 \label{fig:ac_example}
\end{figure}
\par
In this section, I provide a brief overview of \textit{infinite precision} arithmetic coding (AC), i.e.\ I assume that the algorithm can perform infinite precision add and multiply operations. 
While this method cannot be implemented efficiently on a computer, it is sufficient for the thesis.
For finite precision AC, see any standard reference, e.g.\ \citet{cover1999elements}. 
From now on, I refer to the infinite precision variant simply as AC.
\par
For some $M, n \in \posNats$, let $\XSpace = \{1, \hdots, M\}$ be the source alphabet, and $\rvx_{1:n} = \rvx_1, \hdots, \rvx_n$ be $\XSpace$-valued random variables, with joint distribution $P_{\rvx_{1:n}}$ and probability mass function $p(x_{1:n})$.
The core idea of AC is to treat the whole source string as a single random variable and represent it as a unique real number in the interval $[0, 1)$.
Concretely, AC compresses a given a sequence $x_{1:n}$ in two steps:
\begin{enumerate}
\item \textbf{Establishing the bounds:} Let $F(x_k \mid x_{1:k - 1})$ be the \textit{cumulative distribution function} (CDF) of the $k$th member of the sequence:
\begin{align*}
F(x_k \mid x_{1:k - 1}) = \sum_{\substack{x \in \XSpace \\ x \leq x_k}} p(x \mid x_{1:k - 1}),
\end{align*}
and compute
\begin{align*}
U &= \sum_{k = 1}^n F(x_k \mid x_{1:k - 1}) \cdot p(x_{1:k - 1}) \\
L &= U - p(x_{1:n}). 
\end{align*}
By the definition of the CDF, $0 < p(x_k \mid x_{1:k - 1}) \leq F(x_k \mid x_{1:k - 1}) \leq 1$, hence $0 \leq L < U \leq 1$. 
Figure \cref{fig:ac_first_stage} shows a visual example of this step.
\item \textbf{Coding the bounds:} Once the algorithm has computed $L$ and $U$, it finds the largest \textit{dyadic interval} $I \subseteq [L, U)$. 
A dyadic interval is an interval whose length is a power of two, e.g. $I = [3/16, 5/16)$, where $\lvert I \rvert = 2^{-3}$. 
The key observation here is that a binary string of length $m$ can represent a dyadic interval of size $2^{-m}$. 
Concretely, observe that we can associate a binary string $b_{1:m} \in \{0, 1\}^m$ with a real number $\ell = 0.b_1b_2\hdots b_m000\hdots \in [0, 1)$.
Then, given $m$, the string $b_{1:m}$ uniquely specifies the interval $[\ell, \ell + 2^{-m})$.
Therefore, the binary code output by AC is precisely the binary string representing the above dyadic interval $I \subseteq [L, U)$. 
Note also, that $\lvert b_{1:m} \rvert = m = -\lb \lvert I \rvert$. 
Figure \cref{fig:ac_second_stage} represents a visual example of this step. 
\end{enumerate}
From \cref{fig:ac_example}, it should be intuitively clear that the coding procedure of AC is uniquely decodable. 
Showing it formally, however, is not insightful for this thesis. 
Thus, the reader is referred to \citet{cover1999elements}, Section 13.3. 
\par
The final question is: Why is AC efficient? 
For a source string $x_{1:n}$ the size of $[L, U)$ is by definition $U - L = p(x_{1:n})$, and the largest dyadic interval $I$ that fits into $[L, U)$ will be approximately at least half the size of $[L, U)$:
\begin{align*}
\lvert I \rvert \geq \frac{\lvert [L, U) \rvert}{2} = 2^{-(\lceil -\lb p(x_{1:n}) \rceil + 1)}.
\end{align*}
The code specifying $I$ will be exactly $-\lb\lvert I \rvert$, hence 
\begin{align*}
m = -\lb\lvert I \rvert \leq \lceil -\lb p(x_{1:n})\rceil + 1 \leq -\lb p(x_{1:n}) + 2.
\end{align*}
which is the information content of $x_{1:n}$ plus two extra bits. 
Taking expectations gives us the desired upper bound on the efficiency:
\begin{align}
\label{eq:arithmetic_coding_efficiency}
\Exp[m] \leq \Exp_{\rvx_{1:n}}[-\lb p(\rvx_{1:n})] + 2 = \Ent{\rvx_{1:n}} + 2.    
\end{align}
\subsection{Lossy Source Coding}
\label{sec:lossy_source_coding}
\par
Now that we have established that a ``short'' codelength should correspond to the rate, or in other words, the average description length produced by a code, let us return to the general setting and consider how we might relax the requirement that $\rvx = \dec(\enc(\rvx))$ to $\rvx \approx \dec(\enc(\rvx))$.
\par
A natural way to formalise the notion of ``approximate'' or ``close'' is to equip $\XSpace$ with a \textit{distortion measure} $d: \XSpace \times \XSpace \to \nonnegReals$.
Letting $\hat{\rvx} = \dec(\enc(\rvx))$, the quantity $d(\rvx, \hat{\rvx})$ measures the closeness of the reconstruction to the original data.
Though it is not always required, in this thesis, I will also assume that $d(\rvx, \hat{\rvx}) = 0$ if and only if $\rvx = \hat{\rvx}$, i.e.\ perfectly reconstructing the data is the only way to incur no distortion.
\par
How can we now incorporate this into our definition of a compressor? 
A popular choice is the \textbf{worst-case} formulation, which requires that no reconstruction exceed some distortion level $D \geq 0$.
In other words, we wish to find an encoder $\enc$ and decoder $\dec$, such that for every $\rvx \in \XSpace$ we have $d(\rvx, \dec(\enc(\rvx))) \leq D$.
However, for most practical purposes, this is both too draconian and inconvenient because 1) we expect to encounter pathological cases rarely and 2) analysing the worst-case behaviour of $\hat{\rvx}$ is usually intractably difficult.
Therefore, it stands to reason that we should allow randomising either $\enc$ or $\dec$ and consider the \textbf{average-case} distortion instead.
Thus, we can now state the average-case lossy source coding problem informally:
\begin{displayquote}
Given a statistical model of the data $\rvx \sim P$ and a distortion level $D \geq 0$, we wish to find a pair of (possibly random) functions $\enc$ and $\dec$, which minimise the average codelength $\Exp[\abs{\enc(\rvx)}]$ while also satisfying the average distortion constraint $\Exp[d(\rvx, \dec(\enc(\rvx))] \leq D$, where the expectations are taken over the distribution of $\rvx$ and the randomness of $\enc$ and $\dec$.
\end{displayquote}
Motivated by \cref{thm:shannons_theorem} and the low computational complexity of arithmetic coding for the case of lossless coding, three analogous questions arise:
\begin{enumerate}
\item \textbf{Lower bound:} Can we exhibit a non-trivial lower bound $L(P)$, such that for any $(\enc, \dec)$ that solve the above problem, we have $L \leq \Exp[\abs{\enc(M)}]$?
\item \textbf{Achievability:} Is a given lower bound $L(P)$ \textit{achievable}? In other words, can we exhibit a pair $(\enc, \dec)$ that solves the above problem and whose performance gets close to the lower bound, i.e.\ $\Exp[\abs{\enc(M)}] \lesssim L(P)$?
\item \textbf{Runtime:} Given an achievable lower bound $L(P)$, what are the computational limitations on the pairs $(\enc, \dec)$, and can we achieve them?
\end{enumerate}
\added[id={CM}, comment={Addressing the question ``Does this setup differ from rate distortion theory?''}]{%
For classical, quantisation-based approaches to lossy source coding, question 1 is answered by rate-distortion theory \citep{cover1999elements}, and questions 2 and 3 are answered by the vector quantisation algorithms and the theory of lattice coding \citep{zamir2014lattice}.
In this thesis, I take a different approach and consider lossy compression using relative entropy coding (which I define later in this chapter) instead of quantisation.
Later in this chapter, I will answer question 1 in the context of relative entropy coding based on an argument from \citet{li2024channel}, though we will see that this answer is not maximally tight.%
}%
To address this, in \cref{chapter:fundamental_limits}, I will delve deeper into the question of fundamental limits and derive tight results for questions 1) and the first half of 3).
Finally, in \cref{chapter:rec_with_pp,chapter:branch_and_bound}, I will show that the limits I derive are achievable, answering question 2 and the second half of 3.
\par
\textbf{Issues with lossless compression and quantisation.}
Before moving on to the formal definition, I make a few observations.
First, given the discussion so far, we can see that lossless source coding (\cref{sec:lossless_source_coding}) is a special case of lossy source coding when we require that there be no distortion, i.e.\ $D = 0$.
However, this is only possible if $\enc$ is invertible and $\enc^{-1} = \dec$, which immediately forces $\XSpace$ to be countable. 
We encounter a similar issue in the lossy source coding case, i.e.\ if we allow $D > 0$ to work with uncountable $\XSpace$, but require $\dec$ to be deterministic (with $\enc$ possibly randomised).
Since $\{0, 1\}^*$ is countable, the set of reconstructions $\dec(\enc(\XSpace))$ must also be countable.
Sadly, the countability of the set of reconstructions is a fundamental issue for machine learning.
The usual deep learning recipe is to parameterise the encoder and the decoder as neural networks $\enc_\theta$ and $\dec_\phi$ with weights $\theta$ and $\phi$, and use gradient descent to find the weights.
However, gradient descent only works over continuous spaces and thus cannot be applied here directly.
\par
\textbf{Towards relative entropy coding.}
To resolve the issue of countability, we might next consider randomising the decoder only: given the lossy reconstruction $\hat{\rvx}$, the decoder could introduce some noise $\hat{\rvx} + \epsilon$ (say) to the reconstruction, for some random noise $\epsilon$.
While this does make the reconstruction set uncountable, \citet{theis2021advantages} have shown it is also suboptimal: the decoder has only access to $\hat{\rvx}$, but not the original data $\rvx$, which restricts the efficiency of this scheme.
\par
\textbf{Relative entropy coding.}
Finally, the solution is to randomise both the encoder and the decoder \textit{in a synchronised way}, or, in more mathematical language, to couple them.
Concretely, I will assume that the encoder and decoder have access to the same realisation of some random variable $\rvz$, which I will call the \textbf{common randomness}.
Then, \textbf{relative entropy coding} is nothing but performing lossy source coding with such a pair of $(\enc_\rvz, \dec_\rvz)$ coupled via the common randomness $\rvz$.
The beauty of relative entropy coding is that it can now be applied to deep learning via the reparameterisation trick \citep{kingma2014auto}.
Namely, we parameterise $\enc_{\theta, \rvz}$ and $\dec_{\phi, \rvz}$ in such a way that 
\begin{align*}
\dec_{\phi, \rvz}(\enc_{\theta, \rvz}(\rvx)) \sim f(\rvx, \phi, \theta, \epsilon)
\end{align*}
for some function $f$ that is differentiable in $\theta$ and $\phi$, $\epsilon$ is some noise independent of $\rvx$, and $\sim$ denotes equality in distribution.
In practice, we would always start by designing $f$ first and then derive $\enc$ and $\dec$ from it; I will give examples of how this is done in \Cref{chapter:rec_with_pp,chapter:branch_and_bound,chapter:combiner}.
\section{Relative Entropy Coding Through Examples}
\label{sec:rec_through_examples}
\noindent
Before I move on to formally defining relative entropy coding and other mathematical gobbledygook, this section gives the reader 1) some intuition of the relative entropy coding API,\footnote{application programming interface, i.e.\ how one might \textit{use} the framework.} 2) an example ``implementation'', and 3) an overview of the most relevant use cases.  
% This is to give the reader some reason why they should care about relative entropy coding and a feel for how most algorithms (and in particular the ones presented in \cref{chapter:rec_with_pp}) work.
%
\subsubsection{The Relative Entropy Coding API}
In practice, a relative entropy coding algorithm has a similar interface to lossless source coding algorithms, such as Huffman coding or arithmetic coding.
% However, given a source symbol $\rvx$, instead of encoding $\rvx$ itself, it ``selects'' a conditional distribution $P_{\rvy \mid \rvx}$ from which a relative entropy coding algorithm encodes \textbf{one sample}: ${\rvy \sim P_{\rvy \mid \rvx}}$.
However, given a source symbol $\rvx$, instead of encoding $\rvx$ itself, the coding algorithm encodes \textbf{one sample} ${\rvy \sim P_{\rvy \mid \rvx}}$ from the conditional distribution $P_{\rvy \mid \rvx}$.
More concretely, the ingredients / moving parts are:
\begin{itemize}
\item \textbf{The channels:} A collection of ``selectable'' conditional probability distributions ${\DistFamily = \{P_{\rvy \mid \rvx} \mid \rvx \in \XSpace\}}$.
\added[id={CM}, comment={Addressing ``how do we pick $x \to y$? That's a design parameter, right?''}]{%
In practice, we get to pick this collection at our convenience, as it best fits our application; this is the real practical power of relative entropy coding.%
}
For example, $\DistFamily$ may consist of Bernoulli distributions with $\rvx = p$ representing the success probability, as was done by \citet{kobus2024universal}.
As another example, $\DistFamily$ could be the set of Gaussian distributions where $\rvx = (\mu, \sigma^2)$ such that $P_{\rvy \mid \rvx} = \Normal(\mu, \sigma^2)$, as will be the case in \cref{chapter:combiner}.
I will refer to an element of $\DistFamily$ as a \textit{channel} and assume that both the encoder and the decoder have agreed on the same set $\DistFamily$.
\item \textbf{The coding distribution:} a distribution $P_\rvy$, defined over the same space as the $P_{\rvy \mid \rvx}$s; it corresponds to the coding/source distribution in entropy coding.
I will assume that both the encoder and the decoder know $P_\rvy$, and as such, it cannot depend on anything the decoder does not already know prior to communication.
\par
\textbf{Remark.} In the traditional source coding literature, the coding distribution is also often called the source distribution.
However, in relative entropy coding, these two terms are generally different because we think of $\rvx$ as the source but are interested in encoding a randomly chosen $\rvy$.
\item \textbf{The pseudo random number generator (PRNG) seed:} For a fixed source symbol $\rvx$, the encoded sample $\rvy$ is still random.
We simulate this randomness on a computer using a PRNG with some seed $\rvz$.
However, this might appear contradictory since, given a symbol $\rvx$ and a seed $\rvz$, we will always get back the same sample $\rvy$.
Thus, technically, the randomness of $\rvy$ arises by randomly choosing a different seed $\rvz$ each time we encode a sample.
Therefore, mathematically, I will model the seed $\rvz$ as a random variable.
\par
I will assume that the encoder always prepends $\rvz$ to their message, and thus, they can always synchronise the seed with the decoder.
In practice, $\rvz$ usually is a 32- or 64-bit integer and thus contributes a constant number of bits to the total codelength.
\added[id={CM}, comment={Clarifying ``In the one-shot setting, you need to send the shared randomness, and this will
impose a fixed cost''}]{However, these 32 or 64 bits ar negligible to the several kilo- or megabytes of code that is required to encode $\rvy$ in practice, e.g., when encoding high-resolution images or video.}
As such, I will not account for the cost of communicating $\rvz$ in the results I present in the thesis.
\end{itemize}
Given the above ingredients, \textbf{the encoder's API} consists of a single function $\enc_{\rvz}(\rvx)$, which outputs a short message $m$ in the form of a binary string that corresponds to a sample $\rvy \sim P_{\rvy \mid \rvx}$ for the distribution $P_{\rvy \mid \rvx} \in \DistFamily$.
How exactly does the message $m$ correspond to the sample $\rvy$?
The answer is given by the \textbf{decoder's API}, which consists of a single function $\dec_\rvz(m)$, which has the property that $\dec_\rvz(\enc_\rvz(\rvx)) \sim P_{\rvy \mid \rvx}$.
Note that the encoder and decoder's seeds are synchronised.
Thus, a relative entropy coding algorithm is a sampling algorithm that outputs a code $m$ for the sample from which it can be reconstructed at any time, given the seed $\rvz$.
\subsubsection{Implementing \texorpdfstring{$\enc$}{the encoding} and \texorpdfstring{$\dec$}{the decoding function}}
\begin{figure}
\centering
\includegraphics{2-SourceCoding/img/selection_sampling_sketch.tikz}
\caption[Illustration of relative entropy coding using a selection sampler.]{Illustration of relative entropy coding for the pair of dependent random variables $\rvx, \rvy \sim P_{\rvx, \rvy}$ using a selection sampler.
The sender \textbf{A} and the receiver \textbf{B} share a sequence of i.i.d.\ $P_\rvy$-distributed samples as their common randomness $\rvz$.
Then, upon receiving a source sample $\rvx \sim P_{\rvx}$, \textbf{A} uses a selection rule $K$ that selects one of the samples in the shared sequence such that $\rvy_K \sim P_{\rvy \mid \rvx}$.
Since the selected index $K$ is discrete, \textbf{A} uses an appropriate entropy coding algorithm to efficiently encode $K$ and transmit it to \textbf{B}.
Finally, \textbf{B} can recover a $P_{\rvy \mid \rvx}$-distributed sample by decoding $K$ and selecting the $K$th sample in the shared sequence.
See the main text for details on the selection rule and the encoding process.
}
\label{fig:selection_sampling_sketch}
\end{figure}
\textbf{An Inefficient Implementation of $\enc$ and $\dec$.}
Given a source symbol $\rvx$ and a seed $\rvz$, the most straightforward implementation of $\enc_\rvz(\rvx)$ is to use the seed $\rvz$ to directly simulate a sample $\rvy \sim P_{\rvy \mid \rvx}$ from the channel of interest.
Then, we can use a lossless source coding algorithm such as arithmetic coding to encode $\rvy$ with the coding distribution $P_\rvy$.
Then, the decoder can recover $\rvy$ by using the same coding distribution $P_\rvy$, and they did not even need the shared seed $\rvz$!
But what is the coding efficiency of this solution?
Since we assume to have used an entropy coder, the average codelength will be approximately $\Ent{\rvy}$.
However, we can do much better, as I describe next.
\par
\textbf{An Efficient Implementation of $\enc$ and $\dec$: Selection Sampling.}
Surprisingly, there are schemes that can encode $\rvy$ using $\MI{\rvx}{\rvy} \leq \Ent{\rvy}$ bits on average.
This is especially significant when $\rvy$ is continuous since the mutual information can still be finite in this case, while the entropy is always infinite.
I now sketch a general way of constructing such a scheme, which I will refer to as \textit{selection sampling}.
However, note that it is in no way the only high-level recipe to implement relative entropy coding.
Nonetheless, it will be helpful for the reader to keep the following construction in mind, as it captures the essence of how the schemes I discuss in \cref{chapter:rec_with_pp} will work.
\par
I illustrate the ``algorithm'' in \cref{fig:selection_sampling_sketch}, and it works as follows:
At step $n$, the encoder simulates a new sample $\rvy_n \sim P_\rvy$ using the shared seed $\rvz$ and appends it to the previously collected samples, resulting in the list $\rvy_{1:n} = (\rvy_1, \hdots, \rvy_n)$.
Then, the algorithm uses a (possibly randomised) function $\mathtt{stop}(\rvy_{1:n})$, which outputs either \texttt{True} or \texttt{False} to decide whether to stop the algorithm or to move onto step $n + 1$.
If $\mathtt{stop}(\rvy_{1:n}) = \mathtt{True}$, then the algorithm uses a second function $\mathtt{select}(\rvy_{1:n})$, which outputs an integer $K \in [1:n]$.
Importantly, we construct $\mathtt{stop}$ and $\mathtt{select}$ in such a way that $\rvy_K \sim P_{\rvy \mid \rvx}$.
Finally, the encoder outputs a code for $K$ that they communicate to the decoder.
On the decoder's side, they decode $K$ and use it in conjunction with the shared seed $\rvz$ to simulate the sample $\rvy_K$.
\par
\textbf{An Example of Selection Sampling: Rejection Sampling.}
Perhaps the simplest selection sampling scheme is rejection sampling \citep{neumann1951various}.
Denote the density ratio/Radon-Nikodym derivative as $r_\rvx = \frac{dP_{\rvy \mid \rvx}}{dP_\rvy}$ and assume it is bounded, i.e.\ $\norm{r_\rvx}_\infty < \infty$.
Then, for an independently drawn uniform random variable ${U_n \sim \Unif(0, 1)}$, we set%
\begin{align*}
\mathtt{stop}(\rvy_{1:n}) &= \Ind[U_n \leq r_\rvx(\rvy_n) / \norm{r_\rvx}_\infty] \\
\mathtt{select}(\rvy_{1:n}) &= n
\end{align*}
In this case, we see that for every $n$
\begin{align*}
\Prob[\mathtt{stop}(\rvy_{1:n})] = \Prob[U_n \leq r_\rvx(\rvy_n) / \norm{r_\rvx}_\infty] = \Exp[r_\rvx(\rvy_n) / \norm{r_\rvx}_\infty] = 1 / \norm{r_\rvx}_\infty.
\end{align*}
Then, it follows that
\begin{align*}
\Prob[\rvy_{\mathtt{select}(\rvy_{1:n})} \in A \mid \mathtt{stop}(\rvy_{1:n})] &= \frac{\Prob[\rvy_n \in A, \mathtt{stop}(\rvy_{1:n})]}{\Prob[\mathtt{stop}(\rvy_{1:n})]} \\
&= \norm{r_\rvx}_\infty \cdot \Exp[\Ind[\rvy_n \in A] \cdot r_\rvx(\rvy_n) / \norm{r_\rvx}_\infty] \\
&= P_{\rvy \mid \rvx}(A),
\end{align*}
as desired.
\added[id={CM}, comment={Partially addressing the comment ``I would also recommend that the candidate make key definitions of algorithms clearer typographically (perhaps with boxes) or consider the use of algorithm environments''}]{%
I give the pseudocode for this procedure in \cref{alg:global_rs} in \cref{sec:global_rs}.%
}%
\par
\textbf{Encoding the Selected Index.}
How many bits do we need to encode the returned index $N$?
The best we can do is $\Ent{N \mid \rvz}$, i.e., encoding $N$ given our knowledge of the shared randomness $\rvz$.
While taking $\rvz$ into account when encoding $N$ is possible \citep{braverman2014public}, in most cases considering the ``marginal'' entropy $\Ent{N} \geq \Ent{N \mid \rvz}$ is usually much easier and fortunately already leads to optimal encoding by using a clever scheme devised by \citet{li2018strong}.
The high-level idea is to collect some statistics about the problem to construct a good coding distribution.
Specifically, they showed that if the encoder and the decoder know $\Exp[\lb N]$, then there is a coding distribution whose compression rate $R$ is 
\begin{align*}
R = \Exp[\lb N] + \lb (\Exp[\lb N] + 1) + 2 \geq \Ent{N}.
\end{align*}
For completeness, I repeat their argument in \cref{sec:li_el_gamal_bound_on_pos_int_random_variable}.
Now, if ${\Exp[\lb N] \approx \MI{\rvx}{\rvy}}$, then we have a scheme as efficient as I promised!
\par
So, how does rejection sampling fare?
First, note the distribution of $N$ for an input $\rvx$ is geometric with rate $\Prob[\mathtt{stop}(\rvy_{1:n})] = 1 / \norm{r_\rvx}_\infty$, since the events that the algorithm stops at step $n$ are i.i.d., hence $\Exp[N \mid \rvx] = \norm{r_\rvx}_\infty$.
Then, we have
\begin{align}
\Exp[\lb N \mid \rvx] &\leq e^{1 / \norm{r_\rvx}_\infty} \cdot \lb (\norm{r_\rvx}_\infty + 1), \tag{see \cref{sec:geom_log_expectation_bound}}\\
&\leq \lb \norm{r_\rvx}_\infty + 2 \tag{$1 \leq \norm{r_\rvx}_\infty$}\\
\Rightarrow \quad \quad  \Exp[\lb N] &\leq \Exp_{\rvx}[\lb \norm{r_\rvx}_\infty] + 2. \label{eq:rs_expected_index_log_upper_bound}
\end{align}
% where for the last inequality, I took expectation over both sides and used the facts that $\lb (x + 1) = \lb (x) + \lb(1 + 1/x)$ and $1 \leq \norm{r_\rvx}_\infty$.
We see that the right-hand side of \cref{eq:rs_expected_index_log_upper_bound} is not quite the $\MI{\rvx}{\rvy}$ that we desire. 
Indeed it is always at least as large as the mutual information, so the question is: \textit{when} is it close? 
If $P_{\rvx, \rvy}$ is jointly Gaussian, we do have $\Exp_{\rvx}[\lb \norm{r_\rvx}_\infty] \approx \MI{\rvx}{\rvy}$ (I will show this in \cref{sec:numerical_experiments}, see \cref{eq:gauss_gaus_exp_inf_div}), and thus relative entropy coding with rejection sampling in the Gaussian case is essentially optimal.
\par
Another interesting case is to consider the problem from a computational perspective.
A natural requirement is that the average runtime of the sampling algorithm be finite, that is, $\Exp[N] \leq C$ for some constant $C > 1$.
As we saw above, average runtime of rejection sampling is $\Exp[N] = \Exp_{\rvx}[\norm{r_\rvx}_\infty]$. 
Hence, using Jensen's inequality, we obtain that the upper bound on $\Exp[N]$ implies that $\Exp_{\rvx}[\lb\norm{r_\rvx}_\infty] \leq \lb C$, and hence $\Exp_{\rvx}[\lb\norm{r_\rvx}_\infty] \leq \MI{\rvx}{\rvy} + \lb C'$ for some $C' \leq C$.
From this, using minimal manipulation we find that in this case
\begin{align*}
R \leq \MI{\rvx}{\rvy} + \lb(\MI{\rvx}{\rvy} + 1) + \lb C' + \lb(\lb C' + 3) + 2
\end{align*}
I reiterate, that we achieve the above coding rate by directly encoding the runtime $N$ of rejection sampling.
With some minimal additional ``post-processing'' we can reduce this further to
\begin{align*}
R < \MI{\rvx}{\rvy} + \lb(\MI{\rvx}{\rvy} + 1) + 2\lb (\lb C' + 1) + 8.31
\end{align*}
We can achieve this by including the uniforms $U_n$ we used for the acceptance decisions in the common randomness $\rvz$ and encoding the \textit{sorted position} of $U_N$ among the first $N$ (technically, the first $\exp_2(\lceil\lb N\rceil)$) uniforms; see \Cref{sec:rej_samp_with_sorted_unifs} for the details.
While this technically acheives our desired efficiency, the disadvantage is that the bound on the codelength depends on $C'$.
Depending on the application, this may or may not be acceptable. 
Hence, in the rest of the thesis, I will develop algorithms whose rate is independent of such constants $C'$.
\par
Finally, in the general case, it is challenging to control $\Exp_{\rvx}[\lb \norm{r_\rvx}_\infty]$ and thus rejection sampling is not the best selection sampler to use for relative entropy coding.
After developing the foundations of relative entropy coding in this chapter, I devote the entirety of \cref{chapter:rec_with_pp,chapter:branch_and_bound} to developing selection samplers with more sophisticated $\mathtt{stop}$ and $\mathtt{select}$ functions that are maximally efficient both in terms of codelength and their runtime.
\subsubsection{Applications of Relative Entropy Coding}
As I explained in \cref{sec:lossy_source_coding}, relative entropy is best thought of as a lossy source coding method and an alternative to quantisation.
But what concrete benefits does it offer compared to quantisation?
I collect its most salient advantages below. 
\begin{itemize}
\item \textbf{Learned lossy source coding.}
Recall from \cref{sec:lossy_source_coding} that in lossy source coding, we are given some data $\rvx \sim P_\rvx$ and a notion 
\replaced[id={PL}, comment={fixed ungrammatical sentence, PL/email/2}]{%
of distance $d$.
Then, we wish to design a \textit{lossy compression mechanism} that outputs some $\hat{\rvx}$ that is cheap to encode and makes the \textit{distortion} $d(\rvx, \hat{\rvx})$ small.
}{ 
distance $d$ we wish to design a \textit{lossy compression mechanism} that outputs some $\hat{\rvx}$ that is cheap to encode and such that the \textit{distortion} $d(\rvx, \hat{\rvx})$ is small.}
\par
We can employ relative entropy coding to obtain such a mechanism: we design a conditional distribution $P_{\hat{\rvx} \mid \rvx}$ such that encoding a sample $\hat{\rvx} \sim P_{\hat{\rvx} \mid \rvx}$ has low distortion; for example, we can require that on average the distortion does not exceed some pre-specified level $D$: $\Exp_{\rvx}[\Exp_{\hat{\rvx} \mid \rvx}[d(\rvx, \hat{\rvx})] \leq D$.
What is the best we can do under this setup?
For a fixed source distribution $P_\rvx$ and conditional distribution $P_{\hat{\rvx}\mid \rvx}$, relative entropy coding requires approximately $\MI{\rvx}{\hat{\rvx}}$ bits to encode a sample, hence the best we can do is
\begin{align}
\label{eq:source_coding_rate_distortion_function}
R(D) = \inf_{P_{\hat{\rvx} \mid \rvx} \in \DistFamily} \MI{\rvx}{\hat{\rvx}} \quad \text{subject to}\quad \Exp_{\rvx}[\Exp_{\hat{\rvx} \mid \rvx}[d(\rvx, \hat{\rvx})] \leq D.
\end{align}
When $\DistFamily$ contains all possible conditional distributions, the function $R(D)$ is known as the \textit{rate-distortion function} in the learned compression literature.
\added[id={CM}, comment={First part of answering the question ``What is the difference between the learned lossy source coding you present and the other papers in the literature?''}]{%
Note also that $R(D)$ is closely related to, but not equivalent to, the information rate-distortion function from classical information theory \citet{cover1999elements}.
The two differences are that in learned compression, 1) we do not assume that the joint distribution factorises dimensionwise as $P_{\hat{\rvx}, \rvx} = P_{\hat{\ervx}_1, \ervx_1}^{\otimes \dim(\rvx)}$ in the definition of $R(D)$ and as such 2) we do not (and cannot) analyse the asymptotic behaviour as $\dim(\rvx) \to \infty$.%
}
\par
While the connection to the rate-distortion function is a nice theoretical property, relative entropy coding offers a practical benefit too: we can directly optimise the objective in \cref{eq:source_coding_rate_distortion_function}; this is called ``end-to-end training'' in machine learning parlance.
Concretely, let us restrict $\DistFamily$ to consist of a family of \textit{parametric, reparameterisable distributions}: we assume that every distribution $P_{\hat{\rvx} \mid \rvx} \in \DistFamily$ admits a representation of the form $\hat{\rvx} = \phi(\theta, \rvx, \epsilon)$, for a deterministic function $\phi$ with parameters $\theta$ and independent noise $\epsilon \perp \rvx$.
Furthermore, we assume that $\phi$ is differentiable in its argument $\theta$; an example of this could be the set of Gaussian distributions of the form $P_{\hat{\rvx} \mid \rvx} = \Normal(\theta \cdot \rvx, I)$, which admits the reparameterisable representation $\hat{\rvx} = \theta \cdot \rvx + \epsilon$ where $\epsilon \sim \Normal(0, I)$.
Now, we can write down the Lagrange dual of the objective in \cref{eq:source_coding_rate_distortion_function}:
\begin{align}
\label{eq:source_coding_rate_distortion_objective}
\Loss(\theta, \beta) = \MI{\rvx}{\phi(\theta, \rvx, \epsilon)} + \beta \cdot \Exp_{\rvx, \epsilon}[d(\rvx, \phi(\theta, \rvx, \epsilon))] + \mathrm{constant},
\end{align}
where I suppressed the dependence on $D$ since it is constant with respect to $\theta$.
Now, using an initial guess for $\theta_0$, we can directly use gradient descent \citep{rumelhart1986learning} in conjunction with the reparameterisation trick \citep{kingma2014auto} to optimise \cref{eq:source_coding_rate_distortion_objective}, using Monte Carlo estimates for the expectations where appropriate. 
Moreover, using different settings for $\beta$ will correspond to different distortion levels $D$.
Of course, here I ignored many practical concerns, such as computing the coding distribution $P_{\hat{\rvx}}$ that is required for compression with relative entropy coding and computing the mutual information term in \cref{eq:source_coding_rate_distortion_objective}; I will deal with these practical matters in \cref{chapter:combiner}.
\par
\added{\textit{A note on alternative approaches to learned compression.}}
We cannot directly optimise \cref{eq:source_coding_rate_distortion_objective} if we use quantisation as our lossy mechanism, as it does not admit a differentiable reparameterisation since its output is discrete.
\added[id={CM}, comment={Addressing ``What is the difference between the learned lossy source coding you present and
the other papers in the literature?''}]{%
Hence, contemporary learned compression methods that do not use relative entropy coding have to resort to approximations instead.
The two most popular approaches are 1) to perform straight-through estimation (STE) for gradient descent and 2) to use an additive noise approximation.
STE-based approaches, such as those of \citet{theis2017lossy} or \citet{van2017neural}, make the ``approximation'' during training that the gradient of the quantiser $\mathcal{Q}$ is the identity instead of being zero almost everywhere: they ``set'' $\nabla_x \mathcal{Q}(x) = 1$ during backpropagation.
On the other hand, additive noise-based approaches, such as \citet{balle2017end,balle2018variational}, resort to scalar quantisation.
Then, during training only, they make the approximation that $\round{x} \approx x + U$ for $U \sim \Unif(-1/2, 1/2)$ and switch to hard quantisation for compression after training.
While both of these approaches can work well in practice, they ultimately suffer from the weaknesses of quantisation I discussed in \Cref{sec:lossy_source_coding}; their main advantage over relative entropy coding is that their implementation tends to be simpler.%
}%
\begin{figure}
\centering
\includegraphics{2-SourceCoding/img/realism_experiment.tikz}
\caption{The realism distinguishability experiment.
See the item titled ``Compression with realism constraints'' for an explanation.}
\label{fig:realism_experiment}
\end{figure}
\item \textbf{Compression with realism constraints.}
It is important to remember that our mathematical formulation of lossy source coding is only one among several sensible possibilities, and it is useful to question whether it captures what we wish to achieve operationally.
For example, consider video streaming: given that the lossy reconstruction captures the original video's semantic details, we really want it to \textit{look good}.
The rate-distortion formulation in \cref{eq:source_coding_rate_distortion_function} goes back to Shannon's original work \citep{shannon1948mathematical}, but how well does the distortion (e.g., mean squared error) capture the notion of looking good?
It certainly does to a point: if the distortion is significant, the reconstruction is probably not very good.
But what if the distortion is small?
Unfortunately, a well-known issue with mean squared error is that even if it is quite small, image/video compression algorithms optimised for it tend to produce blurry patches in areas of low local contrast.
These particular details are called \textit{compression artefacts} in the literature, compared to the \textit{artefact-free} original.
\par
A landmark insight came from \citet{blau2019rethinking}, who realised that, under mild assumptions, no distortion measure can be artefact-free.
This is because distortion does not capture the notion of an artefact or, from the complementary perspective, realism.
To see how we can extend the framework, we can start with the question we would like to answer: given, say, an image $\rvx$, how realistic does it look? Finding a mathematical formulation of this question is currently one of the holy grails of lossy compression \citep{theis2024what}.
The question is difficult to formalise because it would require a \textit{no-reference} formulation of realism: we are only given $\rvx$ and cannot compare it to anything.
Therefore, we slightly modify the setup to simplify things, illustrated in \Cref{fig:realism_experiment}: Assume we randomly selected an image $\rvx_0$ from a large dataset.
Then, with probability $1/2$, we compress it with our lossy compressor and obtain the reconstruction $\rvx_1 \sim P_{\hat{\rvx} \mid \rvx = \rvx_0}$ and with probability $1/2$, we leave it as is; denote the outcome by $\rvx_{\rvb}$.
Now, we ask the following \textit{forced binary choice} question:
Given only $\rvx_\rvb$, can we say whether it was compressed?
Now, since the answer is yes or no, the best result we can hope for is that no one can tell with probability better than $1/2$, i.e.\ a random guess: we call this \textit{perfect realism}.
However, this is only possible if the input distribution $P_\rvx$ matches the reconstruction distribution $P_{\hat{\rvx}}$ perfectly: $P_\rvx = P_{\hat{\rvx}}$.
Therefore, we can consider the rate-distortion function under a perfect realism constraint:
\begin{align*}
R_p(D) = \inf_{P_{\hat{\rvx} \mid \rvx} \in \DistFamily} \MI{\rvx}{\hat{\rvx}} \quad \text{subject to}\quad \Exp_{\rvx}[\Exp_{\hat{\rvx} \mid \rvx}[d(\rvx, \hat{\rvx})] \leq D \text{ and } P_{\rvx} = P_{\hat{\rvx}}.
\end{align*}
This formulation now presents the second benefit of relative entropy coding over quantisation:
\citet{theis2021advantages} have shown that there exist cases where $R_p(D)$ is \textit{strictly} lower if we use relative entropy coding than if we quantise.\footnote{Technically, they show that there are cases where allowing the encoded noise to depend on the input (i.e., using relative entropy coding) results in a strictly lower $R_p(D)$ value compared to using decoder-side only randomisation of a quantised representation.}
\item \textbf{Compressing differential privacy mechanisms.}
A particularly relevant aspect of contemporary data science and machine learning is the privacy of users' data.
For example, assume there is a trusted server (e.g., a hospital database) that holds some user data $\Data = \{\rvx_1, \rvx_2, \hdots, \rvx_n\}$ (e.g., patient records).
Now, assume an untrusted third party, e.g., a researcher or data scientist from a biotech company, wishes to analyse the $\Data$.
Differential privacy \citep{dwork2014algorithmic} provides a way to share a privatised version of $\Data$ with the third party, i.e., it allows the third party to extract useful information without learning the exact values of the $\rvx_i$.
We achieve this via a \textit{differential privacy mechanism}: we add noise to $\rvx$ in a controlled way: $\rvy = f(\rvx, \epsilon)$ for some carefully designed function $f$ and noise $\epsilon$.
For example, if $\rvx$ represents a patient's height,
we could add some Gaussian noise to privatise it: $\rvy = \rvx + \epsilon$, where $\epsilon \sim \Normal(0, \sigma^2)$.
We can tune the privacy-accuracy trade-off by adjusting the noise variance $\sigma^2$.
\par
An even stricter setup is when the server holding the users' data is also untrusted, so the users do not want the server to learn too much about their data.
An example of this is in federated learning, where a central server wishes to train a machine learning model using distributed gradient descent.
To this end, the server shares its model and model parameters with its clients, who then compute a local gradient update based on their locally held data.
In non-privatised federated learning, the clients send their gradients to the server, and the server combines them to perform a central model update.
However, in \textit{local differential privacy} \citep{kasiviswanathan2011can}, the users add noise to their data (the gradient update, in this example) in a controlled way so that the server cannot learn too much about their data.
For example, a popular local differential privacy mechanism is \texttt{PrivUnit}$_2$ \citep{bhowmick2018protection}.
Roughly speaking, once a user has computed their local gradient $\vg$, they decompose it into its magnitude and direction: $\vg = \norm{\vg} \cdot \vu$, where $\vu$ is a unit vector.
Then, the \texttt{PrivUnit}$_2$ perturbs the direction $\vu$, which the client then recombines with the magnitude and sends it to the server.
\par
Note that in both cases, one of the parties holds some data $\rvx$ but sends a noisy version $\rvy$ to the recipient.
This then provides a third natural application for relative entropy coding: we can express the privacy mechanism as a conditional distribution $P_{\rvy \mid \rvx}$ and use relative entropy coding to communicate the privatised data directly.
However, we cannot naively apply any relative entropy coding algorithm to the privacy mechanism; we must take additional care that the code (not just the sample it encodes) does not reveal too much about the user's data.
This problem was first studied by \citet{shah2022optimal}, who developed an approximate scheme; recently, \citet{liu2024universal} has developed the first privatised relative entropy coding algorithm that can encode exact samples from arbitrary privacy mechanisms.
\item \textbf{Easier theoretical analysis, better worst-case performance.}
I close this section on a more theoretical/technical point: relative entropy coding is simply \textit{better-behaved} than quantisation.
This is because a quantiser $\quantiser$ is a deterministic function, which makes the error analysis troublesome.
More precisely, given the quantisation error $\Delta(\rvx) = Q(\rvx) - \rvx$, we can perform two types of analyses: 
\begin{enumerate}
\item 
\textit{Worst-case analysis:} We can try to characterise $\sup_{\rvx}\abs{\Delta(\rvx)}$.
Performing this analysis can already be intractable in one dimension and is usually hopelessly difficult in higher dimensions.
\item
\textit{Average-case analysis:} Given the source distribution $P_\rvx$ we might try to characterise the expected error $\Exp_{\rvx}[\abs{\Delta(\rvx)}]$.
Unfortunately, the distribution of $\Delta(\rvx)$ will usually depend on $P_\rvx$ in a highly non-trivial way, which once again makes the analysis impossible except for the simplest of cases.
\end{enumerate}
Relative entropy coding suffers from none of these issues. 
First, we have absolute control over the error distribution: \textit{we pick it to our convenience}.
Second, since each time we encode a randomly perturbed version of $\rvx$, we do not need to worry about the worst-case error occurring, especially in high-dimensional settings where concentration of measure phenomena begin to kick in.
\end{itemize}
\subsubsection{The Current Limitations of Relative Entropy Coding}
If relative entropy coding is the bee's knees with so many advantages over quantisation, why is it not used in practice? First, I should note that this might be a somewhat unfair question since the applicability of relative entropy coding to practical data compression was only recognised a few years ago \citep{havasi2019minimal}.
Nonetheless, relative entropy coding currently possesses two significant disadvantages compared to quantisation: speed and synchronisation.
\par
First, all currently known relative entropy coding algorithms that can be applied to practically relevant problems are either too slow or too limited.
Unfortunately, in \cref{sec:computational_complexity_of_rec}, I show this is not a coincidence: any generally applicable relative entropy coding algorithm is slow.
Therefore, any fast algorithm must use some special structure of the problem class.
However, not all hope is lost: I dedicate all of \cref{chapter:branch_and_bound} to developing relative entropy coding algorithms that \textit{almost} apply to practical problems and are optimally fast.
\par
The second issue of relative entropy coding is that it poses a much greater engineering challenge than quantisation. Since it relies on shared common randomness, which is usually implemented with a shared pseudo-random number generator (PRNG) and a shared seed, we must always ensure that the encoder and decoder's PRNG states are synchronised; desynchronisation can lead to catastrophic decoding errors.
While this is ``only'' a matter of rigorous bookkeeping, it remains a monumental engineering challenge and is currently an under-investigated research area.
\section{Relative Entropy Coding, Formally}
\label{sec:rec_formally}
\noindent
In this section, I formally define the relative entropy coding problem; the rest of the thesis uses this formulation.
Before I begin, however, I need to introduce some further terminology.
I start by establishing the minimum necessary structure of the ambient space required to make sense of relative entropy coding.
\begin{definition}[Polish space]
A set $\XSpace$ is Polish if it is a separable, completely metrisable topological space.
\end{definition}
In this thesis, I define every measure, probability distribution, and random variable in an appropriate Polish space.
As such, every time I say ``space'' or ``random variable'', it should be understood that I mean ``Polish space'' or ``Polish random variable.''
But why use Polish spaces? Essentially, it is a matter of generality and convenience.
Most of the external results I invoke have been developed for Polish spaces.
However, this is no coincidence; Polish spaces possess the minimum necessary mathematical structure required to define probability distributions that we can evaluate and analyse easily and random variables that we can hope to simulate on a computer.
An alternative characterisation of Polish spaces is that it is topologically equivalent (homeomorphic) to a subspace of $[0, 1]^\Nats$ (a $G_\delta$ subset of the Hilbert cube).
Thus, in this sense, Polish random variables are the most general random variables we can simulate with inverse transform sampling, at least in theory, ignoring concerns with computational complexity.
\par
Now, I define the problem setting I will consider for most of the thesis.
\begin{definition}[Channel simulation algorithm]
\label{def:channel_simulation_alg}
Let $\rvx, \rvy \sim P_{\rvx, \rvy}$ be dependent random variables over the space $\XSpace \times \YSpace$ with joint distribution $P_{\rvx, \rvy}$.
Furthermore, let $\rvz \perp \rvx$ be a random variable over $\ZSpace$, such that there exists a pair of computable, measurable functions $\enc: \ZSpace \times \XSpace \to \{0, 1\}^*$ and $\dec: \ZSpace \times \{0, 1\}^* \to \YSpace$ with $\dec_\rvz(\enc_\rvz(\rvx)) \sim P_{\rvy \mid \rvx}$.
Then, I will call the triple $\rvz, \enc_\rvz, \dec_\rvz$ a \textbf{channel simulation algorithm} for $\rvx \to \rvy$.
Furthermore, I will assume that the encoder and decoder synchronise their realisation of $\rvz$; hence, I will call $\rvz$ the \textbf{common randomness} between the encoder and decoder.
\end{definition}
A channel simulation algorithm is thus a sampling algorithm, since $\dec_\rvz(\enc_\rvz(\rvx))$ produces a $P_{\rvy \mid \rvx}$-distributed sample.
However, it has more structure: it also produces a code for the sample via $\enc_\rvz(\rvx)$, from which we can always reconstruct the sample by invoking $\dec_\rvz$.
% As such, channel simulation algorithms are the generalisation of \textit{variable-length codes} from source coding. 
\par
\textbf{Remark.}
In the above definition, for a given $\rvx$, the only other source of randomness in $\enc$ comes from the common randomness $\rvz$.
But is this not too restrictive? Could we further generalise this definition by allowing $\enc$ to depend on some \textit{private} randomness $\rvz'$ too?
Interestingly, when we have \textit{unlimited common randomness} (essentially, we have ``$\Ent{\rvz} = \infty$''), we do not need any private randomness to design maximally efficient schemes. 
Indeed, in the rest of this thesis, I assume that the sender and receiver can synchronise an unlimited amount of common randomness; hence, the above definition is sufficient.
However, it is essential to remember that this is an additional assumption we need to make.
As I discuss in \cref{sec:some_properties_of_rec}, if we have limited common randomness, additional private randomness allows us to design more efficient schemes.
\par
The name \textit{channel simulation} comes from information theory, where people like to think of conditional distributions such as $P_{\rvy \mid \rvx}$ as noisy communication channels with $\rvx$ being the input to the channel and $\rvy$ representing some corrupted version of $\rvx$.
Usually, \textit{nature} corrupts $\rvx$, but in channel simulation, we take things into our own hands and design an algorithm for corrupting $\rvx$ ourselves.
Concretely, we can think of the channel simulation pipeline as a \textit{coupling} $g$ of $\rvx, \rvy$ and $\rvz$ via the relation $\rvy = g(\rvx, \rvz) = \dec_\rvz(\enc_\rvz(\rvx))$.
Such a $g$ contains more information than the conditional $P_{\rvy \mid \rvx}$, as there could be many different couplings $g$ that give rise to the same conditional.
However, can we also go the other way?
Namely, can we always find at least one coupling $g$ for an arbitrary conditional $P_{\rvy \mid \rvx}$?
It turns out that the answer to this question is not only positive, but as \citet{li2018strong} have shown, for any $\rvx, \rvy \sim P_{\rvx, \rvy}$ we can always construct an algorithm with some common randomness $\rvz$ such that $\abs{\enc_\rvz(\rvx)}$ is relatively short on average; this will be the essence of relative entropy coding.
As an interesting side note, the coupling definition of a channel is closer to the definition Shannon used to develop his theory originally \citep{shannon1948mathematical}; presumably, the information theory community only shifted closer to the conditional distribution definition since, for most practical purposes, the exact coupling mechanism is irrelevant.
\par
There is another interesting interpretation of channel simulation algorithms as \textit{invertible sampling algorithms}.
A sampling algorithm usually relies on a pseudo-random number generator (PRNG), which, given some seed (a binary string, usually represented as an integer), generates a random variate from some specified distribution.
In the case of a channel simulation algorithm, we can think of $\dec_\rvz$ as the PRNG and the code output by $\enc_{\rvz}$ as a second, ``private'' seed in comparison to the ``public'' seed $\rvz$.
Furthermore, this sampler is invertible in the sense that $\rvy$ is a deterministic function of $\rvx$ and $\rvz$ and hence we can recover ``the private seed that generated $\rvy$'' via the relation $\dec^{-1}_\rvz(\rvy) = \enc_\rvz(\rvx)$.
\par
However, channel simulation algorithms only capture the notion of correctness but leave out the notion of efficiency, as measured by their expected codelength.
This difference is analogous to the information-theoretic notion of a \textit{uniquely decodable code}, which only captures the correctness and \textit{entropy codes} (\cref{thm:shannons_theorem}), which are uniquely decodable codes whose efficiency achieves the fundamental lower bound of the entropy.
Therefore, as I hope it comes as no surprise to the reader, I now show that channel simulation algorithms also satisfy an appropriate lower bound given by an appropriate relative entropy.
I will then define a relative entropy code as a channel simulation algorithm that achieves this lower bound.
\begin{theorem}[Lower bound on the description length of channel simulation algorithms]
\label{thm:channel_simulation_mi_lower_bound}
Let $\rvx, \rvy \sim P_{\rvx, \rvy}$ be a pair of dependent random variables.
Let $\rvz, \enc_\rvz, \dec_\rvz$ be a channel simulation algorithm for the channel $\rvx \to \rvy$.
Then,
\begin{align*}
\MI{\rvx}{\rvy} \leq \Exp[\abs{\enc_\rvz(\rvx)}]
\end{align*}
\end{theorem}
\begin{proof}
Let $\rvm = \enc_\rvz(\rvx)$.
Then,
\begin{align*}
\Exp[\abs{\rvm}] &\stackrel{\text{(a)}}{\geq} \Ent{\rvm} \\
&\stackrel{\text{(b)}}{\geq} \Ent{\rvm \mid \rvz} \\
&\stackrel{\text{(c)}}{=}\Ent{\rvm \mid \rvz} - \Ent{\rvm \mid \rvz, \rvx} \\
&=\MI{\rvx}{\rvm\mid \rvz} \\
&\stackrel{\text{(d)}}{=}\MI{\rvx}{(\rvm, \rvz)} \\
&\stackrel{\text{(e)}}{\geq} \MI{\rvx}{\rvy},
\end{align*}
where equation (a) holds by Shannon's source coding theorem (\cref{thm:shannons_theorem}); (b) holds, since conditioning reduces the entropy; (c) holds since $\rvm$ is a deterministic function of $\rvz$ and $\rvx$; (d) holds by the chain rule of mutual information and since $\MI{\rvx}{\rvz} = 0$ due to independence; and (e) holds by the data processing inequality, since $\rvx \to (\rvm, \rvz) \to \rvy$ is a Markov chain.
\end{proof}
As I will show in \Cref{thm:global_gprs_codelength}, this lower bound is achievable within a logarithmic factor, which now warrants the following definition.
\begin{definition}[Exact relative entropy coding]
\label{def:rec}
Let $\rvz, \enc_\rvz, \dec_\rvz$ be a channel simulation algorithm for $\rvx, \rvy \sim P_{\rvx, \rvy}$.
I call the triple $\rvz, \enc_\rvz, \dec_\rvz$ an (exact) relative entropy coding algorithm if we also have
\begin{align}
\label{eq:rec_achievability_upper_bound}
\Exp[\abs{\enc_\rvz(\rvx)}] \leq \MI{\rvx}{\rvy} + \Oh(\lb (\MI{\rvx}{\rvy} + 1)).
\end{align}
\end{definition}
I will slightly relax this definition in \cref{sec:approximate_sampling} and consider approximate relative entropy coding algorithms whose output does not follow $P_{\rvy \mid \rvx}$ exactly.
However, for most of the theoretical part of the thesis, I will consider only \cref{def:rec}.
\subsection{Some Properties of Relative Entropy Coding}
\label{sec:some_properties_of_rec}
\textbf{Relation to Source Coding.}
For a discrete random variable $\rvx \sim P_\rvx$, let us set $\rvy = \rvx$.
Now, simulating the channel $\rvx \to \rvy$ is equivalent to encoding $\rvx$ so that it can be reconstructed exactly, which is precisely the lossless source coding problem.
Pleasingly, in this case $\MI{\rvx}{\rvy} = \MI{\rvx}{\rvx} = \Ent{\rvx}$, hence combining \cref{thm:channel_simulation_mi_lower_bound,def:rec} \textit{almost} recovers the lower bound and achievable upper bound from Shannon's source coding theorem \citep{shannon1948mathematical}.
Only almost, because the required upper bound \cref{eq:rec_achievability_upper_bound} is looser than the lower bound, as it features an additional $ \Oh( \lb(\MI{\rvx}{\rvy} + 1))$ factor.
In other words, while the upper and lower bounds' scaling matches in the first order, they differ in the second order.
Therefore, it is natural to wonder whether this requirement is due to the looseness of the analysis or if it is fundamental.
Surprisingly, \citep{braverman2014public, li2018strong} constructed pairs of dependent random variables $\rvx, \rvy \sim P_{\rvx, \rvy}$ such that the efficiency of any relative entropy coding algorithm is lower-bounded by at least ${\MI{\rvx}{\rvy} + \lb(\MI{\rvx}{\rvy} + 1) - \Oh(1)}$, showing that the additional logarithmic factor is indeed necessary for general algorithms.
Furthermore, in \cref{sec:communication_complexity_of_rec}, I will also demonstrate a situation where neither the lower nor the upper bound is tight!
\par
There are a few further similarities and differences between source coding and relative entropy coding worth pointing out.
First, while in the $\rvx = \rvy$ case channel simulation seems to collapse onto lossless source coding, it technically does not because in lossless source coding, there is no need for the encoder and the decoder to synchronise between each other some source of common randomness $\rvz$, as required by channel simulation.
Thus, do we need the common randomness $\rvz$ for channel simulation after all?
Second, a reader familiar with source coding will recognise a few characteristics of the above definition similar to analogous concepts in coding theory: the above definition requires a \textit{one-shot}, \textit{variable-length} code.
Thus, we may again ask whether other settings make sense, such as performing asymptotic instead of one-shot analysis or using fixed-length descriptions.
In the rest of this section, I delve into these questions in more detail.
\par
\textbf{The Role of Common Randomness.}
Unlike lossless source coding, relative entropy coding algorithms make heavy use of the common randomness $\rvz$.
But is it strictly necessary, or could it be possible to develop an algorithm that needs no access to common randomness while retaining the same codelength efficiency?
\citet{cuff2013distributed,kumar2014exact} investigated this question and found that common randomness is indeed fundamental.
In fact, \citep{kumar2014exact} have shown that the achievable lower-bound for channel simulation when no common randomness is available is given by a quantity they call the ``common entropy'' $G(\rvx, \rvy) \geq \MI{\rvx}{\rvy}$. 
They also show that if we allow $R$ bits of common randomness instead, the achievable lower bound starts decreasing.
Then, there is a rate $R_0$, at which the achievable lower bound reaches $\MI{\rvx}{\rvy}$ and further common randomness does not help.
However, without precise knowledge of the source distribution $P_\rvx$, it is impossible to determine $R_0$; hence, it is easier to consider the unlimited common randomness case.
Note that unlimited common randomness is not a practical issue, so long as we are happy to accept that the output of the pseudo-random number generator (PRNG) is ``random enough''.
If we make this leap of faith, then the encoder can share their PRNG seed with the decoder, using which they can generate a practically infinite amount of pseudo-random numbers.
Hence, in the rest of the thesis, I will assume that unlimited common randomness is available between the encoder and the decoder.
\par
\textbf{The One-shot vs.\ the Asymptotic Setting.}
Operationally, the reader may think of the one-shot setting as measuring the scheme's efficiency by measuring the codelength $\abs{\enc_\rvz(\rvx)}$ and averaging this number across several experiments, e.g.\ compressing individual images and averaging their codelength.
We can contrast this with the \textit{asymptotic} setting, where we create ``bigger and bigger'' experiments and normalise by the ``experiment size'', e.g.\ we compress a big dataset of images and measure the codelength of the compressed dataset normalised by the number of images in the dataset.
Formally, let $\rvx, \rvy \sim P_{\rvx, \rvy}$. 
In the one-shot setting, we consider a sequence $(\rvz_1, \rvz_2, \hdots)$ of different, i.i.d.\ realisations of the same common randomness and a sequence $(\rvx_1, \rvx_2, \hdots)$ of i.i.d.\ $P_{\rvx}$-distributed variables and we measure the average codelength by $\lim_{n \to \infty}\frac{1}{n}\sum_{k = 1}^n \abs{\enc_{\rvz_n}(\rvx_n)} = \Exp[\abs{\enc_\rvz(\rvx)}$.
On the other hand, in the asymptotic setting, for each $n \in \posNats$ we create the $n$-fold copy of the joint distribution $P_{\rvx, \rvy}^{\otimes n}$ and consider the relative entropy coding problem for the channel $\rvx^n \to \rvy^n$ and measure the scheme's efficiency by $\abs{\enc_\rvz(\rvx^n)} / n$.
Then, we are interested in quantifying $\lim_{n \to \infty} (\abs{\enc_\rvz(\rvx^n)} / n)$.
Most previous work in information theory has studied relative entropy coding from the asymptotic perspective, as tools such as the law of large numbers and the asymptotic equipartition property \citep{cover1999elements} allow a simplified analysis.
However, the one-shot setting is more relevant for implementing practical schemes and, besides a slight detour in \cref{sec:rec_asymptotic_behaviour}, I only consider the one-shot setting in this thesis.
\par
\textbf{Variable- vs. Fixed-length Description.}
Similarly, we may contrast the variable-length description with the \textit{fixed-length} (also called fixed-rate) description, i.e.\ a \textit{rate} $R$ is fixed ahead of time, and we are only allowed to use $R$-bit-long codewords.
However, it turns out that one-shot, fixed-rate schemes require the rate $R$ to be much larger than $\MI{\rvx}{\rvy}$ to work \citep{cubitt2011zero}, and hence I do not consider them further in this thesis.
However, this question becomes relevant again if we relax the requirement that a relative entropy coding algorithm has to output an exact sample, which is what I consider next.  
\par
\textbf{Exact vs. Approximate Schemes.}
In \cref{def:channel_simulation_alg}, I required that for any input $\rvx$ the reconstruction distribution $Q_{\rvy \mid \rvx} \sim \dec_\rvz(\enc_\rvz(\rvx))$ has to match the ``true'' conditional exactly: $Q_{\rvy \mid \rvx} = P_{\rvy \mid \rvx}$.
However, this might be too draconian a requirement in some cases, and an approximate sample might suffice.
But how should we measure the quality of a sample?
As I argue in \cref{sec:approximate_sampling}, the right notion (at least when applying relative entropy coding for lossy source coding) is the \textit{total variation distance} (TV-distance), defined between two probability measures $\mu, \nu$ over the measurable space $\Omega$ with $\sigma$-algebra $\mathcal{F}$ as
\begin{align}
\label{eq:total_variation_distance}
\TVD{\mu}{\nu} = \sup_{A \in \mathcal{F}}\abs{\mu(A) - \nu(A)}.
\end{align}
In \cref{sec:approximate_sampling}, I demonstrate that the TV-distance captures the notion of ``one-shot distinguishability,'' i.e., how difficult it is for an observer to say whether a given sample $\rvy$ was drawn from $\mu$ or $\nu$.
As I will show, the advantage of approximate schemes is that they admit efficient fixed-length descriptions with small expected TV-error $\Exp_\rvx[\TVD{Q_{\rvy \mid \rvx}}{P_{\rvy \mid \rvx}}]$. 
\subsection{Relation to Bits-back Coding}
\label{sec:bits_back_vs_rec}
\par
Relative entropy coding was introduced to machine learning by Marton Havasi and his collaborators \citep{havasi2019minimal}, inspired by the bits-back argument of \citet{hinton1993keeping}.
Coincidentally, around the same time, \citet{townsend2019practical} introduced a lossless source coding technique based on the bits-back argument called \textit{bits-back coding}.
This has created some confusion in the community regarding the relation between relative entropy coding and bits-back coding, which I hope to clear up here.
\par
First, to be clear, what I am about to describe and to which I will refer as ``bits-back coding'' is more general than the schemes described by \citet{hinton1993keeping} and \citet{townsend2019practical}.
This is to demonstrate both the generality of bits-back coding and to show that it is, as a technique, \textit{orthogonal} to relative entropy coding.
Essentially, bits-back coding can be described via the following heuristic ``equation'':
\begin{align*}
\text{error-correction with side information } + \text{ invertible sampling } = \text{ bits-back coding.}     
\end{align*}
In the rest of this section, I give meaning to this ``equation.''
\par
\textbf{Error Correction.}
Error correction (also known as channel coding) is a vast topic, so I will only give a basic description here; the interested reader should consult the excellent book of \citet{mackay2003information}.
The core idea behind error correction is that if we have some data we wish to communicate over a noisy channel, we build some redundancy into the data so that even if it gets slightly corrupted, we can recover the original data.
Perhaps the simplest example of this idea is the repetition code. 
For example, if we wish to communicate a single bit over a channel, we map $0 \mapsto 000$ and $1 \mapsto 111$.
Then, so long as no more than one bit of the message gets corrupted, we can recover the original message by mapping the strings $\{000, 100, 010, 001\}$ to $0$ and the strings $\{111, 011, 101, 011\}$ to $1$.
\par
Usually, when we design error-correcting codes, we assume the channel $P_{\rvy \mid \rvx}$ over which data is transmitted is fixed and out of our control.
Then, we try to design some code, represented by a source distribution $P_\rvx$, so we can correct it with high probability.
Concretely, we aim to design the code so that when an input $\rvx$ is passed through the noisy channel and corrupted to $\rvy$, it is decoded to something other than $\rvx$ with as low a probability as possible.
However, we could turn this problem on its head: we could assume that the source distribution $P_\rvx$ is fixed and we design a noisy channel $P_{\rvy \mid \rvx}$ such that we can \textit{always} recover $\rvx$ from $\rvy$.
To continue with the repetition code example, if we have some distribution over the strings $\{000, 111\}$, we can set $P_{\rvy \mid \rvx}$ to be the distribution that changes at most one bit of $\rvx$.
Then, we can always recover $\rvx$ by decoding $\rvy$ to $000$ or $111$, depending on which one is closer.
But why would we do such a thing? Why would we abandon the usual setup of error correction to this more artificial one?
The answer is to do bits-back coding and to see this; I now describe its second ingredient.
\par
\textbf{Invertible sampling.}
I already briefly discussed invertible sampling in \cref{sec:rec_formally}, but here, I expose it in a related but different sense.
The primitive we need to discuss is an arbitrary prefix-free entropy code with an encoding function $\enc$ and a decoding function $\dec$.
However, for this section only, I need to slightly extend the definitions of $\enc$ and $\dec$.
For the rest of the section, assume that $\rvm'$ denotes a sequence of already compressed messages that I will call the \textit{stream}.
Then, given a new message $\rvx$ and a probability distribution $P$, the encoding function $\enc_P(\rvx, \rvm')$ encodes $\rvx$ into the stream using a prefix code derived from $P$, producing the updated stream $\rvm = \enc_P(\rvx, \rvm')$.
Then, $\dec_P$ will denote the inverse of this generalised encoder, in the sense that $(\rvx, \rvm') = \dec_P(\enc_P(\rvx, \rvm'))$ for any input $\rvx$ and stream $\rvm'$.
Arithmetic coding is an example of such a pair of $\enc, \dec$ (\cref{sec:arithmetic_coding}).
\par
Now, consider the following puzzle: if we know the source distribution $P$ and use it to encode a source symbol $\rvx \sim P$ into an empty stream $m = \epsilon$,\footnote{$\epsilon$ denotes the empty string.} what is the distribution of the encoded message $\enc_P(\rvx, \epsilon)$?
Roughly speaking, we know that, on average, the encoded message will contain $\Exp[\abs{\enc_P(\rvx, \epsilon)}] \approx \Ent{\rvx}$ bits.
Furthermore, since $\enc_P(\cdot, \epsilon)$ is an invertible function, we also know that $\Ent{\enc_P(\rvx, \epsilon)} = \Ent{\rvx}$.
However, the entropy of a random binary string can only match its length if its bits are all independent, fair coin tosses.
\par
Given this observation, let us now consider what happens if we have an infinite sequence $\rvm$ of fair coin tosses and we apply $\dec_P$ and obtain $\rvx', \rvm' \gets \dec_P(\rvm)$?
By the invertibility of the above process, we must have $\rvx' \sim P$!
Thus, we can think of a prefix entropy code as an invertible sampler in the sense that if we think of $\rvm$ as a ``random seed,'' then we can ``decode'' a sample $\rvx' \sim P$ by computing $\rvx', \rvm' \gets \dec_P(\rvm)$, but we can also recover the seed by ``encoding'' $\rvx'$: $\rvm = \enc_P(\rvx', \rvm')$.
Here, I have omitted a few important technical points, such as the fact that we usually need a different number of random bits to obtain a $P$-distributed sample.
For a precise formulation of this argument, I refer the interested reader to the excellent paper of \citet{knuth1976complexity}.
\par
\textbf{Bits-back coding.}
Given some source distribution $P_\rvx$, and a prefix entropy code $\enc, \dec$, we could encode a source symbol $\rvx \sim P_\rvx$ into a stream $\rvm$ directly by computing $\enc_{P_\rvx}(\rvx, \rvm)$.
However, combining the above ideas of error correction and invertible sampling, we obtain a second method for encoding $\rvx$, called bits-back coding:
\begin{legal}
\item \textbf{Before communication.}
Given a source distribution $P_\rvx$, the sender and receiver design a channel $P_{\rvy \mid \rvx}$ so that the decoder can always recover $\rvx$ from $\rvy$.
Furthermore, they also compute the marginal $P_\rvy$.
\item \textbf{Bits-back Encoding}
\begin{legal}
\item 
Given a stream $\rvm$ and a source symbol $\rvx \sim P_\rvx$, the sender decodes a sample from the stream $\rvy, \rvm' \gets \dec_{P_{\rvy \mid \rvx}}(\rvm)$.
\item \textbf{Encoding Step 2.}
The sender encodes $\rvy$ into the stream using the marginal $P_\rvy$: $\rvm'' \gets \enc_{P_\rvy}(\rvy, \rvm')$.
\end{legal}
\item \textbf{Bits-back Decoding}
\begin{legal}
\item
Given $\rvm''$, the receiver decodes $\rvy, \rvm' \gets \dec_{P_\rvy}(\rvm'')$.
\item 
%\textbf{Decoding Step 2.}
The receiver uses the error correction to recover $\rvx$ from $\rvy$.
\item 
%\textbf{Decoding Step 3.}
The receiver encodes $\rvy$ back into the stream $\rvm \gets \enc_{P_{\rvy \mid \rvx}}(\rvy, \rvm')$.
\end{legal}
\end{legal}
Thus, bits-back coding is fully invertible, which hinges precisely on the receiver's ability to recover $\rvx$ from $\rvy$.
But how does its efficiency compare to encoding $\rvx$ directly?
First, decoding $\rvy \sim P_{\rvy \mid \rvx}$ decreases the message length by $\approx -\lb P_{\rvy \mid \rvx}(\rvy \mid \rvx)$ bits.
Then, encoding $\rvy$ into the stream using $P_\rvy$ increases the message length by $\approx -\lb P_\rvy(\rvy)$ bits.
In total, by applying Bayes' rule and noting that $\rvx$ is a deterministic function of $\rvy$, we find the message length increases by 
\begin{align}
\abs{m''} - \abs{m} 
\approx -\lb P_\rvy(\rvy) + \lb P_{\rvy \mid \rvx}(\rvy \mid \rvx)
= -\lb \frac{P_{\rvx}(\rvx)}{P_{\rvx \mid \rvy}(\rvx \mid \rvy)} 
= -\lb P_\rvx(\rvx) \text{ bits,}  
\label{eq:bits_back_no_side_info_codelength}
\end{align}
which is the same as the direct encoding technique.
However, if the efficiencies of direct and bits-back coding match, why should we bother with bits-back coding?
We bother because we can encode $\rvx$ without directly modelling its marginal distribution $P_\rvx$!
This is particularly powerful when $\rvx$ represents the equivalence class of $\rvy$ under some symmetry; this is the essence of the work of \citet{kunze2024entropy}.
\par
However, there are many cases where designing an error-correcting code that can always recover $\rvx$ from $\rvy$ only is challenging or impossible.
Instead, we introduce additional flexibility into the framework and require that the decoder can recover $\rvx$ from $\rvy$ given some side information $\rvs$. 
\par
As an example, I now sketch the method Jiajun He and I propose in our paper \citep{he2024getting}, which was partially inspired by the work of \citet{kunze2024entropy}. 
Consider a neural network $f$ whose weight parameterisation is rotationally invariant, i.e.\ for weights $\rvx$ and an arbitrary orthogonal matrix $Q$, we have $f(\cdot \mid \rvx) = f(\cdot \mid Q \rvx)$.
If we wish to compress such a model, it would be desirable to encode the rotational equivalence class of the weights 
\begin{align*}
[\rvx] = \{\rvx' \mid \text{ there exists an orthogonal matrix } Q \text{ such that } Q\rvx' = \rvx\}
\end{align*}
instead of $\rvx$. 
However, it is not straightforward to design such code directly.
Instead, we use bits-back coding to encode a \textit{randomly rotated} matrix.
To do this, without loss of generality, assume that the singular value decomposition (SVD) is of the form $\rvx = DV^\top$, i.e., the usual left-rotation that appears in the SVD is the identity matrix.
If it is not in this form already, we can always bring $\rvx$ in this form.
Now, we can define $\rvy = Q\rvx$ where $Q$ is a uniformly randomly drawn orthogonal matrix, i.e.\
$P_{\rvy \mid \rvx}$ randomly rotates $\rvx$.
Then, when the decoder receives $\rvy$, they compute its SVD $\rvy = \widetilde{Q}D\widetilde{V}^\top$.
Importantly, the SVD of a matrix is unique up to the sign of the rows of $\widetilde{Q}$; in other words, there is a diagonal matrix $\sigma$ with the diagonal entries taking values in $\{-1, 1\}$ such that $\sigma \widetilde{Q} = Q$.
\par
Thus, we can construct the following bits-back code: 
\begin{enumerate}
\item \textbf{Encoding step 1.} The sender ``standardises'' the weights: Given a weight setting $\rvx'$, they compute its SVD $\rvx' = UDV^{\top}$ and set $\rvx \gets DV^{\top}$.
\item \textbf{Encoding step 2.} Given the stream $\rvm$ and the standardized $\rvx$, the sender decodes a random rotation matrix $Q, \rvm' \gets \dec(\rvm)$.
\item \textbf{Encoding step 3.} They compute $\rvy = Q\rvx$ and encode it into the stream $\rvm'' \gets \enc(\rvy, \rvm')$.
\item \textbf{Encoding step 4.} They also compute the SVD of $\rvy = \widetilde{Q}D\widetilde{V}^\top$, and compute the sign matrix $\sigma$ such that $\sigma \widetilde{Q} = Q$ and encodes $\sigma$ into the stream $\rvm''$ as side information; this costs one bit per diagonal entry.
\end{enumerate}
The receiver decodes the sign matrix $\sigma$ from the stream first.
Then, using this side information, they can recover $Q$ for $\rvy$ and reverse the bits-back encoding process.
This procedure was used by \citet{he2024getting} to encode the weights of large language models with rotationally symmetric weight parameterisation, resulting in gains of up to a $5-7\%$ reduction in storage size.
For practical details, the interested reader should consult the original paper.
\par
\textbf{Connection to the relative entropy.}
Let us now consider the coding efficiency of general bits-back coding.
Compared to \cref{eq:bits_back_no_side_info_codelength}, the ``pure bits-back'' message length, i.e.\ discounting the communication cost of the side information, is on average
\begin{align*}
\Exp_{\rvy}[\abs{\rvm''} - \abs{\rvm}] 
\approx \Exp_{\rvy}[-\lb P_\rvy(\rvy) + \lb P_{\rvy \mid \rvx}(\rvy \mid \rvx)]
= \KLD{P_{\rvy \mid \rvx}}{P_\rvy} \text{ bits.}  
\end{align*}
This observation is precisely the bits-back argument of \citet{hinton1993keeping} when applied to the weights of a Bayesian neural network.
\par
\textbf{Connection to relative entropy coding.}
Note that the two connections of bits-back coding to relative entropy coding have been superficial so far.
First, they both rely on the concept of ``invertible sampling,'' but these concepts are distinct between the two frameworks.
Second, the relative entropy $\KLD{P_{\rvy \mid \rvx}}{P_\rvy}$ appears in the coding efficiency of both methods.
However, this is a deceptive connection in both frameworks' cases:
In the case of bits-back coding, the relative entropy discounts the communication cost of the side information. 
In contrast, in the case of relative entropy coding, it is a lower bound on the efficiency, which is not guaranteed to be tight in every case.
Indeed, relative entropy coding and bits-back coding are \textit{orthogonal} frameworks, so we can use bits-back coding to construct a relative entropy coding algorithm.
We can achieve this by setting the common randomness $\rvz$ of relative entropy coding as the side information $\rvs$ used by bits-back coding.
One example is \textit{bits-back quantisation}, a scheme I constructed with Lucas Theis to encode certain uniform random variables.
However, describing this scheme falls outside the scope of this thesis; the interested reader should consult Appendix C in \citet{flamich2023adaptive}.
\subsection{A Brief History of Relative Entropy Coding}
\label{sec:rec_history}
\par
\textbf{General relative entropy coding.}
Randomised codes are as old as information theory itself. 
Indeed, Shannon used them in his seminal work \citep{shannon1948mathematical} to prove his noiseless source coding theorem by using randomisation to show that good codes must exist asymptotically.
However, the study of channel simulation was started by
\citet{bennett2002entanglement}.
They studied communication problems over quantum channels, assuming the communicating parties have a shared quantum entanglement.
They were the first to observe that the classical analogue of quantum entanglement is common randomness and studied channel simulation for discrete memoryless channels.
The next important significantnt in the theory came with the work of \citet{harsha2010communication}, who constructed the first general-purpose relative entropy coding algorithm for arbitrary discrete channels.
The first fully general relative entropy coding algorithm for arbitrary Polish random variables was constructed by \citet{li2018strong}, who called their construction the ``Poisson functional representation.''
However, it turned out that a more general variant of this algorithm, solely for sampling, was proposed by \citet{maddison2014sampling}, who called it ``A* sampling.''
I have worked on and extended all of these works.
Concretely, in \citet{flamich2022fast}, among other contributions, I showed that the more general, branch-and-bound variants of A* sampling also give rise to fast relative entropy coding algorithms. I expose and analyse this algorithm in \cref{sec:global_a_star,sec:bnb_a_star}.
Later, in \citet{flamich2023adaptive} and \citet{flamich2023grc}, I generalised Harsha's greedy rejection sampler to arbitrary Polish spaces.
Furthermore, in \citet{flamich2023gprs} I proposed a new relative entropy coding algorithm inspired by the greedy rejection sampler, which I called the greedy Poisson rejection sampler (GPRS), which I will discuss and analyse in detail in \cref{sec:global_gprs,sec:bnb_gprs}.
\par
\textbf{Applying relative entropy coding to data compression.}
As I note at the beginning of \cref{sec:bits_back_vs_rec}, it was \citet{havasi2019minimal} who introduced relative entropy coding to the machine learning community and proposed an approximate relative entropy coding algorithm, called minimal random coding \citep[MRC;][]{havasi2019minimal,theis2022algorithms}, to compress variational Bayesian neural networks.
In my MPhil thesis, I worked with Marton Havasi to develop the first image compression algorithm that used MRC in conjunction with variational autoencoders \citet{flamich2020compressing}.
I later combined the approaches from both papers to develop general-purpose compression algorithms using Bayesian implicit neural representations in \citet{guo2023compression,he2024recombiner}; I devote the entirety of \cref{chapter:combiner} to discussing these.
While people did study channel simulation and relative entropy coding before, such as in the works of \citet{harsha2010communication,cuff2013distributed,braverman2014public,liu2018rejection,li2018strong}, these focused solely on proving existence and achievability results and did not pay any attention to developing computationally efficient and practically applicable algorithms.
\par
\textbf{Dithered Quantisation.}
Concurrently with my work, \citet{agustsson2020universally} have taken an alternative approach: they developed a very fast relative entropy coding algorithm for uniform channels only.
Concretely, they observed that dithered quantisation \citep[DQ;][]{ziv2003universal} can be viewed as a relative entropy coding algorithm when the conditional $P_{\rvy \mid \rvx}$ is uniformly distributed.
In the appendix of their paper, \citet{agustsson2020universally} already suggested a way to extend DQ to simulate one-dimensional Gaussian distributions; however, their work was later significantly generalised by \citet{hegazy2022randomized} and Lucas Theis and me in \citet{flamich2023adaptive} to work for general one-dimensional, unimodal channels.
As the name suggests, dithered quantisation works by randomising a quantiser in a highly controlled fashion, so that we achieve the desired output distribution.
As such, the development and theory of this line of relative entropy coders relies more on classical lattice quantisation \citep{zamir2014lattice} than general sampling theory, and hence I do not discuss DQ-based approaches further in this thesis.
Finally, very recently, there has been some progress in developing fast relative entropy coding algorithms for discrete multivariate distributions by \citet{sriramufast}, who achieved this by establishing a striking connection between error-correcting codes and relative entropy coding.
However, developing a fast algorithm for the practically most important channel, the multivariate Gaussian when $P_{\rvy \mid \rvx} = \Normal(\rvx, \sigma I)$, is still an open problem.

%% file: 3-FundamentalLimits/fundamental_limits.tex
%!TEX root = ../thesis.tex
%
%
\chapter{The Fundamental Limits of Relative Entropy Coding}
\label{chapter:fundamental_limits}
This chapter explores the fundamental limits on the communication and computational complexity of relative entropy coding.
Furthermore, on the way to obtaining the characterisations, I also encounter concepts, such as the width function and selection sampling, that will be crucial in deriving and analysing the relative entropy coding algorithms in \cref{chapter:rec_with_pp,chapter:branch_and_bound}.
\section[The Communication Complexity of Relative Entropy Coding]{The Communication Complexity of \\ Relative Entropy Coding}
\label{sec:communication_complexity_of_rec}
\par
As I showed in \cref{thm:channel_simulation_mi_lower_bound}, the mutual information provides a lower bound to the average codelength of any relative entropy coding algorithm, and in \cref{chapter:rec_with_pp}, I show that it is also achievable up to a logarithmic factor as required by \cref{def:rec}.
However, as the first significant result of this section, I show that there is an even tighter lower bound.
To do this, I first define a new statistical distance that I call the \textit{channel simulation divergence}, then show that it gives a tighter lower bound to the expected codelength of any relative entropy coding algorithm.
Then, as the second significant result of the section, I will use this tighter lower bound to demonstrate that relative entropy coding has a non-trivial asymptotic coding efficiency.
\subsection{The Channel Simulation Divergence}
\label{sec:channel_simulation_divergence}
\begin{figure}[t]
\centering
\includegraphics{3-FundamentalLimits/img/gauss_width_function_illustration.tikz}
\caption[Illustration of the Gaussian width function.]{Illustration of the width function (\cref{def:width_function}) and the canonical representation of the density ratio (\cref{lemma:width_function_properties}.3 for the case when $Q = \Normal(0.3, 0.5^2)$ and $P = \Normal(0, 1)$.
\textbf{Left:} the density ratio/Radon-Nikodym derivation $\frac{dQ}{dP}$.
The $P$-measure of the {\color{red} red} interval is the value of the width function $w_P(h)$ at $h = 1.5$.
\textbf{Middle:} The width function $w_P(h)$, with the value at $h = 1.5$ marked out corresponding to the left plot.
It is also the probability density function of the associated random variable $H$ (\cref{lemma:width_function_properties}.1).
The {\color{red} shaded red area} is the survival function $S(h)$ evaluated at $h = 1.5$ (\cref{lemma:width_function_properties}.2).
\textbf{Right:} The canonical representation $\eta(p)$ of the density ratio (\cref{lemma:width_function_properties}.3).
The length of the red interval is the value of the width function $w_P(h)$ evaluated at $h = 1.5$, corresponding to the interval marked in the left plot.
}
\label{fig:gauss_width_function_illustration}
\end{figure}
\par
I define a new statistical distance called the \textit{channel simulation divergence} and establish its basic properties.
I first studied this quantity with Daniel Goc in our paper \citep{goc2024channel}, in which we established a lower bound to the coding efficiency of a subclass of relative entropy coding algorithms called \textit{causal rejection samplers}.
In this section, I give new, simplified proofs of the properties of the channel simulation divergence.
Then, I will use it to prove exciting new results that significantly extend our work in \citet{goc2024channel}, showing that it provides a general lower bound on the rate of any relative entropy coding algorithm in the next section.
To begin, I define an object that will play a crucial role over the next three chapters:
\begin{definition}[Width function]
\label{def:width_function}
Let $Q \ll P$ be probability measures over some space $\Omega$ with Radon-Nikodym derivative $r = dQ/dP$.
Then, the $P$ and $Q$ width functions are
\begin{align*}
w_P(h) &= \Prob_{Z \sim P}[r(Z) \geq h] \\
w_Q(h) &= \Prob_{Z \sim Q}[r(Z) \geq h]
\end{align*}
\end{definition}
The $P$-width function has a nice geometric interpretation: we can think of it as the size of the ``horizontal slice'' of $r$ at height $h$ weighted by $P$; see the left panel in \cref{fig:gauss_width_function_illustration} for an illustration.
Before I continue, let me collect some useful properties of width functions; see \cref{fig:gauss_width_function_illustration} for an illustration.
\par
\begin{lemma}
\label{lemma:width_function_properties}
Let $w_P$ and $w_Q$ be width functions as in \cref{def:width_function}.
Then, we have the following:
\begin{enumerate}
\item \textbf{Associated random variable.} There exists a non-negative random variable $H$ with probability density function $w_P$.
\item \textbf{Survival function.} Let $H$ denote the random variable with density $w_P$, and let $S(h) = \Prob[H \geq h]$ be its survival function.
Then,
\begin{align}
\label{eq:width_survival_function_identity}
S(h) = w_Q(h) - h \cdot w_P(h)
\end{align}
from which
\begin{align}
\label{eq:width_function_wP_wQ_identity}
dw_Q(h) = h \cdot dw_P(h)
\end{align}
in the Lebesgue-Stieltjes sense.
\item \textbf{Canonical representation of $r$.} 
$w_P$ is monotonically decreasing on $\Reals^+$, and hence in particular it admits the generalised inverse function \citep{embrechts2013note}:
\begin{align}
\label{eq:width_function_generalised_inverse}
\eta(p) = w_P^{-1}(p) = \sup\{h \in \supp(w_P) \mid w_P(h) \geq p \} = \lambda(w_P \geq p),
\end{align}
where $\lambda$ denotes the standard Lebesgue measure on $\Reals$.
Then, $\eta$ is monotonically decreasing and defines a valid probability density on $[0, 1]$.
\item \textbf{Relative entropy.}
We have the following three representations for $\KLD{Q}{P}$:
\begin{align*}
\KLD{Q}{P} &= \Exp[\lb H] + \lb(e) \\
&= -\DiffEnt{\eta}
\end{align*}
\end{enumerate}
\end{lemma}
\begin{proof}
The proofs follow by direct calculation.
\par
\textbf{(1)}
The function $w_P: [0, \infty) \to [0, 1]$ is non-negative on its domain since it is a probability. 
It remains to show that $w_P$ integrates to $1$.
To see this, note that
\begin{align*}
\int_0^\infty w_P(h) \, dh = \int_0^\infty \Prob_{Z \sim P}[r(Z) \geq h] \, dh \,\stackrel{\text{(a)}}{=}\, \Exp_{Z \sim P}[r(Z)] = \Exp_{Z \sim Q}[1] = 1,
\end{align*}
where in eq.\ (a), I used the Darth Vader rule \citep{muldowney2012darth}, which states that the expectation of a non-negative random variable equals the integral over its survival function.
\par
\textbf{(2)}
The proof follows by direct calculation:
\begin{align*}
S(h) &= \int_h^\infty w_P(\eta) \, d\eta \\
&= \int_\Omega \int_h^{\infty} \Ind[r(x) \geq \eta] \, d\eta \,dP(x) \\
&= \int_\Omega (r(x) - h) \Ind[r(x) \geq h] \, dP(x) \\
&= w_Q(h) - h w_P(h).
\end{align*}
For the second part, observe that the above equation implies the following equality of integral representations:
\begin{align*}
1 - \int_0^h w_P(\eta) \, d\eta &= \int_0^h \, d\big(w_Q(\eta) - \eta w_P(\eta)\big) \\
&= \int_0^h \, d w_Q(\eta) -  \int_0^h \, d\big(\eta w_P(\eta)\big) \\
&=\int_0^h \, d w_Q(\eta) - \int_0^h \eta\, dw_P(\eta) - \int_0^h w_P(\eta) \, d\eta,
\end{align*}
from which we have
\begin{align*}
1 + \int_0^h \eta\, dw_P(\eta) &= \int_0^h \, d w_Q(\eta),
\end{align*}
for which I will use the shorthand differential notation $h\cdot dw_P(h) = dw_Q(h)$. 
This can be formalised using the integration-by-parts formula for Lebesgue-Stieltjes integrals \citep{hewitt1960integration}.
Note that informally this relation is merely expressing that $r(z)\, dP(z) = dQ(z)$.
\par
\textbf{(3)}
That $w_P$ is monotonically decreasing on its domain follows from its definition as a survival probability.
To show that $\eta$ is a valid density on $[0, 1]$, observe that it is positive since it only takes values in $\supp(w_P) \subseteq \Reals^+$.
To see that it integrates to $1$, note that
\begin{align*}
\int_0^1 \eta(p) \, dp &= \int_0^1 w^{-1}_P(p) \, dp \\
&= \int_{w_P^{-1}(0)}^{w_P^{-1}(1)} w^{-1}_P(w_P(h)) \, dw_P(h) \\
&= -\int_0^\infty h \, dw_P(h) \\
&= \int_0^\infty w_P(h) \, dh \\
&= 1.
\end{align*}
\par
\textbf{(4)}
This follows from direct calculation:
\begin{align}
\KLD{Q}{P} &= \int_\Omega r(x) \lb (r(x)) \, dP(x) \nonumber\\
&= \int_0^\infty \int_\Omega \delta(r(x) - h) h \lb (h) \, dP(x) \, dh \nonumber\\
&= -\int_0^\infty h \lb (h) \, dw_P(h) \label{eq:width_fn_properties_proof_kl_diff_ent_identity}\\
&= \lb(e) + \int_0^\infty \lb (h) \cdot w_P(h)\, dh \nonumber\\
&= \lb(e) + \Exp[\lb (H)].
\end{align}
Similarly,
\begin{align}
-\DiffEnt{\eta} &= \int_0^1 \eta(p) \lb \eta(p) \, dp \nonumber\\
&= \int_0^\infty \int_0^1 \delta(\eta(p) - h) h \lb (h) \, dp \, dh \nonumber\\
&= -\int_0^\infty h \lb (h) \, dw_P(h), \nonumber\\
&= \KLD{Q}{P}, \tag{by \cref{eq:width_fn_properties_proof_kl_diff_ent_identity}}
\end{align}
where the third equality follows since the width functions of $\eta$ and $r$ are the same.
\end{proof}
The critical insight in \cref{lemma:width_function_properties} is that given two distributions $Q \ll P$ with $r = dQ/dP$, we can always create the random variable $r(Z)$ where $Z \sim P$.
Then, the $\KLD{Q}{P}$ can be expressed as the entropy of randomly selecting ``vertical slices'' of $\log r$: by definition, $\KLD{Q}{P} = \Exp_{Z \sim Q}[\log r(Z)]$. 
However, $r(Z)$ induces two new variables: 1) $H$ with PDF $w_P$, which allows us to compute the same information by considering ``horizontal slices'' instead; and 2) $\rho \sim \eta$, which defines a canonical representation of $r$ in the sense that $\KLD{Q}{P} = \KLD{\eta}{\Unif(0, 1)} = -\DiffEnt{\rho}$.
\par
Now, given the relation $\KLD{Q}{P} = -\DiffEnt{\rho}$, a natural question is to ask whether $\DiffEnt{H}$ also leads to a statistical distance.
The answer leads directly to the definition of the channel simulation divergence.
\begin{definition}[Channel Simulation Divergence]
\label{def:channel_simulation_divergence}
Let $Q \ll P$ with $r = dQ/dP$ and $P$-width function $w_P$, and associated random variable $H$.
Then, the channel simulation divergence of $Q$ from $P$ is
\begin{align*}
\CSD{Q}{P} &= -\int_0^\infty w_P(h) \lb w_P(h) \, dh \\
&= \DiffEnt{H}.
\end{align*}
\end{definition}
I had better justify the choice of name for the quantity $\CSD{Q}{P}$.
The rather humorous fact is that despite their prevalence in information theory, there isn't a commonly agreed-upon set of axioms that define a divergence.
However, one might expect a few properties from any divergence: it should be non-negative and vanish when and only when $Q = P$.
% Another property one might expect is that the divergence satisfies a ``data processing inequality,'' i.e.\ applying some (possibly random) function to $Q$ and $P$ does not increase the divergence.
As I show in \cref{lemma:properties_of_csd}, the channel simulation divergence fulfils these properties; thus, I call it a divergence.
To justify the ``channel simulation'' part, I will show in \cref{thm:csd_rec_lower_bound} that $\CSD{Q}{P}$ is intimately linked to the channel simulation problem by showing that it provides a tighter lower bound to the best possible codelength.
To set up for this result, I now collect some useful properties of $\CSD{Q}{P}$.
\begin{lemma}[Properties of the channel simulation divergence]
\label{lemma:properties_of_csd}
Let $Q \ll P$ be probability measures over $\Omega$ with sigma algebra $\sigmaAlgebra$, and $\CSD{Q}{P}$ as in \cref{def:channel_simulation_divergence}.
Furthermore, let $r = dQ/dP$ and $\eta$ be its canonical representation.
Then,
\begin{enumerate}
\item \textbf{Divergence.} $\CSD{Q}{P} \geq 0$, with equality if and only if $Q = P$.
% \item \textbf{Data Processing Inequality.} Let $\subSigmaAlgebra \subseteq \sigmaAlgebra$ be a sub-sigma algebra.
% Then,
% \begin{align*}
% \CSD{Q}{P} \geq \CSD{Q\vert_{\subSigmaAlgebra}}{P\vert_\subSigmaAlgebra}.
% \end{align*}
\item \textbf{Integral Representation.} We have
\begin{align*}
\CSD{Q}{P} &= \lb(e) \cdot \left(-1 + \int_0^1 \frac{1}{v} \int_0^v \eta(p) \, dp\, dv\right).
\end{align*}
\item \textbf{KL bound.} We have
\begin{align}
\label{eq:csd_kl_sandwich_bound}
\KLD{Q}{P} \leq \CSD{Q}{P} \leq \KLD{Q}{P} + \lb(\KLD{Q}{P} + 1) + 1.
\end{align}
\item \textbf{Self-information.} Let $\Omega$ be countable, and for $\rvx \in \Omega$ and $A \in \sigmaAlgebra$, denote the point-mass at $\rvx$ as ${Q_\rvx(A) = \Ind[\rvx \in A]}$.
Then,
\begin{align*}
\CSD{Q_\rvx}{P} = -\lb P(\rvx).    
\end{align*}
\end{enumerate}
\end{lemma}
\begin{proof}
\par
\textbf{(1)}
Let $\varphi(x) = -x\lb (x)$, and observe that it is non-negative on $[0, 1]$.
However, since $w_P$ is a survival probability, its range is a subset of $[0, 1]$.
Hence, $\varphi \circ w_P$ is non-negative on its entire domain.
Since $\CSD{Q}{P}$ is the integral of $\varphi \circ w_P$ over its domain, it is also non-negative by the above.
\par
To see that the divergence vanishes when and only when $Q = P$, recall that since the integrand $\varphi \circ w_P$ is non-negative, the integral vanishes if and only if the integrand is almost surely zero over its support.
Now, since $\varphi(x) = 0 \Leftrightarrow x \in \{0, 1\}$, we see that $w_P$ must only take the values $\{0, 1\}$.
Pairing this with the facts that $w_P$ is monotonically decreasing and integrates to one, we must have $w_P(h) = \Ind[h \in [0, 1]]$, from which we see that $r(Z)$ must be almost surely $1$.
Since $Q \ll P$, this implies that $Q = P$.
\par
\textbf{(2)}
The result follows from direct calculation, namely:
\begin{align*}
\CSD{Q}{P} &= \int_0^\infty \varphi(w_P(h)) \, dh \\
&= \int_{w_P(0)}^{w_P(\infty)} \varphi(p) \, d\eta(p) \\
&= \underbrace{[\varphi(p) \eta(p)]_1^0}_{=0} + \int_0^1 \varphi'(p) \eta(p) \, dp \quad \big(= -1 - \Exp_{\rho \sim \eta}[\lb \rho]\big) \\
&=\underbrace{[\varphi'(p) \int_0^p\eta(q)\,dq]_0^1}_{= -\lb(e)} - \int_0^1 \varphi''(p)\int_0^p \eta(v) \, dv \, dp \\
&= \lb(e) \cdot \left(-1 + \int_0^1 \frac{1}{p} \int_0^p \eta(v) \, dv \, dp\right).
\end{align*}
\par
\textbf{(3)}
Loosely speaking, the lower bound expresses that $\KLD{Q}{P}$ is a first-order approximation of $\CSD{Q}{P}$ in the following sense.
Note, that by \cref{eq:width_survival_function_identity} in \cref{lemma:width_function_properties}, for $S(h) = \Prob[H \geq h]$, we have
\begin{align*}
1 - S(h) = h \cdot w_P(h) + (1 - w_Q(h)) \geq h \cdot w_P(h),
\end{align*}
where the inequality follows since $1 - w_Q(h) = \Prob_{Z \sim Q}[r(Z) < h] \geq 0$.
Then, we have
\begin{align*}
\lb(e) &= -\int_0^1 \lb (u) \, du \\
&= -\int_0^{\infty} w_P(h)\lb (1 - S(h)) \, dh \\
&\leq -\int_0^{\infty} w_P(h)\lb (h \cdot w_P(h)) \, dh \\
&= \CSD{Q}{P} - \Exp[\lb H],
\end{align*}
where in the second equality I used the substitution $u = 1 - S(h)$.
Rearranging the two sides yields the desired lower bound.
\par
For the upper bound, I rely on the following general result, which I prove in \cref{sec:bound_on_pos_bounded_rv}:
\begin{restatable}{lemma}{boundOnPosBoundedRV}
\label{lemma:bound_on_pos_bounded_rv}
Let $\rh$ be a positive, non-atomic random variable supported on $[0, \infty)$ admitting a density $p_\rh$ with respect to Lebesgue measure.
Furthermore, assume that $\norm{p_\rh}_\infty \leq 1$ and $-\lb(e) < \Exp[\lb (\rh)]$.
Then,
\begin{align*}
\DiffEnt{\rh} \leq \Exp[\lb (\rh)] + \lb(\Exp[\lb (\rh)] + 1 + \lb(e)) + 1 + \lb(e).
\end{align*}
\end{restatable}
Note, that by \cref{lemma:width_function_properties} (4), $H$ 
fulfils the criteria of \cref{lemma:bound_on_pos_bounded_rv}, from which we immediately obtain
\begin{align*}
\CSD{Q}{P} = \DiffEnt{H} &\leq \Exp[\lb (H)] + \lb(\Exp[\lb (H)] + 1 + \lb(e)) + 1 + \lb(e) \\
&= \KLD{Q}{P} + \lb(\KLD{Q}{P} + 1) + 1.
\end{align*}
\par
\textbf{(5)}
Follows from direct calculation.
\end{proof}
\subsubsection{Connection to Other Information Measures}
To my knowledge, the channel simulation divergence is not equivalent to any widely-used information measure.
However, it does turn out to be a very special case of a somewhat lesser-known quantity, called the \textit{cumulative residual entropy} \citep{rao2004cumulative}.
The cumulative residual entropy of a random vector $\rvx \in \Reals^n$ is defined as
\begin{align*}
\CumResEnt{\rvx} = -\int_{\Reals_+^n} \Prob(\abs{\rvx} > \lambda) \log \Prob(\abs{\rvx} > \lambda) \, d\lambda,
\end{align*}
where $\abs{\rvx}$ denotes taking absolute value elementwise.
Then, for $Q \ll P$ with $r = dQ/dP$, putting $R = dQ/dP(Y)$ with $Y \sim P$, we find that\footnote{I thank Yanxiao Liu for bringing this connection to my attention.}
\begin{align*}
\CSD{Q}{P} = \CumResEnt{R}
\end{align*}
However, despite this special relation, the two quantities have quite different applications.
On the one hand, the cumulative residual entropy was conceived as an alternative to Shannon entropy and differential entropy as a measure of uncertainty.
On the other hand, as we shall soon see, the main purpose of the channel simulation divergence is to provide a tighter characterisation of the communication complexity of relative entropy coding.
As such, I do not explore the above connection any further in this thesis.
\subsubsection{Computing the Channel Simulation Divergence}
Unfortunately, the channel simulation divergence is not as simple to evaluate as the relative entropy, even for reasonably simple densities, such as the Gaussian.
The difficulty is that even for relatively simple density ratios, the width function $w_P(h)$ can already be a non-elementary function.
Then, computing the differential entropy of the width function appears to be fairly hopeless; it is an interesting open question whether it admits an accurate approximation that is easier to evaluate.
Nonetheless, once we obtain the width function $w_P$, we can compute the differential entropy accurately via numerical integration.
In this section, I present some explicit cases where the width function and sometimes even the channel simulation divergence are available in analytic form.
I have computed the examples I present here in my papers \citet{flamich2023adaptive, flamich2023gprs} and \citet{goc2024channel}.
\par
\textbf{One-dimensional Laplace-Laplace case.}
To illustrate the difficulty, let me begin with a simple example: the one-dimensional Laplace distribution.
For a location parameter $m$ and scale parameter $s > 0$, it admits a density with respect to the Lebesgue measure:
\begin{align*}
\Laplace(x \mid m, s) = \frac{1}{2 s}\exp\left(-\frac{\abs{x - m}}{s}\right)
\end{align*}
Then, setting $Q \gets \Laplace(m, s)$ and $P \gets \Laplace(0, 1)$, we can immediately compute the density ratio $r = dQ/dP$ to be
\begin{align*}
r(x) &= \frac{1}{s}\exp\left(-\frac{\abs{x - m}}{s} + \abs{x}\right).
\end{align*}
However, the width function, assuming $0 < s < 1$ to guarantee that the density ratio is bounded, is quite involved even in this simple case \citep[Appendix G.5;][]{flamich2023gprs}:
\begin{align*}
w_P(h) &= 
\begin{cases}
1 - \exp\left(\frac{s}{1 - s} \ln(s \cdot h)\right) \cdot \cosh\left(\frac{m}{1 - s}\right) &\text{if } \ln h \leq - \frac{\abs{\mu}}{s} - \ln s \\
\exp\left(\frac{s^2}{1 - s^2 \cdot (\ln(s \cdot h) - \abs{m})} \right) \cdot \sinh\left(-\frac{s \cdot (\ln(s \cdot h) - \abs{m})}{1 - s^2 }  \right) &\text{if } \ln h > - \frac{\abs{\mu}}{s} - \ln s.
\end{cases}
\end{align*}
Sadly, there does not appear to be an analytic formula for $\CSD{Q}{P}$ at this level of generality.
However, for the special case when $m = 0$, we obtain \citep{goc2024channel}:
\begin{align*}
\CSD{\Laplace(0, s)}{\Laplace(0, 1)} \cdot \ln 2 &= s + \psi(1/s) + \gamma - 1,
\end{align*}
where $\gamma \approx 0.577$ is the Euler-Mascheroni constant and $\psi(x) = \Gamma'(x)/\Gamma(x)$ is the digamma function.
Comparing this with
\begin{align*}
\KLD{\Laplace(0, s)}{\Laplace(0, 1)} \cdot \ln 2 &= s +\ln(1 / s) - 1
\end{align*}
and using the fact that $-x \leq \psi(1/x) + \ln x \leq -x/2$ for $x > 0$ from eq.\ 2.2 in \citep{alzer1997some}, we find that
\begin{align*}
\gamma - s \leq (\CSD{\Laplace(0, s)}{\Laplace(0, 1)} - \KLD{\Laplace(0, s)}{\Laplace(0, 1)}) \cdot \ln 2 \leq \gamma - s/2
\end{align*}
Thus, we see that $\CSD{\Laplace(0, s)}{\Laplace(0, 1)} \to \KLD{\Laplace(0, s)}{\Laplace(0, 1)}) + \gamma \lb e$ as $s \to 0$, which I verify numerically in panel (A) of \cref{fig:csd_numerical_examples}.
\par
\textbf{Multivariate Isotropic Gaussian Case.}
A case that is significantly more practically relevant is the $n$-dimensional isotropic Gaussian case, when ${Q \gets \Normal(\mu, \sigma^2 I)}$ and ${P \gets \Normal(0, I)}$.
From this, we can compute the density ratio $r = dQ/dP$ to be \citep[Appendix G.4;][]{flamich2023gprs}:
\begin{gather*}
r(x) = \sigma^{-n} \exp\left(-\frac{\norm{x - m}^2}{2 s^2} + \frac{\norm{\mu}^2}{2(1-\sigma^2)}\right) \\
m = \frac{\mu}{1 - \sigma^2},\quad \quad s^2 = \frac{\sigma^2}{1 - \sigma^2},
\end{gather*}
From which the width function can be computed to be
\begin{align*}
w_P(h) &= \Prob\left[\chi^2_n(\norm{m}^2) \leq s^2\left(-2\lb (h) - n \lb (\sigma^2) + \frac{\norm{\mu}^2}{1 - \sigma^2}\right)\right],
\end{align*}
where $\chi^2_n(\lambda)$ is a noncentral chi-squared random variable with $n$ degrees of freedom and noncentrality parameter $\lambda$.
Unfortunately, there does not appear to be a closed-form solution for $\CSD{Q}{P}$ for any non-trivial setting of the parameters.
Nonetheless, we can compute it to arbitrary precision using numerical integration. 
\subsection[A Tighter Fundamental Lower Bound for the Communication Complexity]{A Tighter Fundamental Lower Bound for \\ the Communication Complexity}
\par
Armed with the definition of the channel simulation divergence and having examined its properties in \cref{sec:channel_simulation_divergence}, I now return to relative entropy coding.
In this section, I prove that the channel simulation divergence is a lower bound to the coding efficiency of any channel simulation / relative entropy coding algorithm.
\par
A weaker version of the following theorem appears as Proposition 1 in \citet{li2018strong} for channels $\rvx \to \rvy$ when $\rvy$ is discrete.%
\footnote{In fact, the lower bound that appears in Proposition 1 of \citet{li2018strong} seems more complicated at first sight. 
However, I realised that this quantity may be rewritten as the expected channel simulation divergence, which I showed in \citet{goc2024channel}; see Appendix B in the \texttt{arXiv} version of the paper.
This observation motivated \cref{thm:csd_rec_lower_bound}.}
Unfortunately, several steps in the proof use the discreteness assumption, which is the primary factor that complicates its extension to arbitrary $\rvy$.
\par
The high-level proof strategy in \cref{thm:csd_rec_lower_bound} is to fix some arbitrary common randomness $\rvz$ with $\rvz \perp \rvx$ such that $\rvy = g(\rvx, \rvz)$, as is usual in the channel simulation setup.
Then, I will use the conditional entropy $\Ent{\rvy \mid \rvz}$ as a proxy for the shortest achievable expected codelength for encoding $\rvy$ using $\rvz$; note that I use the Shannon entropy here rather than the differential entropy.
That there are cases where the conditional entropy can be finite even when $\rvy$ is continuous is \textit{the} non-trivial property of channel simulation.
Indeed, the rejection sampling scheme I construct in \cref{sec:rec_through_examples} already shows that this lower bound is non-trivial in many cases, in the sense that there exists some $\rvz$ for which $\Ent{\rvy \mid \rvz} < \infty$.
However, in \cref{chapter:rec_with_pp,chapter:branch_and_bound}, I construct  common randomness $\rvz$ and a coupling $g$ for which the conditional entropy is finite whenever $\KLD{Q}{P} < \infty$.
\par
\textbf{Definition of the conditional entropy.}
\textit{What is the definition} of the conditional entropy in this case?
When $\rvy$ is discrete with support $\YSpace$ with conditional probability mass function $p(y \mid z)$, the expectations over $\rvy$ and $\rvz$ in the definition are exchangeable, which I denote as identity (*) below:
\begin{align}
\label{eq:cond_entropy_discrete_exchangeability}
\Ent{\rvy \mid \rvz} = -\Exp_{\rvz}\left[\sum_{y \in \YSpace} p(y \mid \rvz) \lb p(y \mid \rvz) \right] \stackrel{(*)}{=} -\sum_{y \in \YSpace} \Exp_{\rvz}\left[p(y \mid \rvz) \lb p(y \mid \rvz) \right]
\end{align}
However, the difficulty is that when $\rvy$ is not discrete, each realisation of the common randomness $\rvz$ might select a different support for the conditional $\rvy \mid \rvz$, and so the identity (*) does not hold; in fact, the right-most expression does not even make sense.
Overcoming this issue is the main technical difficulty in extending the proof of \citet{li2018strong}. 
Therefore, I now describe the mathematical machinery necessary to handle such changing supports.
However, while I will give as much intuition for the construction as possible, full formal development of the necessary theory is well beyond the scope of this thesis; the interested reader should consult \citet{daley2007introduction} for a comprehensive introduction. 
\par
To reiterate, the difficulty is that a generalised definition of \cref{eq:cond_entropy_discrete_exchangeability} requires some language to talk about averaging probability masses with different supports.
The big idea, then, is instead of considering averaging over the conditional distribution $p(\rvy \mid \rvz)$, I define a random probability measure $\pi_\rvz(\rvy)$ that is equivalent to $p(\rvy \mid \rvz)$ when it exists, but makes sense when the usual definition of the conditional does not.
Then, I will directly average over the random measure $\pi_\rvz$ instead of $\rvz$.
To achieve this, I first need to be able to talk about the ``set of probability masses over $\rvy$.'' 
More precisely, let $\YSpace$ denote the Polish space in which $\rvy$ takes values.
Then, define
\begin{align*}
\PMSpace_\YSpace = \{\pi \in \mathcal{M}_\YSpace^{\#} \mid \pi(\YSpace) = 1, \pi \text{ purely atomic} \},
\end{align*}
where $\mathcal{M}_\YSpace^{\#}$ denotes the space of boundedly finite measures over $\YSpace$ \citep{daley2007introduction}.
The precise definition of $\mathcal{M}_\YSpace^{\#}$ is unimportant for my purposes. 
The only thing I need is that it contains $\PMSpace_\YSpace$, the set of all probability masses over $\YSpace$ as a subset. Thus, I can utilise all the powerful results that hold in general for boundedly finite measures.
Now, for a pair of dependent random variables $\rvy, \rvz \sim P_{\rvy, \rvz}$, the idea is to push $\rvz$ forward to a probability measure $\Palm^\rvz$ over $\PMSpace_\YSpace$.
Furthermore, $\Palm^\rvz$ should satisfy marginalisation constraint: in analogy to the fact that for a set $A \subset \YSpace$, we have $P_\rvy(A) = \Exp_\rvz[P_{\rvy \mid \rvz}(A)]$, we also require 
\begin{align*}
P_\rvy(A) = \int_{\PMSpace_\YSpace} \pi_\rvz(A) \, d\Palm^\rvz(\pi_\rvz).
\end{align*}
Now, I am finally ready to define the generalised definition of the conditional entropy:
\begin{definition}[Conditional Shannon Entropy for Continuous Random Variables]
\label{def:continuous_conditional_shannon_entropy}
Let $\rvy, \rvz \sim P_{\rvy, \rvz}$, and let $\PMSpace_\YSpace$ be as above.
Then, assuming $\rvz$ can be pushed forward to a probability measure $\Palm^\rvz$ over $\PMSpace_\YSpace$ such that $P_\rvy(A) = \int_{\PMSpace_\YSpace} \pi_\rvz(A) \, d\Palm^\rvz(\pi_\rvz)$, we define the Shannon entropy of $\rvy$ conditioned on $\rvz$ as
\begin{align*}
\Ent{\rvy \mid \rvz} = \int_{\PMSpace_\YSpace} \int_\YSpace - \lb \pi_\rvz(\rvy) \, d\pi_\rvz(\rvy) \, d\Palm^\rvz(\pi_\rvz).
\end{align*}   
\end{definition}
Note that when $\rvy$ is discrete, $\pi_\rvz(\rvy)$ is equal to the standard conditional ${p(\rvy \mid \rvz)}$, and the outer average can be taken over $\rvz$ instead of $\Palm^\rvz$; thus, \cref{def:continuous_conditional_shannon_entropy} becomes equivalent to the usual definition.
However, the power of \cref{def:continuous_conditional_shannon_entropy} is that it allows us to generalise \cref{eq:cond_entropy_discrete_exchangeability}.
The generalisation will follow how the joint probability of two random variables can be written in two different ways: ``${p(\rvy \mid \rvz)p(\rvz) = p(\rvz \mid \rvy)p(\rvy)}$.''
\Cref{def:continuous_conditional_shannon_entropy} follows the ``$p(\rvy \mid \rvz)p(\rvz)$ factorisation''; now I consider the ``$p(\rvz \mid \rvy)p(\rvy)$ factorisation.''
We already have a candidate for $p(\rvy)$: it should be $dP_\rvy$.
But what should be the analogue of $p(\rvz \mid \rvy)$?
The answer is given by Propositions 13.1.IV and 13.1.V of \citet{daley2007introduction}:
there exists a ``conditional distribution'' $\Palm_\rvy^\rvz$ over probability masses $\pi_\rvz$, that is ``conditioned on the event that $\rvy \in \supp( \pi_\rvz)$.''
Concretely, Proposition 13.1.IV guarantees that there exists a family of measures $\{\Palm_y^\rvz \mid y \in \YSpace\}$, called the \textit{local Palm distributions}, such that
\begin{align*}
\Ent{\rvy \mid \rvz} = \int_{\PMSpace_\YSpace} \int_\YSpace - \lb \pi_\rvz(\rvy) \, d\pi_\rvz(\rvy) \, d\Palm^\rvz(\pi_\rvz) = \int_\YSpace \int_{\PMSpace_\YSpace} - \lb \pi_\rvz(\rvy) \, d\Palm^\rvz_y(\pi_\rvz) \, dP_\rvy(y),
\end{align*}
and Proposition 13.1.V guarantees that
\begin{align}
\forall y \in \YSpace: \quad \supp \Palm_\rvy^\rvz = \{\pi_\rvz \in \supp \Palm^\rvz \mid \pi_\rvz(y) > 0 \}, \label{eq:local_palm_kernel_support_guarantee}
\end{align}
which finally yields the analogue of \cref{eq:cond_entropy_discrete_exchangeability}.
With this machinery in place, I am now ready to prove the first big result of the thesis.
\begin{importantTheorem}[Lower bound on the efficiency of relative entropy coding.]
\label{thm:csd_rec_lower_bound}
Let $\rvx, \rvy \sim P_{\rvx, \rvy}$.
Then, for any common randomness $\rvz$ such that $\rvz \perp \rvx$ and $\rvy = g(\rvx, \rvz)$, we have
\begin{align*}
\Exp_{\rvy}[\CSD{P_{\rvx \mid \rvy}}{P_{\rvx}}] \leq \Ent{\rvy \mid \rvz}. 
\end{align*}
\end{importantTheorem}
\begin{proof}
I begin by verifying that the requirements of \cref{def:continuous_conditional_shannon_entropy} apply.
To this end, assume the common randomness $\rvz$ takes values in $\ZSpace$ and has distribution $P_\rvz$.
Let $g_z(x) = g(x, z)$ and set $P_{\rvy \mid \rvz = z} = g_z \pushfwd P_\rvx$.
Next, consider the map $f: \ZSpace \to \PMSpace_\YSpace$ defined as $f(z) = P_{\rvy \mid \rvz = z}$ and set $\Palm^\rvz = f \pushfwd P_\rvz$.
Then, by construction and the law of the unconscious statistician, we have the required marginalisation constraint $P_\rvy(A) = \int_{\PMSpace_\YSpace} \pi_\rvz(A) \, d\Palm^\rvz(\pi_\rvz)$.
Next, for this proof only, I establish the following notation for $y \in \YSpace, \rvx \sim P_\rvx$ and an arbitrary measurable function $\gamma(y, \pi_\rvz)$:
\begin{align*}
\rho_y &= \frac{dP_{\rvx \mid \rvy}}{dP_\rvx}(\rvx \mid y) \\
w_y(h) &= \Prob[\rho_y \geq h] \\
\Exp_{\pi_\rvz \mid y}[\gamma(y, \pi_\rvz)]  &= \int_{\PMSpace_\YSpace} \gamma(y, \pi_\rvz) \, d\Palm^\rvz_y(\pi_\rvz).
\end{align*}
Observe that with this notation in place, we have
\begin{align*}
\Ent{\rvy \mid \rvz} = \Exp_\rvy [\Exp_{\pi_\rvz \mid \rvy}[-\lb \pi_\rvz(\rvy)]]
\end{align*}
Therefore, to prove the statement of the theorem, I will show the following pointwise inequality:
\begin{align*}
\forall y \in \YSpace: \quad \Exp_{\pi_\rvz \mid y}[-\lb \pi_\rvz(y)] \geq \CSD{P_{\rvx \mid \rvy = y}}{P_\rvx}.
\end{align*}
This inequality will follow from three steps:
\begin{align}
\Exp_{\pi_\rvz \mid y}[-\lb \pi_\rvz(y)] 
&= \frac{\Exp_{\pi_\rvz \mid y}[ \pi_\rvz(y)] \cdot \Exp_{\pi_\rvz \mid y}[-\lb \pi_\rvz(y)]}{\Exp_{\pi_\rvz \mid y}[ \pi_\rvz(y)]} \label{eq:proof_csd_description_length_lower_bound_multiply_by_one}\\
&\geq \frac{\Exp_{\pi_\rvz \mid y}[ -\pi_\rvz(y) \cdot \lb \pi_\rvz(y)]}{\Exp_{\pi_\rvz \mid y}[ \pi_\rvz(y)]} \label{eq:proof_csd_description_length_lower_bound_covariance_identity} \\
&\geq \frac{\Exp_{\pi_\rvz \mid y}[ \pi_\rvz(y)] \cdot \CSD{P_{\rvx \mid \rvy = y}}{P_\rvx}}{\Exp_{\pi_\rvz \mid y}[ \pi_\rvz(y)]}\label{eq:proof_csd_description_length_lower_bound_csd_identity} \\
&= \CSD{P_{\rvx \mid \rvy = y}}{P_\rvx}. \nonumber
\end{align}
\Cref{eq:proof_csd_description_length_lower_bound_multiply_by_one} follows from multiplying the left-hand side by $1$.
Note, that this multiplication is never troublesome, since $\Exp_{\pi_\rvz \mid y}[ \pi_\rvz(y)] > 0$ almost surely by \cref{eq:local_palm_kernel_support_guarantee}.
Next, note that \cref{eq:proof_csd_description_length_lower_bound_covariance_identity} is equivalent to the statement that $\Cov[\pi_\rvz(y), \lb \pi_\rvz(y)] \geq 0$.
This can be shown to be true by the following direct calculation.%
\footnote{This result and this result only in the proof was first shown by Henry Wilson, with whom I investigated this problem in the summer of 2024.}
First, let $\mu_y = \Exp_{\pi_\rvz \mid y}[ \pi_\rvz(y)]$, then:
\begin{align*}
\Cov[\pi_\rvz(y), \lb \pi_\rvz(y)] 
&= \Exp_{\pi_\rvz \mid y}[ \pi_\rvz(y) \cdot\lb \pi_\rvz(y)] - \mu_y \cdot \Exp_{\pi_\rvz \mid y}[\lb \pi_\rvz(y)] \\
&= \Exp_{\pi_\rvz \mid y}[ (\pi_\rvz(y) - \mu_y) \cdot\lb \pi_\rvz(y)] \\
&= \Exp_{\pi_\rvz \mid y}[ (\pi_\rvz(y) - \mu_y) \cdot\lb \pi_\rvz(y)] - \underbrace{\Exp_{\pi_\rvz \mid y}[ (\pi_\rvz(y) - \mu_y) ]}_{ = 0}\cdot \lb \mu_y \\
&= \Exp_{\pi_\rvz \mid y}[ (\pi_\rvz(y) - \mu_y) \cdot(\lb \pi_\rvz(y) - \lb \mu_y)] \\
&\geq 0,
\end{align*}
where the inequality follows, since by the monotonicity of $\lb$ the sign of $\alpha - \beta$ always matches the sign of $\lb (\alpha) - \lb (\beta)$ for $\alpha, \beta > 0$.
\par
It remains to show \cref{eq:proof_csd_description_length_lower_bound_csd_identity}, which will follow from a second-order stochastic dominance type argument, similar to, but somewhat more technical than the one given in the proof of Proposition 1 by \citet{li2018strong}.
To begin, let $\varphi(x) = -x \lb (x)$.
Then, we have
\begin{align}
\Exp_{\pi_\rvz \mid y}[ -\pi_\rvz(y) \cdot\lb \pi_\rvz(y)] 
&= \int_{\PMSpace_\YSpace} \varphi(\pi_\rvz(y)) \, d\Palm^\rvz_y(\pi_\rvz) \nonumber \\
&= \int_0^1  \varphi(u) \, d\Prob_{\pi_\rvz \mid y}[\pi_\rvz(y) \leq u] \nonumber\\
&= \varphi(1) - \int_0^1  \varphi'(u) \cdot \Prob_{\pi_\rvz \mid y}[\pi_\rvz(y) \leq u] \, du \nonumber\\
&= -\varphi'(1)(1 - \mu_y) + \int_0^1  \varphi''(u) \cdot \int_0^u \Prob_{\pi_\rvz \mid y}[\pi_\rvz(y) \leq v] \, dv \, du \nonumber \\
&= \lb(e) \cdot(1 - \mu_y) + \int_0^1  \varphi''(u) \cdot \int_0^u \Prob_{\pi_\rvz \mid y}[\pi_\rvz(y) \leq v] \, dv \, du, \label{eq:proof_csd_description_length_lower_bound_cond_expectation_integral_representation}
\end{align}
where I used integration by parts twice, used the Darth Vader rule \citep{muldowney2012darth} to obtain $\int_0^1 \Prob_{\pi_\rvz \mid y}[\pi_\rvz(y) \leq v] \, dv = 1 - \mu_y$ and used the facts that $\varphi(1) = 0$ and $\varphi'(1) = -\lb(e)$. 
Next, I will show the following second-order stochastic dominance type result:
\begin{align}
\int_0^u \Prob_{\pi_\rvz \mid y}[\pi_\rvz(y) \leq v] \, dv 
\leq u - \mu_y \int_0^u w_y^{-1}(p) \, dp,
\label{eq:proof_csd_description_length_lower_bound_second_order_stochastic_dominance}
\end{align}
where $w_y^{-1}$ is the generalised inverse of $w_y$ in the sense of \cref{eq:width_function_generalised_inverse}.
Then, using that $\varphi''(u) = -\lb(e)/x$ yields
\begin{align*}
\Exp_{\pi_\rvz \mid y}[ -\pi_\rvz(y) \cdot\lb \pi_\rvz(y)] 
&\stackrel{\text{\cref{eq:proof_csd_description_length_lower_bound_cond_expectation_integral_representation}}}{=} 
\lb(e) \cdot \left(1 - \mu_y - \int_0^1  \frac{1}{u} \cdot \int_0^u \Prob_{\pi_\rvz \mid y}[\pi_\rvz(y) \leq v] \, dv \, du \right)\\
&\stackrel{\text{\cref{eq:proof_csd_description_length_lower_bound_second_order_stochastic_dominance}}}{\geq}
\lb(e) \cdot \left( 1 - \mu_y  - \int_0^1 \frac{1}{u} \cdot \left(u - \mu_y \int_0^u w_y^{-1}(p) \, dp \right) \, du \right) \\
&\stackrel{\hphantom{\text{\cref{eq:proof_csd_description_length_lower_bound_second_order_stochastic_dominance}}}}{=}
\lb(e) \cdot \mu_y \left(-1 + \int_0^1 \frac{1}{u} \int_0^u w_y^{-1}(p) \, dp \, du \right) \\
&\stackrel{\hphantom{\text{\cref{eq:proof_csd_description_length_lower_bound_second_order_stochastic_dominance}}}}{=} \mu_y \cdot \CSD{P_{\rvx \mid \rvy}}{P_\rvx},
\end{align*}
where the last equality follows from the integral representation of the channel simulation divergence from \cref{lemma:properties_of_csd}.3, which shows \cref{eq:proof_csd_description_length_lower_bound_csd_identity} and will finish the proof.
To show that \cref{eq:proof_csd_description_length_lower_bound_second_order_stochastic_dominance} holds, let me first introduce the shorthand $(x)_+ = \max\{0, x\}$.
What follows is a combination of applying the tower rule and Fubini's theorem with one application of Jensen's inequality. 
\begin{align}
\int_0^u \Prob_{\pi_\rvz \mid y}[\pi_\rvz(y) \leq v] \, dv 
&= \Exp_{\pi_\rvz \mid y}[(u - \pi_\rvz(y))_+] \label{eq:proof_csd_description_length_lower_bound_cdf_expectation_identity} \\
&\leq \Exp_{\pi_\rvz \mid y}[(w_y(w_y^{-1}(u)) - \Exp_{\rvx}[\pi_\rvz(y \mid \rvx)])_+] \label{eq:proof_csd_description_length_lower_bound_generalised_inverse_inequality} \\
&= \Exp_{\pi_\rvz \mid y}\left[\left(\Exp_{\rho_y}\big[\Ind[\rho_y \geq w_y^{-1}(u)] - \Exp_{\rvx}[\pi_\rvz(y \mid \rvx) \mid \rho_y] \big]\right)_+\right] \nonumber \\
&\leq \Exp_{\pi_\rvz \mid y}\left[\Exp_{\rho_y}\big[\left(\Ind[\rho_y \geq w_y^{-1}(u)] - \Exp_{\rvx}[\pi_\rvz(y \mid \rvx) \mid \rho_y]\right)_+ \big]\right] \tag{Jensen} \\
&= \Exp_{\pi_\rvz \mid y}\left[\Exp_{\rho_y}\big[\Ind[\rho_y \geq w_y^{-1}(u)] \cdot (1 -  \Exp_{\rvx}[\pi_\rvz(y \mid \rvx) \mid \rho_y]) \big]\right]
\label{eq:proof_csd_description_length_lower_bound_hockeystick_identity} \\
&= \Exp_{\rho_y}\big[\Ind[\rho_y \geq w_y^{-1}(u)] \cdot (1 -  \Exp_{\rvx}[\Exp_{\pi_\rvz \mid y}[\pi_\rvz(y \mid \rvx) \mid \rho_y]]) \big] \tag{Fubini}\\
&= \Exp_{\rho_y}\left[\Ind[\rho_y \geq w_y^{-1}(u)] \cdot\left(\!1 - \Exp_{\rvx}\left[\frac{dP_{\rvx \mid \rvy}}{dP_\rvx}(\rvx \mid y) \cdot \Exp_{\pi_\rvz \mid y}[\pi_\rvz(y) ] \,\Big\vert\, \rho_y \right]\!\right) \right] \nonumber\\
&= \Exp_{\rho_y}\left[\Ind[\rho_y \geq w_y^{-1}(u)] \cdot \left(1 - \mu_y \cdot \Exp_{\rvx}\left[\frac{dP_{\rvx \mid \rvy}}{dP_\rvx}(\rvx \mid y) \,\Big\vert\, \rho_y\right] \right) \right] \nonumber\\
&= \Exp_{\rho_y}\left[\Ind[\rho_y \geq w_y^{-1}(u)] \cdot \left(1 - \mu_y \cdot \rho_y \right) \right] \nonumber\\
&= u - \mu_y \cdot \int_0^u w_y^{-1}(v) \, dv.
\label{eq:proof_csd_description_length_lower_bound_density_ratio_inverse_survival}
\end{align}
In the above, \cref{eq:proof_csd_description_length_lower_bound_cdf_expectation_identity} holds, since for any random variable $V\sim P_V$ supported on $[0, 1]$ we have
\begin{align*}
\,\int_0^u \Prob[V \leq t] \, dt 
&= \int_0^1 \int_0^1 \Ind[v \leq t] \Ind[t \geq u]  \, dt \, dP_V(v) \\
&= \int_0^1 (u - v)_+ \, dP_V(v) \\
&= \Exp_V[(u - V)_+],
\end{align*}
\cref{eq:proof_csd_description_length_lower_bound_generalised_inverse_inequality} holds because \begin{align*}
w_y(w_y^{-1}(u)) = w_y(\sup\{h \in \supp (w_y) \mid w_y(h) \geq u\}) \geq u,
\end{align*}
\cref{eq:proof_csd_description_length_lower_bound_hockeystick_identity} holds because setting $\pi = \Exp_\rvx[\pi_\rvz(y \mid \rvx) \mid \rho_y]$ and noting that $\pi \in (0, 1]$, we have
\begin{align*}
\Exp_{\rho_y}[(\Ind[\rho_y \geq w_y^{-1}(u)] - \pi)_+] &= \Exp_{\rho_y}\left[(\Ind[\rho_y \geq w_y^{-1}(u)] - \pi)\cdot\Ind[\Ind[\rho_y \geq w_y^{-1}(u)] \geq \pi]\right] \\
&= \Exp_{\rho_y}\left[(\Ind[\rho_y \geq w_y^{-1}(u)] - \pi)\cdot \Ind[\rho_y \geq w_y^{-1}(u)]\right] \\
&= \Exp_{\rho_y}\left[(1 - \pi)\cdot \Ind[\rho_y \geq w_y^{-1}(u)]\right],
\end{align*}
and finally \cref{eq:proof_csd_description_length_lower_bound_density_ratio_inverse_survival} holds because
\begin{align*}
\Exp_{\rho_y}\left[\Ind[\rho_y \geq w_y^{-1}(u)] \cdot \left(1 - \mu_y \cdot \rho_y \right) \right] &= -\int_0^\infty \Ind[\rho_y \geq w_y^{-1}(u)] \cdot \left(1 - \mu_y \cdot \rho_y \right) \, dw_y(\rho_y) \\
&= \int_0^1 \Ind[w_y^{-1}(v) \geq w_y^{-1}(u)] \cdot \left(1 - \mu_y \cdot w_y^{-1}(v) \right) \, dv \\
&= \int_0^u \left(1 - \mu_y \cdot w_y^{-1}(v)\right) \, dv,
\end{align*}
which finishes the proof.
\end{proof}
\subsubsection{Computing the Channel Simulation Divergence}
\begin{figure}[t]
\centering
\ref{legend:experiments}
\\
\includegraphics{3-FundamentalLimits/img/experiments.tikz}
\vspace{-0.6cm}
\caption[Numerical demonstration of the looseness of the channel simulation bound]{Numerical demonstration of the looseness of the channel simulation bound in \cref{eq:csd_kl_sandwich_bound}.
\textbf{(A)} We plot $\Delta(Q, P)$ for $Q = \Laplace(0, b)$ and $P = \Laplace(0, 1)$ as a function of the target log-precision $-\ln b$.
\textbf{(B)} We plot $\Delta(Q, P)$ for $Q = \Normal(1, 1/4)^{\otimes d}$ and $P = \Normal(0, 1)^{\otimes d}$ as a function of the dimension $d$.
}
\label{fig:csd_numerical_examples}
\end{figure}

\par
In \cref{thm:channel_simulation_mi_lower_bound}, I showed that $\MI{\rvx}{\rvy}$ bounds from below the best possible description length any relative entropy coding algorithm can achieve for the channel $\rvx \to \rvy$.
The practical utility of \cref{thm:csd_rec_lower_bound} is that by \cref{eq:csd_kl_sandwich_bound}, we can obtain a more precise estimate of the best possible performance by computing the channel simulation divergence instead.
But how much tighter is it exactly?
To simplify the problem, I compute the coding efficiency for the one-shot case (with $Q$ and $P$ fixed) for the one-dimensional Laplace and the isotropic multivariate Gaussian cases in \cref{fig:csd_numerical_examples}.  Concretely, for two distributions $Q$ and $P$, let me denote the difference between the bounds by 
\begin{align*}
\Delta(Q, P) = \CSD{Q}{P} - \KLD{Q}{P},
\end{align*}
this is what I plot in \cref{fig:csd_numerical_examples}.
We see interesting behaviour in both cases: the difference remains bounded for the one-dimensional Laplace case; in fact, it monotonically tends to a constant.
However, more surprisingly, in the isotropic multivariate Gaussian case, we see some strange behaviour: the overhead seems to scale as $\frac{1}{2}\lb (n)$.
In \cref{sec:rec_asymptotic_behaviour}, I reveal that the reason for this surprising phenomenon is the concentration of measure.
However, before moving on, let me also derive a connection between $\Delta(Q, P)$ and the associated random variable $H$.
\begin{lemma}
\label{lemma:delta_logh_diffent_identity}
Let $Q \ll P$ be probability distributions with $r = dQ/dP$ and let $H$ be the associated random variable.
\deleted[id={CM}, comment={Moved definition of differential entropy to \cref{sec:notation}}]{And for a real-valued random variable $X$ with probability density $f$, let
% \begin{align*}
% \DiffEnt{X} = \int_\Reals f(x) \ln f(x) \, dx
% \end{align*}
denote its differential entropy in nats.}%
Then,
\begin{align*}
\Delta(Q, P) \cdot \ln 2 = \DiffEnt{\ln H} \cdot \ln 2  - 1.
\end{align*}
\end{lemma}
\begin{proof}
First, note that by \cref{lemma:width_function_properties} and the definition of the channel simulation divergence, we have
\begin{align*}
\KLD{Q}{P} \cdot \ln 2 &= \Exp[\ln H] + 1 \\
\CSD{Q}{P} \cdot \ln 2 &= \DiffEnt{H} \cdot \ln 2 = -\Exp[\ln w_P(H)]
\end{align*}
Therefore,
\begin{align*}
\Delta(Q, P) \cdot \ln 2 &= (\CSD{Q}{P} - \KLD{Q}{P}) \cdot \ln 2 \\
&= -\int_0^\infty w_P(h) \ln(h \cdot w_P(h)) \, dh - 1 \\
&= - \int_{-\infty}^\infty e^{t}\cdot w_P(e^t) \ln(e^t \cdot w_P(e^t)) \, dt - 1\\
&= \DiffEnt{\ln H} \cdot \ln 2 - 1,
\end{align*}
where the third equality follows via the substitution $t \gets \ln h$ and the last equality follows since for the density $\ln H$ we have
\begin{align*}
\frac{d}{dt}\Prob[\ln H \leq t] = \frac{d}{dt} \Prob[H \leq e^t] = w_P(e^t) \cdot e^t,
\end{align*}
which finishes the proof.
\end{proof}
\subsection{The Asymptotic Behaviour of Relative Entropy Coding}
\label{sec:rec_asymptotic_behaviour}
In this section, I show that the behaviour of the scaling law in \cref{fig:csd_numerical_examples} (B) is no coincidence: 
It arises as a consequence of a concentration of measure phenomenon.
Essentially, this behaviour has to do with the \textit{predictability of the length} of the code produced by a channel simulation algorithm.
The simplest form of predictability is when we have the following deterministic relationship:
\begin{definition}[Singular channel]
Let $\rvx, \rvy \sim P_{\rvx, \rvy}$ be dependent random variables.
Then, the channel $\rvx \to \rvy$ is \textit{singular} if there is a $P_\rvy$-measurable function $g$, such that   
\begin{align*}
\frac{dP_{\rvx \mid \rvy}}{dP_{\rvx}}(x \mid y) = g(y) \quad \text{for } P_{\rvx, \rvy}\text{-almost every } (x, y)
\end{align*}
Otherwise, the channel $\rvx \to \rvy$ is non-singular.
\end{definition}
In other words, the channel $\rvx \to \rvy$ is singular if the \textit{reverse channel} $\rvy \to \rvx$ is a truncated version of $P_\rvy$.
Recently, \citet{sriramu2024optimal} have shown that singular channels behave like \cref{fig:csd_numerical_examples} (A): as the dimensionality $n$ increases, $\Delta(Q^{\otimes n}, P^{\otimes n})$ remains bounded.
In what follows, I now tackle the non-singular case.
\newpage
\begin{importantTheorem}[Second-order behaviour of asymptotic relative entropy coding for non-singular channels.]
\label{thm:second_order_behaviour_of_csd}
Let $\rvx, \rvy \sim P_{\rvx, \rvy}$ be dependent random variables.
Assume that the channel $\rvx \to \rvy$ is non-singular and let $r_{\rvy} = dP_{\rvx \mid \rvy}/dP_\rvx$.
For all $n \geq 1$, let $\rvx^n, \rvy^n \sim P_{\rvx, \rvy}^{\otimes n}$.
Assume that 
$\sigma^2 = \Var_{\rvx, \rvy}\left[\lb \frac{dP_{\rvx, \rvy}}{d(P_\rvx \otimes P_\rvy)}(\rvx, \rvy) \right] < \infty$.
Then, for any sequence of common randomness $\rvz_n$ with $\rvz_n \perp \rvx^n$ and $\rvy^n = g(\rvx^n, \rvz_n)$, we have
\begin{align*}
\lim_{n \to \infty} \frac{\Ent{\rvy^n \mid \rvz_n} - \MI{\rvx^n}{\rvy^n}}{\lb (n)} \geq \frac{1}{2}.
\end{align*}
\end{importantTheorem}
The theorem essentially states that the expected codelength $n \cdot R_n$ of any relative entropy coding algorithm for non-singular channels grows at least as
\begin{align*}
n \cdot R_n \approx \MI{\rvx^n}{\rvy^n} + \frac{1}{2}\lb(\MI{\rvx^n}{\rvy^n} + 1) + \Oh(1).
\end{align*}
\begin{proof}
This proof proceeds more naturally if we measure the information-theoretic quantities in nats rather than bits.
However, to keep the notation consistent with the rest of the thesis, I keep measuring the KL and channel simulation divergences in bits.
Instead, for each $n \geq 1$, I introduce the notation
\begin{align*}
\kappa^n = \KLD{P_{\rvx^n \mid \rvy^n}}{P_{\rvx^n}} \cdot \ln 2, 
\end{align*}
i.e., $\kappa^n$ is the KL divergence in nats for the $n$-length block.
Now, note that \cref{thm:csd_rec_lower_bound} implies that 
\begin{align*}
\Exp_{\rvy^n}[\CSD{P_{\rvx^n \mid \rvy^n}}{P_\rvx^n} - \KLD{P_{\rvx^n \mid \rvy^n}}{P_{\rvx^n}}] \leq \Ent{\rvy^n \mid \rvz_n} - \MI{\rvx^n}{\rvy^n}.
\end{align*}
Thus, if we can show that the left-hand side above scales as $\frac{1}{2}\lb n$ asymptotically, then we have the desired result.
Now, denote by $H_n$ the associated random variable to $P_{\rvx^n \mid \rvy^n}$ and $P_{\rvx^n}$.
Then, using \cref{lemma:delta_logh_diffent_identity}, we see that the above inequality implies
\begin{align*}
\Ent{\rvy^n \mid \rvz_n} - \MI{\rvx^n}{\rvy^n} + \lb(e) 
&\geq \DiffEnt{\ln H_n \mid \rvy^n} \\
&= \DiffEnt*{n^{-1/2} \cdot (\ln H_n - \kappa^n) \,\Big\vert\, \rvy^n} + \frac{1}{2} \lb n,
\end{align*}
where I used the translation invariance and the scaling property of the differential entropy.
In the remainder of the proof, I show that the conditional entropy above converges to a constant as $n \to \infty$, which finishes the proof.
As the specific scaling and translation above suggest, I will use an argument based on the central limit theorem (CLT) to show this result.
To begin, let $s = n^{1/2}$ and $b = \kappa^n$.
Now, note that 
\begin{align}
\Prob\Big[n^{-1/2} \cdot (\ln\, H_n - \kappa^n) \leq t \,\Big\vert\, \rvy^n\Big] 
&= \Prob\left[H_n \leq \exp\left(s \cdot t + b\right) \,\Big\vert\, \rvy^n\right] \nonumber \\
&= \int_0^{\exp\left(s \cdot t + b\right)} w_{P_{\rvx^n}}(h) \, dh \nonumber \\
&= \int_{-\infty}^t s \cdot e^{s \cdot u + b} w_{P_{\rvx^n}}(e^{s \cdot u + b}) \, du,\label{eq:proof_csd_second_order_exp_substitution}
\end{align}
where \cref{eq:proof_csd_second_order_exp_substitution} follows from substituting $h = \exp(s \cdot u + b)$.
Now, let me introduce the shorthand 
\begin{align*}
\Lambda_n(t) = w_{P_{\rvx^n \mid \rvy^n}}(e^{s \cdot t + b}).
\end{align*}
Then, continuing \cref{eq:proof_csd_second_order_exp_substitution}, we have 
\begin{align}
\int_{-\infty}^t s \cdot e^{s \cdot u + b} w_{P_{\rvx^n}}(e^{s \cdot u + b}) \, du 
&= -\int_{-\infty}^t s \cdot e^{s \cdot u + b} \cdot \int_{e^{s \cdot u + b}}^\infty \frac{1}{v} \, dw_{P_{\rvx^n \mid \rvy^n}}(v) \, du \label{eq:proof_csd_second_order_wP_wQ_identity} \\
&= -\int_{-\infty}^t s \cdot e^{s \cdot u + b} \cdot \int_{u}^\infty e^{-(s \cdot \eta + b)} \, d\Lambda_n(\eta) \, du \label{eq:proof_csd_second_order_exp_substitution_2}\\
&= -\int_{-\infty}^t s \int_{u}^\infty e^{-s ( \eta - u)} \, d\Lambda_n(\eta) \, du \nonumber\\
&= -\int_{-\infty}^t \int_{u}^\infty \frac{\partial}{\partial u} e^{-s ( \eta - u) } \, d\Lambda_n(\eta) \, du, \nonumber \\
&= -\int_{-\infty}^t \frac{d}{d u} \int_{u}^\infty \!\!\!e^{-s ( \eta - u) } \, d\Lambda_n(\eta) \, du -\int_{-\infty}^t \!\!\!d\Lambda_n(\eta)\label{eq:proof_csd_second_order_leibniz_rule} \\
&= \int_{-\infty}^t  e^{-s ( \eta - t) } \, d\Lambda_n(\eta) + 1 - \Lambda_n(t), \label{eq:proof_csd_second_order_fun_thm_of_calc}
\end{align}
where \cref{eq:proof_csd_second_order_wP_wQ_identity} follows from \cref{eq:width_function_wP_wQ_identity} in \cref{lemma:width_function_properties}, \cref{eq:proof_csd_second_order_exp_substitution_2} follows from substituting $v = \exp(s \cdot \eta + b)$, \cref{eq:proof_csd_second_order_leibniz_rule} follows from the Leibniz integral rule, and \cref{eq:proof_csd_second_order_fun_thm_of_calc} follows from applying the fundamental theorem of calculus to both terms.
Now, let me study the behaviour of the terms in \cref{eq:proof_csd_second_order_fun_thm_of_calc}.
First, note that
\begin{align}
\lim_{n \to \infty}\int_{-\infty}^t  e^{-s ( \eta - t) } \, d\Lambda_n(\eta) 
&= \lim_{n \to \infty}\int_{-\infty}^t  e^{-n^{1/2} ( \eta - t) } \, d\Lambda_n(\eta) \tag{since $s = n^{1/2}$}\\
&= \int_{-\infty}^t \Ind[\eta = t] \, d\lim_{n \to \infty}\Lambda_n(\eta) \nonumber\\
&=0, \nonumber
\end{align}
since the indicator is $0$ everywhere over the domain of integration except for at the upper limit of integration, where it is equal to $1$.
Second, observe that 
\begin{align*}
1 - \Lambda_n(t) &= 1 - w_{P_{\rvx^n \mid \rvy^n}}(e^{s \cdot t + b})\\
&= \Prob_{\rvx^n \sim P_{\rvx^n \mid \rvy^n}}\left[\frac{dP_{\rvx^n \mid \rvy^n}}{dP_{\rvx^n}}(\rvx^n \mid \rvy^n) < e^{s \cdot t + b} \mid \rvy^n\right] \\
&= \Prob_{\rvx^n \sim P_{\rvx^n \mid \rvy^n}}\left[\prod_{i = 1}^n\frac{dP_{\rvx \mid \rvy}}{dP_{\rvx}}(\rvx_i \mid \rvy_i) < e^{s \cdot t + b} \mid \rvy^n\right] \\
&= \Prob_{\rvx^n \sim P_{\rvx^n \mid \rvy^n}}\left[n^{-1/2} \cdot \sum_{i = 1}^n\left(\ln\frac{dP_{\rvx \mid \rvy}}{dP_{\rvx}}(\rvx_i \mid \rvy_i) - \kappa^n \right) < t \mid \rvy^n\right].
\end{align*}
Thus, assuming the limit exists, we have
\begin{align}
\label{eq:proof_csd_second_order_dist_limit_equality}
\lim_{n \to \infty} \!\frac{\ln H_n - \kappa^n}{\sqrt{n}} 
\,\, \disteq\,\,
\lim_{n \to \infty} \frac{1}{\sqrt{n}} \sum_{i = 1}^n\!\left(\ln\frac{dP_{\rvx \mid \rvy}}{dP_{\rvx}}(\rvx_i \mid \rvy_i) - \kappa^n \right)
\end{align}
$P_{\rvy^\infty}$-almost surely.
I am now ready to apply the central limit theorem to the right-hand side.
However, note that the terms in the sum are only independent but not identically distributed, as $\rvx_i \sim P_{\rvx_i \mid \rvy_i}$.
Hence, I need to apply a more powerful version of the theorem that applies to this case, known as the Lindeberg-Feller CLT.
To do this, I now show that the following Lindeberg condition holds:
\begin{lemma}[Lindeberg condition for the conditional density ratio.]
\label{lemma:eq:proof_csd_second_order_lindeberg_lemma}
Denote the conditional variance of $r_\rvy$ as $\sigma_{\rvy}^2 = \Var_{\rvx \mid \rvy}[\ln r_{\rvy}(\rvx)]$ with $0 < \sigma_\rvy^2$, and let $s_n = \sum_{k = 1}^n \sigma_{\rvy_k}^2$.
Furthermore, let $\kappa_{\rvy} = \Exp_{\rvx \mid \rvy}[\ln r_{\rvy}(\rvx)] = \KLD{P_{\rvx \mid \rvy}}{P_\rvx} \cdot \ln (2)$.
Then, $P_{\rvy^\infty}$-almost surely the following Lindeberg condition holds:
\begin{align}
\label{eq:asymptotic_scaling_proof_lindeberg_condition}
\forall \epsilon > 0: \quad \lim_{n \to \infty}\frac{1}{s_n}\sum_{k = 1}^n \Exp_{\rvx_k \mid \rvy_k}\left[(\ln r_{\rvy_k}(\rvx_k) - \kappa_{\rvy_k})^2 \cdot \Ind[\abs{\ln r_{\rvy_k}(\rvx_k) - \kappa_{\rvy_k}} \geq s_n \cdot \epsilon]\right] = 0.
\end{align}
Then, by the Lindeberg-Feller central limit theorem \citep[Theorem 3.4.5;][]{durrett2019probability}, we have
\begin{align*}
\frac{1}{\sqrt{n}} \sum_{i = 1}^n\left(\ln\frac{dP_{\rvx \mid \rvy}}{dP_{\rvx}}(\rvx_i \mid \rvy_i) - \kappa^n \right) \stackrel{d}{\to} \Normal(0, \Exp_\rvy[\sigma_\rvy^2]) \quad \text{as } n \to \infty.
\end{align*}
\end{lemma}
The proof follows from a ``standard'' argument for showing the Lindeberg condition:
\begin{enumerate}
\item For each $n \geq 1$, I replace the infinite sum in \cref{eq:asymptotic_scaling_proof_lindeberg_condition} with a single quantity that is proportional to the expectation $\Exp[T_n]$ of an appropriate random variable $T_n$.
\item I show that $\plim_{n \to \infty} T_n = 0$, where $\plim$ denotes the limit in probability (if it exists).
\item Then, by the dominated convergence theorem, $\Exp[T_n] \to 0$ as $n \to \infty$.
\end{enumerate}
The core trick in the proof is to observe that if we take expectations of the $n$th partial sum with respect to $\rvy^n$, then the resulting quantity is exchangeable.
I now proceed to prove the statement.
\begin{proof}
Fix $\epsilon >0$.
Then, note that
\begin{align*}
\lim_{n \to \infty}\Exp_{\rvy^n}&\left[\frac{1}{s_n^2}\sum_{k = 1}^n \Exp_{\rvx_k \mid \rvy_k}\left[(\ln r_{\rvy_k}(\rvx_k) - \kappa_{\rvy_k})^2 \cdot \Ind[\abs{\ln r_{\rvy_k}(\rvx_k) - \kappa_{\rvy_k}} \geq s_n \cdot \epsilon]\right] \right] \\
&=\lim_{n \to \infty}\sum_{k = 1}^n \Exp_{\rvx_k, \rvy^n}\left[\frac{1}{s_n^2}(\ln r_{\rvy_k}(\rvx_k) - \kappa_{\rvy_k})^2 \cdot \Ind[\abs{\ln r_{\rvy_k}(\rvx_k) - \kappa_{\rvy_k}} \geq s_n \cdot \epsilon]\right] \\
&=\lim_{n \to \infty} n \cdot \Exp_{\rvx_1, \rvy^n}\left[\frac{1}{s_n^2}(\ln r_{\rvy_1}(\rvx_1) - \kappa_{\rvy_1})^2 \cdot \Ind[\abs{\ln r_{\rvy_1}(\rvx_1) - \kappa_{\rvy_1}} \geq s_n \cdot \epsilon]\right], 
\end{align*}
where the last line follows since the expectations over $(\rvx_k, \rvy^n)$ and $(\rvx_1, \rvy^n)$ are exchangeable.
Now, let $T_n = \frac{n}{s_n^2}(\ln r_{\rvy_1}(\rvx_1) - \kappa_{\rvy_1})^2 \cdot \Ind[\abs{\ln r_{\rvy_1}(\rvx_1) - \kappa_{\rvy_1}} \geq s_n \cdot \epsilon]$.
Now, by the weak law of large numbers, we have $\plim_{n \to \infty} s_n^2/n = \Exp_\rvy[\sigma_\rvy^2]$, hence
\begin{align*}
\plim_{n \to \infty} T_n &= \frac{(\ln r_{\rvy_1}(\rvx_1) - \kappa_{\rvy_1})^2}{\Exp_\rvy[\sigma_\rvy^2]} \cdot \plim_{n \to \infty}\Ind[\abs{\ln r_{\rvy_1}(\rvx_1) - \kappa_{\rvy_1}} \geq s_n \cdot \epsilon] \\
&= \frac{(\ln r_{\rvy_1}(\rvx_1) - \kappa_{\rvy_1})^2}{\Exp_\rvy[\sigma_\rvy^2]} \cdot \plim_{n \to \infty}\Ind\left[\abs{\ln r_{\rvy_1}(\rvx_1) - \kappa_{\rvy_1}} \geq \sqrt{n} \cdot \sqrt{\frac{s_n^2}{n}} \cdot \epsilon\right] \\
&= 0.
\end{align*}
Now, we have by assumption that $\Var_{\rvx, \rvy}[\ln r_\rvy(\rvx)] < \infty$, hence applying the dominated convergence theorem to $T_n$ with $(\ln r_\rvy(\rvx))^2$ as the dominating random variable yields
\begin{align*}
\Exp_{\rvy^\infty}\!&\left[\lim_{n \to \infty}\frac{1}{s_n^2}\sum_{k = 1}^n \Exp_{\rvx_k \mid \rvy_k}\left[(\ln r_{\rvy_k}(\rvx_k) - \kappa_{\rvy_k})^2 \cdot \Ind[\abs{\ln r_{\rvy_k}(\rvx_k) - \kappa_{\rvy_k}} \geq s_n \cdot \epsilon]\right] \right] \\
&=\lim_{n \to \infty}\Exp_{\rvy^n}\left[\frac{1}{s_n^2}\sum_{k = 1}^n \Exp_{\rvx_k \mid \rvy_k}\left[(\ln r_{\rvy_k}(\rvx_k) - \kappa_{\rvy_k})^2 \cdot \Ind[\abs{\ln r_{\rvy_k}(\rvx_k) - \kappa_{\rvy_k}} \geq s_n \cdot \epsilon]\right] \right] = 0.
\end{align*}
Since the term inside the $\rvy^\infty$-expectation is a non-negative quantity whose expectation is $0$, it must be $P_{\rvy^\infty}$-almost surely $0$, which finishes the proof of the lemma.
\end{proof}
Since we assumed the channel $\rvx \to \rvy$ to be non-singular, the variance condition of \cref{lemma:eq:proof_csd_second_order_lindeberg_lemma} is satisfied in the case of \cref{eq:proof_csd_second_order_dist_limit_equality}, hence we have that $P_{\rvy^\infty}$-almost surely
\begin{align*}
\frac{\ln H_n - \kappa^n}{\sqrt{n}} \quad \stackrel{d}{\to} \quad \Normal(0, \Exp_\rvy[\sigma_\rvy^2]) \quad \text{as } n \to \infty
\end{align*}
Then, by adapting Theorem 1 of \citet{barron1986entropy}, it can be shown that in the case of the CLT, the terms converge in relative entropy. 
Hence we have $P_{\rvy^\infty}$-almost surely
\begin{align*}
\lim_{n \to \infty} \DiffEnt*{\frac{\ln H_n - \kappa^n}{\sqrt{n}} \mid \rvy^n} = \DiffEnt*{\Normal(0, \Exp_\rvy[\sigma_\rvy^2])}.
\end{align*}
The right-hand term is a constant, which finishes the proof.
\end{proof}
\Cref{thm:second_order_behaviour_of_csd} demonstrates a non-trivial lower bound for the asymptotic relative entropy coding of non-singular channels.
But how tight is this lower bound?
Indeed, \citet{sriramu2024optimal} have shown that it is tight up to a constant and also that in the case of singular channels, the second order redundancy reduces to $0$.
Hence, combining \cref{thm:second_order_behaviour_of_csd} with \citet[Theorem 1;][]{sriramu2024optimal}, we now have the complete characterisation of asymptotic relative entropy coding:
\begin{corollary}[Characterisation of Asymptotic Relative Entropy Coding.]
Let $\rvx, \rvy \sim P_{\rvx, \rvy}$ be dependent random variables, and assume that $\Var_{\rvx, \rvy}\left[\lb \frac{dP_{\rvx, \rvy}}{d(P_\rvx \otimes P_\rvy)} (\rvx, \rvy) \right] < \infty$.
Define the optimal second-order redundancy (also called the excess functional information by \citet{li2018strong}) as
\begin{align*}
\Psi(\rvx \to \rvy) = \inf_{\rvz \mid \rvz \perp \rvx} \Ent{\rvy \mid \rvz} - \MI{\rvx}{\rvy}.
\end{align*}
For all $n \geq 1$, let $\rvx^n, \rvy^n \sim P_{\rvx, \rvy}^{\otimes n}$.
Then:
\begin{align*}
\lim_{n \to \infty}\frac{\Psi(\rvx^n \to \rvy^n)}{\lb n} = 
\begin{cases}
0 &\text{if } \rvx \to \rvy \text{ is singular}\\
\frac{1}{2} &\text{if } \rvx \to \rvy \text{ is non-singular.}
\end{cases}
\end{align*}
\end{corollary}
\section[The Computational Complexity of Relative Entropy Coding]{The Computational Complexity of \\ Relative Entropy Coding}
\label{sec:computational_complexity_of_rec}
\par
In the previous sections, I investigated the limits of relative entropy coding in terms of its average description length.
However, in practice, another essential factor of compression algorithms is their speed.
Therefore, in this section, I briefly investigate the runtime of channel simulation algorithms.
\par
For this, I need to establish a suitable computational framework that 1) covers a broad and relevant set of channel simulation algorithms and 2) has a natural notion of runtime.
Thankfully, there is already a very general framework that fits the bill perfectly: selection sampling, as I introduced it in \cref{sec:rec_through_examples}!
Therefore, following \citet{flamich2024some}, I now define selection samplers more formally.
\begin{definition}[(Exact) Selection Sampler]
\label{def:exact_selection_sampler}
Let $Q \ll P$ be probability measures, and let $(X_i)_{i \in \Nats}$ be a sequence of i.i.d.\ $P$-distributed random variables.
An (exact) selection sampler for a target distribution $Q$ with proposal distribution $P$ is defined by two random variables, $N$ and $K$, such that the following hold:
\begin{itemize}
\item \textbf{Correctness.} $X_N \sim Q$.
\item \textbf{Computable Termination Criterion.} 
$K$ is a stopping time adapted to $(X_i)_{i \in \Nats}$, i.e., $X_i, X_{i + 1}, \dotsc$ are independent of the event $\{K \geq i\}$.
Furthermore, $1 \leq N \leq K$.
\end{itemize}
Finally, I call:
\begin{itemize}
\item $N$ the selection rule,
\item $K$ the runtime and
\item $\Exp[K]$ the sample complexity
\end{itemize}
of the selection sampler.
\end{definition}
To connect \cref{def:exact_selection_sampler} with the intuitive algorithm from \cref{sec:rec_through_examples}, first note that we can realise the sequence $(X_i)$ by simulating samples from $P$ one-by-one.
Then, at each step $k$, we check if we wish to stop; hence we can ``set'' $K = \min_{k \in \Nats}\{\mathtt{stop}(X_{1:k}) = \mathtt{True}\}$, and $N = \mathtt{select}(X_{1:K})$.
\par
Furthermore, there is a natural notion of runtime for selection samplers, as I already noted in the definition: we count how many samples the algorithm needs to examine from $P$ before it terminates.
However, since $K$ can be random, I instead study the expected runtime, which is also known as the sample complexity \citep{block2023sample}.
The reason why I study this setting is twofold: first, it is already a very general setting with minimal required structure, while also admitting very precise results.
\par
Next, we need to choose an appropriate measure of efficiency.
A natural guess would be that ``the higher the information content, the longer the runtime.'' That is, for a target distribution $Q$ and proposal $P$, and for some monotone function $f$, we would like to make a statement of the form $\Exp[K] = \Oh(f(\KLD{Q}{P}))$.
However, in the general case, it turns out that the scaling of the runtime is not given by the relative entropy but by the so-called R{\'e}nyi $\infty$-divergence, which I discuss next.
\begin{definition}[R{\'e}nyi $\infty$-divergence]
\label{eq:renyi_inf_div}
Let $Q \ll P$ be probability measures over some space $\Omega$ with Radon-Nikodym derivative $r = \frac{dQ}{dP}$.
Then, their R{\'e}nyi $\infty$-divergence is defined as the $\mathcal{L}^\infty(\Omega, P)$-norm of $\lb r$, i.e.,
\begin{align*}
\infD{Q}{P} = \lb \norm{r}_{\infty} = \lb (\esssup\{r\}) = \lb (\inf\{M \in \Reals \mid P(\{x \mid M > r(x)\}) = 0 \}).
\end{align*}
\end{definition}
Note that the essential supremum is needed to avoid measure-theoretic difficulties;
when $\Omega$ is discrete or when $r$ is continuous, the above definition collapses onto the simpler $\infD{Q}{P} = \lb (\sup_{x \in \Omega}\{r(x)\})$.
The relationship of $\KLD{Q}{P}$ and $\infD{Q}{P}$ is analogous to the relationship of the $1$- and $\infty$-norms: the former captures the average behaviour, while the latter captures the worst-case behaviour.
With this intuition, we can derive the following result:
\begin{align*}
\KLD{Q}{P} = \Exp_{Z \sim Q}[\lb r(Z)] \leq \Exp_{Z \sim Q}[\lb \norm{r}_\infty] = \infD{Q}{P}.
\end{align*}
\par
I am now ready to state the lower bound on the runtime of selection samplers.
Unfortunately, this bound not only involves $\infD{Q}{P}$ instead of $\KLD{Q}{P}$, but it also scales very poorly in this quantity.%
\begin{importantTheorem}[Lower bound on the runtime of selection samplers.]
\label{thm:selection_sampler_runtime_lower_bound}
Let $K, N$ be a selection sampler (\cref{def:exact_selection_sampler}) for $Q \ll P$, with $(X_i)_{i \in \Nats}$ defined as above.
Then,
\begin{align*}
\expb(\infD{Q}{P}) \leq \Exp[K].
\end{align*}
\end{importantTheorem}
\begin{proof}
Let $\epsilon > 0$, set $M = \expb(\infD{Q}{P}) = \norm{r}_\infty$ and define $A_\epsilon = r^{-1}([M - \epsilon, \infty))$ to be the set of points $x$ in the sample space for which $r(x) \geq M - \epsilon$.
By the definition of $\infD{Q}{P}$, for any $\epsilon > 0$ we have $P(A_\epsilon) > 0$.
Furthermore, we also have
\begin{align}
\label{eq:proof_of_exponential_runtime_epsilon_sets_bound}
M - \epsilon = \frac{\int_{A_\epsilon}(M - \epsilon)\,dP(x)}{P(A_\epsilon)} \stackrel{(a)}{\leq} \frac{\int_{A_\epsilon}r(x)\,dP(x)}{P(A_\epsilon)} = \frac{Q(A_\epsilon)}{P(A_\epsilon)} \stackrel{(b)}{\leq} \frac{\int_{A_\epsilon}\norm{r}_\infty\,dP(x)}{P(A_\epsilon)} = M,
\end{align}
where inequality (a) follows from the definition of $A_\epsilon$ and inequality (b) follows from the definition of $\infD{Q}{P}$.
I now show that for any $\epsilon > 0$ we have $Q(A_\epsilon)/P(A_\epsilon) \leq \Exp[K]$.
This fact, taken together with \cref{eq:proof_of_exponential_runtime_epsilon_sets_bound}, yields the desired result.
Thus, to finish, observe that
\begin{align}
Q(A_\epsilon) &= \Prob[X_N \in A_\epsilon] \label{eq:proof_of_exponential_runtime_selection_sampler_correctness}\\
&=\sum_{n = 1}^\infty \Prob[N = n, X_n \in A_\epsilon] \nonumber\\
&=\sum_{n = 1}^\infty \Prob[K \geq n, N = n, X_n \in A_\epsilon] \nonumber\\
&= \sum_{n = 1}^\infty \Prob[K \geq n] \cdot P(A_\epsilon) \cdot \Prob[N = n, \mid K \geq n,  X_n \in A_\epsilon] \label{eq:proof_of_exponential_runtime_stopping_time} \\
&\leq \sum_{n = 1}^\infty \Prob[K \geq n] \cdot P(A_\epsilon) \nonumber\\
&= P(A_\epsilon) \cdot \Exp[K], \label{eq:proof_of_exponential_runtime_darth_vader}
\end{align}
where \cref{eq:proof_of_exponential_runtime_selection_sampler_correctness} follows from the correctness requirement of selection samplers, \cref{eq:proof_of_exponential_runtime_stopping_time} follows from the requirement that $K$ be a stopping time adapted to $(X_i)_{i \in \Nats}$ and \cref{eq:proof_of_exponential_runtime_darth_vader} follows from the Darth Vader rule \citep{muldowney2012darth}.
Finally, dividing through by $P(A_\epsilon)$ finishes the proof.
\end{proof}
% \footnote{
\Cref{thm:selection_sampler_runtime_lower_bound} generalises a result of \citet{letac1975building} for selection samplers where the selection rule must match the runtime: $N = K$.%
\footnote{I thank Sam Power for bringing Letac's result to my attention.}
Independently, Daniel Goc discovered the same result and appears as Theorem III.1 in our paper \citet{goc2024channel}, while \cref{thm:selection_sampler_runtime_lower_bound} appears as Theorem A.2 in the same paper.
% }
\par
While \cref{thm:selection_sampler_runtime_lower_bound} might seem discouraging at first sight, perhaps it should not come as a surprise.
I made virtually no assumptions on the structure of $Q$ and $P$, and sampling in such a general setting is known to be hard in other computational frameworks as well.
However, we do not need to solve such a general problem in practice.
As I will demonstrate in \cref{chapter:combiner}, it is sufficient to develop a relative entropy coding algorithm that can encode samples from a highly structured distribution, such as a multivariate Gaussian.
Then we can use an expressive machine learning model to transform the encoded samples to follow the right distribution.
Therefore, after developing some general sampling algorithms in \cref{chapter:rec_with_pp}, I will specialise these algorithms in \cref{chapter:branch_and_bound} to exploit the structure of certain problems that will result in an exponential reduction in the algorithms' average runtime.
\section[Conclusion, Future Directions and Open Questions]{Conclusion, Future Directions and \texorpdfstring{\\}{} Open Questions}
\label{sec:fundamental_limits_conclusion}
\par
In this chapter, I established fundamental lower bounds to the communication efficiency and the sample complexity of relative entropy coding.
To complement these lower bounds and demonstrate their achievability, I develop two families of relative entropy coding algorithms in the next chapter.
\par
There are many interesting possible future directions for investigation and theory-building.
One promising direction would be to deepen the theory of the associated random variable $H$ with density $w$ and canonical representation $\eta$. 
These probability densities possess the highly unusual property of being inverses of one another.
This relation further means that their cumulative distribution functions are convex duals of each other.
Thus, it would be interesting to see if we could use convex analysis to simplify some of the results I presented here or prove new ones.
Recently, \citet{esposito2022functional} has successfully applied convex analysis to study more classical information measures such as the Shannon and relative entropy and the R{\'e}nyi family of divergences. As such, their work could provide a blueprint to apply these techniques to the theory of relative entropy coding.
Perhaps a good starting point for this project is with the channel simulation divergence itself, which Daniel Goc and I showed to be convex in its first argument in our paper \citet{goc2024channel}.
Thus, it would be interesting to derive the convex dual/variational representation similar to the Donsker-Varadhan representation for the relative entropy.

%% file: 4-RelativeEntropyCodingWithPoissonProcesses/relative_entropy_coding_with_poisson_processes.tex
%!TEX root = ../thesis.tex
%*******************************************************************************
%****************************** Second Chapter *********************************
%*******************************************************************************

\chapter[Relative Entropy Coding with Poisson Processes]{Relative Entropy Coding with\\ Poisson Processes}
\label{chapter:rec_with_pp}
% \epigraph{Anyone who considers arithmetical methods of producing random digits is, of course, in a state of sin.}{\textit{John von Neumann}}
%
\ifpdf
    \graphicspath{{ChannelSimulationWithPoissonProcesses/img}}
\else
    \graphicspath{{ChannelSimulationWithPoissonProcesses/img}}
\fi
\noindent
After establishing the basic theory in \cref{chapter:source_coding} and the fundamental limits of the theory in \cref{chapter:fundamental_limits}, I now turn my attention to deriving general-purpose relative entropy algorithms.
The chapter begins with establishing the basic properties of Poisson processes, which will form the building blocks of the two families of algorithms I develop: A* coding and greedy Poisson rejection sampling.
Then, I present the analysis of these algorithms and prove their correctness and codelength efficiency.
Next, I develop more advanced results for Poisson processes, which I finally apply to develop fast variants of these algorithms.
I conclude the chapter with a discussion on developing approximate variants of the algorithms and their efficient implementation on a computer.
\par
This chapter aims to familiarise the reader with Poisson processes and the elegant ways we can use them to construct sampling algorithms, which in turn give rise to relative entropy coding algorithms.
In short, it mainly contributes to the development of the ``Poisson process model for Monte Carlo'' \citep{maddison2016poisson}.
As such, \cref{sec:basic_pp_theory} gives some technical background that mainly follows the exposition of \citet{maddison2016poisson}.
Furthermore, the derivations of rejection sampling (\cref{sec:global_rs}) and A* sampling (\cref{sec:global_a_star}) follow the works of \citet{maddison2014sampling,maddison2016poisson} and \citet{li2018strong}, with the exceptions of one of the proofs of \cref{thm:global_rs_runtime} and the proof of \cref{thm:global_a_star_runtime}, which are mine (though not too deep).
The rest of the chapter contains my original work.
\section{Basic Poisson Process Theory}
\label{sec:basic_pp_theory}
\par
This section follows the excellent exposition by \citet{maddison2016poisson}.
I first define Poisson processes, then describe four pointwise operations: thinning, mapping, restriction and superposition that preserve Poisson processes.
Then, in \cref{sec:global_rs,sec:global_a_star,sec:global_gprs}, I show that the first three operations lead to a different sampling algorithm, respectively.
Furthermore, in \cref{sec:superposition_parallelisation} I show how we can use superposition to construct parallelised variants of our samplers.
As such, Poisson processes provide a beautiful, unifying language to construct and analyse sampling algorithms.
\par
But what are Poisson processes, anyway?
Poisson processes are perhaps the most fundamental kind of \textit{point processes}, which formalise the notion of a ``distribution'' over randomly selected points in some space.
Since general sampling algorithms often follow the selection sampling recipe (\cref{def:exact_selection_sampler}), point processes offer a way to capture the entire sampling process with a single mathematical object.
In this thesis, I will be interested only in Poisson processes, perhaps the simplest among point processes, as they are ``completely random'': knowing the location of a point of the process tells us nothing about the location of any of the other points.
More formally, in a Poisson process $\PoissonProcess$ over some space $\Omega$ consists of a countable collection of points, such that the number of points $\rmN(A) = \abs{\PoissonProcess \cap A}$ of $\PoissonProcess$ that fall in some part of space $A \subseteq \Omega$ is independent of the number of points $\rmN(B) = \abs{\PoissonProcess \cap B}$ falling in some other part of space, i.e., when $A$ and $B$ are disjoint.
Note that $\rmN(A) = \abs{\PoissonProcess \cap A}$ specifies a random integer for each set $A$. 
Hence, it is called the \textit{random counting measure} of the process $\PoissonProcess$.
By \citet[Theorem 2.4.V;][]{daley2003introduction}, the independence requirement alone almost completely determines Poisson processes. 
Indeed, we only need two extra requirements: 1) if a set $A$ is bounded, almost surely finitely many points of the process fall in $A$ and 2) \textit{lightning does not strike twice in the same place:} at any point in the ambient space $x \in \Omega$, the probability that two or more points of the process fall on $x$ is almost surely $0$.
Adding in these extra requirements turns out to be equivalent to requiring that for each set $A \subset \Omega$, the counting measure $\rmN(A)$ is Poisson distributed \citep[Theorem 2.4.V;][]{daley2003introduction}, which gives Poisson processes their name, and it is this latter statement that I will take as the formal definition.
\begin{definition}[Poisson Point Process.]
\label{def:poisson_process}
Let $\Omega$ be a Polish space, and let $\PoissonProcess = \{X_1, X_2, \hdots\} \subseteq \Omega$ be a random, countable collection of points over $\Omega$.
For each (Borel) set $A \subseteq \Omega$, define the counting measure of $\PoissonProcess$ as
\begin{align*}
\rmN(A) = \abs{\PoissonProcess \cap A}.
\end{align*}
Furthermore, define the mean measure of $\PoissonProcess$ as
\begin{align*}
\mu(A) = \Exp[\rmN(A)].
\end{align*}
Then, $\Pi$ is a Poisson process if and only if
\begin{enumerate}
\item \textbf{Independence.} For any two disjoint (Borel) sets $A, B \subseteq \Omega, A \cap B = \emptyset$, the counts are independent:
\begin{align*}
\rmN(A) \perp \rmN(B).
\end{align*}
\item \textbf{Poisson-distributed counts.} For a bounded set $A$, the mean measure is bounded and non-atomic, i.e., $\mu(A) < \infty$ and $\mu(\{x\}) = 0$ for any $x \in \Omega$, and the counting measure is Poisson distributed:
\begin{align*}
\Prob[\rmN(A) = n] = \frac{\mu(A)^n \cdot e^{-\mu(A)}}{n!}.
\end{align*}
\end{enumerate}
\end{definition}
\subsection{Poisson Process Operations}
\label{sec:poisson_process_operations}
\par
I now review four operations on Poisson processes that preserve them.
Below, I will only present the statements of the theorems; their proofs can be found in the works of \citet{kingman1992poisson} and \citet{maddison2016poisson}.
However, the verity of the statements is intuitively clear: each describes a ``local, pointwise operation'' and changes at any given point of the process that depends on that point only and not the others.
Hence, starting with a completely random process and performing a local, pointwise operation still leaves us with a completely random process.
\par
\textbf{Restriction.}
Perhaps the most straightforward operation on a Poisson process $\PoissonProcess$ over some space $\Omega$ is to consider only the points of $\PoissonProcess$ that fall in some subset $A \subseteq \Omega$.
The restriction theorem \citep[Section 2.2;][]{kingman1992poisson} guarangees that $\PoissonProcess \cap A$ is also a Poisson process:
\begin{theorem}[Restriction Theorem, \citep{kingman1992poisson}]
\label{thm:restriction_theorem}
Let $\PoissonProcess$ be a Poisson process over $\Omega$ with mean measure $\mu$, and let $A, B \subseteq \Omega$ be measurable under $\mu$.
Then,
\begin{align*}
\PoissonProcess\vert_A = \PoissonProcess \cap A
\end{align*}
is a Poisson process over $\Omega$ with mean measure
\begin{align*}
\mu\vert_A(B) = \mu(A \cap B).
\end{align*}
\end{theorem}
\par
The viewpoint of \cref{thm:restriction_theorem} is global and static: we took all the (possibly infinitely many) points of $\PoissonProcess$, intersected them with some set and looked at the remainder.
However, there is an alternative, more computational view: we can define a function $\mathtt{filter}: \Omega \to \{0, 1\}$.
Then, we examine the points of $\PoissonProcess$ one-by-one and only keep those for which $\mathtt{filter}$ returns $1$ and ``delete'' the others.
Setting $A = \mathtt{filter}^{-1}(1) \subseteq \Omega$ to be the set of points for which $\mathtt{filter}$ returns $1$, we see that the ``computational'' and the ``global'' viewpoints are equivalent:
\begin{align*}
\{x \in \PoissonProcess \mid \mathtt{filter}(x) = 1 \} = \PoissonProcess \cap A.
\end{align*}
This computational view suggests the following useful generalisation of restriction.
\par
\textbf{Thinning.} Instead of using a deterministic filter, we could delete the points randomly.
Then, so long as the probability that we delete any given point of the process is independent of all other points of the process, the thinning theorem guarantees that the remaining points still form a Poisson process:
\begin{theorem}[Thinning theorem \citep{maddison2016poisson}]
\label{thm:thinning_theorem}
Let $\PoissonProcess$ be a Poisson process over $\Omega$ and mean measure $\mu$.
Let $\{\thinOp(x) \}_{x \in \Omega}$ be a family of Bernoulli random variables with index set $\Omega$, and let $\rho(x) = \Exp[\thinOp(x)]$.
Then,
\begin{align*}
\thinOp(\PoissonProcess) = \{x \in \PoissonProcess \mid \thinOp(x) = 1\}
\end{align*}
is a Poisson process with mean measure
\begin{align*}
\mu'(A) = \int_A \rho(x) \, d\mu(x).
\end{align*}
\end{theorem}
\par
\textbf{Mapping.}
Another natural operation is to ``move'' the points of the process.
Then, as long as how we move a given point is independent of all other points of the process and we do not move multiple points on top of each other, the Mapping theorem guarantees that the resulting collection of points will form a Poisson process.
\begin{theorem}[Mapping Theorem \citep{kingman1992poisson}]
\label{thm:mapping_theorem}
Let $\PoissonProcess$ be a Poisson process over some space $\Omega$ with mean measure $\mu$.
Let $f$ be a measurable function such that the pushforward measure $(f\pushfwd \mu)(A) = \mu(f^{-1}(A))$ is non-atomic.
Then,
\begin{align*}
f(\PoissonProcess) = \{f(x) \mid x \in \PoissonProcess\}
\end{align*}
is a Poisson process with mean measure $f\pushfwd \mu$.
\end{theorem}
\par
\textbf{Superposition.}
Finally, as a somewhat dual construction to restriction, overlaying two independent Poisson processes also results in a Poisson process.
\begin{theorem}[Superposition Theorem \citep{kingman1992poisson}]
\label{thm:superposition_theorem}
Let $\PoissonProcess_1$ and $\PoissonProcess_2$ be independent Poisson processes over the same space $\Omega$ with mean measures $\mu_1$ and $\mu_2$, respectively.
Then,
\begin{align*}
\PoissonProcess = \PoissonProcess_1 \cup \PoissonProcess_2
\end{align*}
is a Poisson process with mean measure
\begin{align*}
\mu = \mu_1 + \mu_2.
\end{align*}
\end{theorem}
At this level of generality, the computational aspect of the superposition theorem is not immediately apparent.
However, it does not mean that it does not have one: one possible interpretation is that if $\PoissonProcess_1$ and $\PoissonProcess_2$ are Poisson processes that we are simulating in parallel, the superposition theorem tells us how ``merging the results'' will behave.
I will use the theorem precisely in this way in \cref{sec:superposition_parallelisation} to design parallelised versions of the sampling algorithms I describe in this chapter.
\subsection{Spatio-Temporal Poisson Processes and How to Simulate Them}
\par
So far, the discussion has focused on Poisson processes over general spaces. 
However, I now introduce further structure that will be crucial in the constructions I perform in the next section. 
Indeed, in this thesis, I will always assume that the Poisson processes $\PoissonProcess$ are defined over $\Omega = \YSpace \times \nonnegReals$, for some Polish space $\YSpace$.
Therefore, each point $\pi \in \PoissonProcess$ of such a process consists of two coordinates $\pi = (Y, T)$, with $Y \in \YSpace$ and $T \in \nonnegReals$.
Since the second coordinate $T$ is a positive real number, it will be helpful to associate it with some imaginary concept of ``time.''
Then, as a complementary notion, we can associate the first coordinate with ``location.''
Accordingly, each point $\pi \in \PoissonProcess$ describes a point in some ``space-time,'' hence I will call such processes \textbf{spatio-temporal Poisson processes}.
For the thesis, I will solely consider spatio-temporal Poisson processes. 
\textbf{Hence, from now I drop the ``spatio-temporal'' qualifier and refer to these processes as ``Poisson processes'' for brevity}.
\par
Among other things, this additional structure defines a natural total order over the points of $\PoissonProcess$ by inheriting the order structure from $\nonnegReals$: for $\pi, \pi' \in \PoissonProcess$ with $\pi = (Y, T)$ and $\pi' = (Y', T')$, we define 
\begin{align*}
\pi \leq \pi' \quad \Leftrightarrow  \quad T \leq T'.
\end{align*}
Therefore, we can present the points of $\PoissonProcess$ as an ordered sequence instead of just as a set: $\PoissonProcess = ((Y_1, T_1), (Y_2, T_2), \hdots)$ with $T_i \leq T_j$ for $i \leq j$.
Due to the ordering, another helpful/popular association is with races.
That is, we could imagine that the spatial components $Y_i$ of points $\pi_i$ are entities competing in some race, and $T_i$ are times at which they ``cross the finish line.''
Therefore, I will often refer to the $T_i$ as \textbf{arrival times}, and for $n \geq 1$, I will refer to the $n$th point $\pi_n$ in the ordered presentation of $\PoissonProcess$ as the \textbf{$n$th arrival}.
By convention, I also define the ``zeroth arrival time'' as $T_0 = 0$.
\par
\textbf{Arrival time distribution.}
Let $\PoissonProcess$ now be a Poisson process over $\YSpace \times \nonnegReals$ with counting measure $\rmN$ and mean measure $\mu$.
Given the $n$th arrival time $T_n$, what can we say about the distribution of $T_{n + 1}$?
It turns out that this question can be answered by using the following often-used trick: for some time $t \geq T_n$, we have $T_{n + 1} > t$ if and only if no points arrive between $T_n$ and $t$.
Therefore:
\begin{align}
\label{eq:pois_process_inter_arrival_identity}
\Prob[T_{n + 1}\! >\! t \mid T_n] = \Prob[\rmN(\YSpace\! \times\! (T_n, t])\! =\! 0 \mid T_n] = \exp\big(\!-\mu(\YSpace \!\times\! (T_n, t])\big) \cdot \Ind[t > T_n]
\end{align}
where the second equality follows, since the number of points in $ \YSpace \times (T_n, t]$ is independent of $T_n$ by the independence property of Poisson processes and $\rmN( \YSpace \times (T_n, t]) \mid T_n$ is therefore Poisson distributed with the corresponding mean measure.
\par
This chapter will show how we can use Poisson processes to construct sampling algorithms.
Thus, I will always assume that $\YSpace$ is the sample space of interest and I have access to some probability measure $P$ over $\YSpace$ that I call the \textbf{proposal distribution} and assume that I can simulate $P$-distributed samples.
Thus, I will always start with a Poisson process $\PoissonProcess$ whose mean measure is given by a product measure: $\mu = P \otimes \lambda $, where $\lambda$ is the usual Lebesgue measure.
The product measure is characterised by the property that for $A \subseteq \YSpace$ and $B \subseteq \nonnegReals$, we have $\mu(A \times B) = P(A)\cdot \lambda(B)$.
In this situation, since $P$ is a probability measure, we get that
$\mu(\YSpace \times (T_n, t]) = t - T_n$, hence \cref{eq:pois_process_inter_arrival_identity} simplifies to
\begin{align*}
\Prob[T_{n + 1} > t \mid T_n] = \exp\big(-(t - T_n)\big) \cdot \Ind[t \geq T_n].
\end{align*}
In other words, we have the recursive relationship 
\begin{align}
\label{eq:time_homogeneous_process_arrival_identity}
T_{n + 1} \sim T_n + \Exponential(1),
\end{align}
where $\Exponential(\rho)$ denotes an exponentially distributed random variable with rate $\rho$.
Such a process is called \textit{time-homogeneous}, since the distribution of the difference of arrival times (also called the inter-arrival time) $T_{n + 1} - T_n$ is always $\Exponential(1)$.
In particular, it does not depend on $T_n$.%
\par
\added[id={CM}, comment={Corrected argument for why conditioning on a point from a Poisson process over a disjoint set does not destroy the Poisson structure.}]{\textbf{Restriction up to a random point of the process.}}
Before moving on to more computational matters, I show the following ``obvious'' result, which I will repeatedly use in the following sections.
Its essence is that if we condition on a randomly selected arrival time $T_K$ of a spatio-temporal Poisson process, the arrivals before $T_K$ still form a Poisson process.
\begin{lemma}
\label{lemma:random_restriction}
Let $\PoissonProcess = \{(Y_i, T_i)\}_{i = 1}^\infty$ be a Poisson process over $\Omega = \YSpace \times \nonnegReals$ with mean measure $\mu = P \otimes \lambda$. 
Let $K = K(\Pi) \in \Nats$ be a random variable that is allowed to depend on $\PoissonProcess$.
Now, consider the restriction of $\PoissonProcess$ to $A = \YSpace \times [0, T_K)$.
Then, conditioned on the event $\{T_K = t\}$, the restriction theorem (\cref{thm:restriction_theorem}) still applies and we have that $\PoissonProcess\vert_A$ is a Poisson process with mean measure $\mu\vert_A$.
\end{lemma}
At a high level, the result follows from the fact that Poisson processes are \textit{completely random} point processes, and hence conditioning on an arrival does not affect the probability structure of the remaining points.
\begin{proof}
To begin, let $y \in \YSpace$ be an arbitrary element in the support of $Y_{K} \mid T_K = t$, and consider conditioning $\PoissonProcess$ on the $K$th arrival $\{Y_K = y, T_K = t\}$.
Furthermore, let $\rmN$ and $\rmN_{(y, t)}$ be the random counting measures of $\PoissonProcess$ and $\PoissonProcess \mid \{Y_K = y, T_K = t\}$, respectively.
Note that by Proposition 13.1.V of \citet{daley2007introduction} (see also Exercise 13.1.9), $\rmN_{(y, t)}$ exists and is also a random counting measure.
Furthermore, by Proposition 13.1.VII of \citet{daley2007introduction}, since $\PoissonProcess$ is a Poisson process, $\rmN_{(y, t)}$ has the following special structure:
\begin{align}
\label{eq:mecke_equation}
\rmN_{(y, t)} = \rmN * \delta_{(y, t)},
\end{align}
where $\delta_{(y, t)}(A) = \Ind[(y, t) \in A)]$ is the point measure on $(y, t)$.
This identity is a form of what is known as Mecke's theorem or the Mecke equation \citep{mecke1967stationare}.
\par
Consider now any $\rmN_{(y, t)}$-measurable subsets of $B, C \subseteq A = \YSpace \times [0, T_k)$ with $B \cap C = \emptyset$.
By \cref{eq:mecke_equation}, we have 
\begin{align*}
\rmN_{(y, t)}(B) = (\rmN * \delta_{(y, t)})(B) = \rmN(B)
\end{align*}
where the second equality follows from the fact that by construction $(y, t) \not\in B$.
From this, we also have $\rmN_{(y, t)}(B) = \rmN(B) \perp \rmN(C) = \rmN_{(y, t)}(C)$.
Finally, recall that since $y$ was arbitrary, integrating out $Y_{K} \mid T_K = t$ finishes the proof.
\end{proof}
\par
\textbf{Simulating Poisson Processes.}
So far, I have only discussed Poisson processes as abstract objects satisfying \cref{def:poisson_process}, but how can we simulate them?
\Cref{eq:time_homogeneous_process_arrival_identity} suggests a neat approach to \textit{lazily}\footnote{Here, I use the term to mean \textit{lazy evaluation}, also known as \textit{call-by-need} in programming language theory, which refers to not generating all values immediately, but only when they are needed.} generate the points of a time-homogeneous Poisson process $\PoissonProcess$ with mean measure $P \otimes \lambda$ in order.
First, assume we already generated the first $n \geq 0$ points $((Y_i, T_i))_{i=1}^n$ of $\PoissonProcess$. 
To generate the next point, we first simulate an exponential random variable $\Delta_{n + 1} \sim \Exponential(1)$ and set the next arrival time as $T_{n + 1} = T_n + \Delta_{n + 1}$.
Since the mean measure of the process is a product measure, the time and location are independent: $T_{n + 1} \perp Y_{n + 1}$.%
\begin{wrapfigure}[17]{r}{0.5\textwidth}
\begin{minipage}{0.5\textwidth}
\begin{algorithm}[H]
\SetAlgoLined
\DontPrintSemicolon
\SetKwInOut{Input}{Input}\SetKwInOut{Output}{Output}
\textbf{Input:}\;
Spatial distribution $P$\;
\;
$T_0 \gets 0$\;
\For{$n = 1, 2, \hdots$}{
$\Delta_n \sim \Exponential(1)$\;
$T_n \gets T_{n - 1} + \Delta_n$\;
$Y_n \sim P$ \;
\textbf{yield } $(Y_n, T_n)$\;
}
\caption{Generating a time-homogeneous Poisson process with mean measure $P \otimes \lambda$.
See main text for the explanation of \textbf{yield}.}
\label{alg:global_pp_simulation}
\end{algorithm}%
\end{minipage}%
\end{wrapfigure}%
Hence, to obtain the location of the ${(n + 1)}$st arrival, we simulate a $P$-distributed random variate and set it as $Y_{n + 1}$.
I capture this procedure in \cref{alg:global_pp_simulation}, which uses the \texttt{yield} keyword from the \texttt{Python} programming language for defining a \textit{generator}.
This means that once we make a call to $\cref{alg:global_pp_simulation}$, it returns an object $\PoissonProcess$, which represents a Poisson process and which maintains an internal state that keeps track of the points we generated so far.
Then, to generate the next point of the process, we call $\mathtt{next}(\PoissonProcess)$, which executes the code in \cref{alg:global_pp_simulation} until it hits the \texttt{yield} keyword.
At this point, $\mathtt{next}(\PoissonProcess)$ returns the generated point and updates the internal state of the generator, which then pauses execution and waits until we call $\mathtt{next}(\PoissonProcess)$ again.
\par
While there are other, more general algorithms that can simulate essentially arbitrary Poisson processes, such as Algorithm 1 of \citet{maddison2016poisson}, I emphasize that the importance of \cref{alg:global_pp_simulation} is that it generates the points of a process in order, which will be crucial for the analysis of the sampling algorithms I construct in the next section.
\section[Sampling as Search: Developing Sampling Algorithms using Poisson Processes]{Sampling as Search: Developing Sampling \texorpdfstring{\\}{} Algorithms using Poisson Processes}
\label{sec:sampling_as_search}
\begin{figure}[t]
\centering
\begin{minipage}[t]{.33\textwidth}
\begin{algorithm}[H]
\SetAlgoLined
\DontPrintSemicolon
\SetKwIF{If}{ElseIf}{Else}{if}{:}{else if}{else}{end}
\SetKwFunction{SimulatePP}{SimulatePP}
\textbf{Input:}\;
Proposal distribution $P\hspace{-1cm}$\;
Density ratio $r = \nicefrac{dQ}{dP}\hspace{-3mm}$\;
Upper bound $M$ for $r\hspace{-1cm}$\;
\textbf{Output:}\;
$Y \sim Q$ and its index $N\hspace{-1cm}$\;
\;
\tcp{Call \cref{alg:global_pp_simulation}.}
$\Pi \gets \SimulatePP(P)$\;
\;
\For{$n = 1, 2, \hdots$}{
\;
$(Y_n, T_n) \gets \mathtt{next}(\Pi)\hspace{-1cm}$\;
\;
\;
$U_n \sim \Unif(0, 1)$\;
\;
\If{$U_n < r(Y_n) / M\hspace{-1mm}$}{
\Return{$Y_n, n$}
}
}
\vspace{0.65mm}
\caption{\\Rejection sampler.}%
\label{alg:global_rs}
\end{algorithm}
\end{minipage}%
\hfill
\begin{minipage}[t]{.33\textwidth}
\begin{algorithm}[H]
\SetAlgoLined
\DontPrintSemicolon
\SetKwIF{If}{ElseIf}{Else}{if}{:}{else if}{else}{end}
\SetKwInOut{Input}{Input}\SetKwInOut{Output}{Output}
\textbf{Input:}\;
Proposal distribution $P\hspace{-1cm}$\;
Density ratio $r = \nicefrac{dQ}{dP}\hspace{-3mm}$\;
Upper bound $M$ for $r\hspace{-1cm}$\;
\textbf{Output:}\;
$Y \sim Q$ and its index $N\hspace{-1cm}$\;
\;
\tcp{Call \cref{alg:global_pp_simulation}.}
$\Pi \gets \SimulatePP(P)$\;
$N, Y^*, U \gets (0, \perp, \infty)\hspace{-2mm}$\;
\For{$n = 1, 2, \hdots$}{
\;
$(Y_n, T_n) \gets \mathtt{next}(\Pi)\hspace{-1cm}$\;
\If{$T_n / r(Y_n) < U$}{
$U \gets T_n / r(Y_n)\hspace{-3mm}$\;
$Y^*, N \gets (Y_n, n)\hspace{-1cm}$
}
\If{$U < T_n / M\hspace{-1mm}$}{
\Return{$Y^*, N\hspace{-1mm}$}\;
}
}
\vspace{0.15mm}
\caption{\\A* sampler.}%
\label{alg:global_a_star}
\end{algorithm}
\end{minipage}%
\hfill
\begin{minipage}[t]{.33\textwidth}
\begin{algorithm}[H]
\parbox{0.96\linewidth}{\caption{Greedy Poisson rejection sampler.
}%
\label{alg:global_gprs}%
\SetAlgoLined
\DontPrintSemicolon
\SetKwIF{If}{ElseIf}{Else}{if}{:}{else if}{else}{end}
\SetKwInOut{Input}{Input}\SetKwInOut{Output}{Output}
\textbf{Input:}\;
Proposal distribution $P$\;
Density ratio $r = \nicefrac{dQ}{dP}\hspace{-1cm}$\;
Stretch function $\sigma$\;
\textbf{Output:}\;
$Y \sim Q$ and its index $N$\;
\;
\tcp{Call \cref{alg:global_pp_simulation}.}
$\Pi \gets \SimulatePP(P)$\;
\;
\For{$n = 1, 2, \hdots$}{
\;
$(Y_n, T_n) \gets \mathtt{next}(\Pi)\hspace{-1cm}$\;
\;
\;
\;
\;
\If{$T_n < \sigma\left(r(Y_n)\right)\hspace{-1mm}$}{
\Return{$Y_n, n$} \;
}
}%
}%
\end{algorithm}
\end{minipage}%
\caption{The samplers I present in \cref{chapter:rec_with_pp}, with their common parts aligned.}%
\label{fig:global_pp_algorithms}
\end{figure}%
\noindent
This section shows how the Poisson process operations I described in \cref{sec:poisson_process_operations} give rise to sampling algorithms.
\Cref{fig:global_pp_algorithms} provides a summary of the whole section by presenting rejection sampling, A* sampling \citep{maddison2014sampling} and greedy Poisson rejection sampling \citep{flamich2023gprs} side-by-side in \Cref{alg:global_rs,alg:global_a_star,alg:global_gprs}, respectively, showcasing their common structure.
My presentation and analysis of rejection sampling and A* sampling in \cref{sec:global_rs,sec:global_a_star} roughly follows the work of \citet{maddison2016poisson}, who was the first to connect rejection sampling and A* sampling.
The analysis of the description length of A* sampling is due to \citet{li2018strong}, who developed and analysed A* sampling independently of \citep{maddison2014sampling} in the context of lossy source coding.
In \cref{sec:global_gprs}, I complete the picture with my own work by developing greedy Poisson rejection sampling \citep{flamich2023gprs}.
\par
An interesting aspect of using Poisson processes to develop sampling algorithms is that 
\replaced[id={PL}, comment={PL/email/3}]{%
we can recast random variate simulation as a search problem over the points of an appropriate Poisson process.}{%
they allow to recast random variate simulation as a search problem over the points of a Poisson process.}
\par
\textbf{The sampling-as-search recipe.}
While \citet{aaronson2014equivalence} has explored similar ideas in the complexity theory literature, searching over point processes was proposed by \citet{maddison2014sampling}.
The high-level recipe to develop sampling algorithms within this ``sampling-as-search'' framework is as follows:
\begin{enumerate}
\item We wish to simulate a random variate with some target distribution $Q$ over some space $\YSpace$, but doing so ``directly'' is infeasible.
\item We have access to a second distribution $P \gg Q$, called the proposal, from which we can simulate an arbitrary number of samples.
\item We construct a time-homogeneous Poisson process $\PoissonProcess$ over $\YSpace \times \nonnegReals$ with mean measure $P \otimes \lambda$.
Then, we use one of the Poisson process operations from \cref{sec:poisson_process_operations} to transform $\PoissonProcess$ into a new process $\PoissonProcess'$, such that the location $Y$ of the first arrival of $\PoissonProcess'$ is $Q$-distributed.
\item Computationally, we simulate the original process $\PoissonProcess$ and find which of its points correspond to the first arrival of $\PoissonProcess'$.
\end{enumerate}
\par
\textbf{Channel simulation with Poisson processes.}
A benefit of these algorithms is that we can immediately turn them into channel simulation algorithms (\cref{def:channel_simulation_alg}); A* sampling and greedy Poisson rejection sampling will even turn out to be relative entropy coding algorithms (\cref{def:rec})!
To see how, let $\rvx, \rvy \sim P_{\rvx, \rvy}$ be dependent random variables over $\XSpace \times \YSpace$, and assume we wish to develop a channel simulation algorithm for $\rvx \to \rvy$.
We assume that the encoder and decoder know $P_\rvy$, and, given some source symbol $\rvx \sim P_\rvx$, the encoder wishes to transmit $\rvy \sim P_{\rvy \mid \rvx}$ to the decoder.
Now let $\Pi$ be a Poisson process over $\YSpace \times \nonnegReals$ with mean measure $P_\rvy \otimes \lambda$ and set it as the common randomness $\rvz \gets \Pi$.
In practice, this means that the encoder and decoder initiate \cref{alg:global_pp_simulation} with the same PRNG seed to simulate the points of $\Pi$.
Then, the encoder uses one of the algorithms in \cref{fig:global_pp_algorithms} to simulate a sample from $P_{\rvy \mid \rvx}$ using the shared process $\Pi$.
Conveniently, all these algorithms are selection samplers (\cref{def:exact_selection_sampler}), meaning they not only return a sample but also their index $N$ within $\Pi$.
Finally, the sender encodes $N$ using the entropy coding technique I already outlined in \cref{sec:rec_through_examples}, but repeat here for completeness.
First, \cref{lemma:li_el_gamal_bound_on_pos_int_random_variable} tells us that by using the Zeta distribution $\zeta(n \mid \alpha) \propto n^{-\alpha}$ with exponent $\alpha = 1 + 1 / \Exp[\lb N]$ as the coding distribution to entropy code $N$, then the average description length will be bounded above by ${\Exp[\lb N] + \lb(\Exp[\lb N] + 1) + 1}$.
Therefore, if we can show that $\Exp[\lb N] \leq \MI{\rvx}{\rvy} + \Oh(1)$, then this scheme yields a relative entropy coding algorithm!
\par
In the following three sections, I present rejection sampling, A* sampling and greedy Poisson rejection sampling. The layout of the sections is roughly the same: I start with a mathematical construction based on a Poisson process operation, after which I describe an algorithmic construction to simulate the resultant process.
Finally, I analyse the runtime/sample complexity of the algorithms.
\subsection{Thinning: Rejection Sampling}
\label{sec:global_rs}
\noindent
Let us now see how we can ``rediscover'' rejection sampling by thinning a Poisson process.
First, let me recap our assumptions: we have two distributions $Q \ll P$ over some space $\YSpace$ with Radon-Nikodym derivative $r = dQ/dP$.
We wish to simulate a sample from $Q$ but can only draw $P$-distributed samples.
\par
To this end, let $\Pi$ be a Poisson process over $\Omega = \YSpace \times \nonnegReals$ with mean measure $P \otimes \lambda$.
Now, the thinning theorem (\cref{thm:thinning_theorem}) tells us that if we define some function $\rho: \YSpace \times \nonnegReals \to [0, 1]$ and randomly delete points of $\PoissonProcess$, such that we keep a point $\pi \in \PoissonProcess$ with probability $\rho(\pi)$, then the remaining points form a Poisson process $\thinOp(\PoissonProcess)$ with mean measure
\begin{align*}
A \subseteq \YSpace,\, B \subseteq \nonnegReals: \quad \mu'(A \times B) = \int_B \int_A \rho(t, y) \, dP(y) \, dt. 
\end{align*}
While we can pick any appropriate $\rho$, for simplicity, let us consider functions that depend only on their spatial argument: $\rho(y)$.
This way, the mean measure of the thinned process simplifies to a product measure:
\begin{align*}
\mu'(A \times B) = \int_B \int_A \rho(y) \, dP(y) \, dt = \lambda(B) \cdot \int_A \rho(y) \, dP(y).
\end{align*}
Observe the striking similarity to the change of measure formula $Q(A) = \int_A r(y) \, dP(t)$: it seems that by setting $\rho = r$, the thinned process $\PoissonProcess'$ will have mean measure $Q \otimes \lambda$, and we are already done!
This is almost right, but there is one slight discrepancy: $r$'s range is a subset of $\nonnegReals$, not $[0, 1]$.
This highlights a crucial additional assumption we need for rejection sampling to work: we must require that $r$ be bounded.
Therefore, we assume we have some upper bound $M \geq \norm{r}_\infty$.
Then, since $1\geq r(y) / M$ for every $y$, we can set $\rho = r / M$ to obtain a valid thinning probability and lead to a thinned process $\thinOp(\PoissonProcess)$ with mean measure $1/M \cdot Q \otimes \lambda$.
\par
\textbf{Constructing the rejection sampler.}
The above argument tells us that if we find any point $\pi$ of $\thinOp(\PoissonProcess)$, the location of $\pi$ will be $Q$-distributed.
Perhaps the simplest such point we can attempt to find is the first arrival of $\thinOp(\PoissonProcess)$, which is equivalent to finding the first point of the original process $\PoissonProcess$ that was not deleted.
Thus, to implement rejection sampling, we can use \cref{alg:global_pp_simulation} to simulate the points of $\PoissonProcess$ in order.
In particular, at step $n$, we generate the next point of the process $(Y_n, T_n)$.
Then, we flip a biased coin with probability $\rho(Y_n) = r(Y_n) / M$ of landing heads.
If the coin lands tails, it means we deleted the point, and we move onto step $n + 1$; otherwise, if it lands heads, then it means we found the first point of $\thinOp(\PoissonProcess)$, and we terminate.
This procedure is captured in \cref{alg:global_rs}, and now I turn to the analysis of its runtime.
\par
\textbf{The runtime of rejection sampling.}
Due to the relative simplicity of rejection sampling, we can analytically compute the distribution of the number of samples it examines before terminating:
\begin{theorem}[Runtime of rejection sampling.]
\label{thm:global_rs_runtime}
Let $Q \ll P$ be distributions over $\YSpace$ with Radon-Nikodym derivative $r = dQ/dP$, and assume $r$ is bounded above by some constant $M > 1$.
Let $K$ denote the number of proposal samples simulated by \cref{alg:global_rs} before terminating.
Then,
\begin{align*}
K \sim \Geom(1/M).
\end{align*}
\end{theorem}
Due to its simplicity and beauty, I first present the ``usual proof'' of this statement.
\begin{proof}[The ``standard'' proof.]
For the $i$th arrival $(Y_i, T_i)$ in $\PoissonProcess$, let $\beta_i \sim \Bern(r(Y_k) / M)$ denote the Bernoulli coin flip that determines whether the point is kept or not.
Then, the probability that the point is kept is
\begin{align*}
\Exp_{Y_i \sim P}[\Prob[\beta_i = 1 \mid Y]] = \Exp_{Y \sim P}[r(Y)] / M = 1 / M.
\end{align*}
Then, since the locations $Y_i$ of the points in $\PoissonProcess$ are independent, the probability that the algorithm simulates more than $k$ samples is
\begin{align*}
\Prob[K > k] = \Prob[\forall i \in [1:k]: \beta_i = 0] = \prod_{i = 1}^k \Prob[\beta_i = 0] = (1 - 1/M)^k
\end{align*}
This is the survival function of a geometric random variable with rate $1/M$, as claimed.
\end{proof}
While the above proof is elegant, it makes no use of the Poisson process structure of \cref{alg:global_rs}.
Thus, I now provide a second proof, which also allows me to showcase a useful trick that I will use in analysing the algorithms in the coming sections too: 
\textit{conditioning on the first arrival of the transformed process, then integrating it out}.
\begin{proof}[Proof using the Poisson process structure.]
As we saw above, by the thinning theorem, the mean measure of the thinned process $\thinOp(\PoissonProcess)$ is $1/M \cdot Q \otimes \lambda$.
Hence, by \cref{eq:pois_process_inter_arrival_identity},
the first arrival time $\tau$ of $\thinOp(\PoissonProcess)$, has survival function
\begin{align*}
\Prob[\tau > t] = \exp\big(1/M \cdot  Q(\YSpace) \cdot \lambda([0, t)) \big) = \exp(-t/M),
\end{align*}
which shows that $\tau \sim \Exponential(1 / M)$.
\par
Now, recall that the first arrival time of $\thinOp(\PoissonProcess)$ is also the $K$th arrival time of $\PoissonProcess$: $\tau = T_K$.
Then, there are $K - 1$ points of $\PoissonProcess$ that arrived earlier than $T_K$.
\replaced[id={CM}, comment={replaced troublesome step}]{%
However, by \cref{lemma:random_restriction}, the number of points of $\PoissonProcess$ that fall in $\YSpace \times [0, T_K)$%
}{However, by the definition of Poisson processes, the number of points of $\PoissonProcess$ that fall in $\YSpace \times [0, T_K)$ is independent of $T_K$ and}
is Poisson distributed with rate $(P \otimes \lambda)(\YSpace \times [0, T_K)) = T_K$!
\par
Collecting the information we derived so far, we have that
\begin{align*}
T_K \sim \Exponential(1 / M) \quad \text{and}\quad (K - 1 \mid T_K) \sim \Pois(T_k).
\end{align*}
Therefore,
\begin{align*}
\Prob[K - 1 = k] 
&= \Exp_{T_K}[\Prob[K - 1 = k \mid T_K]] \\
&= \frac{1}{k!}\Exp_{T_K}[T_K^k \cdot e^{-T_K}] \\
&= \frac{1}{M}\frac{1}{k!}\int_0^\infty t^k e^{-t(1 + 1/M)} \, dt \\
&= \frac{1}{M}\left(1 - \frac{1}{M}\right)^k,
\end{align*}
where the last equality follows from integrating by parts $k$ times.
Therefore, this shows that $\Prob[K = k] = \frac{1}{M}\left(1 - \frac{1}{M}\right)^{k - 1}$, which shows that $K$ is a geometric random variable with failure probability $1/M$, as desired.
\end{proof}
While this second proof method is slightly longer in this case than the first one, its power lies in its applicability even when the termination probabilities are not independent. 
Indeed, I will leverage this argument to analyse the algorithms in the following two sections as well.
\par
As a small side note, we could generalise rejection sampling and re-introduce dependence between the acceptance criteria by making the thinning probability $\rho$ depend on the arrival time as well.
Then, it would be interesting to investigate whether the necessary criteria for $\rho$ to define a valid rejection sampler for $Q$ could be related to the ones derived by \citep{casella2004generalized} for their version of generalised rejection samplers.
\par
As a second side note, observe that since rejection sampling is a selection sampler, \cref{thm:global_rs_runtime} shows that the lower bound I derived for the sample complexity of selection samplers in \cref{thm:selection_sampler_runtime_lower_bound} is maximally tight and exactly achievable: if we know the optimal lower bound $M = \norm{r}_\infty$, we have $\Exp[K] = \norm{r}_\infty = \expb(\infD{Q}{P})$, and we cannot do better.
\par
\textbf{Channel simulation with rejection sampling.}
\Cref{thm:global_rs_runtime} shows that the distribution of the index $K$ returned by \cref{alg:global_rs} is geometric with rate $1/M$.
Now, if we wish to use rejection sampling for channel simulation, we also have to worry about how efficiently we can encode $K$.
Since in channel simulation we get to design the channel, we might as well do it in such a way that we can work out the best possible upper bound $M$ to use, which is $M = \norm{r}_\infty$.
In particular, for $\rvx, \rvy \sim P_{\rvx, \rvy}$ and $r_\rvx = dP_{\rvy \mid \rvx}/dP_\rvy$,
I have already shown in \cref{eq:rs_expected_index_log_upper_bound} that
encoding $K$ with a Zeta distribution is not efficient since we can only obtain the loose bound $\Exp[\lb K] \leq e \cdot \Exp_{\rvx}[\lb \norm{r_\rvx}_\infty] + e$.
However, it is instructive to consider why $\Exp[\lb K]$ is not lower and how we could design algorithms for which it is lower.
Here, it turns out that the very thing that enabled the simple proof of \cref{thm:global_rs_runtime} is the culprit: the independent acceptance criteria.
To compress the index better, we would ideally like it to be more predictable; in the context of the Zeta coding approach I outlined in the introduction of this section, this is essentially equivalent to ensuring that the distribution $K$ is skewed towards $0$ as much as possible.
To illustrate why rejection sampling goes against this desideratum, assume the sample space $\YSpace$ is finite.
Then, it is perfectly possible that rejection sampling proposes the same element of the sample space $y \in \YSpace$ twice in a row, rejecting $y$ the first time and accepting it the second time.
Note that this could happen precisely because the acceptance decisions are independent, hence the algorithm has no way of ``remembering'' to instantly reject $y$ the second time it is proposed.
This discrepancy is what A* sampling and greedy Poisson rejection sampling eliminate, which will also lead to further benefits, as I show when developing approximate sampling algorithms in \cref{sec:approximate_sampling}. 
\subsection{Mapping: A* Sampling}
\label{sec:global_a_star}
\noindent
Analogously to \cref{sec:global_rs}, I now describe how we could ``rediscover'' A* sampling, whose construction uses mapping instead of thinning.
The setting is the same as before: we have two probability measures $Q \ll P$ over $\YSpace$ with $r = dQ/dP$; we want to simulate a $Q$-distributed sample but can only draw samples from $P$.
\par
First, I specialise in the mapping theorem specialised for spatio-temporal processes.
\begin{proposition}[A special case of \cref{thm:mapping_theorem}.]
Let $\PoissonProcess$ be a Poisson process over $\Omega = \YSpace \times \nonnegReals$ with mean measure $P \otimes \lambda$, and let $f: \YSpace \times \nonnegReals \to \YSpace \times \nonnegReals$ be a measurable function such that $f \pushfwd (P \otimes \lambda)$ is non-atomic.
Then, the points
\begin{align*}
f(\PoissonProcess) = \{f(\pi) \in \Omega \mid \pi \in \PoissonProcess\}
\end{align*}
form a Poisson process with mean measure $f \pushfwd (P \otimes \lambda)$.
\end{proposition}
This result holds for very general maps $f$. 
Hence, similarly to \cref{sec:global_rs}, we now consider some ``common sense'' specialisation for it.
First, observe that if $f$ is one-to-one, then $f \pushfwd (P \otimes \lambda)$ is guaranteed to be non-atomic; hence, let us only consider one-to-one maps from now on.
Then, by definition
\begin{align*}
A \subseteq \YSpace,\,B \subseteq \nonnegReals: \quad (f \pushfwd (P \otimes \lambda))(A \times B) = \int_{f^{-1}(A \times B)} d(P \otimes \lambda)(y, t).
\end{align*}
What makes the pushforward measure complicated is that $f^{-1}(A \times B)$ does not ``factorise,'' which means that we cannot easily split the integral on the right-hand side into two.
Thus, let us consider how we might restrict $f(y, t)$ further to obtain such a factorisation: since $f$ maps $(y, t) \mapsto (y', t')$, what should $y'$ be?
As we are interested in constructing a map to create a sampling algorithm, let us approach this from a practical, computational viewpoint. 
Since we are not assuming any a priori structure on $\YSpace$ that would imply that it possesses some special element or would tell us how its elements can be transformed, the only choice we have is the identity map $y' = y$, making $f(y, t) = (y, t')$ essentially the only possible choice!
What does this imply for the mean measure?
For some fixed $y \in \YSpace$, define $f_y(t) = f(y, t)$.
Then, note that since $f(A \times B) = \{(y, f_y(t)) \in \Omega \mid (y, t) \in A \times B \}$, we have $f^{-1}(A \times B) = \{(y, f_y^{-1}(t)) \in \Omega \mid (y, t) \in A \times B \}$, where $f_y^{-1}$ is unique, since we are assuming $f$ is one-to-one.
Excitingly, we can use this now to split the integral in the pushforward measure:
\begin{align*}
A \subseteq \YSpace,\,B \subseteq \nonnegReals: \quad (f \pushfwd (P \otimes \lambda))(A \times B) 
&= \int_{f^{-1}(A \times B)} d(P \otimes \lambda)(y, t) \\
&= \int_{\{(y, f_y^{-1}(t)) \in \Omega \mid (y, t) \in A \times B \}} d(P \otimes \lambda)(y, t) \\
&= \int_A \int_{f_y^{-1}(B)} \, dt \, dP(y).
\end{align*}
Now, notice again the similarity of the last equation with the change of measure formula $(Q \otimes \lambda)(A \times B) = \int_A \int_B \, r(y) \,dt \, dP(y)$.
This suggests that we should pick $f_y$ such that 
\begin{align*}
\int_{f_y^{-1}(B)} \, dt = \int_B \, r(y) \,dt = r(y) \cdot \lambda(B),
\end{align*}
meaning that $f_y^{-1}$ should scale the size of $B$ by a factor of $r(y)$.
However, if we now finally also require that $f_y$ be continuous, the only map that scales an arbitrary subset of $\nonnegReals$ by a fixed constant is, well, the scaling map $f_y^{-1}(t) = t \cdot r(y)$, which means we have found our function!
To reiterate, we now see that by the mapping theorem, picking $f_y^{-1}(t) = t \cdot r(y)$, hence $f(y, t) = (y, t / r(y))$ will ensure that the mapped process $f(\PoissonProcess)$ is a Poisson process with mean measure $Q \otimes \lambda$.
\par
\textbf{Constructing the A* sampler.}
As was the case for rejection sampling, when it comes to designing a sampling algorithm based on the above construction, we are free to search for any point of $f(\PoissonProcess)$, so long as we can always determine their arrival index in $\PoissonProcess$.
However, from a computational and compression perspective, finding the first arrival of $f(\PoissonProcess)$ makes the most sense.
The arrival times of $f(\PoissonProcess)$ are given by $T_n / r(Y_n)$ for $(Y_n, T_n) \in \PoissonProcess$, hence searching for the first point of $f(\PoissonProcess)$ can be formulated as an optimisation problem:
\begin{align*}
(Y^*, T^*) = \argmin_{(Y, T) \in \PoissonProcess}\{T / r(Y)\}.
\end{align*}
Due to the spatio-temporal structure of $\PoissonProcess$, it is most convenient to search through its points in order.
Thus, let me now reformulate this as a search over indices:
\begin{align*}
(Y^*, T^*) &= (Y_N, T_N) \\
N &= \argmin_{n \in \Nats}\{T_n / r(Y_n) \mid (Y_n, T_n) \in \PoissonProcess \}.
\end{align*}
But $\PoissonProcess$ has infinitely many points, so can we even compute this minimum?
Fortunately, when $r$ is bounded from above, it is indeed possible!
This assumption is not restrictive since \cref{thm:selection_sampler_runtime_lower_bound} tells us that an unbounded $r$ will cause the sample complexity to be infinite.
Hence, assuming the boundedness of $r$ is necessary to obtain a sampler with finite sample complexity (it will also be sufficient).
Therefore, let us assume that there is some upper bound $M > 1$ such that $\norm{r}_\infty < M$.
The trick to constructing a sampling algorithm, then, is to make the following two observations: 
\begin{enumerate}
\item \textbf{Lower bound:} For all points $(Y, T) \in \PoissonProcess$ we have $r(Y) < M$, hence ${T / r(Y) > T / M}$, and in particular $T^* / r(Y^*) > T^* / M$.
\item \textbf{Upper bound:} For any $k \geq 1$, the minimum over the first $k$ arrivals is greater than the global minimum:
\begin{align*}
\min_{n \in \Nats}\{T_n / r(Y_n) \mid (Y_n, T_n) \in \PoissonProcess \} \leq \min_{n \in [1:k]}\{T_n / r(Y_n) \mid (Y_n, T_n) \in \PoissonProcess \}.
\end{align*}
\end{enumerate}
For convenience, let $\tau_k = \min_{n \in [1:k]}\{T_n / r(Y_N) \mid (Y_n, T_n) \in \PoissonProcess \}$ denote the minimum over the first $k$ arrivals of $\PoissonProcess$ and let $\tau = \tau_\infty$ denote the global minimum.
We can now implement A* sampling as follows: at step $k$ we simulate the $k$th arrival $(Y_k, T_k)$ of $\PoissonProcess$ using \cref{alg:global_pp_simulation} and compute $\tau_k$, which we know to be an upper bound to $\tau$.
But how can we tell when to stop?
Observe that at each step, the smallest achievable global minimum increases: if $T_k / r(Y_k) = \tau$ is the global minimum, it has to be at least $T_k / M$ since $r(Y_k)$ is at most $M$.
However, since $T_k \to \infty$ almost surely as $k \to \infty$, we also have $T_k / M \to \infty$ almost surely.
On the other hand, since $f(\PoissonProcess)$ is time-homogeneous with mean measure $Q \otimes \lambda$, by \cref{eq:time_homogeneous_process_arrival_identity} we have that $\tau \sim \Exponential(1)$ and is hence almost surely finite.
These two facts taken together mean that there must exist some step $K$ for which $T_{K + 1} / M > \tau_K$, meaning that $T_{K + 1}$ and later arrivals are now too large and none of the future arrivals can improve upon the current minimum $\tau_K$.
In turn, this means that we must have found the global minimum: $\tau_K = \tau$.
At this point, we terminate and return the index $N$ alongside the arrival $Y_N, T_N$ that achieved the smallest ratio $T_N/r(Y_N) = \tau_K$ in the first $K$ steps.
I describe this procedure in \cref{alg:global_a_star}, and I turn to its analysis next.
\par
\textbf{The runtime of A* sampling.}
The above argument regarding $T_{k} / M$ exceeding $\tau_k$ proves that the runtime of A* sampling is finite.
Next, I fully characterise its runtime.
\begin{theorem}[Runtime of A* sampling.]
\label{thm:global_a_star_runtime}
Let $Q \ll P$ be distributions over $\YSpace$ with Radon-Nikodym derivative $r = dQ/dP$, and assume $r$ is bounded above by some constant $M > 1$.
Let $K$ denote the number of proposal samples simulated by \cref{alg:global_a_star} before terminating.
Then,
\begin{align*}
K \sim \Geom(1/M).
\end{align*}
\end{theorem}
Here, I give a new proof of this theorem that showcases the power of the Poisson-process-based trick in the second proof I gave for \cref{thm:global_rs_runtime}: condition on the first arrival, then integrate it out.
\begin{proof}
Let $\PoissonProcess$ be the Poisson process over $\Omega = \YSpace \times \nonnegReals$ with mean measure $P \otimes \lambda$ that \cref{alg:global_a_star} transforms with the map $f(y, t) =  (y, t / r(y))$.
Moreover, let $\tau$ denote the first arrival of the transformed process, whose mean measure is $Q \otimes \lambda$.
As we saw already, $\tau \sim \Exponential(1)$ by \cref{eq:time_homogeneous_process_arrival_identity}.
Now, by the definition of $K$, it is the number of points $(Y, T) \in \PoissonProcess$ such that $T / M < \tau$, that is, the number of points such that the termination criterion of the algorithm is not violated.
Now, let us remove the point $(Y_N, T_N) \in \PoissonProcess$ that corresponds to the first arrival of the mapped process $f(\PoissonProcess)$.
Then, conditioned on $\tau$, we have
\begin{align*}
K - 1 &= \abs{\{(Y, T) \in \PoissonProcess \mid \tau < T / r(Y) \text{ and } T / M < \tau\}} \\
&= \abs{\{(Y, T) \in \PoissonProcess \mid \tau \cdot r(Y) < T < \tau \cdot M\}}.
\end{align*}
The first condition holds since ${\tau = \min_{(Y', T') \in \PoissonProcess}\{T' / r(Y')\} < T / r(Y)}$ with strict inequality, since $(T, Y)$ is by definition not the point corresponding to the global minimum).
Now, by \replaced[id={CM}, comment={replaced troublesome step}]{\cref{lemma:random_restriction}}{the independence property of Poisson processes}, the number of points that fall in the set $\{(Y, T) \in \PoissonProcess \mid \tau < T / r(Y) \text{ and } T / M < \tau\}$ is Poisson distributed with mean
\begin{align*}
(P \otimes \lambda)(\{(y, t) \in \Omega \mid \tau \cdot r(y) < t < \tau \cdot M\}) &= \int_\YSpace \int_{\tau \cdot r(y)}^{\tau \cdot M} \, dt \, dP(y) \\
&= \int_\YSpace\tau \cdot (M - r(y)) \, dP(y) \\
&= \tau \cdot (M - 1).
\end{align*}
Thus, we have $\tau \sim \Exponential(1)$ and $(K - 1 \mid \tau) \sim \Pois(\tau \cdot (M - 1))$.
Now, if we integrate out the first arrival $\tau$, we get
\begin{equation}
\begin{aligned}
\label{eq:global_a_star_runtime_integrating_out_first_arrival}
\Prob[K - 1 = k] &= \Exp_{\tau}[\Prob[K - 1 = k \mid \tau]] \\
&= \frac{1}{k!}\Exp_{\tau}[(\tau \cdot (M - 1))^k \cdot \exp(-\tau \cdot (M - 1))] \\
&= \frac{1}{k!} \int_0^\infty (t \cdot (M - 1))^k \cdot \exp(-tM) \, dt \\
&= \frac{1}{M} \left(1 - \frac{1}{M}\right)^k
\end{aligned}
\end{equation}
where the last equality follows from integration by parts $k$ times.
Akin to rejection sampling, this shows that $\Prob[K = k] = \frac{1}{M} \left(1 - \frac{1}{M}\right)^{k - 1}$, hence $K$ is a geometric random variable with failure probability $1/M$, as claimed.
\end{proof}
While \cref{thm:global_a_star_runtime} characterises the runtime $K$ of A* sampling, remember that it is not the index $N$ of the selected sample.
However, it turns out that the same trick also works to characterise this index, which I do next.
\begin{theorem}[Conditional distribution of the selected index, \citep{li2021unified}.]
\label{thm:global_a_star_conditional_index_distribution}
Let $Q \ll P$ be distributions over $\YSpace$ with Radon-Nikodym derivative $r = dQ/dP$, and assume $r$ is bounded above by some constant $M > 1$.
Let $N$ denote the index of the sample selected by \cref{alg:global_a_star} before terminating, and let $Y_N \sim Q$ be the location of the selected point.
Then,
\begin{align*}
N \mid Y_N \sim \Geom(1/\Exp_{Y \sim P}[\max\{r(Y_N), r(Y)\}]).
\end{align*}
\end{theorem}
The proof of this was first given by \citet{li2021unified}, which also inspired my proofs of \cref{thm:global_rs_runtime,thm:global_a_star_runtime}.
Thus, I repeat their proof here to showcase the power of the condition-on-the-first-arrival trick further.
\begin{proof}
Let $\PoissonProcess$ be the Poisson process over $\YSpace \times \nonnegReals$ with mean measure $P \otimes \lambda$ that \cref{alg:global_a_star} uses for its proposals, let $f$ denote the map that shifts the arrivals and let $f(\PoissonProcess)$ be the transformed process with mean measure $Q \otimes \lambda$.
Denote the first arrival of $f(\PoissonProcess)$ as $\Upsilon, \tau$.
As I have shown in the proof of \cref{thm:global_a_star_runtime}, $\tau \sim \Exponential(1)$.
Moreover, from the mean measure of $f(\PoissonProcess)$ we also see that $\Upsilon \sim Q$, and note the correspondence $\Upsilon = Y_N$ and $T_N = \tau \cdot r(\Upsilon)$ between the points of $\PoissonProcess$ and $f(\PoissonProcess)$.
Now, consider the $N - 1$ points of $\PoissonProcess$ that arrive before $T_N$:
\begin{align*}
N - 1 &= \abs{\{(Y, T) \in \PoissonProcess \mid \tau < T / r(Y), T < T_N\}} \\
&= \abs{\{(Y, T) \in \PoissonProcess \mid \tau \cdot r(Y) < T < \tau \cdot r(\Upsilon) \}}.
\end{align*}
The first condition holds since $\tau = \min_{(Y', T') \in \PoissonProcess}\{T' / r(Y')\} < T / r(Y)$ with strict inequality, since by definition $T / r(Y)$ is not the minimum. 
Now, by \replaced[id={CM}, comment={replaced troublesome step}]{\cref{lemma:random_restriction}}{Poisson process independence}, the number of points that fall in $\abs{\{(T, Y) \in \PoissonProcess \mid \tau < T / r(Y), T < T_N\}}$ is Poisson distributed with mean
\begin{align*}
(P \otimes \lambda)(\{(y, t) \in \Omega \mid \tau \cdot r(y) < t < \tau \cdot r(\Upsilon) \}) 
&= \int_\YSpace \Ind[r(y) < r(\Upsilon)] \cdot \int_{\tau \cdot r(y)}^{\tau \cdot r(\Upsilon)} \, dt \, dP(y) \\
&= \tau \cdot \Exp_{Y \sim P}[(r(\Upsilon) - r(Y))_+],
\end{align*}
where recall that $(x)_+ = \max\{0, x\}$.
Hence, we have
\begin{align*}
\tau \sim \Exponential(1) \quad \text{and} \quad (N - 1 \mid \Upsilon, \tau) \sim \Pois(\tau \cdot \Exp_{Y \sim P}[(r(\Upsilon) - r(Y))_+])
\end{align*}
Now, set $\alpha = \Exp_{Y \sim P}[(r(\Upsilon) - r(Y))_+]$.
Analogously to \cref{eq:global_a_star_runtime_integrating_out_first_arrival}, we can now integrate out the first arrival $\tau$ to get
\begin{align*}
\Prob[N - 1 = n \mid \Upsilon] &= \Exp_{\tau}[\Prob[N - 1 = n \mid \Upsilon, \tau]] \\
&= \frac{1}{1 + \alpha}\left(1 - \frac{1}{1 + \alpha}\right)^n.
\end{align*}
This demonstrates that $\Prob[N = n \mid \Upsilon] = \frac{1}{1 + \alpha}\left(1 - \frac{1}{1 + \alpha}\right)^{n - 1}$, which shows that $N \mid \Upsilon$ is geometric with failure probability $1 / (1 + \alpha) = 1 / (1 + \Exp_{Y \sim P}[ (r(\Upsilon) - r(Y))_+])$.
Finally, noticing that 
\begin{align*}
1 \!+\! \Exp_{Y \sim P}[(r(\Upsilon) - r(Y))_+] &= \Exp_{Y \sim P}[r(Y) + (r(\Upsilon) - r(Y))_+] = \Exp_{Y \sim P}[\max\{r(\Upsilon), r(Y)\}]
\end{align*}
finishes the proof.
\end{proof}
\par
\textbf{Relative entropy coding with A* sampling.}
To encode a sample with A* sampling, the sender compresses the selected index $N$ rather than the total number of simulated samples $K$.
This already shows that good things are happening: while the runtimes of A* sampling and rejection sampling are equal in distribution by \cref{thm:global_rs_runtime,thm:global_a_star_runtime}, $N \leq K$ always, hence we should already expect a gain in compression efficiency.
As I show next, the description length of the index produced by A* sampling is optimal.
\begin{theorem}[Relative entropy coding with A* sampling.]
\label{thm:global_a_star_codelength}
Let $\rvx, \rvy \sim P_{\rvx, \rvy}$ be dependent random variables over the space $\XSpace \times \YSpace$.
Moreover, let $\PoissonProcess$ be a Poisson process over $\YSpace \times \nonnegReals$ with mean measure $P_{\rvy} \otimes \lambda$.
\par
Consider the following channel simulation protocol: the sender and receiver set $\rvz \gets \PoissonProcess$ as their common randomness.
Then, upon receiving a source symbol $\rvx \sim P_\rvx$, the sender uses A* sampling (\cref{alg:global_a_star}) to simulate $\rvy \sim P_{\rvy \mid \rvx}$ using the points of $\PoissonProcess$ as proposals.
Now, let $N$ denote the index of $\rvy$ in $\PoissonProcess$, which the sender encodes $N$ using a Zeta distribution $\zeta(n \mid \alpha) \propto n^{-\alpha}$ with exponent $\alpha = 1 + 1 / (\MI{\rvx}{\rvy} + 1)$.
Then, the average description length of this protocol is upper bounded by
\begin{align*}
\MI{\rvx}{\rvy} + \lb(\MI{\rvx}{\rvy} + 2) + 3 \text{ bits}.
\end{align*}
\end{theorem}
The argument I present here is based on the works of \citet{li2021unified} and \citet{li2024pointwise}.
\begin{proof}
By \cref{lemma:li_el_gamal_bound_on_pos_int_random_variable}, using a Zeta distribution with exponent $\alpha = 1 + 1/\Exp[\lb N]$ to encode an positive integer-valued random variable $N$ yields a scheme with average description length upper bounded by
\begin{align}
\label{eq:a_star_codelength_zeta_efficiency}
\Exp[\lb N] + \lb(\Exp[\lb N] + 1) + 2.
\end{align}
Hence, all we need to do to show the desired result is to bound $\Exp[\lb N]$.
To this end, let $r_\rvx = dP_{\rvy \mid \rvx}/dP_\rvy$ and let $\Upsilon \sim P_{\rvy \mid \rvx}$ denote the first arrival location of the mapped process, as in the proof of \cref{thm:global_a_star_conditional_index_distribution}.
Now, observe that
\begin{align}
\Exp[\lb N \mid \rvx]
&= \Exp_{\Upsilon \sim P_{\rvy \mid \rvx}}[\Exp[\lb N \mid \Upsilon]\mid \rvx] \nonumber\\
&\leq \Exp_{\Upsilon \sim P_{\rvy \mid \rvx}}[\lb \Exp[N \mid \Upsilon]\mid \rvx] \tag{Jensen}\\
&= \Exp_{\Upsilon \sim P_{\rvy \mid \rvx}}[\lb \Exp_{Y \sim P_{\rvy}}[\max\{r_\rvx(\Upsilon), r_\rvx(Y)\} \mid \Upsilon] \mid \rvx]  \tag{by \cref{thm:global_a_star_conditional_index_distribution}} \\
&\leq \Exp_{\Upsilon \sim P_{\rvy \mid \rvx}}[\lb \Exp_{Y \sim P_{\rvy}}[r_\rvx(\Upsilon) + r_\rvx(Y) \mid \Upsilon]\mid \rvx ] \nonumber \\
&= \Exp_{\Upsilon \sim P_{\rvy \mid \rvx}}[\lb (r_\rvx(\Upsilon) + 1)\mid \rvx] \nonumber \\
&= \KLD{P_{\rvy \mid \rvx}}{P_\rvy} + \Exp_{\Upsilon \sim P_{\rvy \mid \rvx}}[\lb (1 + 1 / r_\rvx(\Upsilon))\mid \rvx] \nonumber \\
&= \KLD{P_{\rvy \mid \rvx}}{P_\rvy} + \Exp_{\Upsilon \sim P_{\rvy \mid \rvx}}[\lb 2^{1 / r_\rvx(\Upsilon)}\mid \rvx] \tag{$1 + x \leq 2^x$ for $x \geq 1$} \\
&=\KLD{P_{\rvy \mid \rvx}}{P_\rvy} + 1. \nonumber
\end{align}
Taking expectation over $\rvx$, we find that $\Exp[\lb N] \leq \MI{\rvx}{\rvy} + 1$, thus plugging it into \cref{eq:a_star_codelength_zeta_efficiency} finishes the proof.
\end{proof}
\par
\textbf{A historical note.}
A* sampling was first discovered by \citet{maddison2014sampling}, who described not only the basic version I explain here but also its branch-and-bound version, which I discuss in \cref{sec:branch_and_bound_samplers}.
Originally, \citet{maddison2014sampling} formulated A* sampling using so-called Gumbel processes, which were shown by \citet{maddison2016poisson} to have a one-to-one correspondence to Poisson processes.
The advantage of formulating A* sampling (and all other algorithms in this thesis, for that matter) is that the theory is more straightforward and more intuitive.
On the other hand, simulating Gumbel processes on a computer is more numerically stable.
Hence, as I discuss in \cref{sec:implementation_considerations}, a good recipe is to design algorithms using Poisson processes but implement them using Gumbel processes.
Independently, \Cref{alg:global_a_star} was discovered by \citet{li2018strong}, who showed that it could be used for relative entropy coding.
\subsection{Restriction: Greedy Poisson Rejection Sampling}
\label{sec:global_gprs}
Finally, I present greedy Poisson rejection sampling (GPRS), which completes the picture of Poisson process operations giving rise to sampling algorithms.
Once again, the setting is the same: we have two probability measures $Q \ll P$ over $\YSpace$ with $r = dQ/dP$ and wish to simulate a $Q$ distributed sample using samples from $P$.
\begin{wrapfigure}[24]{r}[0pt]{0.5\textwidth}%
\includegraphics{4-RelativeEntropyCodingWithPoissonProcesses/img/global_gprs_illustration.tikz}
\caption[Illustration of greedy Poisson rejection sampling]{Illustration of greedy Poisson rejection sampling for a Gaussian target ${Q = \Normal(1, 0.25^2)}$ and Gaussian proposal distribution $P = \Normal(0, 1)$, with the time axis truncated to the first $17$ units.
The algorithm searches for the first arrival of a spatio-temporal Poisson process $\PoissonProcess$ with mean measure $P \otimes \lambda$ under the graph of $\varphi = \sigma \circ r$ indicated by the \textbf{thick dashed black line} in each plot.
Here, $r = dQ/dP$ is the target-proposal density ratio, and $\sigma $ is given by \Cref{eq:stretch_function_integral_identity}.
The green circle (\tikzcircle[green, fill=green]{3pt}) shows the first point of $\Pi$ that falls under $\varphi$, and is accepted.
All other points are rejected, as indicated by red crosses ({\color{red}\xmark}).
In practice, \cref{alg:global_gprs} does not simulate points of $\Pi$ that arrive after the accepted arrival.
}
\label{fig:global_gprs_illustration}
\end{wrapfigure}
\par
I constructed GPRS in my paper \citet{flamich2023gprs} by following the recipe I already outlined in \cref{sec:global_rs,sec:global_a_star}.
Starting from a base Poisson process, I applied the restriction operation, worked out the mean measure of the resultant process, and set up the restriction such that the resultant process had a point with $Q$-distributed location.
However, GPRS breaks a pattern compared to the first two algorithms I presented. When constructing rejection and A* sampling, I used thinning and mapping to transform a Poisson process with mean measure $P \otimes \lambda$ into another process whose mean measure was a product measure, and in particular, the spatial component had the desired distribution $Q$.
This meant that we could choose to find any point of the resultant process; it was only out of computational considerations that we decided to search for the first arrival in both cases.
In stark contrast, while GPRS's transformed process is still a Poisson process, its mean measure is no longer a product measure.
Furthermore, we have no freedom in choosing which point to search for: the first arrival is the only point in the resultant process with a $Q$-distributed location.
\par
\textbf{The construction with restriction.}
To begin, I first specialise the restriction theorem for spatio-temporal processes.
\begin{proposition}[A special case of \cref{thm:restriction_theorem}.]
\label{prop:restriction_theorem}
Let $\PoissonProcess$ be a Poisson process over $\Omega = \YSpace \times \nonnegReals$ with mean measure $P \otimes \lambda$, and let $A, B \subseteq \Omega$ be measurable under $P \otimes \lambda$.
Then,
\begin{align*}
\PoissonProcess\vert_A = \PoissonProcess \cap A
\end{align*}
is a Poisson process over $\Omega$ with mean measure
\begin{align*}
(P \otimes \lambda)\vert_A(B) = (P \otimes \lambda)(A \cap B).
\end{align*}    
\end{proposition}
As before, this result is far too general on its own, and this time we need to look perhaps a bit further for some inspiration.
The first clue comes from the domain and range of the density ratio $r: \YSpace \to \nonnegReals$.
In particular, consider the graph of $r$:
\begin{align*}
\graph(r) = \{(y, r(y)) \in \YSpace \times \nonnegReals \}.
\end{align*}
Note that \replaced[id={CM}, comment={fixed typo}]{$\graph(r) \subseteq \Omega$}{$G(r) \subseteq \Omega$}, i.e.\ we can embed the graph of $r$ in the same space in which $\PoissonProcess$'s points live!
This now leads to our second clue: let $F$ be some probability distribution over $[0, 1]$ with density $f: [0, 1] \to \nonnegReals$.
Now, consider the \textit{hypograph} of $f$, i.e.\ the set of points that fall under the graph of $f$:
\begin{align*}
\hypograph(f) = \{(x, y) \in [0, 1] \times \nonnegReals \mid y \leq f(x)\}.
\end{align*}
A standard fact is that if we draw a pair of coordinates $(X, Y) \sim \Unif(\hypograph(f))$, then $X \sim F$.
To see this, note that for a subset $A \subseteq [0, 1]$, we have
\begin{align*}
\Prob[X \in A] = \Prob[(X, Y) \in (A \times \nonnegReals) \cap \hypograph(f)] = \int_A \int_0^{f(x)} \, dy \, dx = \int_A f(x) \, dx.
\end{align*}
Can we replicate this somehow with Poisson processes?
To this end, consider the hypograph of $r$:
\begin{align*}
\hypograph(r) = \{(y, t) \in \Omega \mid t \leq r(y)\}.
\end{align*}
This allows me to define the following terminology: I say that ``a point $\pi$ falls under the graph of $r$'' to mean that $\pi \in \hypograph(r)$.
Furthermore, let me introduce the following shorthand notation for restriction for an arbitrary function $\varphi: \YSpace \to \nonnegReals$:
\begin{align*}
\PoissonProcess\vert_\varphi = \PoissonProcess\vert_{\hypograph(\varphi)} \quad \text{and} \quad (P \otimes \lambda)\vert_\varphi = (P \otimes \lambda)\vert_{\hypograph(\varphi)}.
\end{align*}
\par
What happens if we restrict our base process $\PoissonProcess$ to $\hypograph(r)$?
Unfortunately, this cannot possibly work: by \cref{prop:restriction_theorem}, mean measure of all of $\Omega$ under the restricted process is
\begin{align*}
(P \otimes \lambda)\vert_{r}(\Omega) = (P \otimes \lambda)(\hypograph(r)) = \int_{\YSpace}\int_0^{r(y)} \, dt \,dP(y) = \int_{\YSpace}r(y) \,dP(y) = 1.
\end{align*}
This, in turn, means that the probability that no points fall under the graph of $r$ is $1 / e > 0$, and we cannot base an algorithm on a process that might not even have any points! 
However, this prompts a new question: could some other function work even if restricting under $r$ does not work?
\par
\textbf{Restriction to the hypograph of a function.}
In general, consider an arbitrary measurable function $\varphi: \YSpace \to \nonnegReals$, such that its image has the same cardinality as $r$'s image: $\abs{\varphi(\YSpace)} = \abs{r(\YSpace)}$.
By definition, this means that there is a bijection $f$ between the images $\varphi(\YSpace)$ and $r(\YSpace)$.
Finally, since the $\varphi(\YSpace), r(\YSpace) \subseteq \nonnegReals$, there is a total order on the domain and range of the bijection $f$.
Hence, there exist strictly monotonically increasing bijections $\tilde{\sigma}: r(\YSpace) \to \varphi(\YSpace)$ between $r(\YSpace)$ and $\varphi(\YSpace)$.
This means that there always exists a strictly monotonically increasing $\tilde{\sigma}$ such that $\varphi = \tilde{\sigma} \circ r$.
Finally, since $\tilde{\sigma}$ is strictly monotonically increasing, it always has a continuous, monotonically increasing extension $\sigma: \nonnegReals \to \nonnegReals$ to the entirety of $\nonnegReals$ and is therefore also invertible on the entirety of $\nonnegReals$.
The significance of this argument is that we can reduce the investigation of the properties of arbitrary functions $\varphi: \YSpace \to \nonnegReals$ to invertible functions $\sigma: \nonnegReals \to \nonnegReals$.
\par
Now, fix some $\varphi: \YSpace \to \nonnegReals$ and consider $\hypograph(\varphi) = \{(y, t) \in \Omega \mid t \leq \varphi(y)\}$. 
What can we say if we now restrict the process $\PoissonProcess$ with mean measure $P \otimes \lambda$ to $\hypograph(\varphi)$?
What is the distribution of the first arrival of the restricted process?
Before we can derive it, we need to address the issue of existence: as we have seen before, if $(P \otimes \lambda)(\hypograph(\varphi)) < \infty$, there is a positive probability that the restricted process has no points, which is something we need to avoid.
Hence, we need to require that $(P \otimes \lambda)(\hypograph(\varphi)) = \infty$, for which a necessary and sufficient condition is that $\varphi$ is unbounded.
To see why, assume $\varphi$ is bounded above by some constant $M > 0$.
Then, $(P \otimes \lambda)(\hypograph(\varphi)) \leq (P \otimes \lambda)(\YSpace \times [0, M]) = M < \infty$.
\par
Under these assumptions, I now characterise the distribution of the first arrival of the restricted process under $\varphi$.
\begin{lemma}[Distribution of the first arrival of a Poisson process restricted to an unbounded hypograph]
\label{lemma:first_arrival_rd_derivative_of_restricted_process}
Let $\varphi: \YSpace \to \nonnegReals$ be unbounded and let $\PoissonProcess$ be a Poisson process over $\Omega = \YSpace \times \nonnegReals$ with mean measure $\mu = P \otimes \lambda$. 
Let $\PoissonProcess\vert_\varphi$ be the restriction of $\PoissonProcess$ under the graph of $\varphi$, and let $\tilde{\rmN}$ and $\tilde{\mu} = (P \otimes \lambda)\vert_\varphi$ be its counting and mean measures, respectively.
Since by the above assumption and argument $\tilde{\mu}(\Omega) = \infty$, almost surely $\PoissonProcess\vert_\varphi$ consists of infinitely many points; hence, its first arrival $(\tilde{Y}, \tilde{T})$ exists almost surely.
Then, we have that
\begin{align}
\label{eq:gprs_restricted_first_arrival_rd_derivative}
\frac{d\Prob[\tilde{T} = t, \tilde{Y} = y]}{d\mu} = \Ind[t \leq \varphi(y)] \cdot \Prob[\tilde{T} > t].
\end{align}
\end{lemma}
\begin{proof}
First, note that for measurable $A \subseteq \YSpace$ and $B\subseteq \nonnegReals$, it holds that
\begin{align*}
\tilde{\mu}(A \times B) = \int_B \int_A \Ind[\varphi(y) \geq t] \, dP(y) \, dt = \int_B \Prob_{Y \sim P}[\varphi(Y) \geq t, Y \in A] \, dt.
\end{align*}
Thus, by inspection we find the conditional measure $\tilde{\mu}_t(A) = \Prob_{Y \sim P}[\varphi(Y) \geq t, Y \in A]$.
Hence,
\begin{align*}
\Prob[\tilde{Y} \in A \mid \tilde{\rmN}(\YSpace \times \{t\}) = 1] 
&= \frac{\Prob[\tilde{Y} \in A, \tilde{\rmN}(\YSpace \times \{t\}) = 1]}{\Prob[\tilde{\rmN}(\YSpace \times \{t\}) = 1]} \\
&= \frac{\Prob[\tilde{\rmN}(A \times \{t\}) = 1,  \tilde{\rmN}(A^C \times \{t\}) = 0]}{\Prob[\tilde{\rmN}(\YSpace \times \{t\}) = 1]} \\
&= \frac{\Prob[\tilde{\rmN}(A \times \{t\} ) = 1] \cdot \Prob[\tilde{\rmN}(A^C \times \{t\}) = 0]}{\Prob[\tilde{\rmN}(\YSpace \times \{t\}) = 1]} \\
&= \frac{\tilde{\mu}_t(A) e^{-\tilde{\mu}_t(A)} \cdot e^{-\tilde{\mu}_t(A^C)}}{\tilde{\mu}_t(\YSpace) \cdot e^{-\tilde{\mu}_t(\YSpace)}} \\
&= \frac{\tilde{\mu}_t(A)}{\tilde{\mu}_t(\YSpace)},
\end{align*}
which is well-defined for every $t$ since $\varphi$ is assumed to be unbounded hence $\tilde{\mu}_t(\YSpace) > 0$.
Furthermore, since $\Prob_{Y \sim P}[\varphi(Y) \geq t, Y \in A]$ is right-continuous, non-negative and decreasing in $t$, its integral $\tilde{\mu}((0, t) \times A)$ is concave and almost everywhere differentiable in $t$.
Therefore, for almost every $t$ we have
\begin{align*}
\Prob[\tilde{T} \in dt] 
&= -\frac{d}{dt} \Prob[\tilde{T} > t] \\
&= -\frac{d}{dt} \Prob[\tilde{\rmN}(\YSpace \times (0, t]) = 0] \\
&= -\frac{d}{dt} \exp\left(-\tilde{\mu}(\YSpace \times (0, t])\right) \\
&= \tilde{\mu}_t(\YSpace) \cdot \exp\left(-\tilde{\mu}(\YSpace \times (0, t])\right) \\
&= \tilde{\mu}_t(\YSpace) \cdot \Prob[\tilde{T} > t].
\end{align*}
Therefore, we find that
\begin{align*}
\Prob[\tilde{T} \in dt, \tilde{Y} \in A] &= \Prob[\tilde{T} \in dt] \cdot \Prob[\tilde{Y} \in A \mid \tilde{T} \in dt] \\
&= \Prob[\tilde{T} \in dt] \cdot \Prob[\tilde{Y} \in A \mid \tilde{\rmN}(\YSpace \times \{t\}) = 1] \\
&= \tilde{\mu}_t(A) \cdot \Prob[\tilde{T} > t].
\end{align*}
Noting that $\frac{d\tilde{\mu}_t}{dP}(y) = \Ind[\varphi(y) \geq t]$ finishes the proof.
\end{proof}
\par
\textbf{Deriving the stretch function $\sigma$.}
Let $\tilde{Q}(A) = \Prob[\tilde{Y} \in A]$.
Then, integrating out the arrival time in \cref{lemma:first_arrival_rd_derivative_of_restricted_process}, we find that for $P$-almost every $y$ we have
\begin{align*}
\frac{d \tilde{Q}}{dP}(y) 
= \int_0^\infty \Ind[t \leq \varphi(y)] \cdot \Prob[ \tilde{T} > t] \, dt = \int_0^{\varphi(y)} \Prob[ \tilde{T} > t] \, dt.
\end{align*}
Thus, if we can find a $\varphi$ such that $d\tilde{Q}/dP = dQ/dP = r$, we will have the desired mathematical foundation for a sampling algorithm!
To this end, recall that we may write $\varphi = \sigma \circ r$, where $\sigma: \nonnegReals \to \nonnegReals$ is an increasing, continuous, invertible function.
Hence, it is enough for us to derive the appropriate $\sigma$.
To this end, I set the left-hand side in the above inequality to $r$, which yields
\begin{align}
r(y) &= \int_0^{\sigma(r(y))} \Prob[ \tilde{T} > t] \, dt \nonumber\\
\Rightarrow \quad \sha(\tau) &= \int_0^{\tau} \Prob[ \tilde{T} > t] \, dt, \label{eq:shrink_function_integral_identity}
\end{align}
where I introduced $\tau = \sigma(r(y))$ and $\sha = \sigma^{-1}$.
Now, differentiating both sides, we get
\begin{align*}
\sha'(\tau) &= \Prob[\tilde{T} > t] \\
&= \Prob[\tilde{T} > t, \varphi(\tilde{Y}) \geq t] \\
&= \Prob[\varphi(\tilde{Y}) \geq t] - \Prob[\varphi(\tilde{Y}) \geq t \geq \tilde{T}] \\
&= \Prob[r(\tilde{Y}) \geq \sha(t)] - \Prob[r(\tilde{Y}) \geq \sha(t) \geq \sha(\tilde{T})] \\
&= w_Q(\sha(t)) - \sha(t) \cdot w_P(\sha(t)),
\end{align*}
where $w_Q$ and $w_P$ are the width functions of \cref{def:width_function}, and where the second term in the last equation follows from
\begin{align}
\Prob[r(\tilde{Y}) \geq \sha(t) \geq \sha(\tilde{T})] 
&= \int_\YSpace \int_0^t \Ind[r(y) \geq \sha(t)]\Ind[r(y) \geq \sha(\tau)] \Prob[\tilde{T} > \tau] \, d\tau \, dP(y) \nonumber \\
&= \int_\YSpace \Ind[r(y) \geq \sha(t)] \int_0^t \Prob[\tilde{T} > \tau] \, d\tau \, dP(y) \tag{since $t \geq \tau$}
\end{align}
\begin{align}
\hphantom{\Prob[r(\tilde{Y}) \geq \sha(t) \geq \sha(\tilde{T})]} &= \int_\YSpace \Ind[r(y) \geq \sha(t)] \sha(t) \, dP(y) \tag{by \cref{eq:shrink_function_integral_identity}} \\
&= \sha(t) \cdot w_P(\sha(t)) \nonumber.
\end{align}
Thus, to summarise, I have shown that the shrink function $\sha = \sigma^{-1}$ is the solution to the differential equation
\begin{align}
\label{eq:shrink_function_diffeq_identity}
\sha' = w_Q(\sha) - \sha \cdot w_P(\sha) \quad \text{with} \quad \sha(0) = 0, \, \sha'(0) = 1.
\end{align}
Using the inverse function theorem, we find
\begin{align}
\sigma'(h) &= \frac{1}{w_Q(h) - h \cdot w_P(h)} \nonumber \\
&= \frac{1}{\Prob[H \geq h]} \tag{by \cref{eq:width_function_wP_wQ_identity}},
\end{align}
where $H$ is the associated random variable of $r(Z)$ with $Z \sim P$ from \cref{lemma:width_function_properties}.
Integrating, we finally obtain
\begin{align}
\label{eq:stretch_function_integral_identity}
\sigma(h) &= \int_0^h \frac{1}{\Prob[H \geq h]} \, dh.
\end{align}
\par
\textbf{Constructing the greedy Poisson rejection sampler.}
While the mathematical construction above was more involved than for rejection sampling or A* sampling, creating a sampling algorithm is now straightforward.
Let $\PoissonProcess$ be a Poisson process with mean measure $P \otimes \lambda$.
Then, by the construction in the above paragraph, our algorithm needs to find the first arrival of $\PoissonProcess$ under the graph of $\varphi = \sigma \circ r$, where $\sigma$ is given by \cref{eq:stretch_function_integral_identity}.
To this end, we simulate the arrivals $(Y_n, T_n)$ of $\PoissonProcess$ in order and return the first arrival that satisfies the condition $\sigma(r(Y_n)) \geq T_n$; this is the greedy Poisson rejection sampler (GPRS) that I describe in \cref{alg:global_gprs}.
\par
\textbf{Implementation challenges.}
Observe that \cref{alg:global_gprs} requires more information than rejection sampling or A* sampling: we need to be able to evaluate $\sigma$ for any $h \in r(\YSpace)$.
However, analytically solving the integral in \cref{eq:stretch_function_integral_identity} is usually difficult.
This is not an immediate issue, as we can use high-precision numerical integrators to solve \cref{eq:stretch_function_integral_identity}.
Instead, the practical challenge with numerical integration is that the integrand $1/\Prob[H \geq h]$ diverges as $h \to \infty$, and this makes evaluating $\sigma(h)$ for $h \gg 1$ troublesome.
Fortunately, there is a simple fix: the condition $\sigma(r(Y_n)) \geq T_n$ is equivalent to $r(Y_n) \geq \sha(T_n)$, which we can use as the termination criterion instead.
This is more desirable because $\sha$ and $\sha'$ are a lot better behaved: $\sha'$ is a survival probability. 
Hence it is bounded between $[0, 1]$, and $\sha \leq \norm{r}_\infty$ since $\sha$ maps into the image of $r$.
Therefore, in practice, I implement \cref{alg:global_gprs} by integrating \cref{eq:shrink_function_diffeq_identity} and checking the condition $r(Y_n) \geq \sha(T_n)$, with which I encountered no stability issues.
\par
\textbf{The runtime of greedy Poisson rejection sampling.}
Unfortunately, giving a full characterisation of the runtime of GPRS akin to \cref{thm:global_a_star_conditional_index_distribution} does not seem possible.
Instead, I now turn to analysing its sample complexity.
\begin{importantTheorem}[Runtime of greedy Poisson rejection sampling.]
\label{thm:global_gprs_runtime}
Let $Q \ll P$ be distributions over $\YSpace$ with Radon-Nikodym derivative $r = dQ/dP$, and assume ${\norm{r}_\infty < \infty}$.
Let $K$ be the number of proposal samples simulated by \cref{alg:global_gprs} before terminating.
Then,
\begin{align*}
\Exp[K] = \norm{r}_\infty.
\end{align*}
\end{importantTheorem}
Here, I utilise the same trick I employed for the analysis of rejection sampling and A* sampling: condition on the arrival time and average.
\begin{proof}
Let $\PoissonProcess$ be the Poisson process with mean measure $P \otimes \lambda$ that \cref{alg:global_gprs} uses to propose candidate points. 
Let $\sigma^{-1}$ be the shrink function, i.e., the solution of \cref{eq:shrink_function_diffeq_identity} and $\varphi = \sigma \circ r$.
Finally, let $(\tilde{T}, \tilde{Y})$ be the first arrival of $\PoissonProcess\vert_\varphi$, i.e.\ the first point of $\PoissonProcess$ that falls under $\varphi$.
Since $\tilde{T}$ is a nonnegative random variable, by the Darth Vader rule \citep{muldowney2012darth}, we have
\begin{align}
\Exp[\tilde{T}] 
&= \int_0^\infty \Prob[\tilde{T} > t] \, dt \nonumber \\
&= \lim_{t \to \infty} \sigma^{-1}(t) \tag{by \cref{eq:shrink_function_integral_identity}} \\
&= \norm{r}_\infty, \label{eq:gprs_expected_arrival_time}
\end{align}
where the last equality follows from $\lim_{h \to \norm{r}_\infty} \sigma(h) = \infty$, since $\varphi$ is unbounded, and from the continuity of $\sigma^{-1}$.
Next, consider the following ``probability integral transform-type'' identity:
\begin{align}
\Exp[\tilde{\mu}( \YSpace \times (0, \tilde{T}])] 
&= \int_0^\infty \tilde{\mu}(\YSpace \times (0, t])\cdot \tilde{\mu}_t(\YSpace) \cdot e^{-\tilde{\mu}(\YSpace \times (0, t])} \, dt \nonumber\\
&= \int_0^\infty u e^{-u} \, du  \tag{substitute $u \gets \tilde{\mu}(\YSpace \times (0, t])$}\\
&= 1. \label{eq:gprs_expected_arrival_mean_time}
\end{align}
Now, consider the points of $\PoissonProcess$ that did not fall under the graph of $\varphi$: $\PoissonProcess\vert_\varphi^C = \PoissonProcess \cap (\hypograph \varphi)^C$.
By the restriction theorem, $\PoissonProcess\vert_\varphi^C$ is also a Poisson process with mean measure $\nu(A) = \mu(A \cap (\hypograph \varphi)^C)$ for $A \subseteq \Omega$.
Observe that by the independence property of Poisson processes $\tilde{T} \perp \PoissonProcess\vert_\varphi^C$.
Hence, conditioned on $\tilde{T}$, the first $K - 1$ points of $\PoissonProcess$ that did not fall under the graph of $\varphi$ are Poisson distributed with mean
\begin{align*}
\Exp[K - 1 \! \mid\! \tilde{T}] = \int_0^{\tilde{T}} \!\!\!\int_\YSpace\!\! \Ind[t > \varphi(y)] \, dP(y) \, dt
= \int_0^{\tilde{T}}\!\!\! \int_\YSpace (1 - \Ind[t \leq \varphi(y)]) \, dP(y) \, dt
= \tilde{T} - \tilde{\mu}(\tilde{T})
\end{align*}
Therefore,
\begin{align}
\Exp[K] &= 1 + \Exp_{\tilde{T}}[\Exp[K - 1 \mid \tilde{T}]] \nonumber\\
&= 1 + \Exp_{\tilde{T}}[\tilde{T} - \tilde{\mu}(\tilde{T})] \nonumber\\
&= \norm{r}_\infty, \tag{by \cref{eq:gprs_expected_arrival_time,eq:gprs_expected_arrival_mean_time}}
\end{align}
which finishes the proof.
\end{proof}
\par
\textbf{Relative entropy coding with greedy Poisson rejection sampling.}
Encoding samples with GPRS proceeds as before: we encode the index $K$ of the selected sample using an appropriate Zeta distribution.
Similarly to A* sampling, the average description length of GPRS is also optimal.
\begin{importantTheorem}[Relative entropy coding with greedy Poisson rejection sampling]
\label{thm:global_gprs_codelength}    
Let $\rvx, \rvy \sim P_{\rvx, \rvy}$ be dependent random variables over the space $\XSpace \times \YSpace$.
Moreover, let $\PoissonProcess$ be a Poisson process over $\YSpace \times \nonnegReals$ with mean measure $P_{\rvy} \otimes \lambda$.
\par
Consider the following channel simulation protocol: the sender and receiver set $\rvz \gets \PoissonProcess$ as their common randomness.
Then, upon receiving a source symbol $\rvx \sim P_\rvx$, the sender uses greedy Poisson rejection sampling (\cref{alg:global_gprs}) to simulate $\rvy \sim P_{\rvy \mid \rvx}$ using the points of $\PoissonProcess$ as proposals, and let $K$ denote the index of $\rvy$ in $\PoissonProcess$.
Finally, the sender encodes $K$ using a Zeta distribution $\zeta(k \mid \alpha) \propto k^{-\alpha}$ with exponent ${\alpha = 1 + 1 / (\MI{\rvx}{\rvy} + 1)}$.
Then, the average description length of this protocol is upper bounded by
\begin{align*}
\MI{\rvx}{\rvy} + \lb(\MI{\rvx}{\rvy} + 2) + 3 \text{ bits}.
\end{align*}
\end{importantTheorem}
The proof I present below is an improved version of the one in \citet{flamich2023gprs}.
The improvement is due to Daniel Goc, who suggested the use of the inequality 
\begin{align}
\label{eq:daniels_fraction_inequality}
a,b,c,d \geq 0: \quad \frac{a + c}{b + d}\, \leq\, \max\left\{\frac{a}{b}, \frac{c}{d}\right\},
\end{align}
with the right-hand side equal to $\infty$ if either $b$ or $d$ equal $0$.
\begin{proof}
Similarly to the proof of \cref{thm:global_a_star_codelength}, I start by noting that by \cref{lemma:li_el_gamal_bound_on_pos_int_random_variable}, the average description length of a scheme using a Zeta distribution with exponent $\alpha = 1 + 1 / \Exp[\lb K]$ to encode a positive integer-valued random variable $K$ is upper bounded by
\begin{align}
\label{eq:global_gprs_zeta_coding_bound}
\Exp[\lb K] + \lb(\Exp[\lb K] + 1) + 2.
\end{align}
Hence, I will show that $\Exp[\lb K] \leq \MI{\rvx}{\rvy} + 1$, which will finish the proof.
First, however, I need to set up a few quantities and derive some preliminary results.
Thus, given $\rvx \sim P_\rvx$, let $\PoissonProcess$ be the Poisson process with mean measure $P_\rvy \otimes \lambda$ that \cref{alg:global_gprs} uses for its proposals, let $r_\rvx = dP_{\rvy \mid \rvx}/dP_\rvy$, let $\sigma$ be given by \cref{eq:stretch_function_integral_identity}, let $\sha = \sigma^{-1}$ and let $\varphi = \sigma \circ r$.
Then, let $\PoissonProcess\vert_\varphi$ be the restriction of $\PoissonProcess$ under the graph of $\varphi$ with mean measure $\tilde{\mu}$.
Finally, let $\tilde{T}$ denote the first arrival time of $\PoissonProcess\vert_\varphi$ and let $H_\rvx$ be the associated random variable of $r_\rvx(\rvy), \rvy \sim P_\rvy$.
To begin, observe that
\begin{align}
\Prob[\tilde{T} > t] &= w_{P_{\rvy \mid \rvx}}(\sha(t)) - \sha(t) \cdot w_{P_{\rvy}}(\sha(t)) \tag{by \cref{eq:shrink_function_diffeq_identity}} \\
&= \Prob[H_\rvx > \sha(t)] \tag{by \cref{eq:width_function_wP_wQ_identity}} \\
&= \Prob[\sigma(H_\rvx) > t], \nonumber\\
\Rightarrow\quad \tilde{T} &\sim \sigma(H_\rvx) \label{eq:global_gprs_arrival_time_disteq_sigma_of_h}
\end{align}
Next, I show the following ``two-sided Chernoff-type'' inequality:
\begin{align}
\label{eq:daniels_inequality}
\frac{1 + \sigma(h) - \tilde{\mu}(\sigma(h))}{h} \leq \max\left\{\frac{1}{\Prob[H_\rvx > h]}, \frac{1}{\Prob[H_\rvx \leq h]} \right\}.
\end{align}
To see this, note first that
\begin{align}
\frac{1 + \sigma(h) - \tilde{\mu}(\sigma(h))}{h} 
&= \frac{1 + \sigma(h) - \tilde{\mu}(\sigma(h))}{h + \Prob[H_\rvx \leq h] - \Prob[H_\rvx \leq h]} \nonumber \\
&\leq \max\left\{\frac{ \sigma(h) - \tilde{\mu}(\sigma(h))}{h - \Prob[H_\rvx \leq h]}, \frac{1}{\Prob[H_\rvx \leq h]} \right\}. \tag{by \cref{eq:daniels_fraction_inequality}}
\end{align}
One concern that might arise is that $h - \Prob[H_\rvx \leq h]$ becomes negative, while \cref{eq:daniels_fraction_inequality} requires all involved quantities to be nonnegative.
However, fortunately this is never so: recall that $\Prob[H_\rvx \in dh] = w_{P_\rvy}(h) \leq 1$, meaning $\Prob[H_\rvx \leq h]$ grows sublinearly.
Furthermore, $\Prob[H_\rvx \leq 0] = 0$, hence for $h > 0$ we must have $h - \Prob[H_\rvx \leq 0] \geq 0$.
To proceed, observe that
\begin{align*}
\tilde{\mu}(\sigma(h)) &= \int_0^{\sigma(h)} \Prob_{Y \sim P_\rvy}[\varphi(Y) > t] \, dt \\
&= \int_0^{\sigma(h)} \Prob_{Y \sim P_\rvy}[r(Y) > \sha(t)] \, dt \\
&= \int_0^h \sigma'(h) \cdot w_{P_\rvy}(h) \, dh.
\end{align*}
Hence,
\begin{align}
\sigma(h) - \tilde{\mu}(\sigma(h)) 
&= \int_0^h \sigma'(\eta)(1 - w_{P_\rvy}(\eta)) \, d\eta \nonumber \\
&\leq \sigma'(h)\int_0^h (1 - w_{P_\rvy}(\eta)) \, d\eta \tag{$\sigma'$ is increasing} \\
&= \frac{h - \Prob[H_\rvx \leq h]}{\Prob[H_\rvx > h]}. \nonumber
\end{align}
Therefore,
\begin{align*}
\frac{1 + \sigma(h) - \tilde{\mu}(\sigma(h))}{h} &\leq \max\left\{\frac{ \sigma(h) - \tilde{\mu}(\sigma(h))}{h - \Prob[H_\rvx \leq h]}, \frac{1}{\Prob[H_\rvx \leq h]} \right\} \\
&\leq \max\left\{\frac{1}{h - \Prob[H_\rvx \leq h]} \cdot \frac{h - \Prob[H_\rvx \leq h]}{\Prob[H_\rvx > h]}, \frac{1}{\Prob[H_\rvx \leq h]} \right\},
\end{align*}
which finally yields \cref{eq:daniels_inequality}.
\par
Given the above setup, I am ready to bound $\Exp[\lb K]$.
For this, I begin by developing a pointwise bound for $\Exp[\lb K \mid \rvx]$ as follows.
First, recall from the proof of \cref{thm:global_gprs_runtime} that $(K - 1) \mid \tilde{T}, \rvx$ is Poisson distributed with mean $\tilde{T} - \tilde{\mu}(\tilde{T})$.
Hence,% we have
\begin{align}
\Exp[\lb K \mid \rvx] 
&\leq \Exp_{\tilde{T}}[\lb \Exp[ K \mid \tilde{T}] \mid \rvx] \tag{Jensen} \\
&= \Exp_{\tilde{T}}[\lb(1 + \tilde{T} - \tilde{\mu}(\tilde{T})) \mid \rvx] \nonumber\\
&= \Exp_{H_\rvx}[\lb(1 + \sigma(H_\rvx) - \tilde{\mu}(\sigma(H_\rvx)))]. \tag{by \cref{eq:global_gprs_arrival_time_disteq_sigma_of_h}}\\
&\leq \Exp_{H_\rvx}[\lb(H_\rvx)] + \int_0^\infty w_{P_\rvy}(h) \lb \left(\max\left\{\frac{1}{\Prob[H_\rvx > h]}, \frac{1}{\Prob[H_\rvx \leq h]} \right\} \right)\, dh \tag{by \cref{eq:daniels_inequality}} \\
&= \Exp_{H_\rvx}[\lb(H_\rvx)] + \int_0^1 \lb \left(\max\left\{\frac{1}{u}, \frac{1}{1-u} \right\}\right) \, du \tag{substitute $u \gets \Prob[H_\rvx \leq h]$} \\
&= \KLD{P_{\rvy \mid \rvx}}{P_\rvy} + 1. \tag{by \cref{lemma:width_function_properties}/4}
\end{align}
Taking expectation over $\rvx$ yields the desired result.
\end{proof}
\par
Comparing \cref{thm:global_a_star_runtime,thm:global_gprs_runtime} and \cref{thm:global_a_star_codelength,thm:global_gprs_codelength}, we see that the sample complexities and our bounds on the channel simulation efficiencies of A* sampling and GPRS match exactly.
However, as I discussed, GPRS can be much more challenging to implement, as computing $\sigma$ or $\sha$ can be difficult.
Indeed, I am not aware of any convincing case where \cref{alg:global_gprs} would be more advantageous to use than \cref{alg:global_a_star}.
A possible scenario in the discrete case might be speculative decoding \citep{leviathan2023fast}, which was recently connected to greedy rejection sampling by \citet{kobus2025speculative}.
However, from my perspective, the real redemption for GPRS will come in \cref{sec:branch_and_bound_samplers}, where I develop fast variants of A* sampling and GPRS.
As I will show, branch-and-bound GPRS will be a fully optimal algorithm, while A* sampling will be slightly inefficient.
\par
\textbf{A historical note.}
The inspiration for greedy Poisson rejection sampling came from the greedy rejection sampler developed by \citet{harsha2010communication}.
However, their algorithm is formulated for discrete random variables only.
Taking it as a starting point, I first generalised this algorithm to arbitrary Polish spaces in \citet{flamich2023adaptive} and developed a fast variant of it (not presented in this thesis) in \citet{flamich2023grc}.
However, these algorithms were all developed by positing a certain ``greedy acceptance criterion'' that accepts samples with as high a probability as possible and then deriving the rest of the algorithm such that the output distribution is correct.
In contrast, greedy Poisson rejection sampling derives the whole algorithm as an inevitable consequence of its geometric construction.
\subsection{Superposition: Parallelisation}
\label{sec:superposition_parallelisation}
\par
Rather than construct a conceptually new algorithm, this section applies the superposition theorem to develop parallelised variants of the samplers I described up to this point.
The problem to solve remains the same: there are probability measures $Q \ll P$ over $\YSpace$ with $r = dQ/dP$, and we wish to simulate a $Q$-distributed sample using $P$-distributed candidates.
But now I also assume we have access to $J$ \textit{parallel threads}.
I am intentionally vague about the computational meaning of a ``thread'', as this may differ between environments and applications.
Essentially, the two properties threads should have are that we can perform independent computation in them, but it is also possible to synchronise data between them.
In the analysis I present below, I make the simplifying assumption that the operations t synchronise data between the threads require negligible time.
Hence, I will not account for threads potentially blocking for a long time to read from or write to some shared resource.
\begin{wrapfigure}[22]{r}[0pt]{0.5\textwidth}%
\centering
\includegraphics{4-RelativeEntropyCodingWithPoissonProcesses/img/parallel_gprs_illustration.tikz}
\caption[Illustration of parallelised greedy Poisson rejection sampling]{Illustration of parallelised greedy Poisson rejection sampling (\cref{alg:parallel_gprs}) for a Gaussian target $Q = \Normal(1, 0.25^2)$ and Gaussian proposal distribution $P = \Normal(0, 1)$, with the time axis truncated to the first $17$ units.
The algorithm in the illustration uses two independent Poisson processes $\PoissonProcess_1$ and $\PoissonProcess_2$, both with mean measures $1/2 \cdot (P \otimes \lambda)$.
{\color{blue}Blue} points are arrivals in $\PoissonProcess_1$ and {\color{orange} orange} points are arrivals in $\PoissonProcess_2$.
Crosses (\xmark) indicate rejected, and circles (\tikzcircle[black, fill=black]{3pt}) indicate accepted points by each thread.
The algorithm accepts the earliest arrival across the two processes, which in this case is marked by the blue circle (\tikzcircle[blue, fill=blue]{3pt}).
By the superposition theorem, this is equivalent to finding the first arrival of $\PoissonProcess = \PoissonProcess_1 \cup \PoissonProcess_2$.
}
\label{fig:parallel_gprs_illustration}
\end{wrapfigure}
\par
\textbf{Parallelised construction with superposition.}
The starting point for the construction is that \cref{alg:global_rs,alg:global_a_star,alg:global_gprs} all use the same Poisson process $\PoissonProcess$ over $ \YSpace \times \nonnegReals$ with mean measure $P \otimes \lambda$ to propose candidates.
The superposition theorem allows us to ``divide'' this process into subprocesses and ``merge the results'' at the end, as its following specialisation shows.%
\begin{proposition}[A special case of \cref{thm:superposition_theorem}]
\label{prop:superposition_theorem}
Let $J$ be a positive integer, and let $\PoissonProcess_1,\hdots, \PoissonProcess_J$ be Poisson processes such that the mean measure of each $\PoissonProcess_j$ is $\frac{1}{J} \cdot (P \otimes \lambda)$.
Then,
\begin{align*}
\PoissonProcess = \bigcup_{j = 1}^J \PoissonProcess_j
\end{align*}
is a Poisson process with mean measure $P \otimes \lambda$.
\end{proposition}
\begin{figure}[t]
\centering
\begin{minipage}[t]{0.49\textwidth}
 \begin{algorithm}[H]
\SetAlgoLined
\DontPrintSemicolon
\SetKwInOut{Input}{Input}\SetKwInOut{Output}{Output}
\SetKwFor{ParFor}{in parallel for}{do}{end}
\SetKw{TermThread}{terminate thread}
\SetKw{SingleWhile}{while}
\textbf{Input:}\;
Proposal distribution $P$\;
Density ratio $r = dQ/dP\hspace{-1cm}$\;
Upper bound $M$ on $r$\;
Number of parallel threads $J$\;
Time-rates $\{c_j\}_{j = 1}^J$\;
\textbf{Output:}\;
Sample $Y \sim Q$ and its code $(j^*, N_{j^*})\hspace{-2mm}$\;
\;
$U, Y^*, j^*, N_{j^*} \gets \infty, \perp, \perp, \perp \hspace{-1cm}$\;
\ParFor{$j = 1, \hdots, J$}{
$\Pi_j \gets \SimulatePP(1/J, P)$\;
\For{${n_j} = 1, 2, \hdots$}{
\vspace{0.1cm}
$Y^{(j)}_{n_j}, T^{(j)}_{n_j} \gets \mathtt{next}(\Pi_j)$\; 
$T^{(j)}_{n_j} \gets T^{(j)}_{n_j} \big/ c_j$\;
\;
\If{$T^{(j)}_{n_j} / r(Y^{(j)}_{n_j}) < U$}{
\vspace{0.1cm}
$U \gets T^{(j)}_{n_j} / r(Y^{(j)}_{n_j})$\;
$Y^*, j^*, N_{j^*} \gets Y^{(j)}_{n_j}, j, n_j\hspace{-1cm}$\;
}
\If{$U < T_{n_j}^{(j)} / M$}{
\TermThread $j$.\;
}
}
}
\Return{$Y^*, (j^*, N_{j^*})$}
%\vspace{0.239cm}
\caption{Parallelised A* sampling with $J$ available threads.}
\label{alg:parallel_a_star}
\end{algorithm}   
\end{minipage}%
\hfill
\begin{minipage}[t]{0.49\textwidth}
 \begin{algorithm}[H]
\SetAlgoLined
\DontPrintSemicolon
\SetKwInOut{Input}{Input}\SetKwInOut{Output}{Output}
\SetKwFor{ParFor}{in parallel for}{do}{end}
\SetKw{TermThread}{terminate thread}
\SetKw{SingleWhile}{while}
\textbf{Input:}\;
Proposal distribution $P$\;
Density ratio $r = dQ/dP\hspace{-1cm}$\;
Stretch function $\sigma$\;
Number of parallel threads $J$\;
Time-rates $\{c_j\}_{j = 1}^J$\;
\textbf{Output:}\;
Sample $Y \sim Q$ and its code $(j^*, N_{j^*})\hspace{-2mm}$\;
\;
$T^*, Y^*, j^*, N_{j^*} \gets \infty, \perp, \perp, \perp$\;
\ParFor{$j = 1, \hdots, J$}{
$\Pi_j \gets \SimulatePP(1/J, P)$\;
\For{${n_j} = 1, 2, \hdots$}{
\vspace{0.1cm}
$Y^{(j)}_{n_j}, T^{(j)}_{n_j} \gets \mathtt{next}(\Pi_j)$\; 
$T^{(j)}_{n_j} \gets T^{(j)}_{n_j} \big/ c_j$\;
\;
\If{$T^* < T_{n_j}^{(j)}$}{
\TermThread $j$.\;
}
\If{$T^{(j)}_{n_j} < \sigma(r(Y^{(j)}_{n_j}))$}{
$T^*\!\!, Y^*\!\!\!, j^*\!, N_{j^*} \gets T^{(j)}_{n_j}\!\!, Y^{(j)}_{n_j}\!\!\!, j, n_j\hspace{-1cm}$\;
\TermThread $j$.\;
}
}
}
\Return{$Y^*, (j^*, N_{j^*})$}
\vspace{1.45mm}
\caption{Parallel GPRS with $J$ available threads.\vspace{1.05mm}}
\label{alg:parallel_gprs}
\end{algorithm}   
\end{minipage}%
\end{figure}%
Therefore, \cref{prop:superposition_theorem} suggests the following design principle for parallelising our sampling algorithms over $J$ available threads: make each thread $j$ simulate $\PoissonProcess_j$ and search for the distinguished point specified by the construction of the sampling algorithm, keeping the termination criterion synchronised across the threads.
Then, once we locate the sought-after point in one of the $j$ threads, \cref{prop:superposition_theorem} shows that it is also the sought-after point in $\PoissonProcess$, ensuring the correctness of the procedure.
\par
The case for rejection sampling is quite simple and uninteresting since its termination criterion does not depend on the arrival time.
Hence, parallelising requires no modification to the original algorithm and amounts to just running \cref{alg:global_rs} $J$ times in parallel.
The situation is more interesting in the case of A* sampling and GPRS, and hence I describe the parallelised algorithms in \cref{alg:parallel_a_star,alg:parallel_gprs}; for a visual example of parallelised GPRS, see \cref{fig:parallel_gprs_illustration}.
\par
\textbf{Runtime of the parallelised samplers.}
If we use $J$ parallel threads to speed up the sampling process, we would ideally like to see a $J$ times improvement in the expected runtime.
As I show below, this is almost the case, bar a very small inefficiency.
\begin{theorem}[Expected runtime of parallelized A* sampling]
\label{thm:parallel_a_star_runtime}
Let $Q \ll P$ be distributions over $\YSpace$ with Radon-Nikodym derivative $r = dQ/dP$ and assume $r$ is bounded above by some constant $M > 1$.
Assume we run \cref{alg:parallel_a_star} with $J$ parallel threads, and let $K_j$ denote the number of samples the algorithm simulates in thread $j$.
Then, for all $j \in [1:J]$, we have
\begin{align*}
K_j \sim \Geom\left(\frac{J}{J + M - 1}\right)
\end{align*}
and thus, for the total number of simulated samples $K$, we have
\begin{align*}
\Exp[K] = \Exp\left[\sum_{j = 1}^J K_j\right] = M + J - 1.
\end{align*}
\end{theorem}
\textbf{For this proof, I ignore the possible delay synchronisation between threads and assume that information in one thread is immediately propagated to other threads.}
The proof follows the steps of the proof of the A* sampling runtime (\cref{thm:global_a_star_runtime}) mutatis mutandis.
\begin{proof}
Let $\PoissonProcess_1, \hdots \PoissonProcess_J$ be the Poisson processes with mean measures $\frac{1}{J} \cdot P \otimes \lambda$ \cref{alg:parallel_a_star} simulates in parallel, and let $\PoissonProcess = \bigcup_{j = 1}^J \PoissonProcess_j$.
Let $f(y, t) = (y, t/r(y))$ denote the shift function, $f(\PoissonProcess)$ the shifted process and $\tau$ the first arrival time of $f(\PoissonProcess)$, that is, the first arrival time across all shifted processes $f(\PoissonProcess_j)$.
Then, by the same reasoning as I gave in the proof of \cref{thm:global_a_star_runtime}, the number of points simulated in thread $j$ given the global first arrival time $\tau$ is $K_j - 1 \mid \tau \sim \Pois(\frac{\tau \cdot (M - 1)}{J})$.
Then, integrating out $\tau \sim \Exponential(1)$, we have the desired result.
\end{proof}
\begin{theorem}[Expected runtime of parallelized greedy Poisson rejection sampling]
\label{thm:parallel_gprs_runtime}
Let $Q \ll P$ be distributions over $\YSpace$ with Radon-Nikodym derivative $r = dQ/dP$ and assume $\norm{r}_\infty \infty$.
Assume we run \cref{alg:parallel_gprs} with $J$ parallel threads, and let $K_j$ denote the number of samples simulated in thread $j$.
Then,
\begin{align*}
\Exp[K_j] = \frac{\norm{r}_\infty - 1}{J} + 1.
\end{align*}
Hence, for the total number of simulated samples $K$, we have
\begin{align*}
\Exp[K] = \Exp\left[\sum_{j = 1}^J K_j\right] = \norm{r}_\infty + J - 1
\end{align*}
\end{theorem}
\begin{proof}
The proof follows the same idea as the proof of \cref{thm:parallel_a_star_runtime}, but executed as a modification of \cref{thm:global_gprs_runtime}, and hence I omit it.
The reader can find the full proof in \citep[Appendix C;][]{flamich2023gprs}.
\end{proof}
\par
\textbf{Partitioning-based parallelisation.}
Note that the superposition theorem I gave in \cref{prop:superposition_theorem} is a heavy specialisation of \cref{thm:superposition_theorem}, and we could have chosen other sensible ways to split the original process $\PoissonProcess$ into subprocesses.
For example, we could have partitioned the sample space $\YSpace$ into sets $B_1, \hdots B_J$, and created restricted subprocesses $\PoissonProcess_j = \PoissonProcess \cap (B_j \times \nonnegReals)$.
By the restriction theorem, $\PoissonProcess_j$ is a Poisson process, and by the superposition theorem, we see that $\cup_{j = 1}^J \PoissonProcess_j = \PoissonProcess$.
This way, each subprocess searches over a different part of the sample space.
This way of parallising A* sampling is the basis of the algorithm Jiajun He and I developed in \citet{he2024accelerating}.
%
% %
\section{Approximate Samplers from Exact Ones}
\label{sec:approximate_sampling}
So far, I have only considered exact samplers, i.e., algorithms whose output is guaranteed to be distributed according to the target distribution $Q$ we specify.
However, we have also seen that all these samplers are \textit{Las Vegas} algorithms: while they always return an exact answer and terminate almost surely, \textit{their runtime is random}.
Unfortunately, a random runtime can be quite problematic in practical situations unless its variance is small.
Furthermore, an exact sample is not necessary for good performance in many scenarios, such as the learned compression setting I describe in \cref{chapter:combiner}.
Therefore, we might prefer to use \textit{Monte Carlo} algorithms in these situations instead: algorithms with deterministic runtime which return approximate results only.
\par
Thus, in this section, I develop approximate sampling algorithms from exact ones based on my paper \citet{flamich2024some}.
My constructions will be elementary: I take an exact sampling algorithm, such as \cref{alg:a_star_sac}, and specify a ``runtime budget'' $k$.
Then, I run the exact sampler as usual, except if its runtime exceeds $k$ steps, I terminate it and return one of the proposal samples the algorithm examined previously.
The exciting insight in my paper \citet{flamich2024some} is that there is a serendipitous ``alignment of interests'' between approximate sampling and relative entropy coding. We can reuse the bounds on the coding efficiency of an exact relative entropy coding algorithm to obtain optimal bounds on the sample quality of its approximate variants.
However, before I proceed, I need to establish a reasonable measure of approximation.
\subsection{The Total Variation Distance}
\par
Naturally, one should choose a measure of error that fits well for their application.
In this thesis, the application I am interested in is one-shot lossy source coding, where we reconstruct the data from a single encoded sample.
Hence, I need a measure that can quantify the quality of a single sample.
Fortunately, there is an excellent candidate for this: the total variation distance, which I define next.
\begin{definition}[Total variation distance]
\label{def:tv_distance}
Let $\mu, \nu$ be probability measures over the measurable space $\YSpace, \sigmaAlgebra$.
Then, the total variation (TV) distance between $\mu$ and $\nu$ is defined as
\begin{align*}
\TVD{\mu}{\nu} = \sup_{A \in \sigmaAlgebra}\abs{\mu(A) - \nu(A)}.
\end{align*}
\end{definition}
In the following lemma, I collect some important properties of the TV distance.
\begin{lemma}[Properties of the TV distance]
\hfill
\begin{enumerate}
\item \textbf{Bounded metric.} The total variation distance $\TVD{\mu}{\nu}$ is a metric:
\begin{itemize}
\item $\TVD{\mu}{\nu} \geq 0$, and $\TVD{\mu}{\nu} = 0 \Leftrightarrow \mu = \nu$.
\item $\TVD{\mu}{\nu} = \TVD{\nu}{\mu}$.
\item $\TVD{\mu}{\nu} \leq \TVD{\mu}{\tilde{\mu}} + \TVD{\tilde{\mu}}{\nu}$.
\item $\TVD{\mu}{\nu} \leq 1$.
\end{itemize}
\item \textbf{Characterisation via measurable functions.} Let $\mu, \nu$ be probability measures over the measurable space $(\YSpace, \sigmaAlgebra)$.
Then, for measurable ${\psi: \YSpace \to [0, 1]}$, we have
\begin{align}
\label{eq:tvd_measurable_fun_characterisation}
\TVD{\mu}{\nu} = \sup_{\psi: \YSpace \to [0, 1]}\abs{\Exp_{X \sim \nu}[\psi(X)] - \Exp_{X \sim \mu}[\psi(X)]}.
\end{align}
\item \textbf{Data processing.} The total variation satisfies the data processing inequality: Let $K_{Y \mid X}$ be a transition kernel for the Markov chain $X \to Y$.
Let $P_X, Q_X$ be two source distributions over $X$, and denote the marginals they induce over $Y$ as $P_Y$ and $Q_Y$, respectively.
Then:
\begin{align*}
\TVD{P_X}{Q_X} \geq \TVD{P_Y}{Q_Y}.
\end{align*}
\item \textbf{Two-alternative forced choice characterisation.} For $\mu, \nu$ over the same measurable space, set $X_1 \sim \mu_1$ and $X_2 \sim \mu_2$.
Finally, set $X = X_B$ for $B \sim \Bern(1/2)$.
Now, consider the following task: given $X$, predict the value of $B$, i.e., the observer has to decide whether $X$ is $\mu_1$ or $\mu_2$ distributed.
Denote the observer's prediction as $B'$.
Then, for any observer/decision rule, we have
\begin{align*}
\Prob[B' \neq B] = \frac{1}{2}(1 - \TVD{\mu}{\nu}).
\end{align*}
\end{enumerate}    
\end{lemma}
The proofs of the first two items are standard; I adapt the proof from \citet{nielsen2013hypothesis}.
\begin{proof}
\textbf{Bounded metric.}
Non-negativity and symmetry follow from the non-negativity of the absolute value in the definition of $\TVD{\mu}{\nu}$.
For $\mu = \nu$ we have $\TVD{\mu}{\nu} = 0$ by definition and for the other direction, note that $\TVD{\mu}{\nu} = 0$ implies that $\mu(A) = \nu(A)$ for every possible event $A$, hence $\mu = \nu$.
For the triangle inequality, note that by the triangle inequality for the absolute value, we have
\begin{align*}
\TVD{\mu}{\nu} &= \sup_{A \in \sigmaAlgebra}\abs{\mu(A) - \nu(A)} \\
&\leq \sup_{A \in \sigmaAlgebra}\abs{\mu(A) - \tilde{\mu}(A)} + \sup_{A \in \sigmaAlgebra}\abs{\tilde{\mu}(A) - \nu(A)} \\
&= \TVD{\mu}{\tilde{\mu}} + \TVD{\tilde{\mu}}{\nu}.
\end{align*}
Finally, for the boundedness, note that for any $A \in \sigmaAlgebra$, it holds that $0 \leq \mu(A), \nu(A) \leq 1$, hence $-1 \leq \mu(A) - \nu(A) \leq 1$ and therefore $\abs{\mu(A) - \nu(A)} \leq 1$.
\par
\textbf{Characterisation via measurable functions.}
First, consider setting $\psi_A(x) = \Ind[x \in A]$ for each $A \in \sigmaAlgebra$.
This shows that
\begin{align*}
\TVD{\mu}{\nu} \leq \sup_{\psi: \YSpace \to [0, 1]}\abs{\Exp_{X \sim \nu}[\psi(X)] - \Exp_{X \sim \mu}[\psi(X)]}.
\end{align*}
For the other direction, fix an arbitrary measurable function $\psi: \YSpace \to [0, 1]$.
Now, by the Darth Vader rule, we have
\begin{align*}
\Exp_{X \sim P}[\psi(X)] = \int_0^1 \Prob_{X \sim P}[\psi(X) \geq t] \, dt.
\end{align*}
Hence,
\begin{align*}
\abs{\Exp_{X \sim P}[\psi(X)] - \Exp_{X \sim Q}[\psi(X)]}  
&= \abs*{\int_0^1 \Prob_{X \sim P}[\psi(X) \geq t] - \Prob_{X \sim Q}[\psi(X) \geq t] \, dt} \\
&\leq  \abs*{\int_0^1\TVD{P}{Q} \, dt} \\
&= \TVD{P}{Q}.
\end{align*}
\textbf{Data processing.}
Let $\YSpace$ denote the space over which $Y$ is defined, and fix $A \subseteq \YSpace$.
Then:
\begin{align*}
\abs{P_Y(A) - Q_Y(A)} 
&= \abs{\Exp_{X \sim P_X}[K_{Y \mid X}(A \mid X)] - \Exp_{X \sim Q_X}[K_{Y \mid X}(A \mid X)]} \\
&\leq \TVD{P_X}{Q_X},
\end{align*}
where I applied $\cref{eq:tvd_measurable_fun_characterisation}$ for $\psi(x) = K_{Y \mid X}(A \mid x)$.
\\
\noindent
\textbf{Two-alternative forced choice characterisation.}
Denote the distribution of $X$ as $\nu$.
The decision rule that minimises the probability of error is the MAP Bayes detector \citep{nielsen2013hypothesis}:
\begin{align*}
B' = \argmax_{b} \Prob[X \sim \mu_b \mid X],
\end{align*}
which has probability of error
\begin{align}
\Prob[B \neq B'] 
&= \Exp_{X \sim \nu}[\min\{\Prob[X \sim \mu_1 \mid X], \Prob[X \sim \mu_2 \mid X]\}] \nonumber\\
&= \Exp_{X \sim \nu}\left[\min\left\{\Prob[X \sim \mu_1]\frac{d\mu_1}{d\nu}(X), \Prob[X \sim \mu_2]\frac{d\mu_2}{d\nu}(X)\right\}\right] \tag{Bayes' rule} \\
&=\frac{1}{2} \Exp_{X \sim \nu}\left[\min\left\{\frac{d\mu_1}{d\nu}(X), \frac{d\mu_2}{d\nu}(X)\right\}\right] \nonumber\\
&= \frac{1}{4} \Exp_{X \sim \nu}\left[\frac{d\mu_1}{d\nu}(X) + \frac{d\mu_2}{d\nu}(X) - \abs*{\frac{d\mu_1}{d\nu}(X) - \frac{d\mu_2}{d\nu}(X)}\right] \nonumber \\
&= \frac{1}{2} \left(1 - \TVD{\mu_1}{\mu_2}\right), \nonumber 
\end{align}
where the fourth equality holds since $2\min\{a, b\} = a + b - \abs{a - b}$, and the last equality follows since (using the notation $(x)_+ = \max\{0, x\}$)
\begin{align*}
\Exp\left[\abs*{\frac{d\mu_1}{d\nu}(X) - \frac{d\mu_2}{d\nu}(X)}\right] 
&= \Exp\left[\left(\frac{d\mu_1}{d\nu}(X) - \frac{d\mu_2}{d\nu}(X)\right)_+ \!\!\!\!+ \left(\frac{d\mu_2}{d\nu}(X) - \frac{d\mu_1}{d\nu}(X) \right)_+\right] \\
&= 2\Exp\left[\left(\frac{d\mu_1}{d\nu}(X)\right)_+ \right] \\
&= 2 \TVD{\mu_1}{\mu_2}.
\end{align*}
\end{proof}
\par
\textbf{Why should we choose the total variation as our measure of error?}
When compressing some data $\rvx \sim P_\rvx$, we usually want to directly compare its quality to its reconstruction $\tilde{\rvx} \sim P_{\tilde{\rvx}}$.
However, as I will discuss at greater length in \cref{sec:nonlinear_transform_coding}, it is usually more convenient to transform $\rvx$ into another random variable $\rvy \sim P_{\rvy \mid \rvx}$ that we encode (which will also form the basis for using relative entropy coding), and use some deterministic function $g$ to reconstruct the data: $\tilde{\rvx} = g(\rvy)$.
Then, $P_{\tilde{\rvx} \mid \rvx} = g \pushfwd P_{\rvy \mid \rvx}$.
Now, the advantage of using total variation as the measure of error here is that if we have control over the TV error of the compressor, i.e.\ for some $\epsilon > 0$ it holds that $\TVD{P_{\tilde{\rvx}}}{P_{\rvx}} \leq \epsilon$, then switching to an approximate scheme to encode $\rvy$ also translates to a guarantee in ``data space.''
Namely, assume that we now use an approximate relative entropy coding algorithm whose output $\tilde{\rvy}$ has distribution $P_{\tilde{\rvy} \mid \rvx}$ and we have the TV guarantee that for some $\delta > 0$ it holds that $\TVD{P_{\tilde{\rvy} \mid \rvx}}{P_{\rvy \mid \rvx}} < \delta$.
Then, this translates to an $\epsilon + \delta$ TV error guarantee in ``data space'':
\begin{align}
\Exp[\TVD{P_{\rvx}}{g \pushfwd P_{\tilde{\rvy} \mid \rvx}}]
&\leq \Exp[\TVD{P_{\rvx}}{g \pushfwd P_{\rvy \mid \rvx}} + \TVD{P_{g \pushfwd P_{\rvy \mid \rvx}}}{g \pushfwd P_{\tilde{\rvy} \mid \rvx}}] \tag{triangle inequality} \\
&\leq \Exp[\TVD{P_{\rvx}}{g \pushfwd P_{\rvy \mid \rvx}} + \TVD{P_{P_{\rvy \mid \rvx}}}{P_{\tilde{\rvy} \mid \rvx}}] \tag{data processing} \\
&\leq \epsilon + \delta. \tag{assumptions}
\end{align}
\subsection{Step-limited Selection Samplers}
\par
As I mentioned at the start of this section, we can turn any exact selection sampler into an approximate one by limiting the number of proposal samples it is allowed to examine and returning one of the previously examined samples if it exceeds the limit:
\begin{definition}[Step-limited Selection Sampler]
\label{def:step_limited_selection_sampler}
Let $Q \ll P$ be probability measures, and let $(X_i)_{i\in \Nats}$ be a sequence of i.i.d.\ $P$-distributed random variables.
Furthermore, \added[id={PL}, comment={PL/email/4}]{let} $N, K$ be a selection sampler (\cref{def:exact_selection_sampler}) for $Q$ using $(X_i)_{i\in \Nats}$ and let $m$ be a positive integer.
Then, define the step-limited selection rule as
\begin{align*}
N_m = 
\begin{cases}
N &\text{if } N \leq m \\
f(X_{1:m}) &\text{otherwise, for some arbitrary function $f$,}
\end{cases}
\end{align*}
and denote their output distribution as $X_{N_m} \sim \tilde{Q}_m$.
The sampler is \textbf{$\epsilon$-approximate} if for some $\epsilon > 0$ it holds that
\begin{align*}
\TVD{Q}{\tilde{Q}} \leq \epsilon.    
\end{align*}
\end{definition}
\par
\textbf{Analysing step-limited selection samplers using relative entropy coding.}
Now, let us consider the total variation error of the sampler.
\added[id={CM}, comment={Corrected the argument below}]{Note that by construction, we have that}
\begin{align}
\tilde{Q}_m(A) &= \Prob[X_{N_m} \in A, N \leq m] + \Prob[X_{N_m} \in A, N > m] \nonumber\\
&= \Prob[X_N \in A, N \leq m] + \Prob[X_{N_m} \in A, N > m] \tag{by definition of $N_m$}\\
&= Q(A) - \Prob[X_N \in A, N > m] + \Prob[X_{N_m} \in A, N > m] \tag{$\Prob[X_N \in A] = Q(A)$} \\
&= Q(A) + \Prob[N > m] \cdot (\Prob[X_{N_m} \in A\mid N > m] - \Prob[X_N \in A \mid N > m]) \label{eq:step_limited_dist_identity}
\end{align}
Therefore:
\begin{align}
\TVD{\tilde{Q}_m}{Q} 
&= \sup_{A \in \sigmaAlgebra}\abs*{\tilde{Q}_m(A) - Q(A)} \tag{definition}\\
&= \sup_{A \in \sigmaAlgebra}\abs*{\Prob[N > m] \cdot (\Prob[X_{N_m} \in A\mid N > m] - \Prob[X_N \in A \mid N > m])} \tag{by \cref{eq:step_limited_dist_identity}}\\
&= \Prob[N > m] \cdot \sup_{A \in \sigmaAlgebra}\abs*{\Prob[X_{N_m} \in A\mid N > m] - \Prob[X_N \in A \mid N > m]} \nonumber\\
&\leq \Prob[N > m], \nonumber
\end{align}
where for the inequality, I recognised the supremum term as a total variation distance and used the fact that it is bounded above by $1$.
The importance of this result is that \textit{the tail probability of the selection rule $\Prob[N > m]$ controls approximation quality}, hence so long as it decays sufficiently quickly, it does not matter what our ``contingency selection rule'' $f$ is.
A setting that is almost the same as the step-limited selection sampler (\cref{def:step_limited_selection_sampler}) was first studied by \citet{block2023sample}, who used rejection sampling (\cref{alg:global_rs}) as the basis of their approximate sampler.
The difference between their setting and \cref{def:step_limited_selection_sampler} is that they need to introduce a second approximation step between $Q$ and $\tilde{Q}_m$ to have good control of $\TVD{\tilde{Q}_m}{Q}$; this also means that the bounds that they obtain are slightly inefficient.
However, as I show next, when using a selection sampler that has relative entropy coding guarantees, \cref{def:step_limited_selection_sampler} is enough to obtain essentially optimal approximate samplers.
The reason is that there is a serendipitous alignment between relative entropy coding and approximate sampling. In both cases, we try to skew the selection rule's distribution as much towards $1$ as possible.
Formally, the relationship can be seen from an application of Markov's inequality, which we can apply since $\lb N \geq 0$: 
\begin{align}
\label{eq:approx_sampling_log_Markov_ineq}
\Prob[N > m] = \Prob[\lb N > \lb m] \leq \frac{\Exp[\lb N]}{\lb m}.
\end{align}
This result shows that control over $\Exp[\lb N]$ directly translates to control over $\Prob[N > m]$ and hence $\TVD{\tilde{Q}_m}{Q}$.
In particular, this gives us a recipe for picking the step budget based on the desired level of approximation: for $\epsilon > 0$
\begin{align}
\label{eq:approx_sampling_necessary_budget}
\epsilon = \frac{\Exp[\lb N]}{\lb m} \quad \Leftrightarrow \quad m = 2^{\frac{\Exp[\lb N]}{\epsilon}}
\end{align}
But now, note that $\Exp[\lb N]$ was precisely what I bounded when analysing the average description length of A* coding and GPRS!
Therefore, we have the following corollary based on the above reasoning, taken together with the proofs of \cref{thm:global_a_star_codelength,thm:global_gprs_codelength}.
\begin{importantCorollary}
\label{corollary:approximate_a_star_and_gprs}
Let $Q \ll P$ be probability measures with bounded $r = dQ/dP$, and let $\epsilon > 0$ be the desired maximal total variation error.
Then:
\begin{enumerate}
\item \textbf{$\epsilon$-Approximate A* sampler.} Let $\PoissonProcess$ be the Poisson process over $ \YSpace \times \nonnegReals$ with mean measure $P \otimes \lambda$ used by \cref{alg:global_a_star} to propose points.
Then, define the step-limited A* selection rule as
\begin{align*}
N_m = \argmin_{n \in [1:m]}\{T_n / r(Y_n) \mid (Y_n, T_n) \in \PoissonProcess\},
\end{align*}
and denote $Y_{N_m} \sim \tilde{Q}_m$.
Then, for a budget of
\begin{align*}
m = \left\lceil \exptwo\left(\frac{\KLD{Q}{P} + 1}{\epsilon}\right)\right\rceil,
\end{align*}
we have $\TVD{\tilde{Q}_m}{Q} \leq \epsilon$.
\item \textbf{$\epsilon$-Approximate greedy Poisson rejection sampler.} Let $N$ be the index returned by \cref{alg:global_gprs}.
Then, define the step-limited GPRS selection rule as
\begin{align*}
N_m = \begin{cases}
N &\text{if } N \leq m \\
m &\text{if } N > m,
\end{cases}
\end{align*}
and denote $Y_{N_m} \sim \tilde{Q}_m$.
Then, for a budget of
\begin{align*}
m = \left\lceil \exptwo\left(\frac{\KLD{Q}{P} + 1}{\epsilon}\right)\right\rceil,
\end{align*}
we have $\TVD{\tilde{Q}_m}{Q} \leq \epsilon$.
\end{enumerate}
\end{importantCorollary}
\begin{proof}
The proofs in both cases follow from combining \cref{eq:approx_sampling_log_Markov_ineq,eq:approx_sampling_necessary_budget} with the facts that for both A* sampling and GPRS, I have shown that their selection rule $N$ satisfies
\begin{align*}
\Exp[\lb N] \leq \KLD{Q}{P} + 1.
\end{align*}
\end{proof}
Similar approximation results can be shown for depth-limited variants of branch-and-bound A* sampling and GPRS (see \cref{alg:a_star_sac,alg:gprs_sac} in  \cref{chapter:branch_and_bound} below) and for other relative entropy coding algorithms where we have control over $\Exp[\lb N]$, such as ordered random coding \citep{theis2022algorithms}.
\Cref{corollary:approximate_a_star_and_gprs} improves upon the approximate rejection sampler proposed and analysed by \citet{block2023sample}, as it exhibits two $\epsilon$-approximate sampling algorithms that achieve the same $\epsilon$ TV error guarantee at a lower sample complexity.
In fact, the approximate algorithms in \cref{corollary:approximate_a_star_and_gprs} are almost completely optimal, as the following special case of Theorem 5 of \citet{block2023sample} shows:
\begin{theorem}[Minimum necessary sample complexity of $\epsilon$-approximate selection samplers]
Let $0 < \epsilon \leq 1/4$ and $\delta > \lb (e) - 1$.
Then, there exist a pair of probability measures $Q \ll P$ such that $\KLD{Q}{P} \leq \delta$ and any selection rule $N_m$ over $m$ i.i.d.\ $P$-distributed samples $(X_i)_{i = 1}^m$ with output distribution $X_{N_m} \sim \tilde{Q}_m$ that satisfies $\TVD{\tilde{Q}_m}{Q} \leq \epsilon$ must have
\begin{align*}
m \geq \frac{1}{2}\exptwo\left(\frac{\delta}{2\epsilon}\right).
\end{align*}
\end{theorem}
This shows that the algorithms in \cref{corollary:approximate_a_star_and_gprs} are optimal up to a factor of $2$ in the exponent and up to a factor of $2$ in front of the exponent.
\section[Considerations for Implementation on a Computer]{Considerations for Implementation \\ on a Computer}
\label{sec:implementation_considerations}
This section briefly discusses some practical tricks and considerations when implementing relative entropy coding algorithms on a digital computer.
\subsection{Using Gumbel Variates to Improve Numerical Stability}
\par
When using one of the Poisson process-based selection samplers, such as \cref{alg:global_a_star} or \cref{alg:global_gprs}, we usually need to simulate a large number of arrivals from the proposal process
$\PoissonProcess$ over $\YSpace \times \nonnegReals$ with mean measure $P \otimes \lambda$.
Concretely, if we wish to draw samples from $Q$ with $r = dQ/dP$, then by \cref{thm:selection_sampler_runtime_lower_bound}, we need to simulate at least $\norm{r}_\infty$ number of samples on average.
There is a potential problem with this when $\norm{r}_\infty$ is very large:  the order of magnitude of most of the simulated times will reach $\norm{r}_\infty$, which might become problematic to represent at a certain floating-point precision.
\par
Here, I adopt the standard solution to such problems, which is to perform all computation in the log-domain, i.e.\ to compute log-arrival times.
That is, we maintain the logarithm of the $n$th arrival time $L_n = \ln T_n$, we simulate the next inter-arrival time $\Delta_{n + 1} \sim \Exponential(1)$ and compute 
\begin{align*}
L_{n + 1} = \ln(\exp(L_n) + \Delta_{n + 1}).
\end{align*}
While this is a perfectly sensible thing to do, we can efficiently implement it in a numerically stable way via the \texttt{logaddexp} operation.%
\footnote{$\mathtt{logaddexp}(a, b)$ finds $\mu = \max\{a, b\}$, and computes ${\mu + \ln (\exp(a - \mu) + \exp(b - \mu))}$.
This is equivalent to computing $\ln (\exp(a) + \exp(b))$, but we need to exponentiate much smaller numbers, thus making it more stable numerically.}
The distribution of $L_n$ is related to the well-known Gumbel distribution:
\begin{align*}
\Prob[-L_{n + 1} \leq a \mid L_n] 
&= \Prob[\ln (\exp(L_n) + \Delta_{n + 1}) \geq a \mid L_n] \\
&= \Prob[\Delta \geq \exp(a) - \exp(L_n) \mid L_n] \\
&= \exp(-(\exp(a) - \exp(L_n)))\Ind[a \geq L_n],
\end{align*}
which we recognise as the CDF of a standard Gumbel distribution truncated to the interval $(-\infty, L_n)$.
In fact, \citet{maddison2014sampling} originally formulated A* sampling using ``Gumbel processes,'' which \citet{maddison2016poisson} reformulated in the language of Poisson processes.
\subsection{Efficient Implementation of the Common Randomness in Channel Simulation via Foldable PRNG Seeds}
\par
As mentioned throughout this thesis, the sender and receiver usually implement the common randomness required for channel simulation / relative entropy coding by sharing the seed for their pseudo-random number generator (PRNG).
For example, the sender can include the seed at negligible extra cost in the metadata of the message they send to the receiver.
But what exactly is the structure of a PRNG?
A fairly general model is that the PRNG consists of a single function $\mathtt{sample}(\rvs, P) = (\rvs', \rvx)$ that receives the seed $\rvs$ (e.g.\ a 32 or 64-bit integer) and the description of a distribution $P$ and depending on the support of $P$ outputs a new seed $\rvs'$ and a fixed- or floating-point number $\rvx$ at a certain precision.
Importantly, $\mathtt{sample}$ is deterministic, meaning that $\mathtt{sample}(\rvs, P)$ will always return the same $\rvs', \rvx$ tuple, we only have $\rvx \sim P$ when ``averaging over'' different seeds $\rvs$.
We can already implement our relative entropy coding algorithms \cref{alg:global_a_star,alg:global_gprs} using $\mathtt{sample}$.
First, the sender shares their initial seed $\rvs_1$ with the receiver.
Then, they simulate the arrivals of the base Poisson process in \cref{alg:global_pp_simulation} by chaining the seeds output by $\mathtt{sample}$: at step $n \geq 1$, the sender simulates $\rvs_n', \Delta_n \gets \mathtt{sample}(\rvs_n, \Exp(1))$ and $\rvs_{n + 1}, Y_n \gets \mathtt{sample}(\rvs_n, P)$.
They continue this process until they find the selected index $N$, which they transmit to the receiver.
Then, the receiver can simulate the selected sample $Y_N$ by chaining their $\mathtt{sample}$ functions the same way the encoder did.
\par
\textbf{Improving the speed using foldable PRNGs.}
Observe that in the above procedure, the decoding algorithm's runtime matches the runtime of the encoding algorithm.
However, if we have a slightly more sophisticated PRNG, we can significantly reduce the decoder's runtime.
In particular, assume that the PRNG offers a second function: $\mathtt{fold\_in}(\rvs, n) = \rvs'$, which takes a seed $\rvs$ and an integer $n$ and produces a new seed $\rvs'$.
Then, the sender can modify their simulation procedure as follows.
They share a \textit{base seed} $\rvs$ with the receiver.
Then, at each step $n \geq 1$, they compute $\rvs_n \gets \mathtt{fold\_in}(\rvs, n)$ and simulate $\rvs_n', \Delta_n \gets \mathtt{sample}(\rvs_n, \Exp(1))$ and $\rvs_{n}'', Y_n \gets \mathtt{sample}(\rvs_n, P)$.
Finally, they communicate the selected index $N$ to the receiver as before.
\par
The advantage of this modified procedure is that this way, \textbf{the receiver can decode the selected sample in one step}: given the base seed $\rvs$ and the index $N$ they decoded, they compute $\rvs_N \gets \mathtt{fold\_in}(\rvs, N)$ using which they can immediately simulate $Y_N$!
Such a speedup can be vital in implementing relative entropy coding in low-latency, high-throughput situations, such as video streaming.
In such scenarios, the encoding time is irrelevant from the decoder's perspective. 
However, fast decoding is essential to avoid negatively impacting the user experience.
Hence, I implemented all algorithms I used for the experiments in \cref{sec:numerical_experiments} in \texttt{Python} using the \texttt{jax} numerical computing library \citep{jax2018github}, as it implements the \texttt{Threefry} PRNG \citep{salmon2011parallel}, which offers the folding functionality I described above.
\section{Conclusion and Open Questions}
\label{sec:point_processes_conclusion}
This chapter presented a unifying picture of how Poisson processes can be used to construct sampling and relative entropy coding algorithms.
My contributions in this chapter are the construction and analysis of greedy Poisson rejection sampling in \cref{sec:global_gprs}, parallelised sampling in \cref{sec:superposition_parallelisation} and almost optimal approximate selection sampling in \cref{sec:approximate_sampling}.
\par
There are many possible directions for future research based on what I presented here.
First, it would be interesting to investigate if using \textit{quasi-random number generators} (QRNG) using low-discrepancy sequences, such as Sobol sequences, could be used to develop more efficient relative entropy coding.
Quasi-Monte Carlo methods are known to have strictly better convergence guarantees when used for numerical integration.
However, can they also lead to benefits in the approximate sampling/relative entropy coding case?
\par
Another direction is to investigate whether we could use different point processes to construct sampling algorithms, such as determinantal point processes
\citep{kulesza2012determinantal}.
For example, the original greedy rejection sampling algorithm of \citet{harsha2010communication}, which I generalised to general Polish spaces in \citet{flamich2023adaptive,flamich2023grc}, can be formulated as using a point process that is a rather unnatural approximation of a Poisson process.
The main difficulty in this direction is that giving up the ``complete randomness'' of the point process will likely complicate the analysis significantly.
Nonetheless, such an algorithm could be significantly more efficient than general-purpose selection samplers in some cases.
\par
Finally, it is exciting to consider whether any of the techniques I presented here could be extended for the construction and analysis of Markov chain Monte Carlo algorithms, too.
Furthermore, from the compression perspective, it would be interesting to consider if we could construct a relative entropy coding algorithm using Markov chain Monte Carlo methods.

%% file: 5-BranchAndBound/branch_and_bound.tex
%!TEX root = ../thesis.tex
%*******************************************************************************
%****************************** Second Chapter *********************************
%*******************************************************************************

\chapter[Fast Relative Entropy Coding with Branch-and-Bound Samplers]{Fast Relative Entropy Coding with Branch-and-Bound Samplers}
\label{chapter:branch_and_bound}
While the algorithms I developed in the previous chapter are mathematically elegant, they are general-purpose selection samplers and are all \textit{slow}.
If we cannot assume that our sample space $\YSpace$ possesses any structure beyond being Polish, selection sampling is virtually the only thing we can do, and \cref{thm:selection_sampler_runtime_lower_bound} already shows that selection sampling, and \cref{alg:global_rs,alg:global_a_star,alg:global_gprs} in particular, are optimal.
From a search perspective, we can think of these algorithms as performing linear search: the proposal Poisson process $\PoissonProcess = ((Y_1, T_1), (Y_2, T_2), \hdots)$ is an infinite ordered list, and the samplers I developed examine these points one-by-one.
However, in many cases, we might know a great deal more about both our sample space and the target-proposal density ratio, which, if we incorporate it into our sampling algorithm, can speed it up significantly.
\par
A prevalent advanced search paradigm is binary / branch-and-bound search.
The high-level idea is that we recursively \textit{break up the problem} of finding an item in the whole of space into smaller subproblems over smaller and smaller parts of space.
Furthermore, each time we split up the problem, we use additional information about the problem at hand to \textit{eliminate parts of the search space} that cannot contain the object we are looking for.
The variants of A* sampling and GPRS I develop in \cref{sec:branch_and_bound_samplers} will follow this line of thought; hence, I will collectively refer to them as \textit{branch-and-bound samplers}.
\added[id={CM}, comment={Addressing ``You mentioned in Ch2 that it's optimally fast, can you just say again, which settings is it optimal, i.e., what are the structural assumptions?''}]{After describing them, I analyse their average codelength and runtime.
The highlight of the chapter is \Cref{thm:bnb_gprs_runtime}, which shows that for one-dimensional, unimodal density ratios, branch-and-bound GPRS achieves $\Oh(\KLD{Q}{P})$ runtime. This result shows that branch-and-bound GPRS is the first optimally fast relative entropy coding algorithm, since any such algorithm needs at least $\KLD{Q}{P}$ steps to output the bits of the code for the sample they return. 
Unfortunately, a GPRS requires that we be able to compute the width function of the density ratio, which is a significant limitation for general-purpose sampling.
Fortunately, we can use branch-and-bound A* sampling instead, which I show to have a slightly worse $\Oh(\infD{Q}{P})$ runtime (compared to GPRS) for the same problem class in \cref{thm:bnb_a_star_runtime}, without having to compute the width function.
However, I conjecture that the $\Oh(\infD{Q}{P})$ runtime cannot be improved to $\Oh(\KLD{Q}{P})$ without further structure.}
Now, before I move on to describing these algorithms and their analysis, I will develop the necessary mathematical language in the next section.
\section[Advanced Poisson Process Theory]{Advanced Poisson Process Theory}
\par
This section first develops an alternative way to construct Poisson processes when additional geometric structure is available.
Then, it discusses how to simulate the alternative construction on a computer.
\subsection{Splitting Functions and Binary Space-partitioning Trees}
\label{sec:splitting_functions_and_bsp_trees}
\begin{figure}[t]
\centering
\includegraphics{4-RelativeEntropyCodingWithPoissonProcesses/img/bsp_tree_illustration.tikz}
\caption[Illustartion of a binary space-partitioning tree]{Illustration of various concepts from \cref{sec:splitting_functions_and_bsp_trees}. The \textbf{solid} nodes indicate a \textbf{search tree} $\SubTree$ of the BSP tree induced by the dyadic splitting function (\cref{def:dyadic_splitting_function}).
The \textbf{dashed} nodes indicate the \textbf{frontier} $\Frontier_\SubTree$ (\cref{def:frontier}).
The {\color{red} red} nodes form a \textbf{proper, maximal}, search tree $\SubTree' \lhd \SubTree$.
Each node has three labels associated with it: \textbf{H}: the heap index, \textbf{D}: the depth and \textbf{B}: the bounds of the node.
The heap index enumerates the tree nodes in a breadth-first fashion: the numbers increase top-to-bottom and left-to-right.
The node's depth indicates its distance to the root node.
The picture also illustrates \cref{lemma:bsp_frontier_partitions_space}: the bounds associated to the frontier of every sub-search tree of $\SubTree$ form a partition of $(0, 1]$.
}
\label{fig:bsp_tree_illustration}
\end{figure}%
\par
\textbf{Splitting functions.}
The starting point is a way to break the search space into smaller parts to reduce the larger search problem into subproblems.
This consideration leads me to the following definition.
\begin{definition}[Splitting function]
Let $(\YSpace, \sigmaAlgebra, P)$ be a measure space.
Then, a (possibly random) splitting function $\splitFun: \sigmaAlgebra \to \sigmaAlgebra \times \sigmaAlgebra$ for the measure space takes a measurable set $A \in \sigmaAlgebra$ and returns a $P$-almost partition of $A$:
\begin{align*}
\splitFun(A) = L, R \quad \Rightarrow\quad L, R \subseteq A, \,\, L \cap R = \emptyset, \,\,P(A \setminus (L \cup R)) = 0.
\end{align*}
\end{definition}
The splitting function is the core mechanism that I  leverage to implement the branch-and-bound search algorithms. 
My definition is quite abstract to demonstrate the generality of the idea.
However, I now give the two most important examples, the latter of which I will use for the rest of the thesis.
\begin{definition}[Dyadic splitting function.]
\label{def:dyadic_splitting_function}
Let $\YSpace = (0, 1]$ be the half-open unit interval equipped with the standard Borel $\sigma$-algebra.
Then, for $0 < l < r \leq 1$, the dyadic splitting function is defined as
\begin{align*}
\splitFun((l, r]) = \big(l, (l + r) / 2\big],\,\, \big((l + r)/2, r\big]
\end{align*}
\end{definition}
In plain words, the dyadic splitting function takes half-open intervals and splits them exactly in half at their midpoint.
The reason why I call this function \textit{dyadic} is that if we start by splitting the unit interval $(0, 1]$ and then apply $\splitFun$ repeatedly to the sets it produces, their sizes will be negative powers of two and their endpoints will be dyadic integers, see \cref{fig:bsp_tree_illustration} for an illustration.
A slightly more complex splitting function, though similar in spirit, is what I call the \textit{on-sample} splitting function.
\begin{definition}[On-sample splitting function.]
\label{def:on_sample_splitting_function}
Let $\YSpace$ be a connected subset of $\Reals$ with the standard Borel $\sigma$-algebra and some non-atomic measure $P$ over it.
Then, the on-sample splitting function for some $l, r \in \YSpace$ with $l < r$ is defined as
\begin{align*}
\splitFun((l, r]) = (l, Y],\,\, (Y, r] \quad Y \sim P\vert_{(l , r]}.
\end{align*}
where $P\vert_{A}(B) = P(A \cap B) / P(A)$.
\end{definition}
Thus, the on-sample splitting function randomly splits intervals in two instead of always picking the midpoint as the ``pivot.''
While the randomisation might appear an unnecessary complication at first, it will turn out to be the key to designing the optimally fast variants of A* sampling and GPRS, and thus the samplers and analyses in \cref{sec:branch_and_bound_samplers} will revolve entirely around this splitting function.
\par
As I already mentioned, I will repeatedly apply $\splitFun$ to its own output, which now leads to the following important definition.
\begin{definition}[Induced binary space-partitioning tree]
Let $(\YSpace, \sigmaAlgebra, P)$ be a measure space with splitting function $\splitFun$.
Then, the binary space-partitioning tree (BSP-tree) induced by $\splitFun$ is an infinite, (possibly randomly) labelled, rooted, directed binary tree $\Tree$ with root node $\rho$, with the following labelling rules:
\begin{enumerate}
\item The labels take values in $\sigmaAlgebra$.
Thus, let $B_\nu \in \sigmaAlgebra$ denote the label of an arbitrary node $\nu \in \Tree$.
I will refer to $B_\nu$ as the \textbf{bounds} associated with $\nu$.
\item The label associated with the root $\rho$ is $\YSpace$.
\item 
For each node $\nu \in \Tree$ with children $\alpha, \beta$, set
\begin{align*}
\splitFun(B_\nu) = B_\alpha,\,\, B_\beta.
\end{align*}
\end{enumerate} 
\end{definition}
The utility of the induced BSP-tree of some splitting function $\splitFun$ is that it conveniently collects all the sets we can obtain by repeatedly applying $\splitFun$ to its own output.
Next, I introduce another key concept that I will make heavy use of.
\begin{definition}[Search tree and frontier]
\label{def:frontier}
Let $(\YSpace, \sigmaAlgebra, P)$ be a measure space, with splitting function $\splitFun$ and induced BSP-tree $\Tree$.
Then, I will refer to finite subtrees $\SubTree$ of $\Tree$ as \textbf{search trees}.
Furthermore, let $\Frontier_\SubTree$ collect the descendants of the leaves in $\SubTree$ in $\Tree$:
\begin{align*}
\Frontier_{\SubTree} = \{\nu \in \Tree \mid \nu \text{ is a leaf node in } \SubTree\}.
\end{align*}
Then, I call the $\Frontier_\SubTree$ the \textbf{frontier} of the search tree $\SubTree$.
\end{definition}
The frontier has an intuitive interpretation from the search perspective: assuming we have already examined the nodes in a search tree $\SubTree$, the frontier $\Frontier$ represents the set of nodes we can examine next; see \cref{fig:bsp_tree_illustration} for an illustration.
The following is an important result about frontiers.
\begin{lemma}[The frontier partitions the ambient space]
\label{lemma:bsp_frontier_partitions_space}
Let $(\YSpace, \sigmaAlgebra, P)$ be a measure space, with splitting function $\splitFun$ and induced BSP-tree $\Tree$.
Let $\SubTree$ be a search tree of $\Tree$ with frontier $\Frontier$.
Then, $\bigcup_{\nu \in \Frontier} B_\nu$ is a $P$-almost partition of $\YSpace$.
\end{lemma}
The proof of this fact proceeds by run-of-the-mill induction.
However, it is useful because it mimics an important feature of the search algorithms I develop in \cref{sec:branch_and_bound_samplers}: obtaining a new search tree from an old one by taking a node from its frontier and \textbf{expanding} or splitting that node.
\begin{proof}
Let $\SubTree'$ be a \textit{maximal, proper} sub-search tree of $\SubTree$, if there is precisely one fewer node in $\SubTree'$ and the missing node $\nu$ is in the frontier of $\SubTree'$.
Formally, $\abs{\SubTree'} + 1 = \abs{\SubTree}$ and $\nu \in \Frontier_{\SubTree'}$ for $\nu \in \SubTree \setminus \SubTree'$.
I will denote this relationship by $\SubTree' \lhd \SubTree$, and it implies that we can obtain $\SubTree$ from $\SubTree'$ by \textbf{expanding} one of the nodes on its frontier $\Frontier_{\SubTree'}$ and adding it to $\SubTree'$; see \cref{fig:bsp_tree_illustration} for an illustration.
\par
Now, let $\rho = \SubTree_1 \lhd \SubTree_2 \lhd \hdots \lhd \SubTree_n = \SubTree$ be a chain of maximal, proper, search trees.
Note that such a chain always exists and is finite with length $n = \abs{\Frontier_\SubTree} + 1$: the frontier of the root node $\rho$ has two elements and going from $\SubTree_i$ to $\SubTree_{i + 1}$ always increases the size of the frontier by $1$. 
This is because we remove the leaf node $\nu \in \SubTree_i$ that we expand from the next frontier, but now its two children become leaf nodes in the new tree $\SubTree_{i + 1}$.
\par
To prove the statement of the lemma, we proceed by induction.
For the base case, note that the root node's label $B_\rho = \YSpace$ and hence its descendants' associated bounds are the two outputs of $\splitFun(\YSpace)$, which form a $P$-almost partition of $\YSpace$ by definition.
Then, assuming the statement holds for some $1 \leq j < n$, we can perform the inductive step as follows.
Let $\mu \in \SubTree_j$ the node that we expand to obtain $\SubTree_{j + 1}$ with descendants $\alpha, \beta$.
Now, note that $\bigcup_{\nu \in \Frontier_{\SubTree_i} \setminus \mu} B_\nu$ forms a $P$-almost partition of $\YSpace \setminus B_\mu$ and $\{B_{\alpha}, B_{\beta}\}$ form a $P$-almost partition of $B_\mu$, hence taking the union of these two sets forms a $P$-almost partition of $\YSpace$.
But $(\Frontier_{\SubTree_i} \setminus \mu) \cup \{\alpha, \beta\}$ is precisely $\Frontier_{\SubTree_{j + 1}}$, which finishes the proof.
\end{proof}
For data compression, it will also be important to be able to refer to each node of an induced BSP-tree using a positive integer; hence, I introduce this next.
\begin{definition}[Heap index and depth]
\label{def:heap_index_and_depth}
Let $\Tree$ be an infinite, rooted, directed binary tree.
Then, we can address each node $\nu \in \Tree$ via a unique positive integer that I will call the heap index of $\nu$, defined as follows:
\begin{enumerate}
\item The heap index of the root node $\rho$ is $1$.
\item For a node $\nu \in \Tree$ with heap index $H_\nu$ and children $\alpha, \beta$, we set
\begin{align*}
H_\alpha &= 2 H_\nu \\
H_\beta &= 2 H_\nu + 1.
\end{align*}
\end{enumerate}
Finally, I define the depth of a node $\nu$ in the tree as 
\begin{align*}
D_\nu = \lfloor \lb (H_\nu) \rfloor.
\end{align*}
\end{definition}
Intuitively, the heap index is a breadth-first enumeration of the nodes of the binary tree, and the depth of a node is its distance from the root; see \cref{fig:bsp_tree_illustration} for an illustration.
\par
The fast samplers I develop in \cref{sec:branch_and_bound_samplers} will use the on-sample splitting function, but of course they will not realise its entire induced BSP-tree.
Instead, they will use a \textit{search heuristic} to inform the algorithms which part of space they should investigate next.
In technical terms, these algorithms will maintain a search tree $\SubTree$, and the search heuristic will guide which node on its frontier the algorithms should expand next.
\subsection[Simulating Poisson Processes via Binary Space-partitioning Trees]{Simulating Poisson Processes via \\ Binary Space-Partitioning Trees}
\label{sec:simulating_bnb_poisson_processes}
In this section, I combine Poisson processes with the machinery of splitting functions and BSP-trees I developed in the previous section.
The two central ideas I develop are 1) how a BSP-tree induced by a splitting function can be used to construct an alternative way of simulating a Poisson process and 2) connecting the indices given by ordering of the points of a Poisson process given by their arrival times and the heap indices of the points induced by a BSP-tree.
The following construction shows the first central idea.
\begin{definition}[BSP structure over a Poisson process.]
\label{def:bsp_structure_over_poisson_process}
Let $(\YSpace, \sigmaAlgebra, P)$ be a probability space, and $\PoissonProcess$ a Poisson process over $\YSpace \times \nonnegReals$ with mean measure $P \otimes \lambda$.
Let $\splitFun$ be a splitting function for $\YSpace$ that can depend on $\PoissonProcess$ and denote its induced BSP-tree over $\YSpace$ by $\Tree$.
Then, the \textbf{binary space-partitioning structure} $\Tree_\PoissonProcess$ over $\PoissonProcess$ is the following recursive extension of the BSP-tree $\Tree$:
\begin{enumerate}
\item We adjoin the first arrival $(Y_1, T_1)$ of $\PoissonProcess$ and the \textbf{time index} $1$ to the label of the root node of $\Tree$.
\item For any further node $\nu \in \Tree$ with associated bounds $B_\nu$, we adjoin the first arrival $(Y_\nu, T_\nu)$ of $\PoissonProcess$ restricted to $B_\nu$ that is not already adjoined to an ancestor of $\nu$.
Furthermore, we also adjoin the \textbf{time index} $N_\nu$ of the arrival, that is $N_\nu = n$ if $(Y_\nu, T_\nu)$ is the $n$th arrival of $\PoissonProcess$.
\end{enumerate}
\end{definition}
Importantly, as the following result shows, we almost surely do not ``lose'' any of the points of a Poisson process $\PoissonProcess$ in its BSP structure.
\begin{lemma}[BSP structure contains almost all points of its underlying Poisson process]
\label{lemma:bsp_structure_is_complete}
Let $\YSpace$, $\PoissonProcess$, $\Tree_\PoissonProcess$ as in \cref{def:bsp_structure_over_poisson_process}.
Then, $\Tree_\PoissonProcess$ contains $P$-almost all points of $\PoissonProcess$.
\end{lemma}
An equivalent statement was shown for Gumbel processes in \citet{maddison2014sampling}. 
My argument below is based on Lemma D.1 from \citet{flamich2023grc}.
\begin{proof}
I will prove the statement of the lemma by showing that for each positive integer $n$, the BSP structure contains the $n$-th arrival of $\PoissonProcess$ $P$-almost surely.
I now show this latter statement using induction.
For the base case, note that the first arrival of $\PoissonProcess$ is always associated with the root of the BSP structure, hence for $n = 1$ the statement holds.
\par
Now, I assume the statement holds for $n \leq k$ and show the inductive step for $n = k + 1$.
First, consider the first $k$ arrivals of $\PoissonProcess$: by the inductive hypothesis, the BSP structure $\Tree_\PoissonProcess$ has nodes associated with each arrival $P$-almost surely.
Thus, let $\Gamma_k$ be the subgraph of $\Tree_\PoissonProcess$ formed by the nodes associated with the first $k$ arrivals.
Next, I continue by showing the following sub-result.
\begin{lemma}[The first $k$ arrivals form a search tree]
\label{lemma:first_k_arrivals_form_a_search_tree}
Let $\Gamma_k$, $\PoissonProcess$ and $\Tree_\PoissonProcess$ be as above.
Then, $\Gamma_k$ is a subtree of $\Tree_\PoissonProcess$.
\end{lemma}
\begin{proof}
To prove the result, it is enough to show that $\Gamma_k$ is connected, since any connected subgraph of a tree is a subtree.
To show this, assume the contrary, namely that there are at least two separate components in $\Gamma_k$.
Let $R_k$ denote the component of $\Gamma_k$ that contains the root $\rho$ of $\Tree_\PoissonProcess$, and let $\nu \in \Gamma_k$ node such that $\nu \not\in R_k$.
Since $\rho, \nu \in \Tree_{\PoissonProcess}$, there is a unique path between them.
Furthermore, since $\rho$ and $\nu$ are in different components of $\Gamma_k$, there is a node $c \in \Tree_\PoissonProcess$ along the path that connects $\nu$ and $\rho$ such that $c \not\in \Gamma_k$.
This leads to a contradiction: since $c \not\in \Gamma_k$, its arrival time $T_c$ must exceed the $k$th arrival time of $\PoissonProcess$.
However, since $c$ is an ancestor of $\nu$, by the definition of $\Tree_\PoissonProcess$, we must have $T_c \leq T_\nu$, a contradiction.
\end{proof}
Thus, let $\Gamma_{k}$ be the search tree of $\Tree_\PoissonProcess$ containing the first $k$ arrivals of $\PoissonProcess$, and let $\Frontier_k$ denote its frontier.
By \cref{lemma:bsp_frontier_partitions_space}, the bounds in $\Frontier_k$ form a $P$-almost partition of $\YSpace$.
Combining this with the fact that $T_\nu \geq T_{k}$ for $\nu \in \Frontier$, we see that the first arrival across the nodes in the frontier,
\begin{align*}
T_k^* = \min_{\nu \in \Frontier_k} T_\nu
\end{align*}
is $P$-almost surely the $k + 1$th arrival of $\PoissonProcess$, as desired.
\end{proof}
An important corollary of \cref{lemma:bsp_structure_is_complete} is that it yields a second way of generating the points of a Poisson process:
\begin{corollary}[Simulating a Poisson process via a splitting function]
\label{corollary:simulate_pp_via_splitting_fun}
Let $P$ be a probability measure over $\YSpace$ with a splitting function $\splitFun$ and induced BSP tree $\Tree$ with root $\rho$.
Then, let $\rho = \SubTree_1 \lhd \SubTree_2 \lhd \SubTree_2 \lhd \hdots$ be an infinite chain of (possibly random) maximal, proper search trees of $\Tree$ such that the chain converges to $\Tree$ in the sense that for any heap index $H$ there exists a positive integer $n$ such that $\nu_H \in \SubTree_n$.
Then, consider the following procedure:
\begin{enumerate}
\item Set $T_\rho \sim \Exponential(1)$ and $Y_\rho \sim P$.
\item For each $i > 1$, let $\nu \in \SubTree_i \setminus \SubTree_{i - 1}$ be the ``missing node'' which we expand next, and let $\mu$ be the parent node of $\nu$.
Thus, set $T_{\nu} \sim T_\mu + \Exponential(P(B_\nu))$ and $Y_\nu \sim P\vert_{B_\nu}$.
\end{enumerate}
Then, the set of points $\{(Y_\nu, T_\nu)\}_{\nu \in \Tree}$ form a Poisson process with mean measure $P \otimes \lambda$.
\end{corollary}
\begin{proof}
The proof is immediate from \cref{lemma:bsp_structure_is_complete} and noting that by the independence property of Poisson processes and by \cref{eq:pois_process_inter_arrival_identity} the point $(Y_\nu, T_\nu)$ we construct for each expanded node $\nu$ is equal in distribution to the first arrival of a Poisson process with mean measure $P \otimes \lambda$ restricted to $B_\nu$ that has not been assigned to an ancestor of $\nu$.
Finally, by the convergence requirement, the procedure is guaranteed to simulate every node in the BSP structure eventually, which finishes the proof. 
\end{proof}
\par
\textbf{Simulating a single branch of the BSP structure.}
The algorithms I describe in the following sections will only simulate a single branch of the BSP structure induced by on-sample splitting, which is why they will be more efficient than their general-purpose variants.
Therefore, I introduce some notation for single-branch simulation that will simplify its presentation and discuss a useful perspective for studying the setting.
\par
Thus, let $\YSpace \subseteq \Reals$ be connected, and let $\PoissonProcess$ be a Poisson process over $\YSpace \times \nonnegReals$ with mean measure $P \otimes \lambda$, and let $\Tree_\PoissonProcess$ be the BSP structure over $\PoissonProcess$ induced by on-sample splitting.
Now, let $\Branch \subseteq \Tree_\PoissonProcess$ be an infinite branch of the tree, i.e.\ $\Branch$ is a connected subgraph of $\Tree_\PoissonProcess$, such that the depth of each node is unique.
Now, since the depths are unique natural numbers, this yields another natural way of indexing the arrivals of $\PoissonProcess$ that form part of $\Branch$: 
\begin{displayquote}
For $d \in \Nats$, I will write $\nu^d$ to denote the unique node in $\Branch$ whose depth is $d$.
Similarly, I will write $Y^d, T^d, B^d, H^d$ to denote the arrival location, arrival time, bounds and heap index associated with $\nu^d$.   
\end{displayquote}
\par
Next, for a branch $\Branch$, I show how we can ``reparametrise'' its arrivals $\{(Y^d, T^d)\}_{d = 0}^\infty$ to a sequence of i.i.d.\ sequences, which will be helpful for the analysis and the implementation.
Starting with the $Y^j$s, note that $Y^j \mid Y^{0:j - 1} \sim P_{B_j}$, where $B_j$ is fully determined by $Y^{0:j - 1}$ within $\Branch$.
I will now use the generalised inverse probability transform to reparametrise these variables.
To this end, let $U^{0:\infty}$ be a sequence of i.i.d.\ random variables with $U^j \sim \Unif(0, 1)$, and let $F$ be the CDF of $P$.
Then, set $Y^0 = F^{-1}(U^0)$.
Furthermore, for step $j > 0$, assume the bounds $B^j = (L^j, R^j)$, as determined by $Y^{0:j - 1}$.
Then, set 
\begin{align}
\label{eq:branch_and_bound_loc_reparam}
Y^j = F^{-1}(F(L^j) + (F(R^j) - F(L^j)) \cdot U^j).
\end{align}
To see that $Y^j$ has the right distribution, note that $\Prob[Y^j \leq y] = \Ind[y \geq L^j] \cdot (F(y) - F(L^j)) / P(B^j)$.
Inverting this CDF and noting that $P(B^j) = F(R^j) - F(L^j)$ yields \cref{eq:branch_and_bound_loc_reparam}. 
Moving on to the $T^j$s now, recall that $T^j \mid T^{j - 1}, B^j$ is the first arrival time of $\PoissonProcess$ restricted to $ B^j \times [T^{j - 1}, \infty)$.
Hence, using \cref{eq:pois_process_inter_arrival_identity}, we find that $(T^j - T^{j - 1} \mid T^{j - 1}, B^j) \sim \Exponential(P(B^j))$.
Therefore, let $E^{0:\infty}$ be a sequence of i.i.d.\ random variables with $E^j \sim \Exponential(1)$, and set
\begin{align}
\label{eq:branch_and_bound_time_reparam}
T^j = \sum_{i = 0}^j \frac{E^i}{P(B^i)}.
\end{align}
From the above reasoning, we can see that $T^j$ will follow the required distribution.
\section[Sampling as Search II: Branch-and-Bound Samplers]{Sampling as Search II: \\ Branch-and-Bound Samplers}
\label{sec:branch_and_bound_samplers}
Finally, I develop and analyse fast variants of A* sampling and greedy Poisson rejection sampling in this section.
In both cases, the important insight is that when the target-proposal density ratio is unimodal, then if we use on-sample splitting (\cref{def:on_sample_splitting_function}) to simulate the proposal Poisson process, we can \textit{always} eliminate one of the resultant regions from the search.
The intuition why this speeds up the search procedure is that the average depth of a search tree $\SubTree_N$ corresponding to the first $N$ arrivals of the proposal Poisson process $\PoissonProcess$ is $\Oh(\lb N)$.
Then, the ability to rule out a region at each step is equivalent to only ever simulating a single branch of the search tree, which will take $\Oh(\lb N)$ steps.
Based on this loose intuition, in this section, I prove that the branch-and-bound variant of A* sampling simulates $\Oh(\infD{Q}{P})$, while GPRS simulates $\Oh(\KLD{Q}{P})$ samples before terminating on average.
These results demonstrate the benefit of leveraging additional information to construct faster sampling algorithms and showcase GPRS's advantage over A* sampling.
As we will see, this discrepancy in the average runtime is due to the greedy nature of GPRS: it stops the moment it encounters the sample it selects.
On the other hand, A* coding needs to check more samples, even if it has already encountered the sample that it eventually returns.
\par
\textbf{Unimodality.}
Before I begin, I need to clarify what I mean by a ``unimodal'' density ratio: it turns out that the relatively weak notion of \textit{quasiconcavity} is enough.
\begin{definition}[Quasiconcave function]
For some $n \geq 1$ and some convex subset $S \subseteq\Reals^n$, a function $f: S \to \nonnegReals$ is quasiconcave if all its superlevel sets
\begin{align*}
\alpha \in \nonnegReals: \quad S_\alpha = \{x \in S \mid f(x) \geq \alpha\}
\end{align*}
are convex.
\end{definition}
In one dimension, quasiconcavity is essentially the weakest notion of unimodality since the only one-dimensional convex sets are intervals.
This means that one-dimensional non-quasiconcave functions have at least one level set that is not an interval (for example, it could be the union of two intervals), corresponding to the intuitive notion of having ``more than one peak.''
In dimensions higher than one, quasiconcavity is not the weakest notion of unimodality, as, for example, it excludes functions with a single, banana-shaped peak (for this, we could relax the requirement on the superlevel sets to be connected only).
However, the higher-dimensional versions of the branch-and-bound samplers I discuss in \cref{sec:general_branch_and_bound_variants} rely on the existence of supporting hyperplanes to the superlevel sets and thus require quasiconcavity.
\subsection{Branch-and-bound A* Sampling}
\label{sec:bnb_a_star}
\par
In this section, I derive and analyse the branch-and-bound variant of A* sampling when the density ratio is unimodal.
Specifically, assume now that our sampling setting is that we have two probability measures $Q \ll P$ over an interval $\YSpace \subseteq \Reals$ with quasiconcave $r = dQ/dP$ such that we have an upper bound $M$ on $r$.
\par
\textbf{Deriving the algorithm.}
Let $\PoissonProcess$ be a Poisson process with mean measure $P \otimes \lambda$, and recall from \cref{sec:global_a_star} that A* sampling searches for 
\begin{align*}
(Y^*, T^*) &= (Y_N, T_N) \\
N &= \argmin_{n \in \Nats} \{ T_n / r(Y_n) \mid (Y_n, T_n) \in \PoissonProcess\}.
\end{align*}
\begin{figure}[H]
\centering
\begin{minipage}[t]{0.49\textwidth}
 \begin{algorithm}[H]
\SetAlgoLined
\DontPrintSemicolon
\textbf{Input:}\;
Proposal distribution $P$, \;
Density ratio $r = dQ/dP$, \;
Upper bound $M$ on $r$, \;
Location $m$ of the mode of $r$.\;
\textbf{Output:}\;
Sample $Y \sim Q$ and its heap index $H\hspace{-2mm}$\;
\;
$U, T^0\!\!, H, B, Y^*\!\!, H^* \gets \infty, 0, 1, \Reals, \perp, 1\hspace{-5mm}$\;
\For{$d = 1, 2, \hdots$}{
$Y^d \sim P\vert_B$\;
$\Delta^d \sim \mathrm{Exp}\left(P(B)\right)$\;
$T^d \gets T^{d - 1} + \Delta^d$\;
\;
\If{$T^d / r(Y^d) < U$}{
$U \gets T^d / r(Y^d)$\;
$Y^*, H^* \gets Y^d, H$\;
}
\If{$U < T^d / M$}{
\Return{$Y^*, H^*$} \;
}
\eIf{$m \leq Y^d$}{
$B \gets B \cap (-\infty, Y^d)$\;
$H \gets 2H$\;
}{
$B \gets B \cap (Y^d, \infty)$\;
$H \gets 2H + 1$\;
}
}%
\parbox{0.89\linewidth}{\caption{Branch-and-bound A* sampling on $\Reals$ with unimodal $r$.}%
\label{alg:a_star_sac}%
}%
\end{algorithm}
\end{minipage}%
\hfill
\begin{minipage}[t]{0.49\textwidth}
 \begin{algorithm}[H]
\SetAlgoLined
\DontPrintSemicolon
\SetKwInOut{Input}{Input}\SetKwInOut{Output}{Output}
\SetKwFunction{GetMode}{GetMode}\SetKwFunction{CalculateProbs}{CalculateProbs}
\SetKwData{RatioMode}{ratio\_mode}
\textbf{Input:}\;
Proposal distribution $P$, \;
Density ratio $r = dQ/dP$, \;
Stretch function $\sigma$, \;
Location $m$ of the mode of $r$.\;
\textbf{Output:}\;
Sample $Y \sim Q$ and its heap index $H\hspace{-2mm}$\;
\;
$T^0, H, B \gets (0, 1, \Reals)$\;
\For{$d = 1, 2, \hdots$}{
$Y^d \sim P\vert_B$\;
$\Delta^d \sim \mathrm{Exp}\left(P(B)\right)$\;
$T^d \gets T^{d - 1} + \Delta^d$\;
\;
\;
\;
\;
\;
\If{$T^d < \sigma(r(Y^d))$}{
\Return{$Y^d, H$} \;
}
\eIf{$m \leq Y^d$}{
$B \gets B \cap (-\infty, Y^d)$\;
$H \gets 2H$\;
}{
$B \gets B \cap (Y^d, \infty)$\;
$H \gets 2H + 1$\;
}
}
\vspace{4.6mm}
\caption{Branch-and-bound GPRS on $\Reals$ with unimodal $r$.\vspace{0.8mm}}
\label{alg:gprs_sac}
\end{algorithm}
\end{minipage}%
\vspace{-0.5cm}
\end{figure}%
Consider what information we have after examining the first arrival $(Y_1, T_1)$ of $\PoissonProcess$.
The only way $N \neq 1$ is if there is some $n > 1$ such that $T_n / r(Y_n) < T_1 / r(Y_1)$.
However, this can only occur if $r(Y_n) > r(Y_1)$, since $T_n > T_1$ for $n > 1$.
We can reformulate the condition that $r(Y_n) > r(Y_1)$ as saying that $Y_n$ must be in some $\alpha$-superlevel set $S_\alpha$ of $r$ for some $\alpha > r(Y_1)$:
\begin{align*}
\exists \alpha > r(Y_1): \quad Y_n \in S_\alpha = \{y \in \YSpace \mid r(y) \geq \alpha \}.
\end{align*}
However, we can now exploit the quasiconcavity of $r$, as it is sufficient to keep searching for $N$ either to the left or right of $Y_1$!
To see this, note that 1) by its definition, $S_\alpha$ does not contain $Y_1$, and 2) $S_\alpha$ is convex and hence an interval by the quasiconcavity of $r$.
Taking these two facts together, it must be the case that $S_\alpha \subseteq (-\infty, Y_1)$ or $S_\alpha \subseteq(Y_1, \infty)$ for any $\alpha > r(Y_1)$. 
Note that these two subsets are precisely the outputs of $\splitFun(\YSpace)$, where $\splitFun$ is the on-sample splitting function!
\par
To summarise, I have shown that simulating the first arrival $(Y_1, T_1)$ of $\PoissonProcess$ and computing ${\splitFun(\YSpace) = L, R}$ with $Y_1$ as the ``pivot'', then $N \neq 1$ implies ${Y_N \in B \in \{L, R\}}$.
Therefore, if $B = L$, we can eliminate $R$ from the search and vice versa.
Determining whether $B = L$ or $B = R$ requires we decide whether the mode $m$ of $r$ is to the left or the right of the sample $Y_1$.
Formally, I will call a point $m \in \YSpace$ a \textit{mode-point} if for all $\alpha \leq \norm{r}_\infty$ it holds that $m \in S_\alpha$.
In this case, $B = L \Leftrightarrow m < Y_1$ and $B = R$ if $Y_1 < m$.
We do not need to know the value of $\norm{r}_\infty$, only the location $m$ to implement this decision rule.
\par
The above argument demonstrates that after simulating the first arrival, we only need to search through the points of $\PoissonProcess$ restricted to $B \times (T_1, \infty)$.
However, by \cref{thm:restriction_theorem}, the restricted process $\PoissonProcess\vert_{B \times (T_1, \infty) }$ is also a Poisson process. 
Repeating the above argument shows once again that either $(Y_N, T_N)$ is the first arrival of $\PoissonProcess\vert_{B \times (T_1, \infty)}$ or otherwise we must have $Y_N \in B' \in \splitFun(B)$.
Indeed, we can repeat the above argument for any number of steps $d$, which finally yields the branch-and-bound A* sampling, described in \cref{alg:a_star_sac}.
To describe the algorithm, first, let $\Branch$ be the branch of the BSP structure over $\PoissonProcess$ such that for every node $\nu \in \Branch$, the mode-point $m \in B_\nu$.
From this point onward, I thus use the single-branch indexing I introduced at the end of \cref{sec:simulating_bnb_poisson_processes} for $\Branch$.
Now, at each step $d \geq 0$, we always know that it is enough to check the first arrival $(Y^d, T^d)$ in some convex set (interval) $B^d$, starting with $B^0 = \YSpace$.
We also maintain a record of the best arrival so far $N^d = \argmin_{n \in [0:d]}\{T^n / r(Y^n) \mid (Y^n, T^n) \in \Branch\}$.
Then, unless the arrival satisfies the termination criterion $T^d / M > T_{N^d} / r(Y_{N^d})$, we compute $B^{d + 1} \in \splitFun(B^d)$ with $Y^d$ as the pivot and carry on the search in $\PoissonProcess\vert_{B^{d + 1} \times (T^d, \infty)}$.
Finally, by the above argument, we have that $(Y^{N^d}, T^{N^d})$ is the optimal arrival not just within the branch we simulated but in all of $\PoissonProcess$, hence $Y^{N^d} \sim Q$ and $T^{N^d} / r(Y^{N^d}) \sim \Exponential(1)$.
\par
\textbf{The runtime of branch-and-bound A* sampling.}
I now turn to analysing the runtime of the sampler I constructed above.
As we will see, we make the exponential savings we expected by simulating only a single branch of the BSP structure.
The proof I present below is new and improves on both the scaling and additive constants compared to the original result presented in \citet{flamich2022fast}. 
\begin{importantTheorem}[Runtime of branch-and-bound A* sampling]
\label{thm:bnb_a_star_runtime}
Let $Q \ll P$ be probability measures over an interval $\YSpace \subseteq \Reals$ with Radon-Nikodym derivative $r = dQ/dP$ and assume $r$ is bounded above by some constant $M \geq 1$.
Furthermore, assume that $r$ is quasiconcave.
Let $D$ denote the number of samples simulated by branch-and-bound A* sampling (\cref{alg:a_star_sac}) that do not violate its termination condition.
Then,
\begin{align*}
\Exp[D] \leq \frac{\lb M + 2}{\lb(e) - 1} < 2.26 \cdot \lb M + 4.52
\end{align*}
\end{importantTheorem}
\begin{proof}
Let $\PoissonProcess$ denote the Poisson process over $\YSpace \times \nonnegReals$ with mean measure $P \otimes \lambda$, and let $\Tree$ denote the BSP structure over it induced by the on-sample splitting function, whose branch $\Branch$ \cref{alg:a_star_sac} simulates.
Then, by the definition of \cref{alg:a_star_sac}, the number of arrivals in $\Branch$ that do not violate the termination condition is
\begin{align}
\label{eq:bnb_a_star_runtime_proof_stopping_time_def}
D = \min_{d \in \Nats}\left\{\frac{T^{d + 1}}{M} > \min_{\delta \in [0:d]}\left\{\frac{T^\delta}{r(Y^\delta)}\right\} \,\,\Big\vert\,\, (T^d, Y^d), (T^{\delta}, Y^{\delta}) \in \Branch \right\},
\end{align}
which only depends on $Y^{0:d}, T^{0:d+1}$. 
However, observe that the $T^{j}$s and $Y^{j}$ all depend on the previous elements in the sequence.
Thus, to simplify the analysis, I now switch to the reparametrise $Y^{0:\infty}, T^{0:\infty}$ to depend on only on the i.i.d.\ sequences $U^{0:\infty}, E^{0:\infty}$, where $U^j \sim \Unif(0, 1)$ and $E^j \sim \Exponential(1)$ and according to \cref{eq:branch_and_bound_loc_reparam,eq:branch_and_bound_time_reparam} I set
\begin{align*}
Y^j = F^{-1}(F(L^j) + (F(R^j) - F(L^j)) \cdot U^j)
\quad\text{and}\quad
T^j = \sum_{i = 0}^j \frac{E^i}{P(B^i)},
\end{align*}
where $F$ is the CDF of the proposal distribution $P$.
\par
Returning to examining $D$, since the event $\{D = d\}$ only depends on $T^{0:d + 1}, Y^{0:d}$, we can use the above reparameterisations to note that in fact the event $\{D = d\}$ depends only on $E^{0:d + 1}, U^{0:d}$.
Technically, this means that $D$ is a stopping time with respect to the filtration $\filtration$ defined by $\sigmaAlgebra_d = \sigma(E^{0:d + 1}, U^{0:d})$.
In particular, note that the event $\{D \leq d - 1\} \perp E^{d + 1:\infty}, U^{d:\infty}$, hence its complement $\{D \geq d\}$ is also independent of $E^{d + 1:\infty}, U^{d:\infty}$.
Thus, we now get
\begin{align}
\Exp\left[T^{D + 1}\right] 
&=\Exp\left[\sum_{d = 0}^\infty \Ind[D = d] T^{d + 1} \right] \nonumber\\
&=\Exp\left[T^1 + \sum_{d = 1}^\infty \Ind[D \geq d] (T^{d + 1} - T^{d}) \right] \tag{summation by parts} \\
&=\Exp\left[E^0 + \frac{E^1}{P(B^1)} + \sum_{d = 1}^\infty \Ind[D \geq d] \frac{E^{d+1}}{P(B^{d + 1})}\right] \tag{reparam.\ / \cref{eq:branch_and_bound_time_reparam}} \\
&= \Exp\left[1 + \frac{1}{P(B^1)}\right] + \Exp\left[\sum_{d = 1}^\infty \Ind[D \geq d] \frac{\Exp[E^{d+1}]}{P(B^{d+1})}\right] \tag{$\{D \geq d\} \perp E^{d + 1}$} \\
&= \Exp\left[\sum_{d = 0}^\infty \Ind[D = d] \sum_{j = 0}^{d + 1}\frac{1}{P(B^{j})} \right] \tag{$\Exp[E^{d + 1}]=1$ \& summation by parts}\\
&\geq \Exp\left[\frac{1}{P(B^D)} + \frac{1}{P(B^{D + 1})} \right] \nonumber. 
\end{align}
From this, we find that
\begin{align}
\label{eq:bnb_a_star_runtime_proof_time_identity}
\Exp\left[T^{D + 1} - \frac{1}{P(B^{D + 1})}\right] \geq \Exp\left[\frac{1}{P(B^D)}\right]
\end{align}
Now, it turns out that we can compute the left-hand side of \cref{eq:bnb_a_star_runtime_proof_time_identity}!
In particular, let $\tau = \min_{n \in \Nats} \{T^n / r(Y^n) \mid (Y^n, T^n) \in \Branch\}$ be the first arrival of the shifted process.
Then, by the definition of $D$ the variable $T^{D + 1}$ is the first arrival time of $\PoissonProcess$ restricted to $B^{D + 1} \times [\tau \cdot M, \infty)$!
Observe, that for any point $(t, y) \in B^{D + 1} \times [\tau \cdot M, \infty)$ we must have $t / r(y) > \tau \cdot M / r(y) > \tau$ by construction.
Thus, the points in $(t, y) \in [\tau \cdot M, \infty) \times B^{D + 1}$ are independent of $\tau$ since they are guaranteed not to be the shifted first arrival.
Hence, the mean measure of $\PoissonProcess\vert_{B^{D + 1} \times [\tau \cdot M, \infty)}$ is $P\vert_{B^{D + 1}} \otimes \lambda\vert_{[\tau \cdot M, \infty)}$.
Therefore, from \cref{eq:time_homogeneous_process_arrival_identity}, we have
\begin{align}
\label{eq:runtime_plus_one_arrival_time_identity}
(T^{D + 1} - \tau \cdot M \mid \tau) \sim \Exponential(P(B^{D + 1})),
\end{align}
from which we find
\begin{align*}
\Exp\left[T^{D + 1} - \frac{1}{P(B^{D + 1})} \mid \tau\right] = M \cdot \tau.
\end{align*}
Finally, noting that $\tau \sim \Exponential(1)$, we get that
\begin{align*}
\Exp\left[T^{D + 1} - \frac{1}{P(B^{D + 1})}\right] = M.
\end{align*}
Combining this with \cref{eq:bnb_a_star_runtime_proof_time_identity} and taking logs, we have so far that
\begin{align}
\label{eq:bnb_a_star_main_inequality}
\lb(M) \geq \lb\left(\Exp\left[\frac{1}{P(B^D)}\right]\right) \geq \Exp\left[\lb \frac{1}{P(B^D)}\right] = \Exp\left[\lb \frac{1}{P(B^{D+1})}\right] + \Exp\left[\lb \left(\frac{P(B^{D + 1})}{P(B^D)}\right)\right]
\end{align}
where the second inequality follows from Jensen's inequality.
Finally, I deal with the two terms on the right-hand side separately.
First, observe that
\begin{align}
\Exp[-\lb P(B^{D+1})] 
&= -\Exp\left[\sum_{d = 0}^\infty \Ind[D = d] \lb P(B^{d + 1})\right] \nonumber\\
&= -\Exp\left[\lb P(B^1) + \sum_{d = 1}^\infty \Ind[D \geq d] \lb \left(\frac{P(B^{d + 1})}{P(B^{d})}\right)\right] \label{eq:bnb_a_star_runtime_proof_minus_log_pb_identity}
\end{align}
Now, $\lb \left(\frac{P(B^{d + 1})}{P(B^d)}\right)$ depends on $U^{0:d}$.
However, it admits approximations that only depend on $U^d$, which I derive next.
First, expand $B^d = (L^d, R^d)$, and assume that the density ratio has mode-location $m$.
Then, note that by definition
\begin{align}
\frac{P(B^{d + 1})}{P(B^d)} 
&= \frac{P((L^d, Y^d))\Ind[m < Y^d] + P((Y^d, R^d))\Ind[Y^d < m]}{P(B^d)} \nonumber\\
&= \frac{(F(Y^d) - F(L^d))\Ind[m < Y^d] + (F(R^d) - F(Y^d))\Ind[Y^d < m]}{F(R^d) - F(L^d)} \nonumber\\
&= U^d\Ind[m < Y^d] + (1 - U^d)\Ind[Y^d < m]. \tag{\cref{eq:branch_and_bound_loc_reparam}}
\end{align}
From this, we obtain two inequalities:
\begin{align}
\frac{P(B^{d + 1})}{P(B^d)} &\leq \max\{U^d, 1 - U^d\} \label{eq:bnb_a_star_bound_ratio_upper_bound}\\
\frac{P(B^{d + 1})}{P(B^d)} &\geq \min\{U^d, 1 - U^d\} \label{eq:bnb_a_star_bound_ratio_lower_bound}
\end{align}
Plugging \cref{eq:bnb_a_star_bound_ratio_upper_bound} into \cref{eq:bnb_a_star_runtime_proof_minus_log_pb_identity}, we get
\begin{align}
\Exp[-\lb P(B^{D + 1})] 
&\geq -\Exp\left[\lb \left(\max\{U^{1}, 1 - U^{1}\}\right) + \sum_{d = 1}^\infty \Ind[D \geq d] \lb \left(\max\{U^{d}, 1 - U^{d}\}\right)\right] \nonumber\\
&= (\lb(e) - 1) -\sum_{d = 1}^\infty \Prob[D \geq d] \Exp\left[\lb \left(\max\{U^{d}, 1 - U^{d}\}\right)\right] \tag{$\{D \geq d\} \perp U^{d}$}\\
&\geq \left(\lb(e) - 1\right)\left(1 + \sum_{d = 1}^\infty \Prob[D \geq d]\right)  \nonumber \\
&= \left(\lb(e) - 1\right) (1 + \Exp[D]), \label{eq:bnb_a_star_first_term_lower_bound}
\end{align}
where the last equality holds by the Darth Vader rule.
For the second term, note that by \cref{eq:bnb_a_star_bound_ratio_lower_bound}
we have
\begin{align*}
\Exp\left[\lb \left(\frac{P(B^{D + 1})}{P(B^D)}\right)\right] \geq \Exp\left[\lb \left(\min\{U^D, 1 - U^D\}\right)\right].
\end{align*}
Now, define $M^d = \lb(\min\{U^d, 1 - U^d\})$ and note that it is a martingale with respect to $U^{0:d - 1}, E^{0:d}$ (technically, the filtration $\filtration$ given by $\sigmaAlgebra_{d} = \sigma(U^{0:d}, E^{0:d+1})$), since for all $d \geq 0$ we have
\begin{align*}
\Exp[M^d \mid U^{0:d - 1}, E^{0:d}] = \Exp[\lb(\min\{U^d, 1 - U^d\})] = -(\lb(e) + 1).
\end{align*}
Now, since $D$ is completely determined by $U^{0:d - 1}, E^{0:d}$ (i.e.\ it is a stopping time adapted to $\filtration$), by the optimal stopping theorem \citep[Theorem 12.5.1;][]{grimmett2020probability} we have
\begin{align}
\label{eq:bnb_a_star_proof_second_term_lower_bound}
\Exp\left[\lb \left(\frac{P(B^{D + 1})}{P(B^D)}\right)\right] \geq \Exp\left[M^D\right] = \Exp\left[M^0\right] = -(\lb(e) + 1).
\end{align}
Finally, putting \cref{eq:bnb_a_star_first_term_lower_bound,eq:bnb_a_star_proof_second_term_lower_bound} into \cref{eq:bnb_a_star_main_inequality}, we get
\begin{align*}
\lb(M) \geq \left(\lb(e) - 1\right)(1 + \Exp[D]) -(\lb(e) + 1).
\end{align*}
Rearranging the terms finishes the proof.
\end{proof}
\par
\textbf{Relative entropy coding with branch-and-bound A* sampling.}
For a pair of dependent random variables $\rvx, \rvy \sim P_{\rvx, \rvy}$, where $\rvy$ is $\Reals$-valued and $\frac{P_{\rvy \mid \rvx}}{P_\rvy}$ is quasiconcave $P_\rvx$-almost surely, we can use branch-and-bound A* sampling to construct a channel simulation protocol to encode a $P_{\rvy \mid \rvx}$-distributed sample using $P_\rvy$ as the proposal.
However, since \cref{alg:a_star_sac} no longer simulates the proposal Poisson process in time order, we can no longer rely on encoding the time index of the returned sample.
Instead, I propose to encode the ``search path'' of the node $\nu$ to which the sample belongs instead. 
This is equivalent to encoding the heap index $H_\nu$ (\cref{def:heap_index_and_depth}) of the accepted node $\nu$!
But how can we efficiently encode the accepted heap index $H_\nu$?
Observe that branch-and-bound A* sampling is structurally very similar to the first stage of arithmetic coding (\cref{sec:arithmetic_coding}, \cref{fig:ac_first_stage}): A* sampling ``establishes the bounds'' $B_\nu$ in which the accepted sample is located.
Therefore, we could try to use the second stage of arithmetic coding (\cref{fig:ac_second_stage}) to encode $B_\nu$?
From \cref{eq:arithmetic_coding_efficiency}, we already know that the second stage of (infinite-precision) AC can encode $B_\nu$ at the efficiency of at most $\Exp[-\lb P(B_\nu)] + 2$ bits.
Hence, if we can show that $\Exp[-\lb P(B_\nu)] \approx \MI{\rvx}{\rvy}$, then we are almost done!
I say almost, since we also need to encode the depth of the selected node $D_\nu$ so that the decoder knows when to stop.
\par
There is an alternative view of the above encoding approach: we first encode the depth $D_\nu$ of the accepted node. 
Then, we encode the path in the on-sample splitting BSP structure: at each depth $1 \leq d \leq D_\nu$, we encode whether we continued the search inside the left or the right bound.
Fortunately, there is a natural probability distribution we can use to encode our left/right decisions: at each depth $d$, we encode our choice with probability $P(B^{d + 1}) / P(B^d)$!
Using arithmetic coding \cref{sec:arithmetic_coding} to encode each choice results in a total of 
\begin{align*}
\approx -\sum_{d = 0}^{D_\nu - 1} \lb\left(\frac{P(B^{d + 1}}{P(B^d)}\right) = -\lb P(B^D) \text{ bits},
\end{align*}
same as above.
\par
It remains to bound $\Exp[-\lb P(B_\nu)]$ and the codelength of $D_\nu$.
I first tackle the negative log-probability term and show that it behaves well in expectation.
\begin{lemma}[Upper bound on the the self-information of bounds returned by branch-and-bound A* sampling]
\label{thm:bnb_a_star_neg_log_bound_size}
Let $Q \ll P$ be probability measures over an interval $\YSpace \subseteq \Reals$ with quasiconcave Radon-Nikodym derivative $r = dQ/dP$ and assume $r$ is bounded from above by some constant $M \geq 1$.
Let $B$ denote the bounds in which \cref{alg:a_star_sac} locates the accepted sample when applied to $Q$ and $P$.
Then:
\begin{align*}
\Exp[-\lb P(B)] \leq \KLD{Q}{P}.
\end{align*}    
\end{lemma}
\begin{proof}
Let $\PoissonProcess$ denote the Poisson process over $\YSpace \times \nonnegReals$ with mean measure $P \otimes \lambda$, and let $\Tree$ denote the BSP structure over it induced by the on-sample splitting function, whose branch $\Branch$ \cref{alg:a_star_sac} simulates.
Let $\delta$ denote the depth of the node in $\Branch$ associated with the sample the sampler accepts.
Recall that by construction, $( Y^\delta, T^\delta / r(Y^\delta))$ is the first arrival of the shifted process. 
Hence, we have
\begin{align*}
\tau = \frac{T^\delta}{r(Y^\delta)} \sim \Exponential(1), \quad Y^\delta \sim Q.
\end{align*}
Therefore, by taking logs and rearranging, we find
\begin{align*}
\Exp[\lb T^\delta] = \Exp[\lb r(Y^\delta)] + \Exp[\lb \tau] = \KLD{Q}{P} - \gamma \cdot \lb(e)
\end{align*}
On the other hand, using the reparameterisations \cref{eq:branch_and_bound_loc_reparam,eq:branch_and_bound_time_reparam}to replace $Y^{0:\infty}, T^{0:\infty}$ with the i.i.d.\ sequences $U^{0:\infty}, E^{0:\infty}$, respectively, we also have by definition that
\begin{align*}
\Exp[\lb(T^\delta)] = \Exp\left[\lb\left(\sum_{k = 0}^K \frac{E^\delta}{P(B^\delta)}\right)\right] \geq \Exp[\lb(E^\delta)] - \Exp[\lb P(B^\delta)].
\end{align*}
Putting these two identities together and rearranging, we have thus far that
\begin{align}
\label{eq:bnb_a_star_codelength_penultimate_inequality}
\Exp[-\lb P(B^\delta)] \leq \KLD{Q}{P} - \gamma \cdot \lb(e) - \Exp[\lb(E^\delta)].
\end{align}
Thus, it remains to compute the last term on the right-hand side to finish the proof.
The difficulty here is that, unlike the runtime, the selected index $\delta$ is not a stopping time.
However, there is a trick, inspired by the proof of \cref{thm:global_a_star_conditional_index_distribution}: while $\delta$ is not a stopping time adapted to the branch $\Branch$, \textit{it is} a stopping time conditioned on the accepted arrival location $Y^\delta = y$ adapted to the ``conditional branch'' $\Branch \mid Y^\delta = y$.
Intuitively, we can consider a computational way of constructing $\Branch \mid Y^\delta = y$ up to step $\delta$: at each step $d \geq 0$, given the available information $\sigmaAlgebra_d$ (given by the $\sigma$-algebra generated by $Y^{0:d - 1}, E^{0:d - 1}$ and the event $\{Y^\delta = y\}$) we flip a coin with probability $\Prob[\delta = d \mid \sigmaAlgebra_d]$.
If it comes up heads, we set $Y^d = y$ and $E^\delta$ to follow the appropriate conditional distribution.
On the other hand, if it comes up tails, we simulate $E^d, Y^d$ according to the distribution $\Prob[(Y^d, E^d) \in \cdot \mid \delta \neq d, \sigmaAlgebra_d]$.
Now, observe furthermore that if we set $\Delta^d = \lb (E^d) - \Exp[\lb (E^d) \mid \sigmaAlgebra_d]$, then $\Delta^d$s form a martingale with respect to the filtration $\filtration$ given by the $\sigmaAlgebra_d$s, since for any $d \geq 0$ it holds that
\begin{align*}
\Exp[\Delta^k \mid \sigmaAlgebra_k] = \Exp[\lb (E^k) - \Exp[\lb (E^k) \mid \sigmaAlgebra_k] \mid \sigmaAlgebra_k] = 0.    
\end{align*}
Hence, by the optional stopping theorem \citep[Theorem 12.5.1;][]{grimmett2020probability}, we have that
\begin{align*}
\Exp_{Y^K}[\Exp[-\lb(E^K) \mid Y^K]] = \Exp_{Y^K}[\Exp[-\lb(E^0) \mid Y^K]] = \Exp[-\lb(E^0)] = \gamma \cdot \lb(e).
\end{align*}
Putting this in \cref{eq:bnb_a_star_codelength_penultimate_inequality} finishes the proof.
\end{proof}
Finally, I show the following result using branch-and-bound A* sampling for relative entropy coding.
\begin{importantTheorem}[Relative entropy coding with branch-and-bound A* sampling.]
\label{thm:bnb_a_star_codelength}
Let $\rvx, \rvy \sim P_{\rvx, \rvy}$ be dependent random variables over the space $\XSpace \times \YSpace$, and assume that $\YSpace$ is a convex subset of $\Reals$.
Furthermore, assume that $P_\rvx$-almost surely $\frac{P_{\rvy \mid \rvx}}{dP_\rvy}$ is quasiconcave and upper bounded by some constant $M_\rvx \geq 1$.
Now, consider the following channel simulation protocol for the channel $\rvx \to \rvy$: the sender and receiver set $\rvz \gets (U^{0:\infty}, E^{0:\infty})$ as their common randomness, where $U^{0:\infty}$ and $E^{0:\infty}$ are sequences of i.i.d.\ $\Unif(0, 1)$ and $\Exponential(1)$ random variables, respectively, used to reparametrise the arrival locations and times of the branch simulated by \cref{alg:a_star_sac}, as given by \cref{eq:branch_and_bound_loc_reparam,eq:branch_and_bound_time_reparam}.
\par
Then, upon receiving a source symbol $\rvx \sim P_\rvx$, the sender uses branch-and-bound A* sampling (\cref{alg:a_star_sac}) to simulate $\rvy \sim P_{\rvy \mid \rvx}$.
Now, let $\nu$ denote the node accepted by the sampler, $H_\nu$ its heap index, $\delta_\nu$ its depth and $B_\nu$ its associated bounds.
The sender encodes $\nu$ by encoding its depth $\delta_\nu$ using a zeta distribution $\zeta(d \mid \alpha) \propto (d + 1)^{-\alpha}$ with exponent 
\begin{align*}
\alpha = 1 + 1 \Big/ \left(\frac{\lb M_\rvx + 2}{\lb(e) - 1} + 2\right),   
\end{align*}
and encodes the heap index $H_\nu$ using arithmetic coding using $P(B_\nu)$ as the probability model.
Finally, let $\mu = \Exp_{\rvx \sim P_{\rvx}}[\lb M_{\rvx}]$.
Then,
\begin{align}
\label{eq:bnb_a_star_rec_codelength}
\Ent{\nu \mid \rvz} \leq \MI{\rvx}{\rvy} + \lb \left(\mu + 1\right) + \Oh(\lb(\lb(\mu + 1))).
\end{align}
\end{importantTheorem}
\begin{proof}
The code has two parts: the one for $\delta_\nu$ and the one for $H_\nu$.
To show the result, consider $\rvx \sim P_\rvx$ fixed for the moment.
Then, starting the codelength of $\delta_\nu$, since it is a nonnegative integer, we encode $\delta_\nu + 1$ instead (as reflected in the theorem statement).
Now, I make a relatively coarse approximation that is nonetheless sufficient to get an \textit{almost} optimal codelength.
Let $D_\nu$ be the \textit{runtime} of \cref{alg:a_star_sac} as in \cref{thm:bnb_a_star_runtime}; then we have $\delta_\nu \leq D_\nu$ always.
By \cref{thm:bnb_a_star_runtime}, we have in this case that
\begin{align}
\label{eq:bnb_a_star_codelength_proof_depth_expectation}
\Exp[\delta_\nu\mid \rvx] \leq \Exp[D_\nu \mid \rvx] \leq \frac{\lb M_\rvx + 2}{\lb(e) - 1}\end{align}
As I have shown in \cref{lemma:li_el_gamal_bound_on_pos_int_random_variable}, using the zeta coding approach will yield an expected description length of
\begin{align}
\Exp[\lb (D_\nu + 1) \mid \rvx] + &\lb(\Exp[\lb (D_\nu + 1) \mid \rvx] + 1) + 1\nonumber\\
&\leq \lb \Exp[D_\nu + 1 \mid \rvx] + \lb(\lb\Exp[D_\nu + 1 \mid \rvx] + 1) + 1 \tag{Jensen}\\
&\leq \lb \left(\frac{\lb M_\rvx + 2}{\lb(e) - 1} + 1\right) + \lb(\lb\left(\frac{\lb M_\rvx + 2}{\lb(e) - 1} + 1\right) + 1) + 1. \tag{\cref{eq:bnb_a_star_codelength_proof_depth_expectation}}
%&= \lb \left(\lb M_\rvx + \lb(e) + 1\right) - \lb(\lb(e) - 1) + \lb(\lb \left(\lb M_\rvx + \lb(e) + 1\right) - \lb(\lb(e) - 1) + 1) + 1
\end{align}
Averaging over $\rvx$ and using Jensen's inequality, we get that the average codelength is bounded above by
\begin{align}
\lb \left(\frac{\mu + 2}{\lb(e) - 1} + 1\right) + &\lb(\lb\left(\frac{\mu + 2}{\lb(e) - 1} + 1\right) + 1) + 1 \nonumber\\
&= \lb \left(\frac{\mu + 2}{\lb(e) - 1} + 1\right) + \Oh(\lb(\lb(\mu + 1))). \nonumber\\
&= \lb \left(\mu + 1\right) + \Oh(\lb(\lb(\mu + 1))),\label{eq:bnb_a_star_codelength_proof_depth_codelength}
\end{align}
where for the last equality, I used logarithm rules as well as the inequality $\lb(x + \alpha + 1) \leq \lb(x + 1) + \lb(\alpha + 1)$ for $x, \alpha \geq 0$ to simplify the first log term and absorb the constants in the $\Oh(\lb(\lb(\mu + 1)))$ term.
Moving onto encoding $H_\nu$ given $D_\nu$ we can use arithmetic coding to encode $H_\nu$ using $-\lb P(B_\nu) + 2$ bits, by \cref{eq:arithmetic_coding_efficiency}.
Taking averages, by \cref{thm:bnb_a_star_neg_log_bound_size}, the average codelength produced by arithmetic coding will thus be
\begin{align}
\label{eq:bnb_a_star_codelength_proof_heap_index_codelength}
\MI{\rvx}{\rvy} + 2.
\end{align}
Putting \cref{eq:bnb_a_star_codelength_proof_depth_codelength,eq:bnb_a_star_codelength_proof_heap_index_codelength} yields the desired result.
\end{proof}
While technically not conforming to my definition of relative entropy coding in \cref{def:rec}, as I show in \cref{sec:numerical_experiments}, \cref{eq:bnb_a_star_rec_codelength} tends to be within a constant of the efficiency stated in \cref{def:rec} in practice, especially since we can usually compute the best possible upper bound $M_\rvx^* = \norm*{\frac{dP_{\rvy \mid \rvx}}{dP_\rvy}}_\infty$.
\subsection{Branch-and-bound Greedy Poisson Rejection Sampling}
\label{sec:bnb_gprs}
\begin{wrapfigure}[25]{r}[0pt]{0.45\textwidth}%
\centering
\includegraphics{5-BranchAndBound/img/bnb_gprs_illustration.tikz}
\caption[Illustration of branch-and-bound greedy Poisson rejection sampling for unimodal density ratio]{Illustration of branch-and-bound greedy Poisson rejection sampling (\cref{alg:gprs_sac}) for a Gaussian target $Q = \Normal(1, 0.25^2)$ and Gaussian proposal distribution $P = \Normal(0, 1)$, with the time axis truncated to the first $17$ units.
The algorithm searches for the first arrival of a spatio-temporal Poisson process $\PoissonProcess$ with mean measure $P \otimes \lambda$ under the graph of $\varphi = \sigma \circ r$ indicated by the \textbf{thick dashed black line} in each plot.
Here, $r = dQ/dP$ is the target-proposal density ratio, and $\sigma$ is given by \Cref{eq:stretch_function_integral_identity}.
The {\color{red}shaded red} areas are never searched or simulated by the algorithm since, given the first two rejections, we know points in those regions cannot fall under $\varphi$. 
}
\label{fig:bnb_gprs_illustration}
\end{wrapfigure}
In this section, I use a similar design principle to \cref{sec:bnb_a_star} to construct an efficient branch-and-bound variant of GPRS (\cref{sec:global_gprs}).
In particular, I assume a similar setting for sampling, namely that we have a target distribution $Q$ and a proposal distribution $P$ with $Q \ll P$ over an interval $\YSpace \subseteq \Reals$ with quasiconcave $r = dQ/dP$ and stretch function $\sigma$.
The significance of branch-and-bound GPRS is that we can use it to construct an optimally fast relative entropy coding algorithm.
Concretely, for dependent random variables $\rvx, \rvy \sim P_{\rvx, \rvy}$ where $\rvy$ is real-valued and $\frac{dP_{\rvy \mid \rvx}}{dP_\rvy}$ is $P_\rvx$-almost surely quasiconcave, the average-case time complexity of the relative entropy coding algorithm induced by the sampler is $\Oh(\MI{\rvx}{\rvy})$.
Thus, branch-and-bound GPRS provides a general solution for this broad class of relative entropy coding problems.
\par
\textbf{Deriving the algorithm.}
To begin, let $\PoissonProcess$ be a Poisson process with mean measure $P \otimes \lambda$ as usual, and let $\varphi = \sigma \circ r$.
Recall from \cref{sec:global_gprs} that GPRS searches for
\begin{align*}
(T^*, Y^*) &= (T_N, Y_N) \\
N &= \min_{n \in \Nats} \{\varphi(Y_n) \geq T_n \mid (T_n, Y_n) \in \PoissonProcess\}.
\end{align*}
The argument for speeding up the algorithm is almost identical to the one I gave for branch-and-bound A* sampling in \cref{sec:bnb_a_star}.
Thus, consider the first arrival $(Y_1, T_1)$ of $\PoissonProcess$; if $\varphi(Y_1) \geq T_1$ we immediately accept it and the algorithm terminates.
However, if we do not accept the first arrival, thanks to the quasiconcavity of $r$, we only need to search through a subregion of the whole sample space.
To see this, note first that since $\sigma$ is increasing, $r(y) \geq \alpha \Leftrightarrow \varphi(y) \geq \sigma(\alpha)$.
Hence, the $\alpha$-superlevel sets of $r$ coincide with the $\sigma(\alpha)$-superlevel sets of $\varphi$; this can also be seen from the more general fact that composition with an increasing function preserves quasiconcavity.
Therefore, if $N > 1$, the only way we can have $\varphi(Y_N) \geq T_N \Leftrightarrow r(Y_N) \geq \sigma^{-1}(T_N) = \alpha$ is if $Y_N$ falls in the $\alpha$-superlevel set $S_\alpha$ of $r$.
Furthermore, if $N > 1$ we also know that $\varphi(Y_1) < T_1 \Leftrightarrow r(Y_1) < \sigma^{-1}(T_1)$.
Since we also have $T_1 < T_N$, we have the following chain of inequalities: $\alpha = \sigma^{-1}(T_N) > \sigma^{-1}(T_1) \geq r(Y_1)$. 
Therefore, $S_\alpha$ does not contain $Y_1$, and by the quasiconcavity of $r$, we have that $S_\alpha$ is convex, hence an interval.
Thus, it must be the case that $S_\alpha \subseteq (-\infty, Y_1)$ or $S_\alpha \subseteq(Y_1, \infty)$.
This argument generalises to later steps as well, analogously to how the argument in \cref{sec:bnb_a_star} generalises; see \cref{fig:bnb_gprs_illustration} for a visual illustration of the above argument.
\par
This motivates the following branch-and-bound variant of GPRS, which I describe in \cref{alg:gprs_sac}.
However, before I proceed, let $\Branch$ be the branch of the BSP structure over $\PoissonProcess$ such that for every node $\nu \in \Branch$, the mode-point $m \in B_\nu$.
From this point onward, I thus use the single-branch indexing I introduced at the end of \cref{sec:simulating_bnb_poisson_processes} for $\Branch$.
Then, \cref{alg:gprs_sac} proceeds as follows, also illustrated in \cref{fig:bnb_gprs_illustration}:
At each step $d \geq 0$, by the above argument, we can always be sure that the first arrival of $\PoissonProcess$ under $\varphi$ will occur in some convex set $B^d$ (an interval), starting with $B^0 = \YSpace$.
Thus, we simulate the next arrival $(Y^d, T^d)$ of $\PoissonProcess$ in $B^d$ and if $\varphi(Y^d) < T^d$, then we compute $\splitFun(B^d) = \{L, R\}$ with $Y^d$ as the pivot.
In the same way as in \cref{alg:a_star_sac}, for a mode-point $m$ of the density ratio $r$, we set $B^{d + 1} = L \Leftrightarrow m < Y^d$ and set $B^{d + 1} = R$ otherwise.
We repeat this procedure until we find the first arrival such that it falls under the graph of $\varphi$.
\par
\textbf{The runtime of branch-and-bound GPRS.}
I now analyse the runtime of the sampling algorithm.
As I show, the additional computational requirement that we can evaluate the width function of the density ratio finally pays off: the runtime of branch-and-bound GPRS improves super-exponentially compared to its general-purpose variant.
However, I first show the following useful intermediate result.%
\begin{lemma}
\label{lemma:bnb_gprs_log_bound_size_kl_inequality}
Let $Q \ll P$ be probability measures over an interval $\YSpace \subseteq \Reals$ with quasiconcave Radon-Nikodym derivative $r = dQ/dP$.
Let $m$ be a mode-point of $r$, i.e.\ $m \in S_{\alpha}$ for all $\alpha \leq \norm{r}_\infty$ and assume it is unique.
Let $\PoissonProcess$ over $\YSpace \times \nonnegReals$ with mean measure $P \otimes \lambda$.
Let $\Tree$ be the BSP structure that the on-sample splitting function induces over $\PoissonProcess$.
Let $\Branch$ be the branch of $\Tree$ such that for every node $\nu \in \Branch$, we have $m \in B_\nu$.
Finally, let $D$ be the runtime of \cref{alg:global_gprs}.
Then, we have
\begin{align*}
\Exp[-\lb P(B^D)] \leq \KLD{Q}{P} + (1 + \gamma) \lb(e) < \KLD{Q}{P} + 2.26
\end{align*}
\end{lemma}
\begin{proof}
First, I use \cref{eq:branch_and_bound_loc_reparam,eq:branch_and_bound_time_reparam} to reparameterise the arrival locations and times in $\Branch$ to depend on the i.i.d.\ sequences $U^{0:\infty}$ and $E^{0:\infty}$, respectively, wher $U^j \sim \Unif(0, 1)$ and $E^j \sim \Exponential(1)$.
In particular, according to \cref{eq:branch_and_bound_time_reparam}, the $d$th arrival time in $\Branch$ can be written as
\begin{align*}
(T^d \mid Y^{0:d - 1}) \sim \sum_{j = 0}^d \frac{E^j}{P(B^j)}.
\end{align*}
Furthermore, by definition $D = d \Leftrightarrow r(Y^d) \geq \sha(T^d)$.
Then, since $\sha'(t) = \Prob[\tilde{T} \geq t] = \exp(-\tilde{\mu}(t))$ is decreasing, we have
\begin{align}
\lb r(Y^D) &\geq \lb \sha(T^D) \nonumber\\
&\geq \lb (T^D \cdot \sha'(T^D)) \tag{$\sha(t) \geq t \cdot \sha'(t)$}\\
&= \lb\left(\sum_{j = 0}^D \frac{E^j}{P(B^j)}\right) -\tilde{\mu}(T^D) \cdot \lb e \nonumber\\
&\geq \lb\left(\frac{E^D}{P(B^D)}\right) -\tilde{\mu}(T^D) \cdot \lb e \nonumber\\
&= \lb E^D - \lb P(B^D) -\tilde{\mu}(T^D) \cdot \lb e. \nonumber
\end{align}
Rearranging, we have $-\lb P(B^D) \leq \lb r(Y^D) - \lb E^D +\tilde{\mu}(T^D) \cdot \lb e$.
Hence, by the tower rule, we have
\begin{align*}
\Exp[-\lb P(B^D)] &= -\Exp_{D, Y^D, T^D, E^D} [\Exp[\lb P(B^D) \mid D, Y^D, T^D, E^D]] \\
&\leq \Exp_{D} [\Exp_{Y^D, T^D, E^D}[\lb r(Y^D) - \lb E^D +\tilde{\mu}(T^D) \cdot \lb e \mid D] \\
&= \Exp_{Y^D}[\lb r(Y^D)] + \Exp_{T^D}[\tilde{\mu}(T^D)] + \Exp_{E^D}[- \lb E^D] \\
&= \KLD{Q}{P} + (1 + \gamma) \lb(e),
\end{align*}
where $\gamma \approx 0.577$ is the Euler-Mascheroni constant.
For the last equality, I used the definition of the relative entropy for the first term and \cref{eq:gprs_expected_arrival_mean_time} for the second term.
To see that $\Exp_{E^D}[- \lb E^D] = \gamma \cdot \lb(e)$, define $G^d = -\ln E^d$, and note that $G^d$ is a martingale with respect to $E^{0:d-1}, U^{0:d-1}$ (technically, with respect to the filtration $\filtration$ generated by the $\sigma$-algebras $\sigmaAlgebra_d = \sigma(E^{0:d}, U^{0:d})$), since
\begin{align*}
\Exp[G^d \mid E^{0:d}, U^{0:d}] = \Exp[-\ln(E^d)] = \gamma.
\end{align*}
Furthermore, note that $D$ only depends on $E^{0:d-1}, U^{0:d-1}$ (it is adapted to the filtration $\filtration$ defined by the $\sigma$-algebras $\sigmaAlgebra_d$), hence by the optimal stopping theorem \citep[Theorem 12.5.1;][]{grimmett2020probability}, we have
\begin{align*}
\Exp[G^D] = \Exp[G^0] = \gamma.
\end{align*}
Multiplying by $1/\ln(2) = \lb(e)$ finishes the proof.
\end{proof}%
Indeed,\cref{lemma:bnb_gprs_log_bound_size_kl_inequality} will be the crucial ingredient in the analysis of both the runtime as well as the codelength proof of branch-and-bound GPRS.
I start by analysing the runtime first. 
The result I present and prove below is tighter than the one I originally derived in \citet{flamich2023gprs}, and the proof is also significantly simpler.%
\begin{importantTheorem}[Runtime of branch-and-bound GPRS]
\label{thm:bnb_gprs_runtime}
Let $Q \ll P$ be probability measures over an interval $\YSpace \subseteq \Reals$ with quasiconcave Radon-Nikodym derivative $r = dQ/dP$.
Let $D$ denote the number of samples simulated by branch-and-bound GPRS \cref{alg:gprs_sac}.
Then,
\begin{align*}
\Exp[D] \leq \frac{\KLD{Q}{P} + 2 + (1 + \gamma)\lb(e)}{\lb(e) - 1} < 2.26 \cdot \KLD{Q}{P} + 9.66.
\end{align*}    
\end{importantTheorem}
\begin{proof}
The proof proceeds similarly to the proof of \cref{thm:bnb_a_star_runtime} after I derive \cref{eq:bnb_a_star_main_inequality}.
Concretely, picking up from \cref{lemma:bnb_gprs_log_bound_size_kl_inequality}, we have
\begin{align}
\label{eq:bnb_gprs_main_inequality}
\KLD{Q}{P} + (1 + \gamma) \lb(e) \geq -\Exp[\lb P(B^D)] =
-\Exp[\lb P(B^{D + 1})] + \Exp\left[\lb \left(\frac{P(B^{D + 1})}{P(B^D)}\right)\right]
\end{align}
I will bound the two terms from below in an analogous way to the proof of \cref{thm:bnb_a_star_runtime}.
Again, I utilise the uniform-exponential reparameterisation of the arrival locations and times from \cref{eq:branch_and_bound_loc_reparam,eq:branch_and_bound_time_reparam}: write these as $U^{0:\infty}, E^{0:\infty}$ as before.
Now, for the first term on the left-hand side of \cref{eq:bnb_gprs_main_inequality}, we get
\begin{align}
-\Exp[\lb P(B^{D + 1})] 
&= -\Exp\left[\sum_{d = 0}^\infty \Ind[D = d] \lb P(B^{d + 1}) \right] \nonumber\\
&= -\Exp\left[\lb P(B^1) + \sum_{d = 1}^\infty \Ind[D \geq d] \lb \left(\frac{P(B^{d + 1})}{P(B^d)}\right) \right] \tag{summation by parts} \\
&\geq -\Exp\left[\lb \left(\max\{U^1, 1 - U^1\}\right)\right] \nonumber\\
&\quad\quad-\Exp\left[\sum_{d = 1}^\infty \Ind[D \geq d] \lb \left(\max\{U^d, 1 - U^d\}\right) \right] \tag{\cref{eq:bnb_a_star_bound_ratio_upper_bound}}\\
&\geq (\lb(e) - 1) - \sum_{d = 1}^\infty \Prob[D \geq d] \Exp\left[\lb \left(\max\{U^d, 1 - U^d\}\right) \right]
 \tag{$\{D\geq d\} \perp U^d$} \\
&= (\lb(e) - 1)(1 + \Exp[D]) \label{eq:bnb_gprs_first_term_lower_bound}.
\end{align}
For the second term, define $M^d = \min\{U^d, 1 - U^d\}$, and observe that it is a martingale with respect to $E^{0:d-1}, U^{0:d-1}$ (technically, with respect to the filtration $\filtration$ generated by the $\sigma$-algebras $\sigmaAlgebra_d = \sigma(E^{0:d}, U^{0:d})$):
\begin{align*}
\Exp[M^d \mid E^{0:d-1}, U^{0:d-1}] = \Exp[M^d] = -(\lb(e) + 1).
\end{align*}
Then, since $D \leq d$ only depends on $E^{0:d}, U^{0:d}$ (it is a stopping time adapted to $\filtration$), by the optimal stopping theorem \citep[Theorem 12.5.1;][]{grimmett2020probability}, we have%
\begin{align}
\label{eq:bnb_gprs_second_term_lower_bound}
\Exp\left[\lb \left(\frac{P(B^{D + 1})}{P(B^D)}\right)\right] \stackrel{\cref{eq:bnb_a_star_bound_ratio_lower_bound}}{\geq} \Exp[M^D] = \Exp[M^0] = -(\lb(e) + 1).
\end{align}
Putting \cref{eq:bnb_gprs_first_term_lower_bound,eq:bnb_gprs_second_term_lower_bound} into \cref{eq:bnb_gprs_main_inequality}, we get%
\begin{align*}
\KLD{Q}{P} + (1 + \gamma) \lb(e) \geq (\lb(e) - 1)(1 + \Exp[D]) - (\lb(e) + 1).
\end{align*}
Finally, rearranging the terms finishes the proof.
\end{proof}
\par
\textbf{Relative entropy coding with branch-and-bound GPRS.}
Relative entropy coding with branch-and-bound GPRS proceeds in much the same way as with its A* sampling counterpart in \cref{sec:bnb_a_star}: we encode the search path of the accepted node $\nu$ using a two-part code. 
Concretely, we encode its depth $D_\nu$ using a zeta distribution, followed by encoding its heap index $H_\nu$ using $P(B_\nu)$ as the probability model.
However, the difference compared to the result I derived for the average codelength of branch-and-bound A* sampling in \cref{thm:bnb_a_star_codelength}, this strategy will yield an optimal codelength.%
\begin{importantTheorem}[Relative entropy coding with branch-and-bound GPRS]
\label{thm:bnb_gprs_codelength}
Let $\rvx, \rvy \sim P_{\rvx, \rvy}$ be dependent random variables over the space $\XSpace \times \YSpace$, and assume that $\YSpace$ is a convex subset of $\Reals$.
Furthermore, assume that $P_\rvx$-almost surely $\frac{dP_{\rvy \mid \rvx}}{dP_{\rvy}}$ is quasiconcave.
Now, consider the following channel simulation protocol for the channel $\rvx \to \rvy$:
the sender and receiver set $\rvz \gets (U^{0:\infty}, E^{0:\infty})$ as their common randomness, where $U^{0:\infty}$ and $E^{0:\infty}$ are sequences of i.i.d.\ $\Unif(0, 1)$ and $\Exponential(1)$ random variables, respectively, used to reparametrise the arrival locations and times of the branch simulated by $\cref{alg:gprs_sac}$, as given by \cref{eq:branch_and_bound_loc_reparam,eq:branch_and_bound_time_reparam}.
\par
Then, upon receiving a source symbol $\rvx \sim P_\rvx$, the sender uses branch-and-bound GPRS (\cref{alg:gprs_sac}) to simulate $\rvy \sim P_{\rvy \mid \rvx}$.
Now, let $\nu$ denote the node accepted by the algorithm, $H_\nu$ its heap index, $D_\nu$ its depth and $B_\nu$ its associated bounds.
The sender encodes $\nu$ by encoding its depth $D_\nu$ using a zeta distribution $\zeta(d \mid \alpha) \propto (d + 1)^{-\alpha}$ with exponent%
\begin{align*}
\alpha = 1 + 1 \Big/ \left(\frac{\KLD{Q}{P} + 2 + (1+\gamma)\lb(e)}{\lb(e) - 1} + 2\right),   
\end{align*}
and encodes the heap index $H_\nu$ using arithmetic coding using $P(B_\nu)$ as the probability model.
Then,%
\begin{align}
\label{eq:bnb_gprs_codelength}
\Ent{\nu \mid \rvz} \leq \MI{\rvx}{\rvy} + \lb(\MI{\rvx}{\rvy} + 1) + \Oh(\lb(\lb(\MI{\rvx}{\rvy} + 1))).
\end{align}
\end{importantTheorem}
\begin{proof}
The proof proceeds almost identically as the proof of \cref{thm:bnb_a_star_codelength}: I bound the two parts of the code independently.
Thus, let $\rvx \sim P_\rvx$ be given for the moment.
Then, starting with the depth $D_\nu$, since it can be $0$, we encode $D_\nu + 1$ instead (as reflected in the theorem statement).
Now, by \cref{thm:bnb_gprs_runtime}, we have in this case that
\begin{align}
\label{eq:bnb_gprs_codelength_proof_depth_inequality}
\Exp[D_\nu + 1 \mid \rvx] \leq \frac{\KLD{P_{\rvy \mid \rvx}}{P_\rvy}}{\lb(e) - 1} + C,
\end{align}
where I set $C = 1 + (2 + (1 + \gamma)\lb(e))\, \big/\, (\lb(e) - 1)$
Hence, as I have shown in \cref{lemma:li_el_gamal_bound_on_pos_int_random_variable}, using the zeta distribution with the exponent given in the theorem statement will yield a scheme with expected description length bounded by
\begin{align}
\Exp[\lb(D_\nu + 1) &\mid \rvx] + \lb(\Exp[\lb(D_\nu + 1) \mid \rvx] + 1) + 1 \nonumber\\
&\leq \lb\Exp[D_\nu + 1\mid \rvx] + \lb(\lb \Exp[D_\nu + 1 \mid \rvx] + 1) + 1 \tag{Jensen}\\
&\leq\lb\left(\frac{\KLD{P_{\rvy \mid \rvx}}{P_\rvy}}{\lb(e) - 1} + C\right) + \lb(\lb \left(\frac{\KLD{P_{\rvy \mid \rvx}}{P_\rvy}}{\lb(e) - 1} + C\right) + 1) + 1 \nonumber 
\end{align}
Averaging over $\rvx$ and using Jensen's inequality, we get that the average codelength is bounded above by
\begin{align}
\lb\left(\frac{\MI{\rvx}{\rvy}}{\lb(e) - 1} + C\right) &+ \lb(\lb \left(\frac{\MI{\rvx}{\rvy}}{\lb(e) - 1} + C\right) + 1) + 1 \nonumber\\
&= \lb\left(\frac{\MI{\rvx}{\rvy}}{\lb(e) - 1} + C\right) + \Oh(\lb(\lb(\MI{\rvx}{\rvy})))\nonumber\\
&= \lb\left(\MI{\rvx}{\rvy} + 1\right) + \Oh(\lb(\lb(\MI{\rvx}{\rvy}))),\label{eq:bnb_gprs_codelength_proof_depth_code}
\end{align}
where for the last equality, I used logarithm rules as well as the inequality $\lb(x + \alpha + 1) \leq \lb(x + 1) + \lb(\alpha + 1)$ for $x, \alpha \geq 0$ to simplify the first log term and absorb the constants in the $\Oh(\lb(\lb(\mu + 1)))$ term.
Next, the sender encodes $H_\nu$ given $D_\nu$ using arithmetic coding, requiring a total of $-\lb P(B_\nu) + 2$ bits according to \cref{eq:arithmetic_coding_efficiency}.
Taking expectations and applying \cref{lemma:bnb_gprs_log_bound_size_kl_inequality}, we find that the average codelength will be upper bounded by
\begin{align}
\label{eq:bnb_gprs_codelength_proof_heap_index_code}
\Exp_{\rvx}[\KLD{P_{\rvy \mid \rvx}}{P_\rvy}] + 2 = \MI{\rvx}{\rvy} + 2.
\end{align}
Finally, putting \cref{eq:bnb_gprs_codelength_proof_depth_code,eq:bnb_gprs_codelength_proof_heap_index_code} yields the desired result.
\end{proof}
\subsection[Higher-Dimensional and More General Branch-and-bound Variants]{Higher-Dimensional and More General \texorpdfstring{\\}{} Branch-and-bound Variants}
\label{sec:general_branch_and_bound_variants}
\par
\textbf{Generalising the algorithms to $\Reals^n$ when the density ratio factorises.}
We can immediately ``generalise'' branch-and-bound A* sampling and GPRS to multivariate problems to encode a random vector $\rvy \in \Reals^n$ where the Radon-Nikodym derivative of the target distribution with respect to the coding distribution factorises
\begin{align*}
\frac{dQ_{\rvy}}{dP_\rvy} = \prod_{m = 1}^n \frac{dQ_{\rvy_m \mid \rvy_{1:m - 1}}}{dP_{\rvy_m \mid \rvy_{1:m - 1}}},
\end{align*}
and each of the $\frac{dQ_{\rvy_m \mid \rvy_{1:m - 1}}}{dP_{\rvy_m \mid \rvy_{1:m - 1}}}$ are almost surely quasiconcave.
In this case, we can apply the branch-and-bound samplers dimensionwise to get fast sampling algorithms.
\par
\textbf{Wait, does this solve multivariate relative entropy coding?}
Unfortunately, only partially.
This is because the channel simulation protocols I describe in \cref{thm:bnb_a_star_codelength,eq:bnb_gprs_codelength} have an $\Oh(1)$ cost per encoded dimension.
This means that if we encode the dimensions of $\rvy$ individually, the average codelength of the scheme has an additional $\Oh(n)$ scaling with dimensionality.
When $n$ is much smaller than the information content of the problem, this extra cost is negligible.
Unfortunately, in most practical machine learning-based compression problems where relative entropy coding has the most direct application and highest potential, the reverse situation is usually true: $\rvy$ will be high-dimensional, and most dimensions will contain very little information.
In \cref{sec:combiner_analysis_and_ablation}, I describe a concrete scenario where this occurs; see the rightmost panels in \cref{fig:combiner_visualisations} for an illustration.
\par
\textbf{Generalising the algorithms to $\Reals^n$ via supporting hyperplanes.}
We can also extend \cref{alg:a_star_sac,alg:gprs_sac} to probability distributions $Q \ll P$ over $\Reals^n$ with quasiconcave $r = dQ/dP$ by replacing on-sample splitting with splitting along \textit{supporting hyperplanes}.
A supporting hyperplane of a set $S \subseteq \Reals^{n}$ is an $(n - 1)$-dimensional hyperplane $H$, such that 1) $S$ is entirely contained in one of the closed half-spaces defined by $H$ and 2) $H$ contains at least one of the boundary points of $S$. 
Importantly, by the quasiconcavity of $r$, all its $\alpha$-superlevel sets $S_\alpha$ are convex, and by the \textit{supporting hyperplane theorem} \citep{boyd2004convex}, every boundary point of $S_\alpha$ has a supporting hyperplane. 
\par
Thus, we can generalise the splitting arguments from \cref{sec:bnb_a_star,sec:bnb_gprs} as follows. Assume we already know that it is sufficient to keep searching in some convex set $B \subseteq \Reals^n$ and we simulate the next arrival $(Y_B, T_B)$ of $\PoissonProcess$ restricted to $B$ that is consistent with our search procedure.
Then, if the search does not terminate, we only need to keep searching in the $\alpha$-superlevel set of $r$ for some appropriate $\alpha > r(Y_B)$.
Then, as a corollary of the supporting hyperplane theorem, we can always find a hyperplane $H(Y_B) \subseteq \Reals^n$ that contains $Y_B$ and is parallel to a supporting hyperplane of $S_\alpha$.
This means that, analogously to the one-dimensional algorithms, we can always discard the half-space defined by $H(Y_B)$ that does not contain $S_\alpha$, thereby speeding up the search.
As before, if $r$ is differentiable in $x$, we can easily find such a hyperplane since the gradient $\nabla_x r(x)$ is always normal to the boundary of $S_{r(x)}$!
\par
However, unfortunately, implementing these variants in practice is difficult because sampling from the set of possible solutions at step $d$, called the \textit{feasible region}, can quickly become difficult.
This is because the feasible region is equal to the convex set defined by intersections of the $d-1$ half-spaces the algorithm has searched through so far; equivalently, the feasible region is defined by $d - 1$ linear inequalities.
This reduces to the problem of sampling from a potentially unbounded convex polytope whose vertices we do not know, which in general is an NP-hard problem \citep{khachiyan2009generating}.
However, good polynomial-time approximate samplers exist, such as the Hit-and-run sampler \citep{belisle1998convergence}, which could be used to implement good approximate branch-and-bound samplers, but I leave this for future work.
\par
\textbf{Using different splitting functions.}
Thanks to \cref{corollary:simulate_pp_via_splitting_fun}, we may define far more general variants of both A* sampling and GPRS that work over an arbitrary Polish space $\YSpace$, so long as we can define a sensible splitting function $\splitFun$ over it.
However, these more general variants now lose the optimality guarantees of \cref{alg:a_star_sac,alg:gprs_sac}. They are also more complex to implement since any splitting mechanism other than on-sample/hyperplane splitting will break the guarantee that we only need to check a single branch of their induced BSP structure over the base process $\PoissonProcess$.
\par
In particular, the branch-and-bound variant A* sampling now needs to keep track of the entire search frontier.
Fortunately, as shown in \citet{maddison2014sampling,maddison2016poisson}, this can be efficiently achieved via a clever data structure called a \textit{priority queue}.
Furthermore, a generalised version of my splitting argument still holds for this variant of A* sampling: we cannot discard one of the half-spaces immediately at each step. 
It allows us to use a different upper bound $M_B$ for each search region $B$, improving its sample complexity compared to the fully general variant (\cref{alg:global_a_star}).
A particularly natural general branch-and-bound variant is the one using the dyadic splitting function (\cref{def:dyadic_splitting_function}), as the shrinkage of the bounds for the feasible region is worst-case optimal: their volume always decreases by a factor of $2$.
We can also use this variant for relative entropy coding by encoding the heap index of the returned sample as before.
With Stratis Markou, I benchmarked this variant on one-dimensional problems in our paper \citep{flamich2022fast} and showed that its performance closely matches the performance of \cref{alg:a_star_sac} on quasiconcave problems.
Furthermore, we also showed that its performance degrades gracefully as the number of modes increases: as a rule of thumb, the runtime increases by one step (compared to the optimal runtime on an ``equivalent'' unimodal problem) for each doubling in the number of modes; e.g.\ sampling from $Q$ where $r=dQ/dP$ has $16$ modes requires roughly four more steps than sampling from an ``equivalent'' $Q'$ with unimodal $r'=dQ'/dP$.
Here, by ``equivalence'' I mean that the width functions of the density ratios coincide: ${\Prob_{Z \sim P}[r(Z) \geq h] = \Prob_{Z \sim P}[r'(Z) \geq h]}$.
\par
While GPRS could be generalised similarly, I took a different approach in \citet{flamich2023gprs}.
Instead of maintaining the whole search frontier, I flip a biased coin at each split with a carefully chosen probability of landing heads to decide in which branch to continue the search.
Then, I realise the rest of the branch \textit{conditioned on the event that the accepted sample must come from it}.
In \citet{flamich2023gprs}, I show that when using this generalised variant of branch-and-bound GPRS using dyadic splitting (\cref{def:dyadic_splitting_function}), it performs comparably and sometimes even better than \cref{alg:gprs_sac}.
Unfortunately, the theoretical analysis of these generalised algorithms is significantly harder than their on-sample splitting counterparts, and I leave it for future research.
\section{Numerical Experiments}
\label{sec:numerical_experiments}
In this section, I investigate the tightness of the results I derived so far on a few synthetic examples.
I implemented all algorithms in the \texttt{Python} programming language using the \texttt{jax} numerical computing library \citep{jax2018github} and using the computational techniques I described in \cref{sec:implementation_considerations}.
I now turn to describe the experiments, whose results I showcase in \cref{fig:sampler_numerical_experiments}.
\par
\textbf{Gaussian relative entropy coding.}
In this first experiment, I experimentally verify the results I derived for the runtimes of the general and branch-and-bound variants of A* sampling and GPRS in \cref{thm:global_a_star_runtime,thm:global_gprs_runtime,thm:bnb_a_star_runtime,thm:bnb_gprs_runtime} and the average description lengths of the relative entropy coding protocols derived from them in \cref{thm:global_a_star_codelength,thm:global_gprs_codelength,thm:bnb_a_star_codelength,thm:bnb_gprs_codelength}.
Namely, in \cref{fig:sampler_numerical_experiments} (A) and (B), I consider simulating a one-dimensional additive white Gaussian noise (AWGN) channel assuming a Gaussian source distribution:
\begin{align}
\rvx &\sim \Normal(0, \sigma^2) \tag{source distribution}\\
(\rvy \mid \rvx) &\sim \Normal(\rvx, \rho^2). \tag{AWGN channel}
\end{align}
From this setup, it follows that marginally $\rvy \sim \Normal(0, \sigma^2 + \rho^2)$ and hence the density ratio is given by
\begin{align}
\frac{dP_{\rvy \mid \rvx}}{dP_\rvy}(y) = r_\rvx(y) &= \sqrt{\frac{\sigma^2 + \rho^2}{\rho^2}} \cdot \exp\left(\frac{\rvx^2}{2 \sigma^2} - \frac{(y - \nu)^2}{2 \kappa^2} \right) \label{eq:gauss_gauss_density_ratio}\\
\nu &= \rvx \cdot \frac{\sigma^2 + \rho^2}{\sigma^2} \nonumber\\
\kappa^2 &= \rho^2 \cdot \frac{\sigma^2 + \rho^2}{\sigma^2}\nonumber
\end{align}
\begin{figure}[H]
\ref*{legend:sampler_numerical_experiments}%
\\
\includegraphics{5-BranchAndBound/img/experiments}
\caption[Numerical comparison of the algorithms described in the thesis]{Numerical comparison of the algorithms I described in \cref{sec:sampling_as_search,sec:branch_and_bound_samplers}.
In each plot, \textit{dashed lines}  indicate the mean, \textit{solid lines} the median and the \textit{shaded areas} the 25 - 75 percentile region of the relevant performance metric.
I computed the statistics over 1000 runs for each setting.
\textbf{(A)} Runtime comparison on a 1D Gaussian channel simulation problem $P_{\rvx, \mu}$, plotted against increasing mutual information $I[\rvx; \mu]$.
The sender receives $\mu \sim \Normal(0, \sigma^2)$ and encodes a sample $\rvx \mid \mu \sim \Normal(\mu, 1)$ to Bob. 
\textbf{(B)}
Average codelength comparison of the algorithms on the same relative entropy coding problem as above.
\textbf{(C)}
One-shot runtime comparison of the branch-and-bound variants of A* coding and GPRS.
The sender encodes samples from a target $Q = \Normal(m, s^2)$ using $P = \Normal(0, 1)$ as the proposal.
I computed $m$ and $s^2$ such that $\KLD{Q}{P} = 2$ bits for each problem, but $\infD{Q}{P}$ increases.
GPRS' runtime stays fixed as it scales with $\KLD{Q}{P}$, while the runtime of A* keeps increasing.
}
\label{fig:sampler_numerical_experiments}
\end{figure}
For simplicity, I set the channel noise variance as $\rho^2 = 1$ in my experiments.
Then, I swept over a range of values for the source variance $\sigma^2$, such that the mutual information $\MI{\rvx}{\rvy}$ ranges between $0.1$ and $12$ bits; that there is a one-to-one relationship between $\sigma^2$ and $\MI{\rvx}{\rvy}$ given $\rho^2$ in this case can be seen from \cref{eq:gauss_gauss_mi}, which I derive later in this section.
For each setting of $\MI{\rvx}{\rvy}$/$\sigma^2$, I sampled $1000$ different locations $\rvx_i \sim \Normal(0, \sigma^2)$ and ran \cref{alg:global_a_star,alg:global_gprs,alg:a_star_sac,alg:gprs_sac} to encode a sample from each $P_{\rvy \mid \rvx_i}$ using $P_\rvy$ as the coding distribution.
I plotted the results in \cref{fig:sampler_numerical_experiments} (A) and (B), which show the average runtime and codelength of the algorithms, respectively.
\par
The first interesting observation that we can make in \cref{fig:sampler_numerical_experiments} (A) is that the runtime of the general variants of A* sampling and GPRS explodes very quickly: already when $\MI{\rvx}{\rvy}$ reaches 6 bits, the algorithms took over an hour to run on a high-end consumer-grade computer!
Interestingly, we can work out analytically why this happens, and it turns out that the situation is even more pathological than the computational results initially suggest. As I show next, we might need to wait for \textit{arbitrarily long} before the samplers terminate!
To begin, by inspection we see from \cref{eq:gauss_gauss_density_ratio} that
\begin{align*}
\norm{r_\rvx}_\infty &= \sqrt{\frac{\sigma^2 + \rho^2}{\rho^2}} \cdot \exp\left(\frac{\rvx^2}{2 \sigma^2}\right).
\end{align*}
Now let $K$ denote the runtime of an arbitrary selection sampler (\cref{def:exact_selection_sampler}).
Using \cref{thm:selection_sampler_runtime_lower_bound}, we can lower bound the expectation of $K$ as follows: 
\begin{align}
\Exp[K] &= \Exp_{\rvx \sim P_{\rvx}}[\Exp[K \mid \rvx]] \nonumber\\
&\geq \Exp_{\rvx \sim P_{\rvx}}[\norm{r_\rvx}_\infty] \tag{\cref{thm:selection_sampler_runtime_lower_bound}} \\
&= \Exp_{\rvx \sim P_{\rvx}}\left[\sqrt{\frac{\sigma^2 + \rho^2}{\rho^2}} \cdot \exp\left(\frac{\rvx^2}{2 \sigma^2}\right)\right] \nonumber\\
&= \infty. \nonumber
\end{align}
This shows that the runtime of any general selection sampler is extremely heavy-tailed: even if the sampler is guaranteed to terminate with probability $1$, we might wait for an arbitrarily long time before the algorithm terminates: this is an instance of the waiting-time paradox \citep[Chapter 4;][]{kingman1992poisson}.
Essentially, the issue is that the tail of $\norm{r_\rvx}_\infty$ grows too quickly, which means that rare source symbols far out in the tails of $P_\rvx$ significantly influence the average runtime.
\par
Perhaps the simplest solution is to cap the number of samples the sampling algorithm can examine, as I have described in \cref{sec:approximate_sampling}.
However, limiting the number of steps means the sampler's output will not follow the desired target distribution $P_{\rvy \mid \rvx}$ exactly.
Thus, in \citet{flamich2023adaptive}, I describe an alternative solution if we desire exact samples and the sample space is $\Reals^n$: 
instead of using $P_{\rvy}$ as the coding distribution, we convolve it with a Gaussian to obtain $Q_\rvy = P_\rvy * \Normal(0, \Delta^2)$ and use this \textit{overdispersed} distribution $Q_\rvy$ as the coding distribution instead.
For $\Delta^2 > 0$, this choice introduces a small codelength overhead but guarantees that the expected runtime of the sampler is finite.
While any $\Delta^2 > 0$ works, as I show in \citet{flamich2023adaptive}, the $\Delta^2$ that achieves the optimal runtime-codelength overhead trade-off is
\begin{align*}
\Delta_{opt}^2 = \sigma \sqrt{\rho^2 + \sigma^2}.
\end{align*}
The interested reader is referred to the paper for details.
\par
Returning to the analysis of the original problem, the second curious observation we might make in \cref{fig:sampler_numerical_experiments} (A) is despite the supposedly much worse $\Oh(\infD{Q}{P})$ scaling for the one-shot runtime of branch-and-bound A* sampling (\cref{thm:bnb_a_star_runtime}), it does not do much worse on average than branch-and-bound GPRS, whose one-shot runtime scales as $\Oh(\KLD{Q}{P})$ (\cref{thm:bnb_gprs_runtime}).
Does this imply that the bound on the codelength of A* sampling could also be significantly improved?
Investigating the matter analytically reveals that the small performance gap is thanks to the fact that the bounds in this particular case are not too far apart.
Concretely, let me compute all other relevant quantities:
\begin{align}
\KLD{P_{\rvy \mid \rvx}}{P_\rvy} &= \frac{\lb(e)}{2}\cdot \left(\frac{\rho^2 + \rvx^2}{\rho^2 + \sigma^2} + \ln\left(\frac{\rho^2 + \sigma^2}{\rho^2}\right) - 1\right) \nonumber\\
\MI{\rvx}{\rvy} &= \frac{1}{2}\cdot \lb\left(\frac{\rho^2 + \sigma^2}{\rho^2}\right) \label{eq:gauss_gauss_mi}\\
\infD{P_{\rvy \mid \rvx}}{P_\rvy} = \lb \norm{r}_\infty &= \frac{\lb(e)}{2}\left(\frac{\rvx^2}{\sigma^2} + \ln\left(\frac{\rho^2 + \sigma^2}{\rho^2}\right) \right) \nonumber\\
\Exp_{\rvx \sim P_\rvx}[\infD{P_{\rvy \mid \rvx}}{P_\rvy}] &= \frac{1}{2}\cdot \left(\lb\left(\frac{\rho^2 + \sigma^2}{\rho^2}\right) +\lb(e) \right). \label{eq:gauss_gaus_exp_inf_div}
\end{align}
Here we see another counter-intuitive feature of the problem: $\Exp_{\rvx \sim P_\rvx}[\infD{P_{\rvy \mid \rvx}}{P_\rvy}] = \MI{\rvx}{\rvy} + \frac{1}{2}\lb(e)$, i.e.\ the average-case runtimes are only a constant apart, which is precisely what we observe \cref{fig:sampler_numerical_experiments}!
This not only demonstrates that the $\Oh(\infD{Q}{P})$ runtime scaling of branch-and-bound A* sampling does not necessarily translate to an undesirable average-case bound but also that the overhead in the description length as given in \cref{thm:bnb_a_star_codelength} is usually negligible in practice.
\par
On the codelength side, we see that the upper bounds on the average description lengths of the general variants of A* sampling and GPRS (\cref{thm:global_a_star_codelength,thm:global_gprs_codelength}) are essentially tight: they are within one bit of the actual performance!
On the other hand, we also see that the results on the codelengths of the branch-and-bound samplers (\cref{thm:bnb_a_star_codelength,thm:bnb_gprs_codelength}) could still be tightened somewhat.
\par
Finally, note that in \cref{fig:sampler_numerical_experiments}, I do not just report the mean performances but also the median and the inter-quartile ranges.
Unlike the mean, these latter statistics are robust to outliers and, as such, provide a more complete picture of the ``usual'' variability of the algorithm's performance.
There are two important takeaways from these statistics: first, we see that the codelength concentrates heavily around its mean: the median codelength is very close to the mean, and the interquartile range is within just two bits of the mean.
Second, we can observe a similar phenomenon for the runtimes of the branch-and-bound samplers: the median is essentially equal to the mean, and the interquartile range is quite tight.
This demonstrates that while the general selection samplers can have pathological behaviour, we can significantly improve performance when additional structure is available.
\par
\textbf{Branch-and-bound sampler runtime comparison.}
As I have shown above, the average-case performance of branch-and-bound A* sampling and GPRS might not differ very much.
However, an important question remains: is the $\Oh(\infD{Q}{P})$ upper bound on the runtime of A* sampling in \cref{thm:bnb_a_star_runtime} tight?
To investigate this, I consider the following set of one-shot problems. 
Fix $P = \Normal(0, 1)$ as the coding distribution and a desired relative entropy $\kappa$; in my experiments, I set $\kappa = 2$ bits.
Then, for a few, regularly spaced $\delta \in [2, 25]$, I constructed a sequence of Gaussian target distributions $Q_\delta = \Normal(\mu_\delta, \sigma^2_\delta)$ such that $\KLD{Q_\delta}{P} = \kappa$ and $\infD{Q_\delta}{P} = \delta$.
As I show in Appendix H of \citet{flamich2023gprs}, this can be achieved by setting
\begin{align*}
b &= 2\ln(2) \cdot \delta - 1 \\
\sigma^2_\delta &= \exp\big(W\big((2\ln(2) \cdot \kappa + b) \cdot \exp(b)\big) - b\big) \\
\mu_\delta &= \sqrt{2(1 - \sigma^2_\delta)\left(\ln(2) \cdot \delta + \frac{1}{2}\ln \sigma^2_\delta\right)},
\end{align*}
where $W$ is the principal branch of the Lambert $W$-function \citep{corless1996lambert}, defined by the relation $W(x) \cdot \exp(W(x)) = x$.
For each setting of $\delta$, I used branch-and-bound A* sampling and GPRS to encode 1,000 samples from $Q_\delta$ using $P$.
\par
In \cref{fig:sampler_numerical_experiments} (C), I report the empirical mean, median and inter-quartile range of these runs.
As we can see, while the runtime of GPRS stays consistently close to $\kappa$ for each chosen infinity divergence $\delta$, the runtime of A* sampling increases essentially linearly with $\delta$.
This also demonstrates an interesting property of branch-and-bound GPRS: unlike all other sampling algorithms I discussed in this thesis, it does not require that the density ratio $r$ be bounded, instead relying on the weakest possible requirement that $\KLD{Q}{P} < \infty$!
Finally, similarly to the results of the first experiment, we see that the runtimes of these algorithms are relatively robust in that we can see that the median runtime is very close to the mean with a tight interquartile range.
\section{Conclusion and Open Questions}
\label{sec:bnb_conclusions}
\par
In this chapter, I constructed fast variants of sampling algorithms I developed in \cref{chapter:rec_with_pp}.
In particular, I showed how branch-and-bound techniques can be adapted to develop fast sampling algorithms, building on the work of \citet{maddison2014sampling}.
Furthermore, I analysed these algorithms and have shown in \cref{thm:bnb_a_star_runtime,thm:bnb_gprs_runtime} that when the target-proposal density ratio is unimodal/quasiconcave, the runtime of the algorithms improves exponentially compared to their general-purpose counterparts, while their relative entropy coding efficiency does not change very much (\cref{thm:bnb_a_star_codelength,thm:bnb_gprs_codelength}).
I presented new, significantly simplified proofs compared to the original ones in \citet{flamich2022fast,flamich2023gprs}, yielding much tighter results than the originals.
Finally, I verified through numerical experiments on synthetic problems that the theoretical results align closely with empirical performance.
I would like to especially thank Lennie Wells, Fanny Seizilles and Peter Wildemann for checking my proof techniques for the runtime and codelength results in \cref{thm:bnb_a_star_runtime,thm:bnb_a_star_neg_log_bound_size,thm:bnb_gprs_runtime,lemma:bnb_gprs_log_bound_size_kl_inequality}. 
\par
There are several exciting directions for future work.
Perhaps the most straightforward is to more precisely characterise the runtime of the algorithms I presented in this chapter.
For example, an interesting question is whether we could tightly characterise the variance of the runtime or use techniques similar to the one I used in \citet[Theorem 3.2;][]{flamich2023gprs} to characterise the $\alpha$-moments of the runtime of general-purpose GPRS for $\alpha \in (0, 1)$.
\par
Another interesting question is whether these branch-and-bound variants could be successfully and efficiently implemented for higher-dimensional problems.
\par
Furthermore, it is also an interesting question to find out if structures other than quasiconcavity admit efficient algorithms.
While there has been some effort in both the sampling \citet{chewi2022rejection} and channel simulation communities \citet{hegazy2022randomized,sriramu2024optimal}, these only work for very specific types of distribution and lack the generality of quasiconcavity.
\par
Finally, the problem of fast, multivariate Gaussian relative entropy coding remains the holy grail of practical relative entropy coding.

%% file: 6-COMBINER/combiner.tex
%!TEX root = ../thesis.tex
%*******************************************************************************
%****************************** Third Chapter **********************************
%*******************************************************************************
\chapter{Compression with Bayesian Implicit Neural Representations}
\label{chapter:combiner}
In this final technical chapter of the thesis, I present an application of relative entropy coding for practical data compression with machine learning.
In fact, most of the theory work I presented in the past three chapters stemmed from practical model- and data compression problems, mainly based on the work of Marton Havasi, who first introduced the idea of relative entropy coding to machine learning \citep{havasi2019minimal}.
Marton supervised my MPhil thesis at Cambridge, and I am forever thankful for his guidance and friendship.
Thus, I dedicate this chapter to him.
\par
Machine learning is revolutionising data compression and information theory: not only does it offer improvements to existing ways to compress data, but it has also opened up entirely new ways of doing so.
The method I present in this chapter, compression with Bayesian implicit neural representations (COMBINER), is one of them.
The high-level idea is to \textit{treat data as a function} and fit a flexible parametric model, such as an artificial neural network, to this function.
Then, encoding the data becomes equivalent to encoding the model parameters. Thus, model compression becomes equivalent to data compression.
The great promise of COMBINER is its energy efficiency: the state-of-the-art ML-based compression methods are extremely compute-intensive and power-hungry.
On the other hand, COMBINER is only a few orders of magnitude more expensive than classical methods, thus representing one of the first plausible candidates to serve as the basis for widespread ML-based compression algorithms.
Now, before presenting the method itself, I first describe nonlinear transform coding, the theoretical framework underpinning it.
\section{Nonlinear Transform Coding}
\label{sec:nonlinear_transform_coding}
\par
As I already discussed in \cref{sec:source_coding,sec:rec_through_examples}, lossy source coding is interested in encoding a source sample $\rvx \sim P_\rvx$ in as few bits as possible while also ensuring that its reconstruction $\hat{\rvx} \sim P_{\hat{\rvx} \mid \rvx}$ is not too bad on average.
Concretely, we wish to ensure for some distortion function $d$ that the average distortion $\Exp_{\rvx, \hat{\rvx}}[d(\rvx, \hat{\rvx})]$ is not too much.
Thus, there are two competing objectives: we simultaneously want to use as few bits for the encoding as possible while also keeping the distortion low.
In \cref{chapter:rec_with_pp,chapter:branch_and_bound}, we saw that for a given source distribution $P_\rvx$ and reconstruction distribution $P_{\hat{\rvx} \mid \rvx}$, we can use relative entropy coding to encode a single reconstruction in approximately $\MI{\rvx}{\hat{\rvx}}$ bits (ignoring the logarithmic overhead).
Hence, we can express the best possible rate $R$ we can achieve for a given distortion $D$ as an optimisation problem over the set of all probability distributions $\DistFamily$ over the range of $\hat{\rvx}$:
\begin{align}
\label{eq:rate_distortion_function}
R(D) = \inf_{P_{\hat{\rvx} \mid \rvx} \in \DistFamily} \MI{\rvx}{\hat{\rvx}} \quad \text{subject to } \Exp_{\rvx, \hat{\rvx}}[d(\rvx, \hat{\rvx})] \leq D.
\end{align}%
The function $R(D)$ is called the \textit{rate-distortion function}.
However, achieving the efficiency of $R(D)$ faces two key practical challenges: 
1) ``training'' is difficult: it is impossible in practice to minimise the objective in \cref{eq:rate_distortion_function} over all possible conditional distributions, and 2) compression is difficult: to perform relative entropy coding, we need to be able to sample from the marginal distribution $P_{\hat{\rvx}}$, which is impossible to compute in any practical scenario.
The issue is that in practice, $\rvx$ is usually high-dimensional: it could be a high-resolution image, and different dimensions of $\rvx$ usually have nontrivial dependencies between each other.
For example, in a portrait depicting someone's face, the pixels corresponding to the left eye correlate with those corresponding to the right eye.
\par
\textit{Transform coding} deals with these issues: we model $\hat{\rvx} = g_\theta(\rvy)$ as a parametric function $g_\theta$ of a new variable $\rvy$. 
Furthermore, we restrict $\DistFamily$ to ``simple'' distributions: for example, we could assume that $P_{\rvy \mid \rvx}$ factorises dimensionwise or that it is Gaussian.
If $\DistFamily$ is not too restricted and $g_\theta$ is expressive enough (e.g.\ it is a large neural network), this is a very mild approximation; though now we also need to optimise over $\theta$.
However, this does not help with the second issue: even when $\DistFamily$ consists of simple distributions, $P_{\rvy}$ is usually still impossible to compute since it would require us to average the conditionals over the unknown data distribution: $P_{\rvy}(A) = \int P_{\rvy \mid \rvx}(A \mid x) \, dP_\rvx(x)$.
The learned compression community has adopted a pragmatic approach to dealing with this: instead of worrying too much about what $P_{\rvy}$ might be, the usual approach is to pick a coding distribution $Q_\rvy$ that is convenient for compression and try to mitigate the issues with the approximation quality.
\par
Most classical transform coding methods choose the transform $g_\theta$ to be linear.
Linear transforms are undoubtedly powerful, are often interpretable, and can greatly simplify the problem.
In particular, while we can tune some of their parameters, these transforms are usually selected to correspond to some inductive bias for the problem.
Perhaps the most well-known is the JPEG image compression format, which uses the discrete cosine transform (DCT) for $g_\theta$.
The motivation for using the DCT is to encode how structures of different scales change in the image instead of encoding the colour intensity for each pixel.
For example, in a picture of a clear sky, most pixels will be blue, and the shade of the blue will only change slowly throughout the picture; thus, we can compress it a lot by focusing primarily on encoding the slowly varying features of the image.
Importantly, the DCT is both interpretable and does not have many parameters, which makes manual tuning feasible.
\par
However, when combined with a simple latent distribution $Q_\rvy$, such as a Gaussian, the reconstruction distribution $P_{\hat{\rvx}} = g_\theta \pushfwd Q_\rvy$ is quite unlikely to match the true data distribution $P_\rvx$ closely.
Thus, a natural idea is to replace the linear transforms using some expressive modern machine learning model, such as an artificial neural network: these methods are collectively known as \textit{nonlinear transform coding} \citep{balle2020nonlinear}.
However, unlike in the classical case, these nonlinear transforms usually have many parameters, and manual tuning is no longer feasible.
Thus, we use the standard machine learning approach: we optimise the model parameters using gradient descent using an appropriate objective function.
The immediate candidate for the objective function is the rate-distortion function in \cref{eq:rate_distortion_function}.
However, it is not suitable for two reasons: 1) it has a hard constraint that is difficult to enforce by construction, and 2) the hard constraint is difficult to interpret when optimised directly.
Concretely, what does the hard constraint $\Exp_{\rvx, \hat{\rvx}}[d(\rvx, \hat{\rvx})] \leq D$ mean, and how should we set the constraint $D$ for a given distortion $d$?
Perhaps even more importantly, distortion is not necessarily what we ultimately care about anyway.
Rather, we want the reconstruction to, for example, look or sound good depending on the data.
Hence, despite the rate-distortion function being widespread in the information theory literature, I will adopt the ML approach here, and optimise its inverse, the \textit{distortion-rate} function:
\begin{align}
\label{eq:distortion_rate_function}
D(R) = \inf_{P_{\hat{\rvx} \mid \rvx} \in \DistFamily}\Exp_{\rvx, \hat{\rvx}}[d(\rvx, \hat{\rvx})] \quad \text{subject to } \MI{\rvx}{\hat{\rvx}} \leq R,
\end{align}
which describes the best possible expected distortion achievable at a given bitrate $R$.
Usually, there are fairly easy-to-interpret notions of rate for each data modality (e.g.\ bits per pixel for images, (kilo) bits per second for audio) that allow us to set the hard constraint $\MI{\rvx}{\hat{\rvx}} \leq R$ much easier.
\par
With the interpretability issue out of the way, I rely on a fairly standard approach to overcome the hard constraint. 
Concretely, I propose to optimise the Lagrange dual of \cref{eq:distortion_rate_function} with a slack variable $\beta > 0$, specialised for the transform coding setting I describe above:
\begin{align}
\label{eq:distortion_rate_lagrange_dual}
\Loss(P_{\rvy \mid \rvx}, \theta, \beta) = \Exp_{\rvx, \rvy}[d(\rvx, g_\theta(\rvy))] + \beta \cdot \MI{\rvx}{g_\theta(\rvy)} + \mathrm{constant}.
\end{align}
Naturally, to use gradient descent, $P_{\rvy \mid \rvx}$ needs to be parametric as well; popular choices in practice are to use Gaussians with $\rvx$-dependent means and variances or uniform distributions with $\rvx$-dependent locations and scales \citep{balle2018variational,flamich2020compressing,balle2020nonlinear}.
% %
\subsection{Implicit Neural Representations}
\par
In the mid-2010s, when machine learning techniques were first successfully introduced into data compression through the break-through papers of \citet{balle2017end} and \citet{theis2017lossy}, they focused on improving the classical methods by replacing the original linear transform, such as the DCT in the case of images, by more expressive nonlinear analogues, such as convolutional neural networks.
These methods were then significantly improved by several highly influential follow-up papers (e.g.\ see the references in the review article of \citet{balle2020nonlinear}), which proposed important architectural innovations but did not veer away from the original recipe of improving classical transform coder architectures.
However, this changed in 2021, with the work of \citet{dupont2021coin}, who proposed the first ``purely'' machine learning-based solution: compression with implicit neural representations (COIN).
\par
The idea of an implicit neural representation (INR) is to conceptualise data as a continuous signal.
For example, we can think of a colour image as a function that maps $(X, Y)$ pixel coordinates to $(R, G, B)$ colour triples or think of weather data as a function that maps from the latitude and longitude $(\phi, \lambda)$ on the Earth to some vector representing some measurements at that location, such as temperature, humidity, wind speed and precipitation.
Now, we may heuristically appeal to some form of the universal approximation theorem for neural networks \citep{kidger2020universal}, which states that any continuous function can be arbitrarily well-approximated by a sufficiently powerful neural net.
This result suggests that if we consider a single datum $\Dataset$, such as an image as a dataset consisting of coordinate-value pairs $(c, v) \in \Dataset$, representing the underlying continuous signal sampled at certain (possibly irregular) intervals, we could fit a neural network $f(c \mid \rvw)$ to $\Dataset$, to memorise it.
Then, once we have fitted the network parameters $\rvw$, we can lossily reconstruct $\Dataset$ by evaluating $f$ at every coordinate location $c$.
Specifically, let $\rmC$ be the set of all input locations appearing in $\Dataset$.
Then, we can obtain a reconstruction by computing $\hat{\Dataset} = f(\rmC \mid \rvw)$, see \cref{fig:recombiner_schematic} for an illustration.
While this idea goes back at least to \citet{stanley2007compositional}, it only recently rose in popularity due to some works in 3D computer graphics such as neural radiance fields \citet{mildenhall2021nerf}.
\par
However, observe now that there is little to no information in the input locations $\rmC$: the data $\Dataset$ is encoded mostly in the weights $w$ of the INR fitted to the data.
This allows us to view INRs as a form of nonlinear transform coding: we have a fixed transform $g(\cdot) = f(\rmC \mid \cdot)$ and we choose the representation to be the conditional distribution $P_{\rvw \mid \Dataset} = \delta_w$, i.e.\ the point-mass at the location $w$.
This suggests a new way to compress data: we fit an INR to it and encode the network weights instead.
This observation (in a somewhat different form) was first made by \citet{dupont2021coin}, with several follow-up works extending their approach by proposing more sophisticated INR architectures \citep{dupont2022coin++,schwarz2022meta,schwarz2023modality}.
However, observe the issue this approach faces: since the representation $\delta_w$ is a point-mass, the mutual information term in \cref{eq:distortion_rate_lagrange_dual} is technically infinite.
COIN-type methods overcame this issue by quantising $w$ to high precision, e.g.,\ $32$ or $16$-bit floating point precision.
This ensures that the mutual information term is now a large but finite constant, hence these methods only optimise the distortion term in \cref{eq:distortion_rate_lagrange_dual}.
However, this is an unsatisfactory solution since, needless to say, restricting the conditional distribution to $32$ or $16$-bit precision point masses significantly limits the power of the transform coding framework.
This observation motivated me to propose compression with Bayesian implicit neural representations instead, where we choose $P_{\rvw \mid \Dataset}$ to be something other than a point-mass. 
The rest of this chapter describes how we can implement this in practice.
\section{The COMBINER Framework}
\label{sec:combiner_method}
\begin{figure}[t]
\centering
\includegraphics[width=\textwidth]{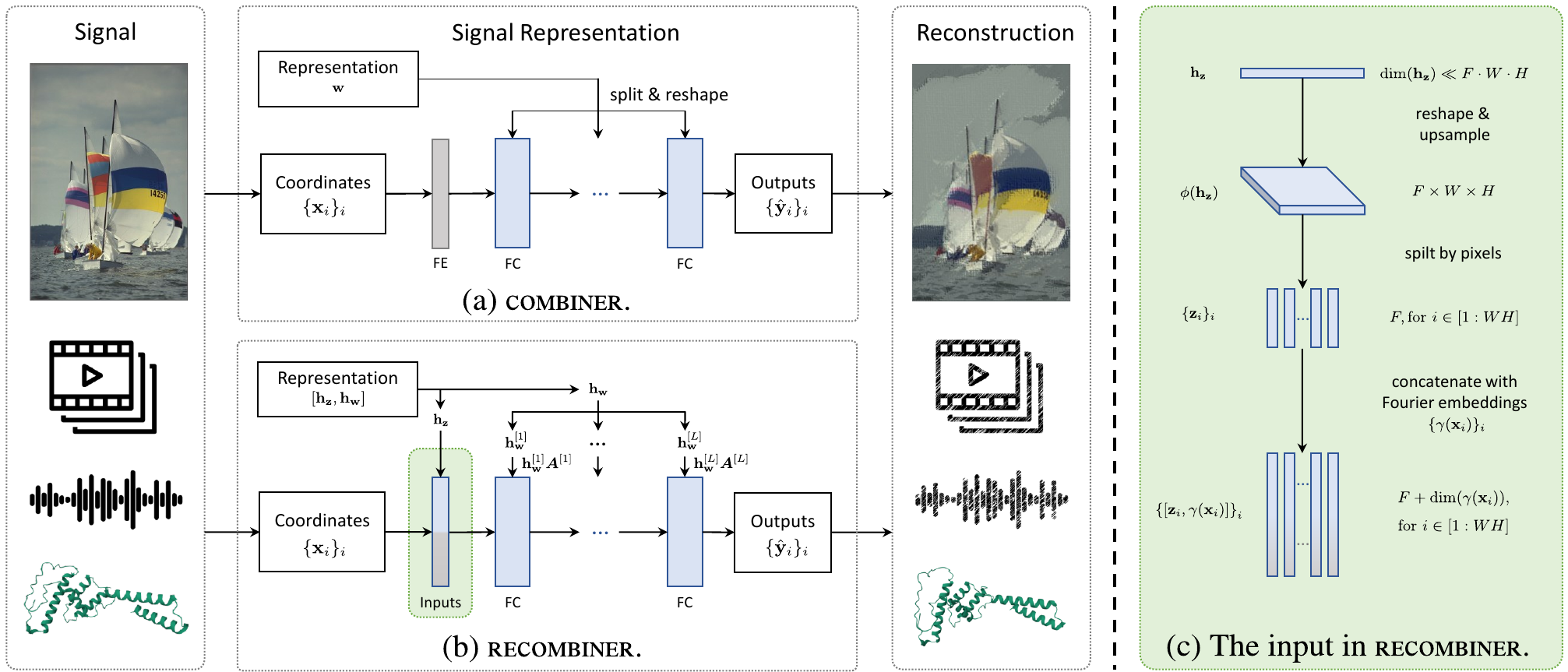}
\caption[Schematic of the COMBINER framework.]{Schematic of (a) the basic COMBINER framework and its extension using (b) the linear weight reparameterisation (\cref{sec:linear_reparam}) and (c) the learned positional embeddings (\cref{sec:recombiner_positional_embeddings}). 
As the INR's input, we upsample $\rvh_\rvz$ to pixel-wise positional encodings and concatenate them to the Fourier embeddings of the coordinates.
\textbf{RECOMBINER} stands for robust and enhanced COMBINER \citep{he2024recombiner}, using the improvements from \cref{sec:recombiner}.
\textbf{FE} stands for Fourier embeddings and \textbf{FC} for fully connected layer.}
\label{fig:recombiner_schematic}
\end{figure}
\par
In this section, I describe the \textit{compression with Bayesian implicit neural representations} (COMBINER) framework that my co-authors and I proposed in \citet{guo2023compression}.
Then, in \cref{sec:recombiner}, I  describe various improvements to the basic framework based on our follow-up work in \citet{he2024recombiner}.
\par
\textbf{The model.}
From a modelling perspective, COMBINER extends the COIN framework of \citet{dupont2021coin} by allowing practitioners to use a conditional distribution for the weights $P_{\rvw \mid \Dataset}$ other than a point-mass; in our experiments, we chose it to be a Gaussian distribution.
Formally, given some source distribution $P_\Dataset$, the COMBINER framework requires us to design an appropriate INR architecture $f$ that takes coordinates $\rmC$ and weights $\rvw$ that it uses to obtain a reconstruction of the data $\hat{\Dataset} = f(\rmC \mid \rvw)$.
Furthermore, we also need to choose an appropriate parametric conditional distribution (technically, a Markov kernel) over the weights $P_{\rvw \mid \Dataset} = P_{\rvw \mid \rvphi}$ for some parameters $\rvphi$.
This immediately defines a marginal distribution over the weights $P_\rvw(A) = \int P_{\rvw \mid \rvphi}(A \mid x) \, dP_\Dataset(x)$.
We also select a distortion function $d$ to measure the quality of the INR's reconstruction.
Then, we can specialise the distortion-rate loss \cref{eq:distortion_rate_lagrange_dual} for the case of an individual datum $\Dataset$ in COMBINER's case:
\begin{align}
\label{eq:combiner_rd_objective}
\Loss(\rvphi \mid \Dataset, \beta) &= \Exp_{\rvw \sim P_{\rvw \mid \rvphi}}[d(\Dataset, \hat{\Dataset})] + \beta \cdot \KLD{P_{\rvw \mid \rvphi}}{P_\rvw}
\end{align}
\par
\textbf{The compression method.}
Then, to encode some data $\Dataset$, we first compute $P_{\rvw \mid \rvphi}$ by minimizing \cref{eq:combiner_rd_objective} in $\rvphi$ using gradient descent.
Once we are satisfied with the $P_{\rvw \mid \rvphi}$ we obtained, we encode a single weight sample $\rvw \sim P_{\rvw \mid \rvphi}$ with a relative entropy coding algorithm, such as the ones I discussed in \cref{chapter:rec_with_pp,chapter:branch_and_bound}, using $P_\rvw$ as the coding distribution.
The appeal of this method is that since the efficiency of relative entropy coding algorithms for individual data points is approximately $\KLD{P_{\rvw \mid \rvphi}}{P_\rvw}$, and hence on average $\Exp_\Dataset[\KLD{P_{\rvw \mid \rvphi}}{P_\rvw}] = \Exp_\Dataset[\KLD{P_{\rvw \mid \Dataset}}{P_\rvw}] = \MI{\Dataset}{\rvw}$ (ignoring the log terms), minimising \cref{eq:combiner_rd_objective} directly optimises the method's distortion-rate performance.
\par
That's all there is to the COMBINER framework.
However, in practice, we still need to overcome several challenges to make it work and to work well.
For example, it is unlikely that we can compute $P_\rvw$ exactly, even if $P_{\rvw \mid \Dataset}$ factorises dimensionwise.
Furthermore, since we are computing $P_{\rvw \mid \Dataset}$ via gradient descent, we must ensure that the encoding process does not take too long.
Thus, in the following sections, I address these issues one by one.
\par
\textbf{A note on the name.}
In the machine learning literature, neural networks with stochastic weights are usually called \textit{Bayesian}.
Hence, following this terminology, in \citet{guo2023compression,he2024recombiner}, we called our networks \textit{Bayesian implicit neural representations}.
However, I should be very clear that it is difficult to argue that these networks are Bayesian in any meaningful sense, as in practice, we neither use the rules of Bayesian inference to compute the posterior over the weights $P_{\rvw \mid \Dataset}$ (in fact, we posit it) nor the posterior predictive $P_{\hat{\Dataset}}$.
Instead, these models are best interpreted from the lossy source coding perspective: they aim to minimise their distortion-rate function (\cref{eq:distortion_rate_function}) by optimising its Lagrange dual (\cref{eq:combiner_rd_objective}).
Thus, in hindsight, perhaps a better name for the method would have been COSINE: compression with \textit{stochastic} implicit neural representations.
\subsection{Training COMBINER}
\label{sec:training_combiner}
\par
As I mentioned above, an important limitation of COMBINER (and any other probabilistic compression method for that matter) is that for any weight posterior $P_{\rvw \mid \Dataset}$ we might choose, the marginal $P_{\rvw}(A) = \int P_{\rvw \mid \Dataset}(A \mid x)\,dP_\Dataset(x)$ is usually impossible to compute, since integrating over the source distribution is not tractable analytically.
To resolve this issue, we can utilise the fact that as we have seen in \cref{chapter:rec_with_pp,chapter:branch_and_bound}, for a given source sample $\Dataset \sim P_{\Dataset}$, we can use any coding distribution $Q_\rvw$ to encode a sample from $P_{\rvw \mid \Dataset}$ at the efficiency of approximately $\KLD{P_{\rvw \mid \Dataset}}{Q_\rvw}$, which will lead to an average coding efficiency of
\begin{align*}
\Exp_{\Dataset \sim P_\rvx}[\KLD{P_{\rvw \mid \Dataset}}{Q_\rvw}] 
&= \KLD{P_{\rvw, \Dataset}}{Q_\rvw \otimes P_\rvx} \\
&= \KLD{P_\rvw}{Q_\rvw} + \KLD{P_{\Dataset \mid \rvw}}{P_\rvx}.
\end{align*}
From the above, we see that the description length is minimised if $\KLD{P_\rvw}{Q_\rvw} = 0$, which occurs precisely when $Q_\rvw = P_\rvw$.
Thus, we can connect the mutual information $\MI{\rvx}{\rvw}$ to the coding efficiency of a relative entropy coding algorithm using an arbitrary coding distribution $Q_\rvw$ via the following variational minimisation problem:
\begin{align}
\label{eq:combiner_codelength_variational_formulation}
\MI{\rvx}{\rvw} = \min_{Q_\rvw \in \DistFamily_\rvw} \Exp_{\Dataset \sim P_\Dataset}[\KLD{P_{\rvw \mid \Dataset}}{Q_\rvw}],
\end{align}
where $\DistFamily_\rvw$ is the set of all valid probability measures over the range of $\rvw$.
Thus, we see that from the compression perspective, the intractability of calculating $P_\rvw$ is equivalent to computing the optimum of the minimisation problem in \cref{eq:combiner_codelength_variational_formulation}.
However, this suggests a principled way of resolving the intractability: instead of minimising the objective over the set of all valid probability measures, we pick a suitable smaller class $\Class$ to minimise over.
\par
The final hurdle is to note that $P_{\rvw \mid \Dataset}$ is not entirely trivial either: for a parametric posterior (which is always going to be the case in this thesis) with parameters $\rvphi$, recall that the posterior is defined as optimum of the distortion-rate objective \cref{eq:combiner_rd_objective} at an implicitly defined rate level given by the trade-off parameter $\beta$:
$P_{\rvw \mid \Dataset} = P_{\rvw \mid \rvphi^*}$, where $\rvphi^* = \argmin_{\rvphi}\Loss(\rvphi \mid \Dataset, \beta)$.
Thus, putting \cref{eq:combiner_rd_objective,eq:combiner_codelength_variational_formulation} together, a principled way to choose the coding distribution is to minimise the average distortion-rate objective
\begin{align*}
Q_\rvw^* = \argmin_{Q_\rvw \in \Class} \left(\min_{\rvphi}\left\{\Exp_{\Dataset, \rvw \sim P_{\Dataset} \otimes P_{\rvw \mid \rvphi}}[d(\Dataset, \hat{\Dataset})] + \beta \cdot \Exp_{\Dataset \sim P_\Dataset}[\KLD{P_{\rvw \mid \Dataset}}{Q_\rvw}]\right\}\right).
\end{align*}
In particular, for a dataset of $M$ points $\{\Dataset_1, \hdots, \Dataset_M\}$, where each $\Dataset_m$ has associated weight parameters $\rvphi_m$, this objective we wish to minimise admits the following Monte Carlo approximation:
\begin{align*}
 \sum_{i = m}^M \Exp_{\rvw \sim P_{\rvw \mid \rvphi_m}}[d(\Dataset_m, \hat{\Dataset}_m)] + \beta \cdot\KLD{P_{\rvw \mid \Dataset_m}}{Q_\rvw}] ,
\end{align*}
which we can now optimise using gradient descent or the expectation-minimisation (EM) algorithm \citep{mackay2003information}.
In the rest of this section, I describe the concrete setting we used in our papers \citet{guo2023compression,he2024recombiner}.
\par
\textbf{Training using fully-factorised Gaussian parameterisation.}
The humble Gaussian is perhaps the simplest sensible parameterisation for the weights: we choose $P_{\rvw \mid \Dataset}$ and $Q_\rvw$ to follow fully-factorised Gaussian distributions.
Concretely, considering the weights as a vector in some high-dimensional Euclidean space, we set 
\begin{align*}
P_{\rvw \mid \Dataset} &= \Normal(\rvmu(\Dataset), \rvsigma^2(\Dataset)) \\
Q_\rvw &= \Normal(\rvnu, \rvrho^2),
\end{align*}
where $\rvmu, \rvnu, \rvsigma^2$ and $\rvrho^2$ are vectors of equal dimension to $\rvw$, and the above notation for Gaussian distributions denotes fully-factorised distributions, i.e.\
\begin{align*}
\Normal(\rvnu, \rvrho^2) &= \bigotimes_{d} \Normal(\nu_d, \rho^2_d).
\end{align*}
Note that now the conditional mean and variance vectors constitute the parameters of the conditional distribution: $\rvphi = \{\rvmu, \rvrho^2\}$. 
The first advantage of this setting is that $\KLD{P_{\rvw \mid \Dataset_m}}{Q_\rvw}]$ admits a closed-form solution.
However, second, we can do even more: for a fixed set of posterior parameters $\{\rvphi_m\}_{m = 1}^M$, the optimum over the posterior parameters $\rvnu, \rvrho^2$ is also available in closed form:
\begin{align}
Q_\rvw^* = \Normal(\rvnu_*, \rvrho^2_*) &= \argmin_{Q_\rvw \in \{\Normal(\rvnu, \rvrho^2)\}} \sum_{m = 1}^M \KLD{P_{\rvw \mid \rvphi_m}}{Q_\rvw}] \nonumber\\
&\nonumber\\
\rvnu_* &= \frac{1}{M}\sum_{m = 1}^M \rvmu(\Dataset_m) \label{eq:combiner_em_mean_update}\\
\rvrho^2_* &= \frac{1}{M}\sum_{m = 1}^M \left[(\rvmu(\Dataset_m) - \rvnu^*)^2 + \rvsigma^2(\Dataset_m)\right]\label{eq:combiner_em_variance_update}
\end{align}
Therefore, in our papers, we used expectation-minimisation \citep{mackay2003information} to train the model hyperparameters: for a dataset $\{\Dataset_m\}_{m = 1}^M$, we alternated between optimising the posterior parameters $\rvphi_m$ and updating the prior parameters using the above identities.
\par
\textbf{Tuning the distortion-rate trade-off parameter $\beta$.}
Recall that $\beta$ arises in the distortion-rate objective as the slack variable associated with the soft constraint enforcing a particular rate/coding budget $R$.
As such, optimising \cref{eq:distortion_rate_lagrange_dual} with $\beta$ set to different values will yield models with different distortion-rate trade-offs.
Hence, in \citet{he2024recombiner}, we incorporated an adaptive tuning strategy in our training loop for $\beta$ that enables practitioners to set an explicit bitrate $R$, which I describe next.
\par
We start with a dataset $\{\Dataset_m\}_{m = 1}^M$, and initial guesses for the coding distribution $Q_\rvw$, the trade-off parameter $\beta$ and for each weight parameter $\rvphi_m$ associated to $\Dataset_m$.
Then, we iterate the following three steps until convergence:
\begin{enumerate}
\item \textbf{Optimise the variational parameters.} We fix the coding distribution $Q_\rvw$ and $\beta$, and use gradient descent to compute 
\begin{align}
\label{eq:combiner_variational weight_update}
\{\rvphi_m\}_{m = 1}^M \gets \argmin_{\{\phi_m\}_{m = 1}^M}\sum_{i = m}^M \Exp_{\rvw \sim P_{\rvw \mid \phi_m}}[d(\Dataset_m, \hat{\Dataset}_m)] + \beta \cdot\KLD{P_{\rvw \mid \phi_m}}{Q_\rvw}]   
\end{align}
\item \textbf{Update the coding distribution.} We fix $\{\rvphi_m\}_{m = 1}^M$ and $\beta$, and set the coding distribution $Q_\rvw \gets \Normal(\rvnu_*, \rvrho^2_*)$ according to \cref{eq:combiner_em_mean_update,eq:combiner_em_variance_update}.
\item \textbf{Adjust the trade-off parameter.} We fix $Q_\rvw$ and $\{\rvphi_m\}_{m = 1}^M$.
Now, we estimate the average codelength of the scheme by computing
\begin{align}
\label{eq:combiner_training_avg_kl_estimate}
\iota = \frac{1}{M}\sum_{m = 1}^M \KLD{P_{\rvw \mid \rvphi_m}}{Q_\rvw}.
\end{align}
Then, if $\iota > C$, we increase the rate penalty by setting $\beta \gets \beta \cdot (1 + \tau)$ and if $\iota < C - \epsilon$, we loosen the rate penalty by setting $\beta \gets \beta \big/ (1 + \tau)$, where $\epsilon$ is a small update threshold parameter and $\tau$ is the update step size.
In our experiments, we set $\tau = 0.5$ and set $\epsilon \in [0.05, 0.5]$, depending on the task.
In practice, we found $\beta$ to stabilise after $30-50$ iterations.
\end{enumerate}
\subsection{Compression with Posterior Refinement}
\label{sec:combiner_posterior_refinement}
\noindent
Once we have obtained the coding distribution using the training procedure I outlined at the end of \cref{sec:training_combiner}, we can share the coding distribution parameters between communicating parties, and we are ready to use COMBINER to compress data!
Hence, from now on, I assume that the coding distribution $Q_\rvw$ is fixed.
Indeed, as I already outlined earlier in this section, compressing some new datum $\Dataset$ consists of initialising the variational weight parameters $\rvphi$ and optimising \cref{eq:combiner_rd_objective}. 
However, to further improve the scheme's performance, we also adopted a progressive posterior refinement strategy, a concept proposed originally in \cite{havasi2019minimal} for Bayesian model compression.
\par
To motivate this strategy, recall that the reason why we restricted the form of the conditional weight distribution $P_{\rvw \mid \Dataset}$ to be Gaussian (or some equally simple distribution) is that otherwise optimising \cref{eq:combiner_rd_objective} is challenging.
However, the fact is that the optimal conditional distribution is quite far from our Gaussian approximation.
However, note that for compression, we only care about encoding a \textbf{single, good quality sample} using relative entropy coding.
Thus, to improve the quality of our encoded sample, \citet{havasi2019minimal} suggest partitioning the weight vector $\rvw$ into $K$ blocks $\rvw_{1:K} = \{\rvw_1, \hdots, \rvw_K\}$.
For example, we might partition the weights per neural network layer with $\rvw_k$ representing the weights on layer $k$ or into a preset number of random blocks; at the extremes, we could partition $\rvw$ per dimension, or we could just set $K = 1$ for the trivial partition.
Furthermore, let $\rvphi_{1:K}$ denote the variational parameters that correspond to the blocks $\rvw_{1:K}$, i.e.\ $\rvphi_k$ denotes the parameters of $P_{\rvw_k \mid \Dataset}$.
Then, to obtain a richer variational approximation given a partition $\rvw_{1:K}$, we start as before and optimise the variational parameters $\rvphi_{1:K}$ in \cref{eq:combiner_rd_objective} to obtain the factorised approximation 
\begin{align}
\rvphi_{1:K} = \argmin_{\phi_{1:K}}\Loss(\phi_{1:K} \mid \Data, \beta).
\end{align}
However, we now deviate from encoding a full, exact sample $\rvw_{1:K} \sim P_{\rvw_{1:K} \mid \Dataset}$ in one go.
Instead, \textbf{we only encode a sample from the first block} $\rvw_1 \sim P_{\rvw_{1} \mid \Dataset}$, and \textbf{refine} the remaining approximation:
\begin{align}
\rvphi_{2:K} = \argmin_{\phi_{2:K}}\Loss(\phi_{2:K} \mid \Data, \beta, \rvw_1),
\end{align}
where $\Loss(\cdot \mid \Dataset, \beta, \rvw_1)$ indicates that $\rvw_1$ is fixed during the optimization.
We now encode $\rvw_2 \sim P_{\rvw_2 \mid \Dataset, \rvw_1}$ to obtain the second chunk of our final sample.
We iterate the refinement procedure $K$ times, progressively conditioning on more blocks, until we obtain our final sample $\rvw = \rvw_{1:K}$.
Note that this way, the sample $\rvw$ has a significantly more complex distribution.
For example, even if we originally parameterised the conditional weight distribution $P_{\rvw_{1:K} \mid \Dataset}$ as a fully-factorised Gaussian distribution, already after the first refinement step, the distribution $P_{\rvw_{2:K} \mid \Dataset, \rvw_1}$ becomes \textit{conditionally factorised Gaussian}. 
This choice makes the model far more flexible, and thus it is guaranteed to achieve a lower distortion-rate score \citep{havasi2020refining}.
\par
\textbf{Combining the refinement procedure with compression:} 
Above, I assumed that after each refinement step $k$, we encode an exact sample $\rvw_k \sim P_{\rvw_{k} \mid \Dataset, \rvw_{1:k - 1}}$.
However, we can also extend the scheme to incorporate approximate relative entropy coding (\cref{sec:approximate_sampling}) by encoding an approximate sample $\tilde{\rvw}_k \sim \tilde{P}_{\rvw_{k} \mid \Dataset, \rvw_{1:k - 1}}$ using an approximate relative entropy coding algorithm; in our experiments, we used the step-limited variant of A* coding (\cref{alg:global_a_star}).
This way, we feed two birds with one scone:
the refinement process allows us to obtain a better overall approximate sample $\tilde{\rvw}_{1:K}$ by extending the variational family and by correcting for the occasional bad quality weight block $\tilde{\rvw}_k$ at the same time, thus making COMBINER more robust. 
\subsection{COMBINER in Practice}
\label{sec:combiner_in_practice}
\noindent
After outlining our posterior refinement strategy in \cref{sec:combiner_posterior_refinement}, I now deal with one more practical issue: the runtime of the compressor.
Concretely, given a partition $\rvw_{1:K}$ of the weight vector $\rvw$, we used step-limited A* coding to encode a sample $\tilde{\rvw}_k$ from each block.
Now, let $\delta_k = \KLD{P_{\rvw_k \mid \Dataset, \tilde{\rvw}_{1:k - 1}}}{Q_{\rvw_k}}$ represent the relative entropy in block $k$ after the completion of the first $k - 1$ refinement steps, where we have already simulated and encoded samples from the first $k - 1$ blocks. 
As I have shown in \cref{sec:approximate_sampling}, step-limited A* sampling needs to simulate on the order of $2^{\delta_k}$ samples from the coding distribution $Q_{\rvw_k}$ to ensure that it encodes a good quality approximate sample $\tilde{\rvw}_k$.
Therefore, for our method to be computationally tractable, it is important to ensure that there is no block with a large divergence $\delta_k$.
To guarantee that COMBINER's runtime is consistent, we would like the divergences across all blocks to be approximately equal, i.e., $\delta_i \approx \delta_j$ for $0 \leq i, j \leq K$.
To this end, we set a bit budget of $\kappa$ bits per block, and below, I describe the techniques we used to ensure $\delta_k \approx \kappa$ for all $k = 1,\hdots, K$.
Unless I state otherwise, we set $\kappa = 16$ bits in our experiments.
\par
Now, I describe how we partition the weight vector based on the training data to approximately enforce our budget on average.
As I explained in \cref{sec:training_combiner}, the practitioner can directly set the average coding budget $R$ during COMBINER's training phase.
Thus, given a block-wise coding budget of $\kappa$ bits, we need at least $\lceil R / \kappa \rceil$ blocks for the budget to be feasible.
However, we cannot divide the dimensions of the weight vector $\rvw$ into equal-sized blocks, as its dimensions do not equally contribute to the total coding cost/KL divergence.
We solve this in two steps:
\begin{enumerate}
\item 
First, we destroy as much correlation as possible between ``neighbouring dimensions,''  such as weights on the same hidden layer corresponding to the same activation.
To achieve this, the sender simulates a uniformly random permutation $\rvalpha$ on $\dim(\rvw)$ elements using their common randomness/shared PRNG seed.
Then, they use it to randomise the order in which the weights are transmitted by computing the permuted weight vector $\rvalpha(\rvw) = (\ervw_{\rvalpha(1)}, \ervw_{\rvalpha(2)}, \hdots \ervw_{\rvalpha(\dim(\rvw)})$, where $\ervw_i$ denotes the $i$th element of the weight vector.
Then, the decoder can sample the same permutation $\rvalpha$ and invert it to get the original ordering of the weight vector back.
\item
Second, to determine the blocks, after we finished training COMBINER, note that we can slightly modify \cref{eq:combiner_training_avg_kl_estimate} to estimate the average KL divergence $c_j$ in each dimension $\ervw_j$.
Then, starting with the first element of $\rvalpha(\rvw)$, we greedily assign the blocks using the next-fit bin packing algorithm \citep{johnson1973near} with item size $c_j$ and bin size $\kappa$.
In practice, we found that this procedure resulted in roughly even-sized bins, and the number of bins $K$ did not exceed the minimum necessary number of $\lceil R / \kappa \rceil$ bins by much.
\end{enumerate}
\par
\textbf{Relative entropy coding-aware fine-tuning:}
Assume we now wish to compress some data $\Data$, and we already selected the desired rate-distortion trade-off $\beta$, ran the training procedure to obtain the coding distribution $Q_\rvw$, fixed a bit budget $\kappa$ for each block and partitioned the weight vector using the procedure from the previous paragraph.
Despite our effort to set the blocks so that the average divergence $\bar{\delta}_k \approx \kappa$ in each block on the training data, if we optimised $P_{\rvw \mid \Dataset}$ using the distortion-rate loss in\cref{eq:combiner_rd_objective}, it is unlikely that the block-wise divergences $\delta_k$ would match $\kappa$.
Therefore, instead of optimise the Lagrange relaxation of a stricter, block-wise variant of the distortion rate function (\cref{eq:distortion_rate_function}): for each block $\rvw_k$ for $k \in [1:K]$ we introduce slack variables $\beta_k$ and obtain the posterior parameters $\rvphi$ by optimising
\begin{align}
\label{eq:combiner_blockwise_rd_objective}
\Loss(\rvphi \mid \Dataset, \beta_{1:k}) &= \Exp_{\rvw \sim P_{\rvw \mid \rvphi}}[d(\Dataset, \hat{\Dataset})] + \sum_{k = 1}^K \beta_k \cdot \KLD{P_{\rvw_k \mid \rvphi}}{P_{\rvw_k}},
\end{align}
where we adjusted the slack variables $\beta$ using the same adaptation procedure as I described at the end of \cref{sec:training_combiner} to ensure that the KL divergences all approximately match the block-wise coding budget $\kappa$.
Finally, we perform the block-wise posterior refinement procedure I described in \cref{sec:combiner_posterior_refinement}, during which we encode each block $\hat{\rvw}_k$ one-by-one using step-limited A* sampling.
\par
\textbf{The comprehensive COMBINER pipeline:}
Thus, a brief summary of the entire COMBINER compression pipeline is as follows:
\begin{enumerate}
\item \textbf{Design and training.} Given a dataset $\{\Data_1, \hdots, \Data_M\}$, we select an INR architecture and run the training procedure (\cref{sec:training_combiner}) with different settings for $\beta$ to obtain coding distributions for a range of rate-distortion trade-offs.
\item \textbf{Compression.} 
To compress a new datum $\Dataset$, we select a prior with the desired rate-distortion trade-off and pick a block-wise coding budget $\kappa$.
Next, we partition the weight vector $\rvw$ based on $\kappa$. 
Finally, we run the relative entropy coding-aware fine-tuning procedure from above, using step-limited A* coding to compress the weight blocks between the refinement steps to encode $\Data$.
\end{enumerate}%
\section[Improving the Basic Framework]{Improving the Basic Framework}
\label{sec:recombiner}
\begin{figure}[t]
\centering
\includegraphics[width=\textwidth, trim={0, 1.4cm, 0, 2.cm}]{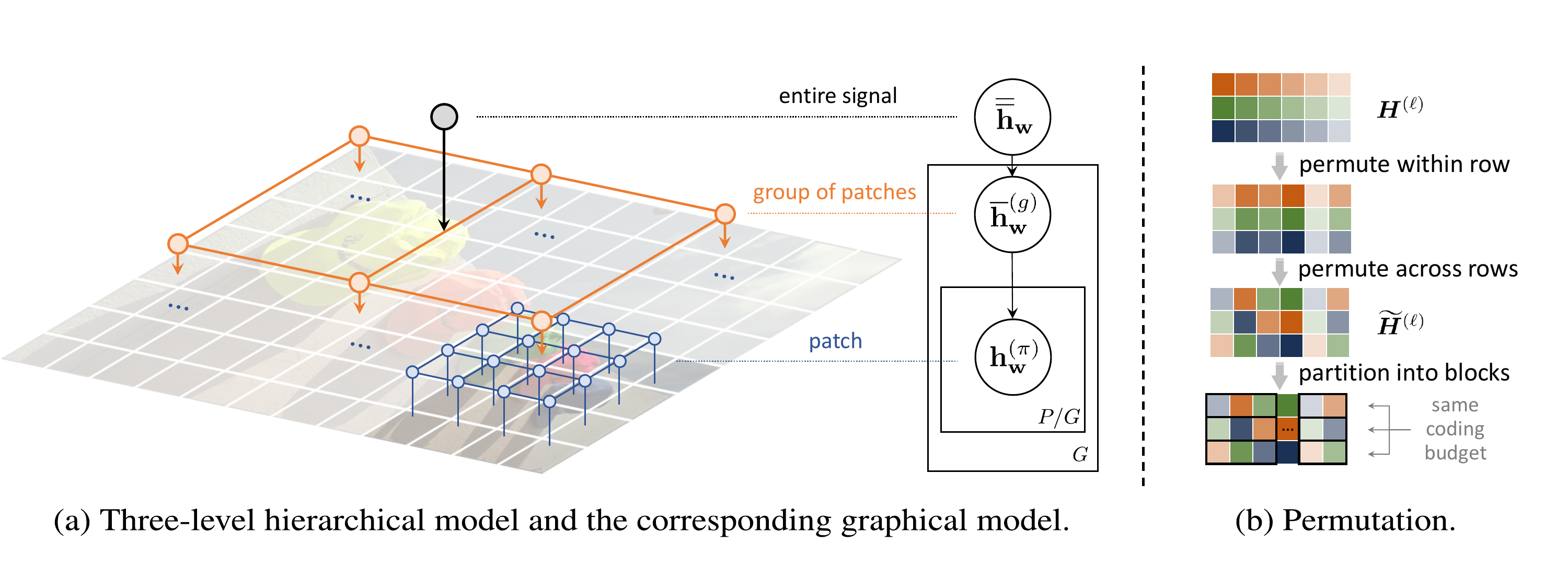}
\caption[Illustration of COMBINER's hierarchical model and advanced permutation strategy.]{Illustration of (a) the three-level hierarchical model and (b) our permutation strategy.}
\label{fig:illustration_hyper_prior}
\end{figure}
\noindent
This section proposes several extensions to the COMBINER framework (\cref{sec:combiner_method}) that significantly improve its robustness and performance. 
Concretely, I
1) introduce a linear reparameterisation for the INR's weights, which yields a richer variational posterior family; 
2) augment the INR's input with learned positional encodings to capture local features in the data and to assist overfitting and
3) scale the method to high-resolution image compression by dividing the images into patches and introducing an expressive hierarchical Bayesian model over the patch-INRs.
Contributions 1) and 2) are illustrated in \cref{fig:recombiner_schematic}, while 3) is shown in \cref{fig:illustration_hyper_prior}.
\subsection{Linear Reparameterization for the Network Parameters}
\label{sec:linear_reparam}
\noindent
We originally implemented COMBINER using a fully-factorised Gaussian weight posterior distribution:
\begin{align*}
P_{\rvw \mid \Dataset} = \Normal(\rvmu, \sigma^2) = \bigotimes_d \Normal(\mu_d, \sigma^2_d),  
\end{align*}
where I suppressed the dependence of $\rvmu, \rvsigma^2$ on the data $\Dataset$.
However, dimension-wise independent weights are known to be quite unrealistic \citep{izmailov2021bayesian} and to underfit the data \citep{dusenberry2020efficient}, which goes directly against our goal of fitting the INR to data as much as possible.
On the other hand, using a full-covariance Gaussian posterior approximation would significantly increase the training and coding time, even for small network architectures.
\par
Hence, I describe a solution we proposed in \citet{he2024recombiner} that lies in between: 
at a high level, we learn a linearly-transformed factorised Gaussian approximation that closely matches the full-covariance Gaussian posterior on average over the training data.
Formally, for each layer $l = 1,\hdots, L$, we model the weights as ${\rvw^{[l]} = \rvh_\rvw^{[l]} \mA^{[l]}}$, where the $\mA^{[l]}$ are square matrices, and we now use a fully-factorised Gaussian posterior and coding distribution for $\rvh_\rvw^{[l]}$ instead.
Then, we incorporate learning each weight transform $\mA^{[l]}$ into the training procedure (\cref{sec:training_combiner}), and we fix them for compression.
Thus, when compressing a new datum, we only infer the factorised posteriors $P_{\rvh_\rvw^{[l]} \mid \Dataset}$.
To simplify notation, I collect the $\mA^{[l]}$ in a block-diagonal matrix $\mA = \diag(\mA^{[1]}, \hdots, \mA^{[L]})$ and the $\rvh^{[l]}_\rvw$ in a single row-vector $\rvh_\rvw = [\rvh^{[1]}_\rvw, \hdots, \rvh^{[L]}_\rvw]$, so that now the weights are given by $\rvw = \rvh_\rvw \mA$. 
Our experiments found this layer-wise weight reparameterisation as efficient as using a joint one for the entire weight vector $\rvw$.
Hence, we used the layer-wise approach, as it is more parameter and compute-efficient.
\par
This simple yet expressive variational approximation has a couple of advantages. 
First, it provides an expressive full-covariance prior and posterior while requiring much less training and coding time.
Specifically, the KL divergence required by \cref{eq:combiner_rd_objective,eq:combiner_blockwise_rd_objective} is still between factorised Gaussians, and we do not need to optimise the full covariance matrices of the posteriors during coding.
Second, this parametrisation has scale redundancy: for any $c \in \Reals$, we have
$\rvh_\rvw \mA = (\nicefrac{1}{c}\cdot\rvh_\rvw)(c\cdot\mA)$.
Hence, if we initialise $\rvh_\rvw$ suboptimally during training, $\mA$ can still learn to compensate for it, making our method more robust.
Finally, note that this reparameterisation is tailored explicitly for INR-based compression and would usually not be feasible in other BNN use cases since we learn $\mA$ while inferring multiple variational posteriors simultaneously.
\subsection{Learned Positional Encodings}
\label{sec:recombiner_positional_embeddings}
\noindent
A challenge for overfitting INRs, especially at low bitrates, is their \emph{global} representation of the data, in the sense that each of their weights influences the reconstruction at every coordinate. 
To mitigate this issue, we extended INRs in \citet{he2024recombiner} to take a learned positional input $z_i$ at each coordinate $c_i$: $f(c_i, z_i \mid \rvw)$.
\par
However, it is usually wasteful to introduce a vector for each coordinate in practice.
Instead, we used a lower-dimensional row-vector representation $\rvh_\rvz$, which we reshape and upsample with a learnable function $\phi$.
For example, in the case of a $W \times H$ image with $F$-dimensional positional encodings,  we could pick $\rvh_\rvz$ such that $\dim(\rvh_\rvz) \ll F \cdot W \cdot H$, then reshape and upsample it to be $F \times W \times H$ by picking $\phi$ to be some small convolutional network.
Then, we set $\rvz_i = \phi(\rvh_\rvz)_{\rvx_i}$ to be the positional encoding at location $\rvx_i$.
We placed a factorised Gaussian prior and variational posterior on $\rvh_\rvz$.
Hereafter, we refer to $\rvh_\rvz$ as the \emph{latent} positional encodings, $\phi(\rvh_\rvz)$ and $\rvz_i$ as the \emph{upsampled} positional encodings.
\subsection{Scaling to High-Resolution Images with Patches}
\label{sec:scaling_with_patches}
With considerable effort, we can scale the basic implementation of COMBINER to high-resolution images by significantly increasing the number of INR parameters.
However, the training procedure becomes very sensitive to hyperparameters, including the initialisation of variational parameters and model size selection. 
Unfortunately, improving the robustness of large INRs using the weight reparameterisation I describe in \cref{sec:linear_reparam} is also impractical because the size of the transformation matrix $\mA$ grows quadratically in the number of weights. 
Therefore, in this section, I describe an alternative approach to scaling the method to high-resolution images using patches, in line with other INR-based works \citep{dupont2022coin++, schwarz2022meta,schwarz2023modality}.
\par
The high-level strategy is to first divide a high-resolution image into a grid of lower-resolution patches, then apply COMBINER separately to each patch to encode them.
However, the patches' INRs are independent by default, hence we re-introduce information sharing between the patch-INRs' weights via a hierarchical model for $\rvh_\rvw$.
Finally, at the end of this section, I discuss how we can take advantage of the patch structure to parallelise data compression and reduce the encoding time.
\par
\textbf{The hierarchical Bayesian model:} 
We posit a global latent representation for the weights $\rvhbar_\rvw$, from which each patch-INR can deviate.
Thus, assuming that we split the data $\Data$ into $P$ patches, for each patch $\pi \in 1,\hdots, P$, we need to define the conditional distributions of patch representations $\rvh_\rvw^{(\pi)} \mid \rvhbar_\rvw$.
However, since we wish to model deviations from the global representation, it is natural to
decompose the patch representation as $\rvh_\rvw^{(\pi)} = \Delta\rvh^{(\pi)}_\rvw + \rvhbar_\rvw$, and specify the conditional distribution of the differences $\Delta\rvh_\rvw^{(\pi)} \mid \rvhbar_\rvw$ instead, without any loss of generality.
Thus, we used a fully-factorised Gaussian coding distribution and variational posterior on the joint distribution of the global representation and the deviations, given by the following product of $P + 1$ Gaussian measures:
\begin{align}
Q_{\,\rvhbar_\rvw, \Delta\rvh_\rvw^{(1:P)}} 
&= \Normal(\rvmubar_{\rvw}, \diag(\rvsigmabar_{\rvw})) \bigotimes_{\pi = 1}^P \Normal(\rvmu_{\Delta}^{(\pi)}, \diag(\rvsigma_{\Delta}^{(\pi)})) \label{eq:patch_prior}%
\\
P_{\,\rvhbar_\rvw, \Delta\rvh_\rvw^{(1:P)} \mid \Dataset} 
&= \Normal(\rvnubar_{\rvw}, \diag(\rvrhobar_{\rvw})) \bigotimes_{\pi = 1}^P \Normal(\rvnu_{\Delta}^{(\pi)}, \diag(\rvrho_{\Delta}^{(\pi)})), \label{eq:patch_posterior}
\end{align}
where $1:P$ is the slice notation, i.e.\ $\Delta\rvh_\rvw^{(1:P)} = \Delta\rvh_\rvw^{(1)}, \hdots, \Delta\rvh_\rvw^{(P)}$.
Importantly, while the posterior in \cref{eq:patch_posterior} assumes that the global representation and the differences are independent, $\rvhbar_\rvw$ and $\rvh_\rvw^{(\pi)}$ remain correlated.
Note that optimising \cref{eq:combiner_rd_objective,eq:combiner_blockwise_rd_objective} requires us to compute $\KLD{P_{\rvh_\rvw^{(1:P)} \mid \Dataset}}{Q_{\rvh_\rvw^{(1:P)}}}$.
Unfortunately, this calculation is infeasible due to the complex dependence between the $\rvh_\rvw^{(\pi)}$s.
Instead, we can minimise an \textit{upper bound} to it by observing that
\begin{align}
\KLD{P_{\rvh_\rvw^{(1:P)} \mid \Dataset}}{Q_{\rvh_\rvw^{(1:P)}}} &\leq \KLD{P_{\rvh_\rvw^{(1:P)} \mid \Dataset}}{Q_{\rvh_\rvw^{(1:P)}}} + 
  \KLD{P_{\,\rvhbar_\rvw \mid \rvh_\rvw^{(1:P)} \mid \Dataset}}{Q_{\,\rvhbar_\rvw \mid \rvh_\rvw^{(1:P)}}} \nonumber\\  
  &=\KLD{P_{\,\rvhbar_\rvw, \rvh_\rvw^{(1:P)} \mid \Dataset}}{Q_{\,\rvhbar_\rvw, \rvh_\rvw^{(1:P)}}} \nonumber\\
  &= \KLD{P_{\,\rvhbar_\rvw, \Delta\rvh_\rvw^{(1:P)} \mid \Dataset}}{Q_{\,\rvhbar_\rvw, \Delta\rvh_\rvw^{(1:P)}}}.\label{eq:augmented_hier_kl}
\end{align}
Hence, when training the patch-INRs, we replace the KL term in \cref{eq:combiner_rd_objective,eq:combiner_blockwise_rd_objective} with the divergence in \cref{eq:augmented_hier_kl}, which is between factorised Gaussian distributions and cheap to compute. 
Finally, we remark that we can view $\rvhbar_\rvw$ as side information, also prevalent in other neural compression codecs \citep{balle2018variational}, or auxiliary latent variables enabling factorisation \citep{koller2009probabilistic}.
\par
While \cref{eq:patch_prior,eq:patch_posterior} describe a two-level hierarchical model, nothing stops us from extending the hierarchical structure by breaking up patches further into sub-patches and adding extra levels to the probabilistic model.
For our experiments on high-resolution audio, images, and video, we found that a three-level hierarchical model worked best, with global weight representation ${\scriptstyle\rvhbbar_\rvw}$, second/group-level representations ${\scriptstyle\rvhbar_\rvw^{(1:G)}}$ and third/patch-level representations ${\scriptstyle\rvh_\rvw^{(1:P)}}$, illustrated in \cref{fig:illustration_hyper_prior}a.
Empirically, a hierarchical model for the learned positional encodings $\rvh_\rvz$ did not yield significant gains. Thus, we only use it for the INR weight representations $\rvh_\rvw$.
\begin{figure}[t]
\centering
\begin{subfigure}{\textwidth}
\centering
\ref*{legend:recombiner_img_rd_comparison}
\includegraphics[width=\textwidth]{6-COMBINER/img/tikz/img_rd_comparison.tikz}

\caption{RD curve on CIFAR-10 (left) and Kodak (right).}
\label{fig:rd_images}
\vspace{10pt}
\end{subfigure}
\includegraphics[width=\textwidth]{6-COMBINER/img/tikz/audio_video_protein_rd_comparison.tikz}
\caption[Rate-distortion curves on image, audio, video and protein data compression.]{Quantitative evaluation of COMBINER with and without the improvements I presented in \cref{sec:recombiner} on image, audio, video, and 3D protein structure. Kbps stands for kilobits per second, RMSD stands for Root Mean Square Deviation, and bpa stands for bits per atom. I also present comparisons with contemporary classical and neural data compression algorithms. 
Solid lines denote INR-based codecs, dotted lines denote VAE-based codecs, and dashed lines denote classical codecs. }
%\vspace{-10pt}
%\vspace{-10cm}
\end{figure}
\begin{figure}
\begin{subfigure}[t]{0.35\textwidth}
    \centering
\includegraphics[height=150pt]{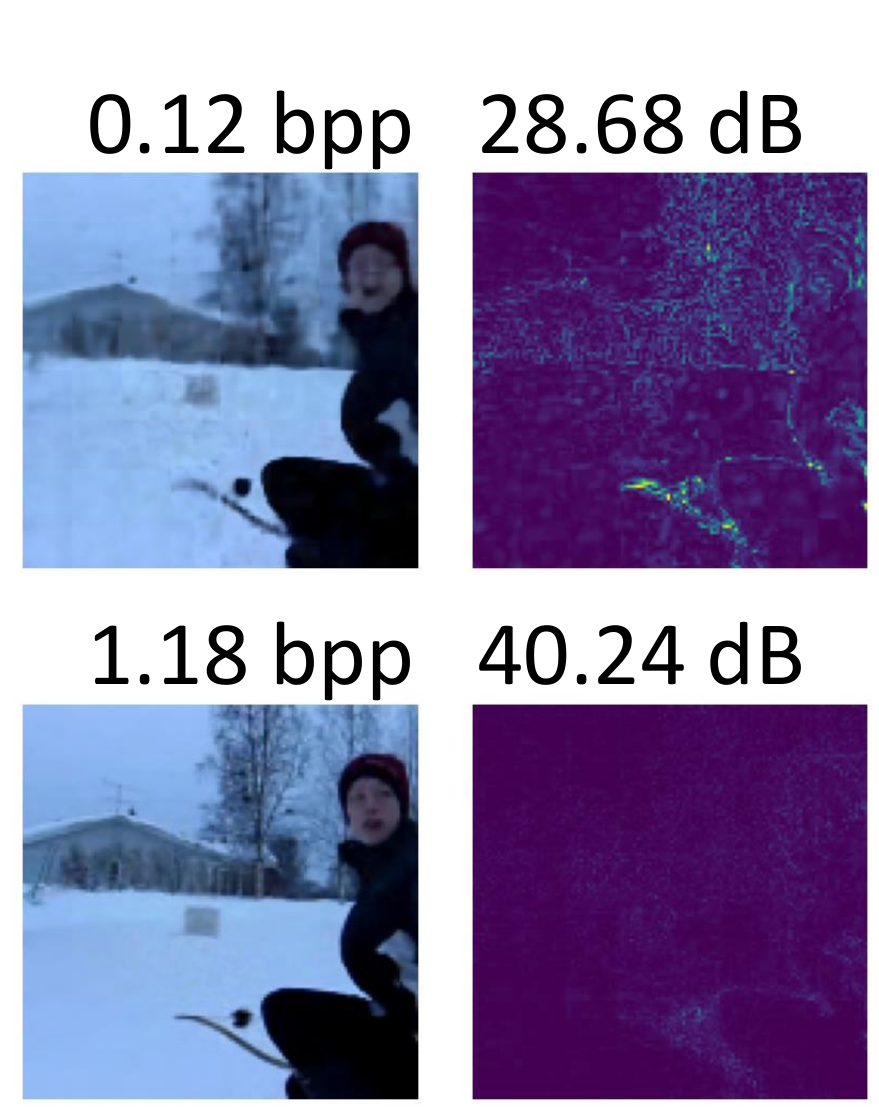}
\caption{Decoded video frames and residuals.}
\label{fig:video_examples}
\end{subfigure}\hspace{-10pt}
\begin{subfigure}[t]{0.65\textwidth}
\centering
\includegraphics[height=110pt]{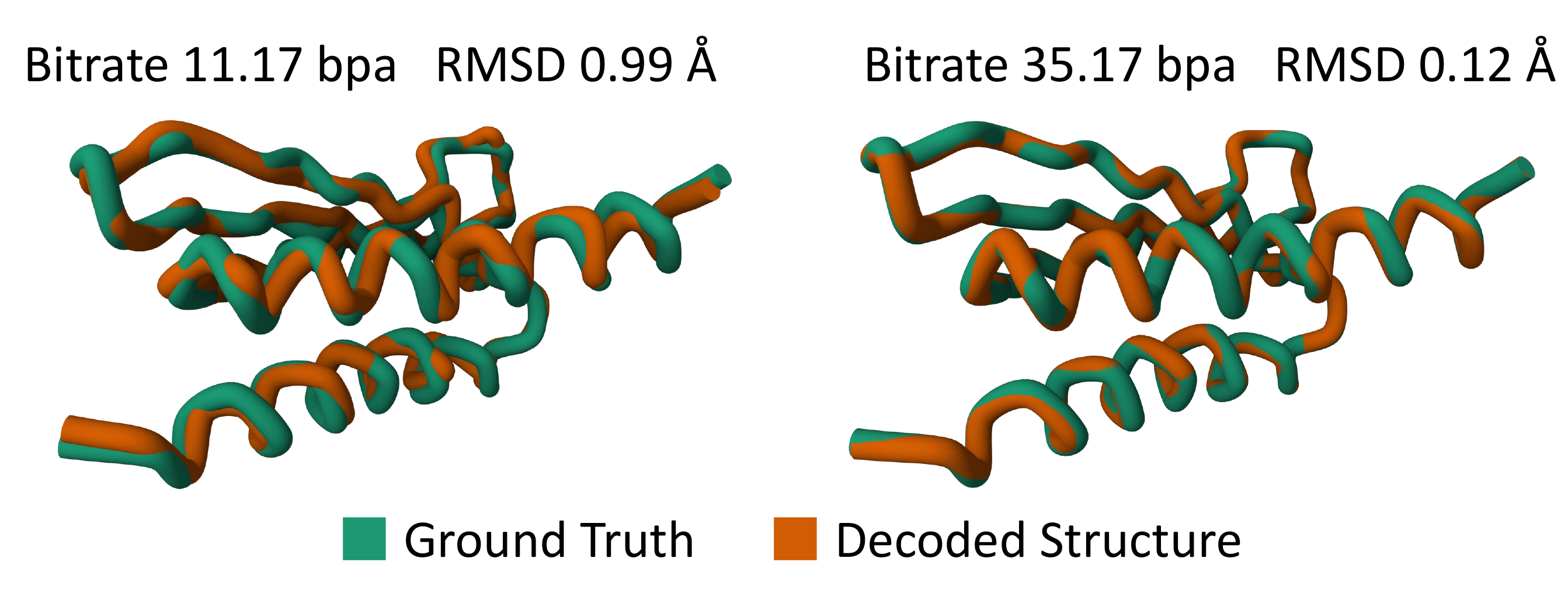}

\caption{Decoded protein structure examples.}
\label{fig:protein_examples}
\end{subfigure}
\caption{Non-cherrypicked examples of video frames and protein structures encoded using RECOMBINER.}
\label{fig:combiner_decoded_examples}
\end{figure}
\par
\textbf{Compressing high-resolution data with RECOMBINER:}
An advantage of patching is that we can compress and fine-tune INRs and latent positional encodings of all patches in parallel.
Unfortunately, compressing  $P$ patches in parallel using COMBINER's procedure is suboptimal since the information content between patches might vary significantly.
However, by carefully permuting the weights \emph{across} the patches' representations, we can 1) adaptively allocate bits to each patch to compensate for the differences in their information content and 2) enforce the same coding budget across each parallel thread to ensure consistent coding times.
Concretely, we stack representations of each patch in a matrix at each level of the hierarchical model.
For example, in our three-level model, we set
\begin{align}
\label{eq:stacked_unpermuted_representations}
\mH^{(0)}_{\pi, :} = [\rvh_\rvw^{(\pi)}, \rvh_\rvz^{(\pi)}],
\quad \mH^{(1)}_{g, :} = \rvhbar_\rvw^{(g)},
\quad \mH^{(2)} = \rvhbbar_\rvw, 
\end{align}
where we use slice notation to denote the $i$th row as $\mH_{i, :}$ and the $j$th column as $\mH_{:, j}$.
Furthermore, let $S_n$ denote the set of permutations on $n$ elements.
Now, at each level $\ell$, assume $\mH^{(\ell)}$ has $\mathcal{C}_\ell$ columns and $\mathcal{R}_\ell$ rows. We sample a single within-row permutation $\kappa$ uniformly from $S_{\mathcal{C}_\ell}$ and for each column of $\mH^{(\ell)}$ we sample an across-rows permutation $\alpha_j$ uniformly from $S_{\mathcal{R}_\ell}$ elements.
Then, we permute $\mH^{(\ell)}$ as $\widetilde{\mH}^{(\ell)}_{i, j} = \mH^{(\ell)}_{\alpha_j(i), \kappa(j)}$.
Finally, we split the $\mH^{(\ell)}$s into blocks row-wise and encode and fine-tune each row in parallel.
We illustrate the above procedure in \Cref{fig:illustration_hyper_prior}b.
\section{Experimental Results}
\label{sec:experiments}
This section presents an empirical study of COMBINER's compression performance on image, audio, video, and 3D protein structure data.
I also compare its performance to contemporary classical and learned lossy data compression methods and demonstrate that COMBINER achieves strong performance across all modalities.
Then, I present extensive ablation studies illuminating how each implementation detail I presented in \cref{sec:combiner_method,sec:recombiner} improves COMBINER's performance.
For convenience, throughout this section:
\begin{enumerate}
\item 
I refer to a model following the basic COMBINER framework (\cref{sec:combiner_method}) only, implemented using a fully-factorised Gaussian posterior and coding distribution directly over the weights as ``\textbf{COMBINER}.''
\item I refer to a model that also implements the improvements I describe in \cref{sec:recombiner} as ``\textbf{RECOMBINER}'', which stands for ``robust and enhanced COMBINER'' \citep{he2024recombiner}.
\end{enumerate}
\subsection{Data Compression Across Modalities}
Unless I state otherwise, we use a 4-layer, 32-hidden unit SIREN network \citep{sitzmann2020implicit} using Fourier feature embeddings for the input coordinates \citep{tancik2020fourier} as the INR architecture. For image compression with RECOMBINER, we used a small 3-layer convolution network as the upsampling network $\phi$.
\par
\textbf{Images:} We evaluated COMBINER and RECOMBINER on the CIFAR-10 \citep{krizhevsky2009cifar} and Kodak \citep{kodak1993dataset} image datasets; I present their rate-distortion (RD) performance in \cref{fig:rd_images}.
We also compared them against contemporary INR and variational autoencoder (VAE)-based methods, as well as VTM \citep{VTM}\footnote{\url{https://vcgit.hhi.fraunhofer.de/jvet/VVCSoftware_VTM/-/tree/VTM-12.0?ref_type=tags}}, BPG \citep{bpg} and JPEG2000.
\par
We observe that COMBINER exhibits competitive performance on the CIFAR-10 dataset, on par with COIN++ and marginally lower than MSCN, despite being conceptually much simpler. 
Furthermore, it achieves impressive performance on the Kodak dataset, surpassing JPEG2000 and other INR-based codecs.
This is likely in part due to our method not requiring an expensive meta-learning loop \citep{dupont2022coin++, strumpler2022implicit,schwarz2023modality}, which would involve computing second-order gradients during training. 
Since COMBINER avoids this cost, we can compress the whole high-resolution image using a single MLP network, thus the model can capture global patterns in the image. 
\par
Taking this further, RECOMBINER displays remarkable performance on CIFAR-10, especially at low bitrates, outperforming even VAE-based codecs. 
On Kodak, it outperforms most INR-based codecs and is competitive with the more complex VC-INR \citep{schwarz2023modality}.
Finally, while RECOMBINER still falls behind VAE-based codecs, it significantly reduces the performance gap.
% \par
As we can clearly see from the image compression experiments, the improvements of \cref{sec:recombiner} significantly improve COMBINER's performance.
Hence, for the experiments on other data modalities, I only present RECOMBINER's performance.
\par
\textbf{Audio:} We evaluated RECOMBINER on the LibriSpeech~\citep{LibriSpeech} dataset. 
In \cref{fig:rd_audio}, we depict its RD curve on the full test set alongside the curves of VC-INR, COIN++, and MP3.
We can see RECOMBINER outperforms both COIN++ and MP3 and matches with VC-INR.
\par
\textbf{Video:} We evaluated RECOMBINER on UCF-101 action recognition dataset \citep{soomro2012ucf101}, following \citet{schwarz2023modality}'s experimental setup. 
However, as they do not report their train-test split and due to the time-consuming encoding process of our approach, we only benchmark our method against H.264 and H.265 on 16 randomly selected video clips.  
\Cref{fig:rd_video} shows RECOMBINER achieves comparable performance to the classic domain-specific codecs H.264 and H.265, especially at lower bitrates. 
However, a gap exists between our approach and H.264 and H.265 when configured to prioritise quality. 
\Cref{fig:video_examples} shows non-cherrypicked frames from a video compressed with RECOMBINER at two different bitrates and its reconstruction errors.
\par
\textbf{3D Protein Structure:}
To further illustrate the applicability of RECOMBINER, we used it to compress the 3D coordinates of C$\alpha$ atoms in protein fragments.
We take domain-specific lossy codecs as baselines, including Foldcomp \citep{kim2023foldcomp}, PDC \citep{zhang2023pdc} and PIC \citep{staniscia2022pic}.
Surprisingly, as shown in \cref{fig:rd_protein}, RECOMBINER's performance is competitive with highly domain-specific codecs.
Furthermore, it allows us to tune its rate-distortion performance, whereas the baselines only support a certain compression rate.  
Since the experimental resolution of 3D structures is  
typically between 1-3 \AA{} \citep{PDBstatistics}, RECOMBINER could help reduce the increasing storage demand for protein structures without losing key information. \Cref{fig:protein_examples} shows non-cherry-picked examples compressed with our method.
\subsection{Analysis and Ablation Studies for COMBINER} \label{sec:combiner_analysis_and_ablation}
\noindent
\textbf{Model Visualizations:}  
To provide more insight into COMBINER's behaviour, I visualise the parameters and information content on the second hidden layer of two small 4-layer variational INRs trained on two CIFAR-10 images with $\beta = 10^{-5}$. 
I use the KL divergence as an estimate of their coding cost and do not encode the weights with A* coding or perform fine-tuning.
\begin{figure*}[t]
\captionsetup[subfigure]{justification=centering}
 \centering
 \begin{subfigure}[t]{0.24\linewidth}
 \includegraphics[scale=0.14, clip, trim=3.5cm 2cm 0.2cm 2.5cm]{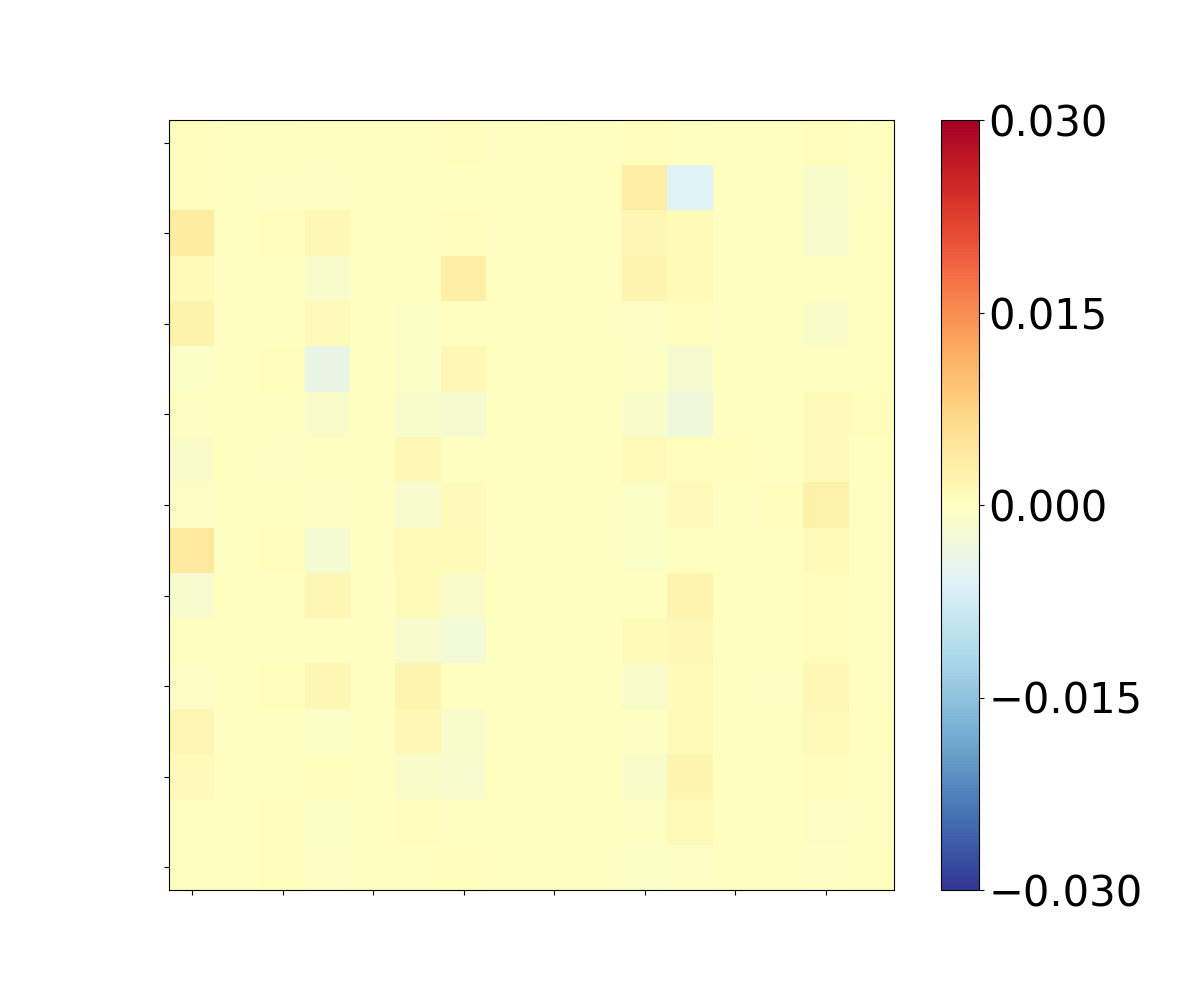}
 \caption*{Prior mean}
 \end{subfigure}
 \hfill
 \begin{subfigure}[t]{0.24\linewidth}
 \includegraphics[scale=0.14, clip, trim=3.5cm 2cm 0.2cm 2.5cm]{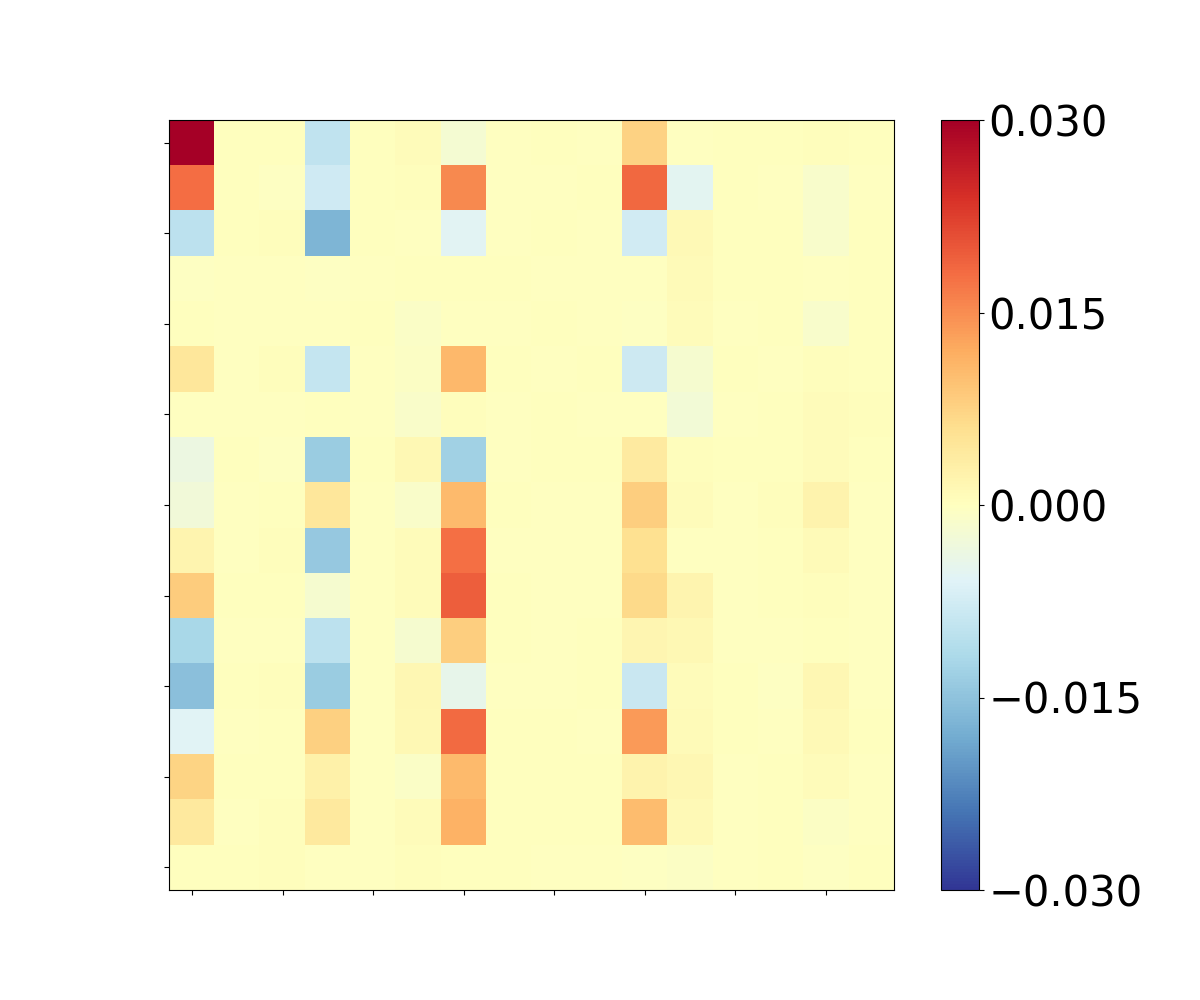}
\caption*{Posterior mean, \\ image 1}
 \end{subfigure}
 \hfill
 \begin{subfigure}[t]{0.24\linewidth}
 \includegraphics[scale=0.14, clip, trim=3.5cm 2cm 0.2cm 2.5cm]{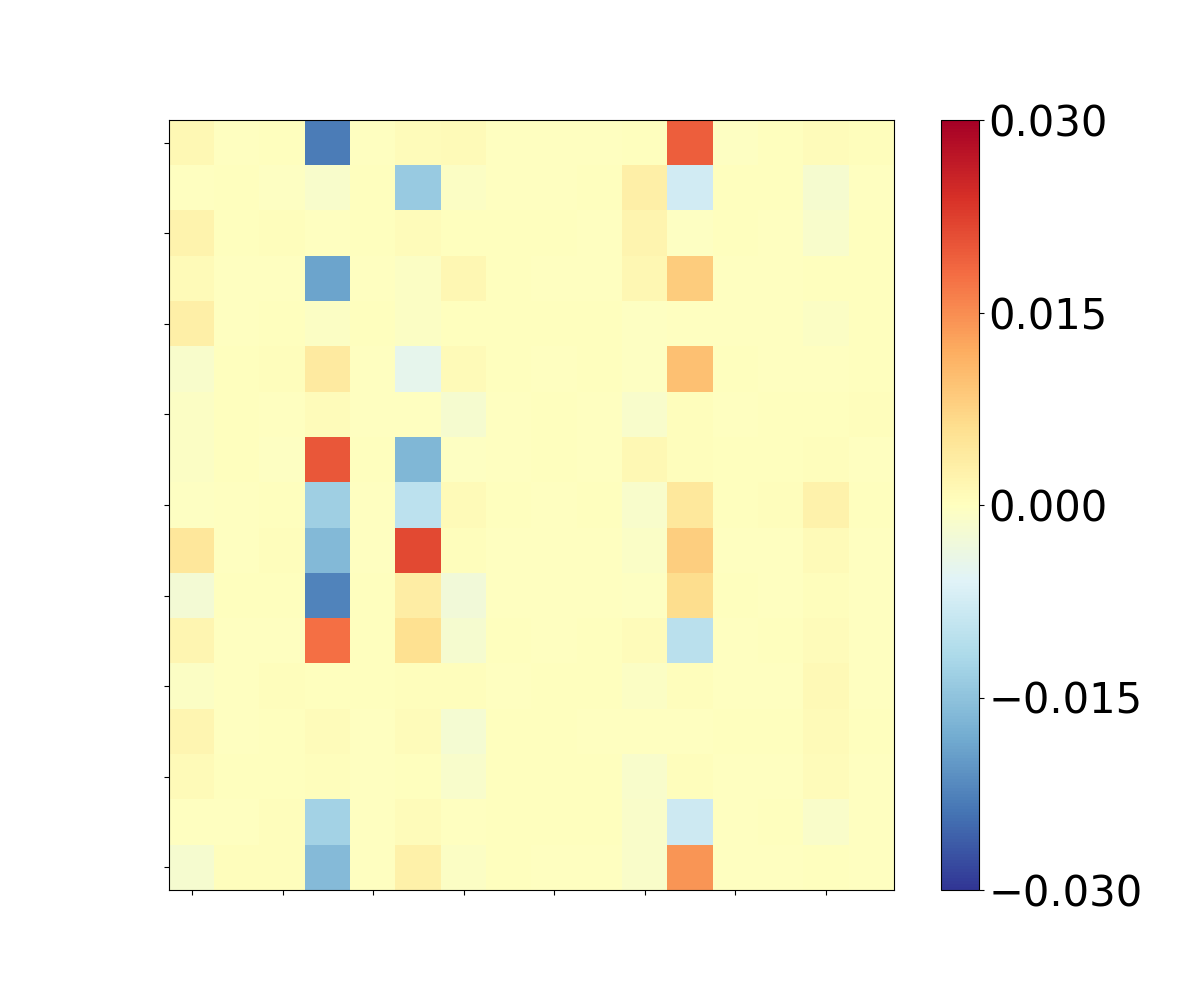}
\caption*{Posterior mean, \\ image 2}
 \end{subfigure}
 \hfill
 \begin{subfigure}[t]{0.24\linewidth}
 \includegraphics[scale=0.14, clip, trim=3.5cm 2cm 0.2cm 2.5cm]{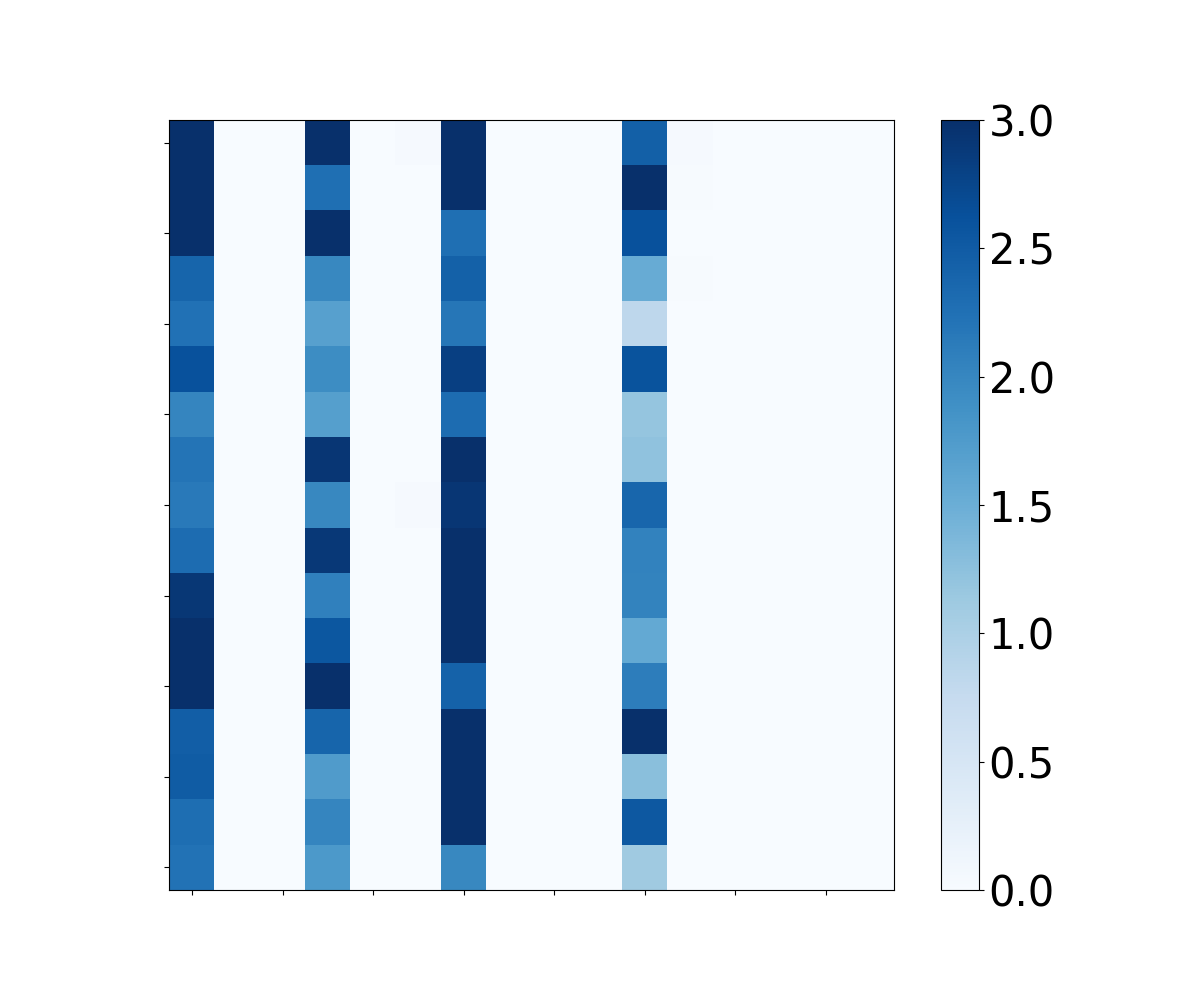}
\caption*{$\rm{KL}$ / bits, image 1 \ }
 \end{subfigure}%
 \vspace{0.4cm}
 \begin{subfigure}[t]{0.24\linewidth}
 \includegraphics[scale=0.14, clip, trim=3.5cm 2cm 0.2cm 2.5cm]{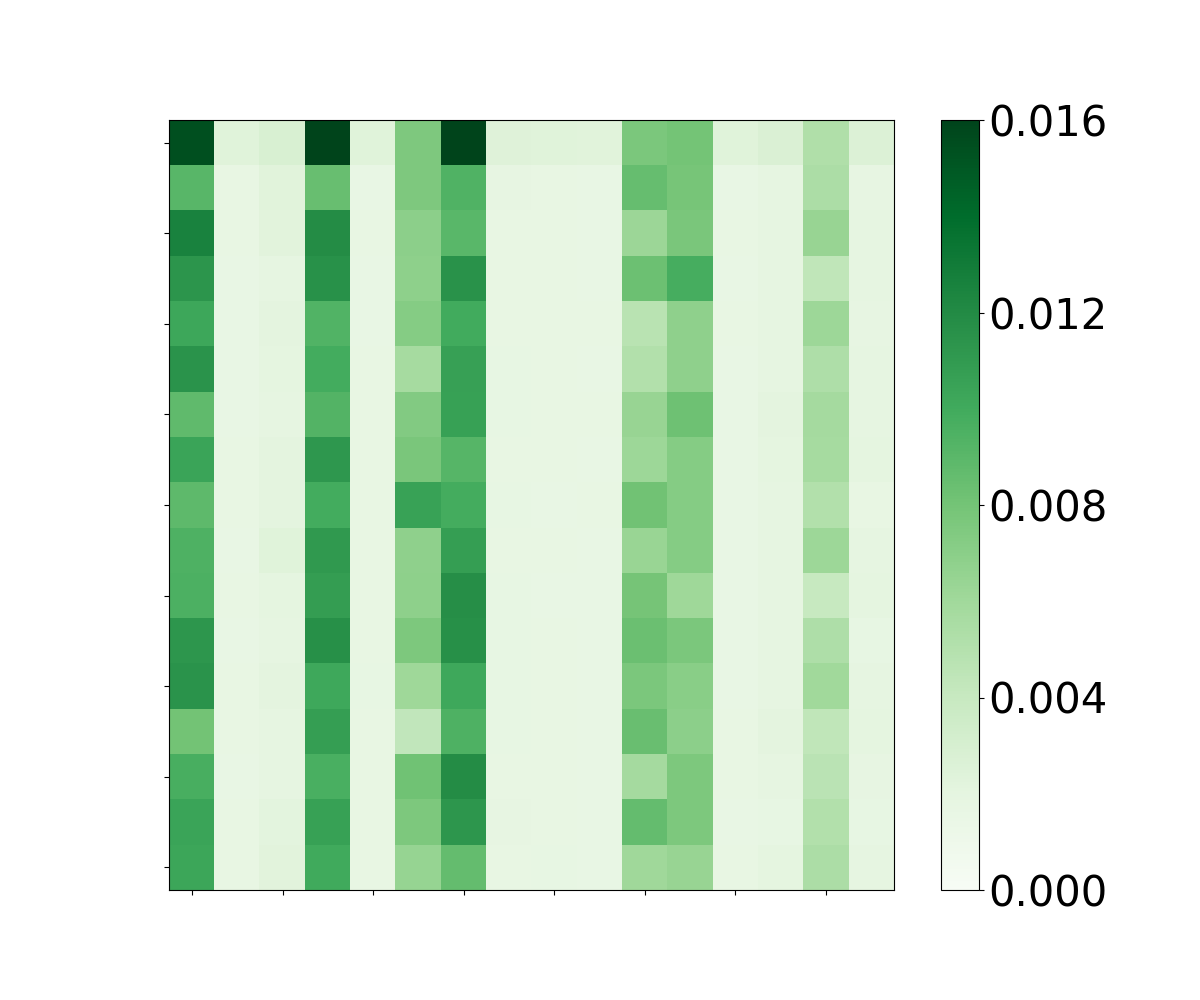}
\caption*{Prior s.d.\ }
 \end{subfigure}
\hfill
 \begin{subfigure}[t]{0.24\linewidth}
 \includegraphics[scale=0.14, clip, trim=3.5cm 2cm 0.2cm 2.5cm]{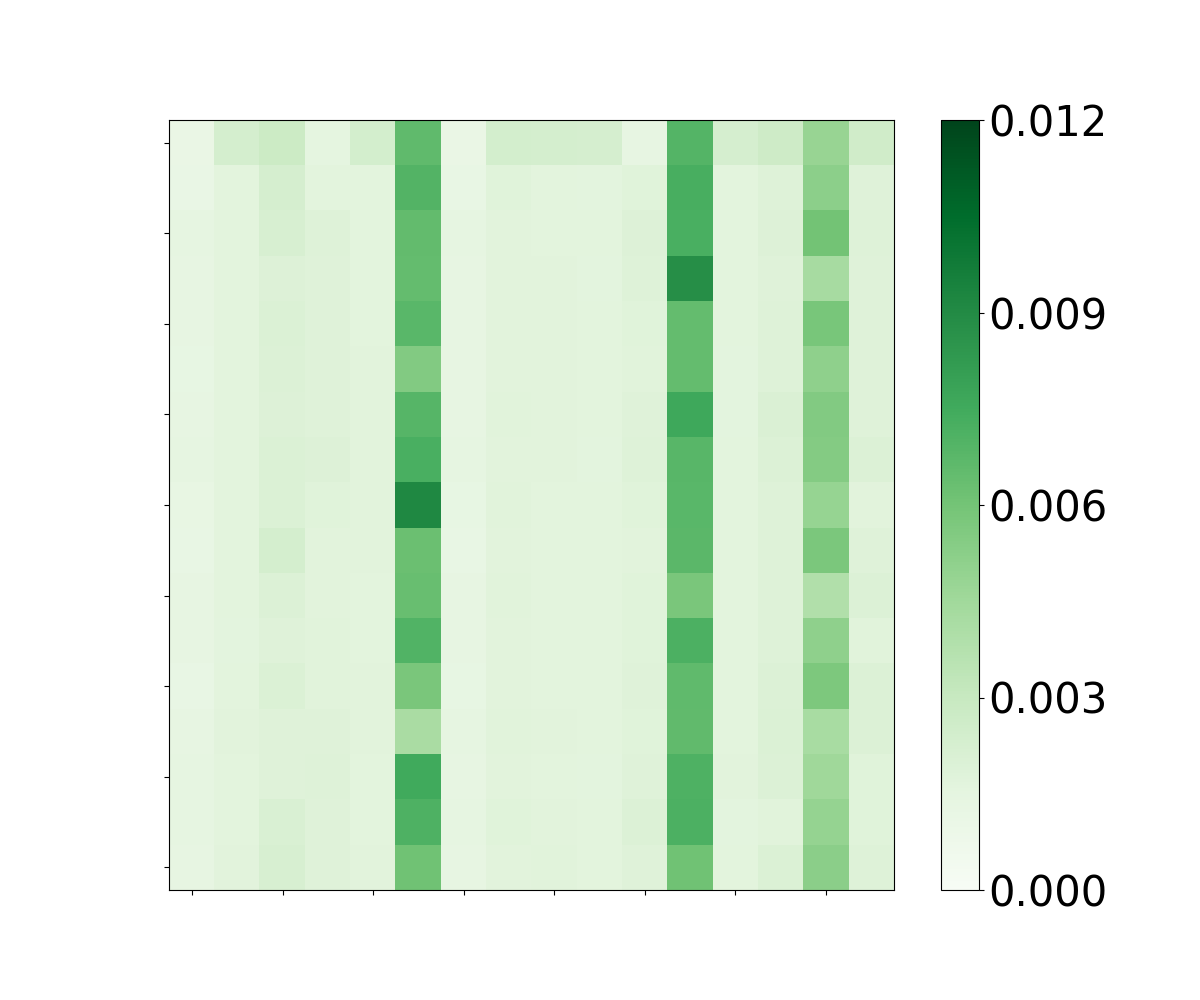}
 \caption*{Posterior s.d., \\ image 1}
 \end{subfigure}
\hfill
 \begin{subfigure}[t]{0.24\linewidth}
 \includegraphics[scale=0.14, clip, trim=3.5cm 2cm 0.2cm 2.5cm]{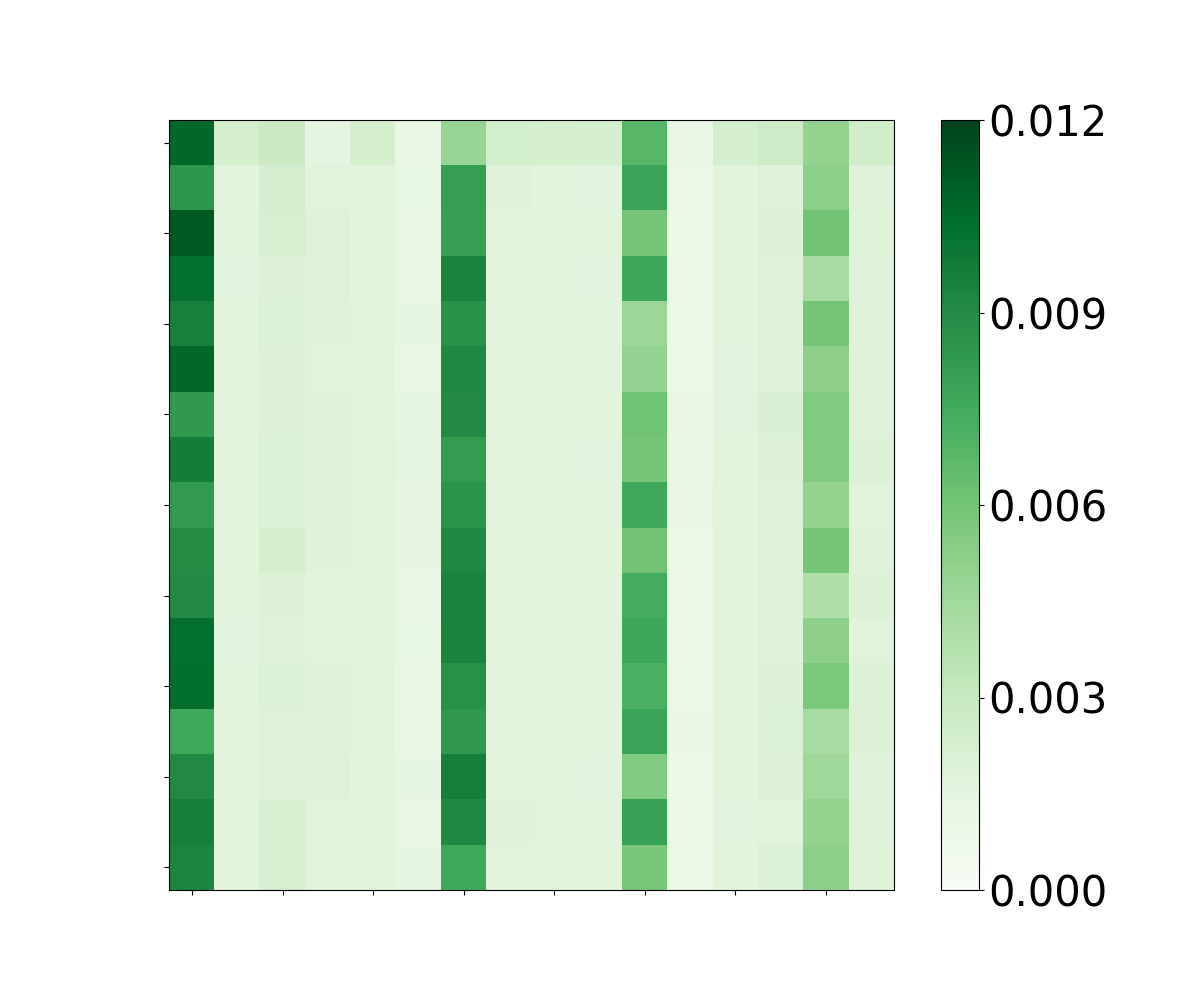}
 \caption*{Posterior s.d., \\ image 2}
 \end{subfigure}
\hfill
 \begin{subfigure}[t]{0.24\linewidth}
 \includegraphics[scale=0.14, clip, trim=3.5cm 2cm 0.2cm 2.5cm]{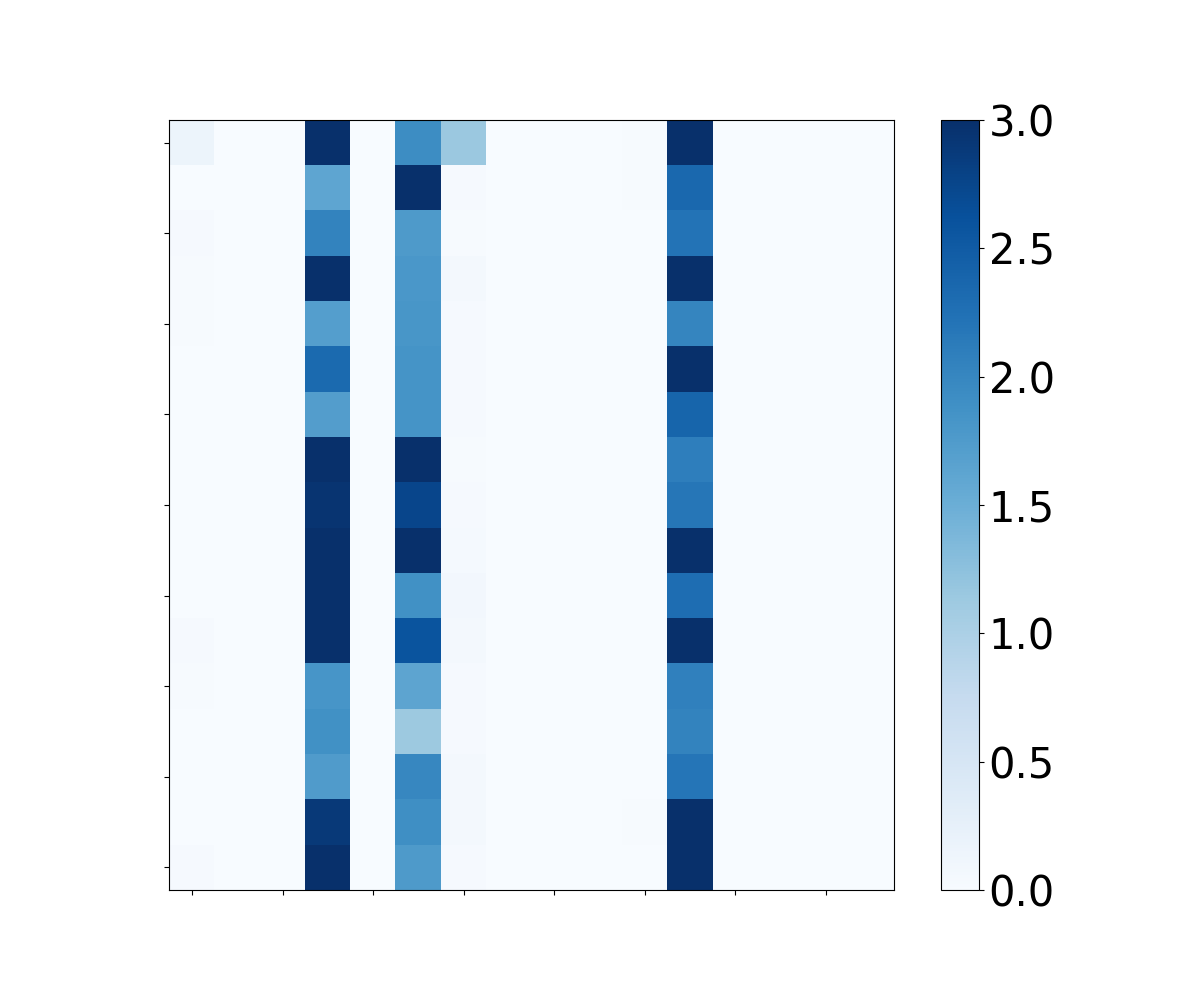}
\caption*{$\rm{KL}$ / bits, image 2 \ }
\end{subfigure}
\vspace{5mm}
 \caption[Visualisations of variational INR parameters and their information content.]{
 Visualisations of the weight prior, posterior and information content of a variational INR trained on two CIFAR-10 images.
 We focus on the INR's weights connecting the first and second hidden layers.
 Each heatmap is $17 \times 16$ because both layers have 16 hidden units, and we concatenated the weights and the biases (last row). 
 The abbreviation \textbf{s.d.}\ stands for standard deviation.
 }
\label{fig:combiner_visualisations}
\vspace{-0.35cm}
\end{figure*}
\par
\Cref{fig:combiner_visualisations}, visualises the parameters $\rvmu_p$ and $\rvsigma_p$ of the coding distribution $Q_\rvw = \Normal(\rvmu_p, \rvsigma_p^2)$ we learned using the training procedure I describe in \cref{sec:training_combiner}.
In the left column, the variational parameters of two distinct images in the second and third columns and the KL divergence $\KLD{P_{\rvw \mid \Dataset}}{Q_{\rvw}}$ in the rightmost column.
Since this layer has 16 hidden units, the weight matrix of parameters has a $17\times16$ shape, where weights and bias are concatenated; the last row represents the bias.
Interestingly, there are seven ``active'' columns within $\bm{\sigma}_p$, indicating that only seven hidden units of this layer would be activated for signal representation at this rate point. 
For instance, when representing image 1, which is randomly selected from the CIFAR-10 test set, four columns are activated for representation. 
This activation is evident in the four blue columns within the KL map, which require a few bits to transmit the sample of the posterior distribution.
Similarly, three hidden units are engaged in the representation of image 2. 
As their variational Gaussian distributions have nearly zero variance, the posterior distributions at these activated columns approach a Dirac delta distribution. 
Therefore, this study demonstrates COMBINER's advantage over previous INR-based methods: by optimising the rate-distortion objective, it can adaptively activate or prune network parameters.
\begin{figure}[t]
\centering
\begin{subfigure}[t]{0.47\textwidth}
\includegraphics[width=0.9\linewidth]{6-COMBINER/img/tikz/combiner_ablation.tikz}
\caption{Ablation study on CIFAR-10 verifying the effectiveness of the fine-tuning and prior learning procedures.
}  
\label{fig:combiner_ablation_study}
\end{subfigure}\hspace{5mm}%
\begin{subfigure}[t]{0.47\textwidth}%
\includegraphics[width=0.95\linewidth]{6-COMBINER/img/tikz/combiner_iternum.tikz}
\caption{COMBINER's performance improvement as a function of the number of fine-tuning steps.}
\label{fig:combiner_iternum}
\end{subfigure}%
\caption{Ablation study of COMBINER and a study of posterior refinement.}
\end{figure}
\par
\textbf{Ablation Studies:}  
We conducted ablation studies on the CIFAR-10 dataset to verify the effectiveness of learning the coding distribution (\cref{sec:training_combiner}) and posterior refinement (\cref{sec:combiner_posterior_refinement}).
In the first ablation study, instead of learning the coding distribution parameters, we followed the methodology of \citep[p.\ 73;]{havasi2021advances} and use a layer-wise zero-mean isotropic Gaussian prior $Q_{\rvw^{[\ell]}} = \Normal(0, \rvsigma_\ell^2 I)$, where $Q_{\rvw^{[\ell]}}$ is the coding distribution for the weights of the $\ell$th hidden layer.
We learned the $\rvsigma_\ell$'s jointly with the posterior parameters by optimising \cref{eq:combiner_rd_objective} using gradient descent and encoding them at 32-bit precision alongside the posterior weight samples we encoded with step-limited A* sampling. 
In the second ablation study, we omitted the fine-tuning steps between encoding blocks with step-limited A* sampling, i.e.\ we never corrected for bad-quality approximate samples.
In both experiments, we compressed each block using 16 bits.
Finally, as a reference, we also compare with the theoretically optimal scenario:
we draw an exact sample from each block's posterior between refinement steps instead of encoding an approximate sample and estimating the sample's codelength with the block's KL divergence.
\par
We compared the results of these experiments with the whole pipeline (\cref{sec:combiner_in_practice}) using the techniques mentioned above in \cref{fig:combiner_ablation_study}.
We found that both learning the coding distribution and posterior refinement contribute significantly to COMBINER's performance.
In particular, fine-tuning the posteriors is more effective at higher bitrates, while learning the coding distribution yields a consistent 4dB gain in PSNR across all bitrates. 
Finally, fine-tuning cannot completely compensate for the occasional bad approximate samples that step-limited A* sampling yields, as there is a consistent 0.8 -- 1.3dB discrepancy between COMBINER's and the theoretically optimal performance.
\par
\textbf{Time Complexity:} 
COMBINER's encoding procedure is slow, as it requires several thousand gradient descent steps to infer the parameters of the INR's weight posterior and thousands more for progressive fine-tuning. 
To get a better understanding of COMBINER's practical time complexity, we evaluated its coding time on both the CIFAR-10 and Kodak datasets at different rates; I report our findings in \cref{tab:combiner_encoding_time_cifar,tab:combiner_encoding_time_kodak}. 
Observe that it can take between 13 minutes (0.91 bpp) and 34 minutes (4.45 bpp) to encode 500 CIFAR-10 images in parallel with a single A100 GPU, including posterior inference (7 minutes) and progressive fine-tuning.
Note that the fine-tuning takes longer for higher bit rates as the weights are partitioned into more groups, and each weight has higher individual information content. 
To compress high-resolution images from the Kodak dataset, the encoding time varies between 21.5 minutes (0.070 bpp) and 79 minutes (0.293 bpp). 
\begin{table}[t]
\centering
\scalebox{0.88}{
\begin{tabular}{@{}ccccc@{}}
\toprule
\multirow{2}{*}{\textbf{bit-rate}} & \multicolumn{3}{c}{\textbf{Encoding (500 images, GPU A100 80G)}} & \multirow{2}{*}{\makecell{\textbf{Decoding} \\ \textbf{(1 image, CPU)}}} \\ \cmidrule(lr){2-4}
& \textbf{Learning Posterior} & \textbf{REC + Fine-tuning} & \textbf{Total} & \\ \midrule
0.91 bpp & \multirow{5}{*}{$\sim$7 min} & $\sim$6 min & $\sim$13 min & 2.06 ms \\
1.39 bpp & & $\sim$9 min & $\sim$16 min & 2.09 ms \\
2.28 bpp & & $\sim$14 min 30 s & $\sim$21 min 30 s & 2.86 ms \\
3.50 bpp & & $\sim$21 min 30 s & $\sim$28 min 30 s & 3.82 ms \\
4.45 bpp & & $\sim$27 min & $\sim$34 min & 3.88 ms \\ \bottomrule
\end{tabular}}
\vspace{0.2cm}
\caption{The encoding time and decoding time of COMBINER on CIFAR-10.}
\label{tab:combiner_encoding_time_cifar}
\end{table}
\begin{table}[t]
\centering
\scalebox{0.88}{
\begin{tabular}{@{}ccccc@{}}
\toprule
\multirow{2}{*}{\textbf{bit-rate}} & \multicolumn{3}{c}{\textbf{Encoding (1 image, GPU A100 80G)}} & \multirow{2}{*}{\makecell{\textbf{Decoding} \\ \textbf{(1 image, CPU)}}} \\ \cmidrule(lr){2-4}
& \textbf{Learning Posterior} & \textbf{REC + Fine-tuning} & \textbf{Total} & \\ \midrule
0.07 bpp & \multirow{3}{*}{$\sim$9 min} & $\sim$12 min 30 s & $\sim$21 min 30 s & 348.42 ms \\
0.11 bpp & & $\sim$18 mins & $\sim$27 min & 381.53 ms \\
0.13 bpp & & $\sim$22 min & $\sim$31 min & 405.38 ms \\
0.22 bpp & \multirow{2}{*}{$\sim$11 min} & $\sim$50 min & $\sim$61 min      & 597.39 ms \\
0.29 bpp & & $\sim$68 min & $\sim$79 min & 602.32 ms \\ \bottomrule
\end{tabular}
}
\vspace{0.2cm}
\caption{The encoding time and decoding time of COMBINER on Kodak.}\label{tab:combiner_encoding_time_kodak}
\end{table}
To assess the effect of the fine-tuning procedure's length, we randomly selected a CIFAR-10 image and encoded it using the whole COMBINER pipeline but varied the number of fine-tuning steps between 2148 and 30260; I report the results of the experiment in \Cref{fig:combiner_iternum}.
We see that running the fine-tuning process beyond a certain point has diminishing returns.
In particular, while we used around 30k iterations in our other experiments, just using 3k iterations would sacrifice a mere 0.3 dB in the reconstruction quality while saving 90\% of the tuning time.
\par
On the other hand, COMBINER has a fast decoding speed since once we decode the compressed weight sample, we can reconstruct the data with a single forward pass through the network at each coordinate, which can be easily parallelised.
Specifically, the decoding time of a single CIFAR-10 image is between 2 ms and 4 ms using an A100 GPU and less than 1 second for a Kodak image. 
\subsection{Analysis and Ablation Studies for RECOMBINER}
This section showcases RECOMBINER's robustness to model size and the effectiveness brought about by each of the techniques I introduced in \cref{sec:recombiner}. 
\begin{figure}[t]
\centering
\begin{subfigure}{0.31\textwidth}
\centering
\includegraphics[width=\textwidth, trim={0, 0.2cm, 0, 1.6cm}, clip]{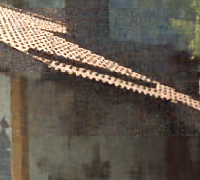}
\caption{w/o positional encodings; \\bitrate 0.287 bpp; PSNR 25.62 dB. }
\end{subfigure}
\begin{subfigure}{0.31\textwidth}
\centering
\includegraphics[width=\textwidth, trim={0, 0.2cm, 0, 1.6cm}, clip]{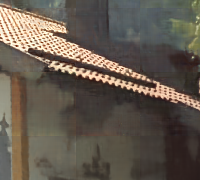}
\caption{with positional encodings; \\bitrate 0.316 bpp; PSNR 26.85 dB.  }
\end{subfigure}
\begin{subfigure}{0.31\textwidth}
\centering
\includegraphics[width=\textwidth, trim={0, 0.2cm, 0, 1.6cm}, clip]{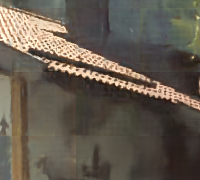}
\caption{with positional encodings;\\bitrate 0.178 bpp; PSNR 25.05 dB.  }
\end{subfigure}
\caption[Ablation study for learned positional encodings.]{Comparison between \texttt{kodim24} details compressed with and without learnable positional encodings. (a)(b) have similar bitrates, and (a) and (c) have similar PSNRs.
}\label{fig:compare_details_w_wo_pe}
%\vspace{-10pt}
\end{figure}
\begin{figure}[t]
\centering
\includegraphics[width=\textwidth]{6-COMBINER/img/tikz/robustness_and_ablation_study.tikz}
\caption[Ablation study for RECOMBINER's improvements.]{(a) RD performances of COMBINER and RECOMBINER with different numbers of hidden units. (b)(c) Ablation studies on CIFAR-10 and Kodak. \textbf{LR:} linear reparameterization; \textbf{PE:} positional encodings;  \textbf{HM:} hierarchical model; \textbf{RP:} random permutation across patches.}
\end{figure}
\par
\textbf{Positional encodings facilitate local deviations:} \Cref{fig:compare_details_w_wo_pe} compares images obtained by RECOMBINER with and without positional encodings (\cref{sec:recombiner_positional_embeddings}) at matching bitrates and PSNRs. 
As we can see, positional encodings preserve intricate details in fine-textured regions while preventing noisy artifacts in other regions of the patches, making RECOMBINER's reconstructions more visually pleasing. 
\par
\textbf{RECOMBINER is more robust to model size:} Using the same INR architecture, \cref{fig:robustness} shows COMBINER and RECOMBINER's RD curves as we vary the number of hidden units. 
RECOMBINER displays minimal performance variation and also consistently outperforms COMBINER. 
Based on experimental evidence \citep[Appendix D;]{he2024recombiner}, this phenomenon is likely due to the linear weight reparametrisation allowing it to prune its weight representations more flexibly.
\par
\textbf{Ablation study:} In \cref{fig:ablations_cifar,fig:ablations_kodak}, I report the results of ablation studies for the linear reparametrisation, positional encodings, hierarchical model, and permutation strategy on
CIFAR-10 and Kodak, with five key takeaways:
\begin{enumerate}[leftmargin=*]
\item Linear weight reparameterisation consistently improves performance on both datasets, yielding up to 4dB gain on CIFAR-10 at high bitrates and over 0.5 dB gain on Kodak in PSNR.
\item Learnable positional encodings provide more substantial advantages at lower bitrates. 
On CIFAR-10, the encodings contribute up to 0.5 dB gain when the bitrate falls below 2 bpp. 
On Kodak, the encodings provide noteworthy gains of 2 dB at low bitrates and 1 dB at high bitrates.
\item Surprisingly, the hierarchical model \emph{without positional encodings} can degrade performance.  
This is possibly because directly applying the hierarchical model poses challenges in optimising \Cref{eq:combiner_rd_objective}.
A potential solution is to warm up the rate penalty $\beta$ akin to the usual strategy for hierarchical VAEs \citep{sonderby2016ladder}.
\item However, positional encodings appear to consistently alleviate this optimisation difficulty, yielding 0.5 dB gain when used with hierarchical models. %
\item The permutation strategy provides significant gains of 0.5 dB at low bitrates and more than  1.5 dB at higher bitrates.
\end{enumerate}
\section{Conclusions, Limitations and Future Directions}
\label{sec:combiner_conclusions}
\noindent
This chapter explored the use of implicit neural representations for data compression.
It proposed using \textit{compression with Bayesian implicit neural representations} (COMBINER), the first INR-based compression framework that allows us to optimise the rate-distortion function directly.
Then, the networks are encoded using a relative entropy coding algorithm; for the experiments in this chapter, we used step-limited A* sampling (\cref{alg:global_a_star}, \cref{sec:approximate_sampling}).
\par
Then, the chapter also proposed several nontrivial improvements to the framework, encompassing the linear reparametrisation for the network weights, learnable positional encodings, and expressive hierarchical Bayesian models for high-resolution signals. 
Finally, with extensive experimental evidence, I demonstrated that COMBINER sets a new state-of-the-art on low-resolution images at low bitrates and consistently delivers strong results across other data modalities. 
\par
The most significant limitation of the work is the complexity of encoding time. 
A possible avenue for solving this issue is to reduce the number of parameters to optimise over and switch from inference over weights to modulations using, e.g.\ FiLM layers \citep{perez2018film}, as is done in other INR-based works.
A second limitation is that while compressing with patches enables parallelisation and higher robustness, it is suboptimal as it leads to block artifacts, as can be seen in \cref{fig:compare_details_w_wo_pe}.
Third, the approximate samples given by step-limited A* sampling significantly impact the method's performance, e.g.\, by requiring more fine-tuning.

%% file: 7-Discussion/discussion.tex
%!TEX root = ../thesis.tex
%*******************************************************************************
%****************************** Second Chapter *********************************
%*******************************************************************************

\chapter{Conclusions and Future Work}
\label{chapter:conclusion}
This thesis contributes to the theoretical, algorithmic, and applied development of relative entropy coding and the broader data compression, sampling theory, and information theory literature.
It is also the first work that comprehensively discusses how relative entropy coding fits into the greater picture of contemporary data compression, what its tight fundamental limits are, how one constructs the algorithms, how one can exploit structure to speed up the algorithms and how one applies relative entropy coding to modern learned data compression problems.
This chapter summarises these contributions and concludes the thesis by discussing the most pressing open questions and exciting future directions.
\section{Summary of Contributions}
The thesis contributes to the information theory and data compression literature and related areas in four major ways, which I summarise below.
\subsection{Tighter Fundamental Limits on the Communication and Computational Complexity of Relative Entropy Coding}
First, \cref{chapter:fundamental_limits} established new, tighter fundamental limits on the communication complexity of relative entropy coding both in the one-shot and (\cref{thm:csd_rec_lower_bound}) asymptotic cases  (\cref{thm:second_order_behaviour_of_csd}).
I achieved these results by introducing and developing several new analytical tools. I introduced a new statistical distance, which I call the channel simulation divergence and exposed its fundamental role in the theory.
Furthermore, I proved \cref{thm:second_order_behaviour_of_csd} by observing \added[id={PL}, comment={PL/email/5}]{that} an appropriately normalised version of the log-Radon-Nikodym derivative obeys a central limit theorem.
Perhaps the most exciting aspect of my results is that they do not follow the ``standard'' information theory approach: for example, at no point in the thesis do I rely on arguments using \replaced[id={PL}, comment={PL/email/6}]{typicality}{typically}.
I hope these ideas can penetrate the wider information theory literature and help establish new, simplified analyses of related problems.
\par
I also introduced the abstract notion of selection sampling and exhibited a fundamental lower bound on their computational efficiency (\cref{thm:selection_sampler_runtime_lower_bound}). 
The lower bound is somewhat surprising since causal rejection samplers (algorithms that do not return to samplers they previously examined) already achieve this lower bound.
Thus, as this thesis examines, their advantages must come from the specific nature or structure of the problem we wish to solve with them, such as using them for relative entropy coding or approximate sampling.
\subsection{The Poisson Process Model of Relative Entropy Coding}
\par
Second, \cref{chapter:rec_with_pp} contributed to the understanding of the Poisson processes' central role in constructing sampling algorithms; \citet{maddison2016poisson} called this the ``Poisson process model of Monte Carlo.''
My two most significant results in this chapter are 1) the derivation of a new sampling algorithm that I call the greedy Poisson rejection sampler (GPRS) and its analysis, which shows that it induces a relative entropy coding algorithm, and 2) illuminating the serendipitous alignment of desiderata in the construction of relative entropy coding and approximate selection samplers.
\par
In particular, GPRS demonstrates that Poisson process restriction also gives rise to a sampling algorithm besides thinning and mapping, which give rise to rejection sampling and A* sampling, respectively.
Furthermore, I also explained how these sampling algorithms can be parallelised based on the superposition theorem.
Finally, I also provided new proofs for the runtimes of rejection sampling and A* sampling using the theory of Poisson processes.
\subsection{Fast Branch-and-bound Samplers for Unimodal Density Ratios}
Third, \cref{chapter:branch_and_bound} is where I reap the benefits of the Poisson process view of sampling algorithms in two ways.
First, they provide the high-level observation that we can recast random variate simulation as a search problem over a random collection of points, which motivated the use of ideas from branch-and-bound search.
Second, I rely on the mature theory of Poisson processes that ensured the correctness of my constructions. 
\par
Results-wise, the highlights of these chapters are \cref{thm:bnb_a_star_runtime,thm:bnb_gprs_runtime}, which show that the runtimes of the branch-and-bound variants of A* sampling and greedy Poisson rejection sampling improve exponentially and superexponentially compared to their general-purpose variants, respectively; as well as \cref{thm:bnb_a_star_codelength,thm:bnb_gprs_codelength}, which show that these algorithms also induce relative entropy coding algorithms.
\par
I also show how we can leverage the theory of stopping times and martingales to analyse sampling algorithms.
Finally, I also briefly discussed how these branch-and-bound samplers generalise to higher-dimensional spaces via splitting up space along supporting hyperplanes.
\subsection[Compression with Bayesian Implicit Neural Representations]{Compression with Bayesian Implicit Neural \texorpdfstring{\\}{} Representations}
Finally, \Cref{chapter:combiner} rounds out the thesis by exposing an exciting application for relative entropy coding to implicit neural representations.
This chapter introduced variational / Bayesian implicit neural representations and discussed how they fit into the nonlinear transform coding framework.
I derived from first principles both the loss function that individual INRs optimise, as well as how their coding distribution can be optimised using lossy source coding theory (\cref{sec:training_combiner}).
Furthermore, I discussed how they can be implemented in practice, such as the posterior refinement procedure (\cref{sec:combiner_posterior_refinement})
and also how this process helps alleviate the bias caused by approximate weights encoded using step-limited A* sampling.
\par
I also proposed several improvements to the framework: 1) linear weight reparametrisations, 2) learned positional encodings and 3) hierarchical coding distributions.
I show via extensive experimental and ablation studies that COMBINER performs as well as, or better than, comparable, state-of-the-art methods and that each proposed improvement is vital to ensure gains.
\section{Where next?}
It is a very exciting time to work on data compression, as machine learning is currently revolutionising the field. It has revitalised dormant parts of the existing theory and inspired new theoretical advances and applications.
Thus, there are a vast number of potential future directions, some of which I already highlighted in \cref{sec:fundamental_limits_conclusion,sec:point_processes_conclusion,sec:bnb_conclusions,sec:combiner_conclusions}.
Therefore, here I take the opportunity to take a step back and discuss higher-level future directions and alternative approaches for which there was no space in this thesis.
\subsection{Dithered Quantisation}
\par
In this thesis, I took the ``general-to-specific'' approach towards solving practical relative entropy coding: I first developed the most general theory possible, then progressively refined it by adding more structure that the sampling algorithms can exploit to perform well.
However, we could consider an alternative, ``specific-to-general'' approach, based on \textit{dithered quantisation} (DQ).
What I mean by this is that DQ starts from immediately assuming a lot of structure: it exploits the vector space structure over $\Reals^d$.
In particular, it works by randomly shifting lattice structures; in its most basic form, it relies on the identity $\lfloor c + U \rceil - U \sim c + U'$ for $c \in \Reals$ and $U, U' \sim \Unif(-1/2, 1/2)$.
\par
This can be extended to layered dithered quantisation \citep{agustsson2020universally,hegazy2022randomized}, which can be used to encode symmetric, unimodal distributions by also randomly scaling the grid.
As the canonical non-trivial example, if we let $S^2 \sim \Gamma(3/2, 1/2)$, then it holds that $S \cdot (\lfloor c / S + U \rceil - U) \sim \Normal(c, 1)$.
Unfortunately, the more general, randomly scaled variants struggle with the same coding inefficiency as applying the branch-and-bound sampler dimensionwise of the samplers I develop (\cref{sec:general_branch_and_bound_variants}): they have a constant $\Oh(1)$ codelength overhead per dimension.
\par
However, among all contemporary approaches, I believe DQ has the highest potential to serve as the basis for a fast, efficient, multivariate relative entropy coding algorithm. 
Recently, there has been some exciting progress towards this goal, if only approximate \citep{kobus2024gaussian}; this work could serve as a starting point for an efficient exact algorithm.
\subsection{The Next Frontier: Markov Chain Monte Carlo}
Markov chain Monte Carlo (MCMC) is the cornerstone of modern statistical inference across the sciences. Yet, it is notably missing not just from this thesis but from the wider relative entropy coding literature at the moment.
Thus, two interesting questions arise: 1) Could we construct relative entropy coding algorithms using Markov chains? 2) Could we use ideas from relative entropy coding to improve and analyse existing MCMC algorithms and construct new ones?
\par
For the first direction, a useful starting point could be perfect sampling, such as coupling from the past \citep{wilson1999layered} and Fill's algorithm \citep{fill1997interruptible}, as these provide a bridge between Markov chain Monte Carlo and exact sampling algorithms.
A more concrete smoking gun is that a central idea of perfect samplers is to consider Markov chains from a coupling perspective.
For a Markov transition kernel $P_{\rvy \mid \rvx}$, perfect samplers require that we specify a ``reparameterisable coupling'' $\rvy = \phi(\rvx, \epsilon)$, where $\rvx$ is the current state of the Markov chain, $\rvy$ is its next state, $\epsilon$ is some random variable independent of $\rvx$ and $\phi$ is a deterministic function.
A simple example is a Gaussian random walk, where $\rvx \sim \Reals$ and $\epsilon \sim \Normal(0, 1)$ with $\phi(\rvx, \epsilon) = \rvx + \epsilon$.
The interesting observation to make then, is that we can ``reparameterise'' a Markov chain $(X_1, X_2, \hdots)$ that admits a coupling $\phi$ using a sequence of i.i.d.\ ``reparameterisation variables'' $(\epsilon_2, \epsilon_3, \hdots)$:
\begin{align}
\label{eq:mcmc_reparam}
 X_1,\, X_2,\, X_3 \hdots \quad\sim \quad X_1,\, \phi(X_1, \epsilon_2),\, \phi(\phi(X_1, \epsilon_3), \epsilon_2), \hdots
\end{align}
Here, the sequence of $\epsilon_i$s seems a good candidate to be set as the common randomness in some relative entropy coding algorithm. 
As an example, we could turn the random walk Metropolis-Hastings algorithm into a channel simulation protocol using the reparametrisation in \cref{eq:mcmc_reparam}: the sender could encode how many state transitions they observed, and how many samples they proposed before each transition occurred.
Such an algorithm is essentially a generalisation of the rejection sampling-based channel simulation protocol I described in \cref{sec:rec_through_examples}.
However, there are major issues that need investigation in this domain: 
When should such a Markov chain-based algorithm stop, and what is the quality of the returned sample? Could we create an algorithm that returns an exact sample? The theory of perfect sampling could help answer these questions.
\par
For the reverse direction, I believe two potential high-level ideas could be fruitful.
First, constructing more efficient Markov chain Monte Carlo algorithms by modifying the Metropolis-Hastings (MH) acceptance rule.
Concretely, an interesting observation is that constructing Markov chains using MH can be viewed as a form of rejection sampling \citep{robert2021rao}.
Hence, it would be interesting to investigate if and when this rejection sampler could be replaced by either A* sampling or GPRS to yield a more efficient scheme.
Second, it would also be interesting to see if some of the ideas I develop for selection sampling, especially for the approximation quality of step-limited selection samplers in \cref{sec:approximate_sampling} could be ported over and be used to analyse the convergence properties of Markov chain methods.
\subsection{Industry-grade Relative Entropy Coding}
\par
While this thesis presented the important first steps towards creating practically applicable relative entropy coding algorithms, there is still a significant gap before any of these algorithms can be applied in practice.
Perhaps the most pressing issue is an engineering solution that ensures the synchronisation of common randomness between the communicating parties: this is an element of the theory that is not needed in traditional entropy coding methods.
Another issue is the security of relative entropy coding.
In practice, we usually implement relative entropy coding algorithms using a shared PRNG seed.
This raises an interesting question: how much should we concern ourselves with this initial piece of information?
For example, it would be interesting to consider whether we could use secure ways of transmitting the seed, e.g.\ by using a secure key exchange protocol, such as the Diffie-Hellman-Merkle protocol \citep{diffie1976new,merkle1978secure}.
\par
While these are considerable challenges, I trust the list of demonstrable benefits of using relative entropy coding for neural data compression and beyond will only accumulate.
Thus, I hope that in the coming decades, relative entropy coding will become the basis of many practical compression algorithms and, as such, a core part of our future global communication ecosystem.
%
% \begin{itemize}
% \item Practical implementation of channel simulation algorithms
% \begin{itemize}
% \item use integer/non-infinite precision arithmetic
% \item issues in neural compression \citep{balle2018integer}
% \end{itemize}
% \end{itemize}

% \textbf{TODOS}
% \begin{itemize}
% \item write the conclusion lol
% \item Illustration for \cref{def:bsp_structure_over_poisson_process}!
% \item run A* coding for \cref{fig:sampler_numerical_experiments}
% \end{itemize}

%% file: Appendix1/appendix1.tex
%!TEX root = ../thesis.tex
% ******************************* Thesis Appendix A ****************************
\chapter{Some Technical Details} 
\label{chapter:technical_details}
\par
\section{Bound on the Entropy of a Positive Integer-valued Random Variable}
\label{sec:li_el_gamal_bound_on_pos_int_random_variable}
In this subsection, I replicate the argument from Appendix B, Proposition 4 from \citet{li2018strong} to bound the entropy of any positive integer-valued random variable.
\begin{lemma}
\label{lemma:li_el_gamal_bound_on_pos_int_random_variable}
Let $K$ be a positive integer-valued random variable.
Then,
\begin{align*}
\Ent{K} \leq \Exp[\lb K] + \lb(\Exp[\lb K] + 1) + 1.
\end{align*}
\end{lemma}
\begin{proof}
Let $q(k) = k^{-\lambda} \big/ Z$, with $\lambda = 1 + 1 / \Exp[\lb K]$ and $Z = \sum_{k = 1}^\infty k^{-\lambda}$, i.e., $q$ is a probability mass function.
Then,
\begin{align}
\Ent{K} &\leq -\sum_{k = 1}^\infty \Prob[K = k]\lb q(k) \nonumber\\
&= \lb Z + \lambda \sum_{k = 1}^\infty \Prob[K = k]\lb k \nonumber\\
&\leq \Exp[\lb K] + \lb\left(\Exp[\lb K] + 1\right) + 1, \label{eq:pos_int_rv_ent_bound_ineq}
\end{align}
where \cref{eq:pos_int_rv_ent_bound_ineq} follows from the fact that
\begin{align*}
Z = \sum_{k = 1}^\infty k^{-\lambda} \leq 1 + \int_1^\infty x^{-\lambda} \, dx = 1 + \frac{1}{\lambda - 1}.
\end{align*}
\end{proof}
\section{Upper Bound on the Expectation of the Logarithm of a Geometric Random Variable}
\label{sec:geom_log_expectation_bound}
\begin{lemma}
Let $K \sim \Geom(p)$ be a geometric random variable with success probability $p$.
Then, we have
\begin{align*}
\Exp[\lb K] \leq e^p \lb\left(1 + \frac{1}{p}\right).
\end{align*}
\end{lemma}
\begin{proof}
\begin{align}
\Exp[\ln K] &= \sum_{k = 1}^{\infty} (1 - p)^{k - 1} p \ln k \nonumber\\
&\leq \int_0^\infty p e^{-p(x - 1)} \ln(x + 1) \, dx \label{eq:geom_bound_int_upper_bound} \\
&= e^p \int_0^\infty \frac{e^{-p x}}{1 + x}  \, dx  \label{eq:geom_bound_int_by_parts}\\
&= e^{2p} \cdot E_1(p) \label{eq:geom_bound_substitution} \\
&\leq e^p \ln \left(1 + \frac{1}{p}\right) \label{eq:geom_bound_exp_int_upper_bound},
\end{align}
where \cref{eq:geom_bound_int_upper_bound} follows since $(1 - p)^k \leq e^{-px}$ and $\ln k \leq \ln(x + 1)$ for $x \in (k - 1, k]$ for every $k \geq 0$; \cref{eq:geom_bound_int_upper_bound} follows from integration by parts; \cref{eq:geom_bound_substitution} follows from substituting the variable of intergration and recognising that the integral is the exponential integral $E_1$ \citep[eq. 5.1.1;][]{abramowitz1968handbook} and \cref{eq:geom_bound_exp_int_upper_bound} follows from the inequality $e^x E_1(x) \leq \ln(1 + 1/x)$ for $x > 0$ \citep[eq. 5.1.20;][]{abramowitz1968handbook}.
Dividing by $\ln(2)$/multiplying by $\lb(e)$ yields the desired result.
\end{proof}
\section{Bound on the Differential Entropy of a Positive Random Variable with Density Bounded by 1.}
\label{sec:bound_on_pos_bounded_rv}
In this section, I perform a similar analysis to the one in \cref{sec:li_el_gamal_bound_on_pos_int_random_variable}, except I bound the differential entropy of a certain positive random variable.
Specifically:
\boundOnPosBoundedRV*
\begin{proof}
Let $q(\rh) = \min\{1, \rh^{-\lambda}\} \big/ Z$, with $\lambda = 1 + 1 / (\Exp[\lb \rh] + \lb(e))$ and 
\begin{align}
Z &= \int_0^\infty\min\{1, \rh^{-\lambda}\}  \,d\rh  \nonumber\\
&= 1 + \int_1^\infty \rh^{-\lambda} \, d\rh \nonumber\\
&= 1 + \frac{1}{\lambda - 1} \label{eq:bound_on_pos_rv_ent_int_cond}\\
&= 1 + \lb(e) + \Exp[\lb (\rh)], \nonumber
\end{align}
where in \cref{eq:bound_on_pos_rv_ent_int_cond} we used that $-\lb(e) < \Exp[\lb (\rh)] \Rightarrow 1 < \lambda$, since otherwise the integral does not converge.
Then,
\begin{align}
\DiffEnt{\rh} &\leq -\int_{0}^\infty p_\rh(\eta)\lb q(\eta) \, d\eta \nonumber\\
&= \lb Z - \int_0^\infty p_\rh(\eta)\lb (\min\{1, \eta^{-\lambda}\}) \, d\eta \nonumber\\
&= \lb (\Exp[\lb (\rh)] + 2) + \lambda \int_1^\infty p_\rh(\eta)\lb (\eta) \, d\eta \nonumber\\
&\leq \lb (\Exp[\lb (\rh)] + 1 + \lb(e)) + \lambda (\Exp[\lb (\rh)] + \lb(e)) \label{eq:bound_on_pos_bounded_rv_density_ineq} \\
&= \Exp[\lb (\rh)] + \lb (\Exp[\lb (\rh)] + 1 + \lb(e)) + 1 + \lb(e) \nonumber
\end{align}
where \cref{eq:bound_on_pos_bounded_rv_density_ineq} follows from the fact that
\begin{align*}
\int_1^\infty p_\rh(\eta)\lb (\eta) \, d\eta &= \Exp[\lb( \rh)] + \int_0^1 p_\rh(\eta)\lb\left(\frac{1}{\eta}\right)\, d\eta \\
&\!\!\!\!\!\!\!\!\!\stackrel{\norm{p_\rh}_\infty \leq 1}{\leq} \Exp[\lb (\rh)] + \int_0^1 \lb\left(\frac{1}{\eta}\right)\, d\eta \\
&= \Exp[\lb (\rh)] + \lb(e).
\end{align*}
\end{proof}

\section{Rejection sampling with sorted uniforms}
\label{sec:rej_samp_with_sorted_unifs}
\par
This section describes how we can turn rejection sampling (\Cref{alg:global_rs}) into an \textit{almost} maximally efficient relative entropy coding algorithm.
To begin, let us recap the basic ingredients of rejection sampling.
Thus, let $Q$ be a target distribution from which we would like to simulate a sample and $P$ a proposal distribution with $Q \ll P$ and with density ratio $r = dQ/dP$ and upper bound $M \geq \norm{r}_\infty$.
Given a sequence of i.i.d.\ pairs $\{(Y_n, U_n)\}_{n = 1}^\infty$ where $Y_n \sim P$ and $U_n \sim \Unif(0, 1)$, rejection sampling returns the index
\begin{align*}
N = \min\{n \in \posNats \mid U_n < r(Y_n) / M\}
\end{align*}
I have shown in \Cref{sec:rec_through_examples} that $N \sim \Geom(1/M)$, from which $\Exp[\lb N] \leq \Exp[\lb M] + 2$, which means that using the Zeta coding method of \citet{li2018strong} will not be maximally efficient here.
\par
However, we can use a simple trick to significantly improve this scheme: we can \textit{sort the uniforms}!
Concretely, for the improved scheme, we set the common randomness to consist of the pairs $\rvz = \{(Y_n, U_n)\}_{n = 1}^\infty$ we used for rejection sampling.
Now, once the sender computed $N$, they also compute $L = \lceil \lb N \rceil$ and $N' = \exp_2(L) \geq N$.
Then, they sort the first $N'$ uniforms in the common randomness to obtain the sorted list $(U_{\sort(1)}, U_{\sort_{N'}(2)}, \dots, U_{\sort_{N'}(N')})$.
Here, $\sort_{N'}(n) \in [1:N']$ and it denotes the bijection (permutation) that maps the indices of the uniforms in the sorted list to the indices in the original, unsorted list.
For example, $\sort_{N'}(1) = \argmin_{n \in [1:N']} U_n$ and $\sort_{N'}(N') = \argmax_{n \in [1:N']} U_n$.
\par
With this notation in place, the improved coding scheme is extremely simple: we encode $\sort_{N'}^{-1}(N)$, that is, the position of $U_N$ in the sorted list instead of $N$ directly!
Concretely, when applying the above procecure for the pair of random variables $\rvx, \rvy \sim P_{\rvx, \rvy}$, where $P_\rvx$ is the source and we are interested in encoding samples $\rvy \sim P_{\rvy \mid \rvx}$, we use \textit{almost} same Zeta distribution that we would use for A* sampling in \Cref{thm:global_a_star_codelength}.
Namely, we set the power parameter as $\alpha = 1 + 1 / (\MI{\rvx}{\rvy} + \lb(5/2))$.
\par
What is the rate of this scheme?
Well, there is a good reason for why I chose $\alpha$ the similar to A* sampling: it turns out that there is a beautiful connection between the arrival times of A* sampling and the sorted uniforms above.
The connection is most elegant when we consider only sorting the first $N$ uniforms: in this case, Lemma 13 of \citet{maddison2016poisson} shows that for all $n \in [1:N]$, we have
\begin{align*}
U_{\sort_N(n)} \sim \frac{T_{n}}{T_{K_{A^*} + 1}}
\end{align*}
where $T_n$ is the $n$th arrival time of the spatio-temporal Poisson process $\PoissonProcess$ with mean measure $P_\rvy \otimes \lambda$ that A* sampling would simulate when applied to sampling the same target distribution as our rejection sampler, and $K_{A^*}$ denotes the A* sampling's runtime.
This demonstartes a beutiful connection between rejection and A* sampling: A* sampling can be thought of as realising the order statistics of rejection sampling directly, so another name for it could as well be ``ordered rejection sampling''!
The upshot of this is that we therefore have
\begin{align}
\label{eq:rs_a_star_index_dist_eq}
\sort_N^{-1}(N) \sim N_{A^*}    
\end{align}
where $N_{A^*}$ is the index A* sampling would select, and hence we can directly apply \Cref{thm:global_a_star_codelength} for encoding $\sort_N^{-1}(N)$!
\par
However, remember that we only wish to encode $L = \lceil \lb N\rceil$ (or equivalently, $N' = \expb(L)$), since we already saw that specifying $N$ directly is inefficient.
However, note that in general we will have $J = \sort_{N}^{-1}(N) \leq \sort_{N'}^{-1}(N) = J'$, so we need to ensure that encoding $J'$ does not cost much more than encoding $J$; fortunately this is indeed the case.
To see this, note that the reason why $J'$ is at least as large as $J$ is because it might easily happen that of the $N' - N$ uniforms in the common randomness that come \textit{after} $U_N$, some of them will be less than $U_N$, and hence the index of $U_N$ in the sorted sequence increases.
Now, note that given $N$ and $U_N$, we have $U_i \perp U_N$ for $i \in [N + 1:N']$, since $N$ is a stopping time adapted to the sequence. 
Thus, the probability that any of the $U_i$s is less than $U_N$ is $\Prob[U_i < U_N \mid N, U_N] = U_N$.
Since there are $N' - N$ uniforms in the sequence after $U_N$, on average we will have $\Exp[J' - J\mid N, U_N, Y_N] = (N' - N) U_N \leq N \cdot U_N$ of them be less than $U_N$.
Now, note that $U_N \mid Y_N \sim \Unif(0, r(Y_N) / M)$, hence $\Exp[J' - J\mid N, Y_N] \leq N \cdot \Exp[U_N \mid Y_N, N] = N \cdot r(Y_N) \big/ 2M$.
Finally, we get
\begin{align}
\Exp[\lb J'] 
&= \Exp[\lb(J + (J' - J))] \nonumber\\
&\leq \Exp_{Y_N}[\lb(\Exp[J] + \Exp_{N \mid Y_N}[J' - J])] \tag{Jensen}\\
&\leq \Exp_{Y_N}[\lb(\Exp[J] + r(Y_N) \cdot \Exp_{N \mid Y_N}[N] \big/ 2M)] \nonumber\\
&\leq \Exp_{Y_N}[\lb(r(Y_N) + 1 + r(Y_N) \cdot \Exp_{N \mid Y_N}[N] \big/ 2M)] \tag{via \cref{eq:rs_a_star_index_dist_eq} \& proof of \cref{thm:global_a_star_codelength}} \\
&\leq \Exp_{Y_N}[\lb r(Y_N)] + \Exp_{Y_N}[\lb(1 + 1 / r(Y_N) + \Exp_{N \mid Y_N}[N] \big/ 2M)] \nonumber\\
&\leq \KLD{Q}{P} + \lb(1 + \Exp_{Y_N}[1 / r(Y_N) + \Exp_{N \mid Y_N}[N] \big/ 2M]) \tag{Jensen} \\
&=\KLD{Q}{P} + \lb(5/2) \nonumber
\end{align}
To encode $L$, we may use Elias gamma coding, which results in a codelength of $2\lfloor \lb L \rfloor + 1 \leq \lb L + 3$ bits.
\par
Therefore, for a pair of dependent random variables $\rvx, \rvy \sim P_{\rvx, \rvy}$, if we encode a sample returned by rejection sampling using $L$ and the sorted index $J' = \sort_{N'}^{-1}(N)$, then for $L$ we achieve a rate of at most
\begin{align*}
2\Exp[\lb L] + 3 = 2\Exp[\lb \lceil \lb N\rceil] + 3 \leq 2\Exp[\lb (\lb N + 1)] + 3 \leq 2\Exp[\lb (\lb \norm{r}_\infty + 1)] + 3 \text{ bits}
\end{align*}
and using the Zeta distribution with $\alpha = 1 + 1 / (\MI{\rvx}{\rvy} + \lb(5/2))$ to encode $J' \mid L$, we achieve a rate of at most
\begin{align*}
\MI{\rvx}{\rvy} + \lb(\MI{\rvx}{\rvy} + 1) + \lb(7/2) + 7/2 \text{ bits}
\end{align*}
Thus, the total coding rate is at most
\begin{align*}
 \MI{\rvx}{\rvy} + &\lb(\MI{\rvx}{\rvy} + 1) + 2\Exp[\lb (\lb \norm{r}_\infty + 1)] + \lb(7/2) + 7/2 + 3 \\
 &\approx \MI{\rvx}{\rvy} + \lb(\MI{\rvx}{\rvy} + 1) + 2\Exp[\lb (\lb \norm{r}_\infty + 1)] + 8.31 \text{ bits}.
\end{align*}